\documentclass[12pt,mycrest]{ociamthesis}

\usepackage[T1]{fontenc}

\usepackage{amsmath}
\usepackage{amssymb}
\usepackage{amsthm}
\usepackage{amsfonts}
\usepackage{amscd}
\usepackage{array}
\usepackage{blindtext}
\usepackage{booktabs}
\usepackage{xcolor}
\usepackage{enumerate}
\usepackage{enumitem}
\usepackage{float}
\usepackage{mdframed}
\usepackage{graphicx}
\usepackage{latexsym}
\usepackage{lmodern}
\usepackage{mathrsfs}
\usepackage{makeidx}
\usepackage{dsfont}
\usepackage{multirow}
\usepackage{tensor}
\usepackage{slashed}
\usepackage{url}
\usepackage{xspace}
\usepackage{tikz-cd}
\usepackage{simplewick}
\usepackage{cite}
\usepackage{setspace}

\definecolor{oxblue}{rgb}{0,0.1294,0.2784}
\definecolor{oxbluel}{rgb}{0.26667,0.4078,0.4902}
\definecolor{pantone279}{rgb}{0.28235,0.5686,0.8627}
\definecolor{pantone562}{rgb}{0,0.46667,0.4392}

\theoremstyle{definition}
\newtheorem{thm}{Theorem}
\newtheorem{deff}{Definition}
\newtheorem{prop}{Proposition}
\newtheorem*{proofsketch}{Proof sketch}

\newtheorem{invtool}{Inversion}

\usepackage{chngcntr}
\counterwithin{rmk}{chapter}
\counterwithin{thm}{chapter}
\counterwithin{deff}{chapter}
\counterwithin{prop}{chapter}
\counterwithin{invtool}{chapter}

\usepackage[pdftex,
            pdfauthor={Johan Henriksson},
            pdftitle={Analytic Bootstrap for Perturbative Conformal Field Theories},breaklinks=true,colorlinks=true,citecolor={oxblue},urlcolor={oxblue},linkcolor=black]{hyperref} 

\usepackage[all]{hypcap}

\AtBeginDocument{ 

\newcommand{\R}{\mathbb R}
\newcommand{\C}{\mathbb C}
\newcommand{\Z}{\mathbb Z}
\renewcommand{\O}{\mathcal O}

\newcommand{\1}{\mathds1}

\newcommand{\SL}[2]{\mathrm{SL}(#1,#2)}
\newcommand{\SO}[1]{\mathrm{SO}(#1)}\newcommand{\SU}[1]{\mathrm{SU}(#1)}
\newcommand{\OO}[1]{\mathrm{O}(#1)}
\newcommand{\UU}[1]{\mathrm{U}(#1)}

\newcommand{\ten}{\tensor}

\DeclareMathOperator*{\Tr}{Tr}

\DeclareMathOperator*{\res}{res}

\newcommand{\df}{\mathrm d} 
\newcommand{\de}{\partial} 

\newcommand{\pp}[1]{\frac{\partial}{\partial#1}}

\newcommand{\munu}{_{\mu\nu}}

\newcommand{\parr}[1]{\left(#1\right)}
\newcommand{\parrk}[1]{\left[#1\right]}
\newcommand{\parrm}[1]{\left\{#1\right\}}
\newcommand{\ket}[1]{\left|#1\right\rangle}

\newcommand{\braccket}[3]{\left\langle#1\middle|#2\middle| #3\right\rangle}

\newcommand{\expv}[1]{\left\langle#1\right\rangle}

\newcommand{\beqa}[1]{\begin{align}#1\end{align}}
\newcommand{\beq}[1]{\begin{equation}#1\end{equation}}

\newcommand{\hb}{{\bar h}}
\newcommand{\NN}[1]{$\mathcal N=#1$}
\newcommand{\zb}{{\bar z}}

\newcommand{\G}{{\mathcal G}}
\newcommand{\alphab}{{\bar\alpha}}

\renewcommand{\AA}{\mathbf A}
\renewcommand{\SS}{\mathbf S_{\boldsymbol 1}}

\renewenvironment{proof}[1][\proofname]{\noindent {\bfseries #1. }}{}

\hyphenation{sing-let}
\hyphenation{sing-lets}

\DeclareMathOperator*{\INV}{INV}

\DeclareMathOperator*{\dsat}{\mathcal{D}_{\mathrm{sat}}}
\DeclareMathOperator*{\cas}{\mathcal{C}}

\newcommand{\Dbar}{{\overline{D}}}

\newcommand{\xit}{{\xi}}

\newcommand{\oprp}{{\widetilde{\mathcal R}}}
\newcommand{\rp}{{\widetilde R}}
\newcommand{\veps}{\epsilon}

} 
\usepackage{layouts}
\newcommand{\dDisc}{\mathrm{dDisc}}
\title{Analytic Bootstrap     
 for   \\[1ex]     {\hspace{-18mm}}Perturbative Conformal Field Theories{\hspace{-18mm}} }

\author{\bf Johan Henriksson}
\college{\bf Balliol College}

\degree{Doctor of Philosophy}
\degreedate{Hilary term MMXX}

\begin{document}


\baselineskip=24pt plus1pt

\setcounter{secnumdepth}{3}
\setcounter{tocdepth}{3}

\maketitle

\begin{romanpages}

\begin{abstractlong}

\noindent \textsf{\textbf{Abstract:}} Conformal field theories (CFTs) play a central role in theoretical physics with many applications ranging from condensed matter to string theory. The conformal bootstrap studies conformal field theories using mathematical consistency conditions and has seen great progress over the last decade. In this thesis we present an implementation of analytic bootstrap methods for perturbative conformal field theories in dimensions greater than two, which we achieve by combining large spin perturbation theory with the Lorentzian inversion formula. In the presence of a small expansion parameter, not necessarily the coupling constant, we develop this into a systematic framework, applicable to a wide range of theories.

The first two chapters provide the necessary background and a review of the analytic bootstrap. This is followed by a chapter which describes the method in detail, taking the form of a practical guide to large spin perturbation theory by means of a step-by-step implementation. The goal is to compute the CFT-data that define a given conformal field theory, and this is achieved by considering contributions from operators in a four-point correlator through the crossing equation. We give a general recipe for determining which operators to consider, how to find their contributions from conformal blocks and how to compute the corresponding CFT-data through the inversion formula.

The second part of the thesis presents several explicit implementations of the framework, taking examples from a number of well-studied conformal field theories. We show how many literature results can be reproduced from a purely bootstrap perspective and how a variety of new results can be derived. We consider in depth how to determine the CFT-data in the $\epsilon$ expansion for the Wilson--Fisher model from crossed-channel operators. All CFT-data to order $\epsilon^3$ follow from only the identity and the bilinear scalar operator, and by considering contributions from two infinite families of operators we generate new results at order $\epsilon^4$. We study in similar depth conformal gauge theories in four dimensions, where we find a five-parameter solution for the most general form of the one-loop four-point correlator of bilinear scalars. For particular parameter values this reproduces the case of the Konishi operator and the stress tensor multiplet in weakly coupled \NN4 super Yang--Mills theory. We then present more briefly four additional examples. These include the critical $\OO N$ model in a large $N$ expansion, a solution for $\phi^4$ theory with any global symmetry, multicritical theories to order $\epsilon^2$ near their critical dimensions, including new results for the central charge, and the four-point correlator of bilinear scalars in the $\epsilon$ expansion. We conclude the thesis with a discussion and some appendices.

\end{abstractlong}

\begin{originalitylong}

\noindent This thesis is based on the results of the following papers, all to which the author contributed substantially:

\vspace{3ex}
{
\renewcommand{\arraystretch}{1.6}
\onehalfspacing
\noindent \begin{tabular}{rp{14cm}}
\cite{Paper1} & 
Johan Henriksson and Tomasz \L ukowski, \emph{Perturbative four-point functions from the analytic conformal bootstrap}, \href{https://doi.org/10.1007/JHEP02(2018)123}{\emph{JHEP} \textbf{02} (2018) 123}, [\href{https://arxiv.org/abs/1710.06242}{\texttt{1710.06242}}].
\\
\cite{Paper2}& 
Luis Fernando Alday, Johan Henriksson and Mark van Loon, \emph{Taming the $\epsilon$-expansion with large spin perturbation theory}, \href{https://doi.org/10.1007/JHEP07(2018)131}{\emph{JHEP} \textbf{07} (2018) 131}, [\href{https://arxiv.org/abs/1712.02314}{\texttt{1712.02314}}].
\\
\cite{Paper3}& 
Johan Henriksson and Mark van Loon, \emph{Critical $ O(N)$ model to order $\epsilon^4$ from analytic bootstrap}, \href{https://doi.org/10.1088/1751-8121/aaf1e2}{\emph{J. Phys.} \textbf{A52} (2019) 025401}, [\href{https://arxiv.org/abs/1801.03512}{\texttt{1801.03512}}].
\\
\cite{Paper4}& 
Luis Fernando Alday, Johan Henriksson and Mark van Loon, \emph{An alternative to diagrams for the critical $ O(N)$ model: dimensions and structure constants to order $1/N^2$}, \href{http://dx.doi.org/10.1007/JHEP01(2020)063}{\emph{JHEP}
  {\textbf{ 01}} (2020) 063}, [\href{https://arxiv.org/abs/1907.02445}{\texttt{1907.02445}}]
\\
\cite{Paper5}&Johan Henriksson, Stefanos Kousvos and Andreas Stergiou, \emph{{Analytic and \mbox{numerical} bootstrap of CFTs with $O(m)\times O(n)$ global symmetry in $3$D}},  {\it
  Arxiv preprint} (2020), [\href{https://arxiv.org/abs/2004.14388}{{\tt
  2004.14388}}].
\end{tabular}
}
 \phantom{
 }
 \vspace{3ex}

\noindent Chapters~\ref{ch:paper2} and \ref{ch:paper1} are modified versions of \cite{Paper2} and \cite{Paper1} respectively. Sections~\ref{sec:ONspectrum} and~\ref{sec:paper4short} summarise \cite{Paper4}, section~\ref{sec:paper5short} is a modified version of part of \cite{Paper5} and finally section~\ref{sec:ONe4} and appendix~\ref{app:subcollinearblocks} are based on material in \cite{Paper3}. Chapters \ref{ch:intro}--\ref{ch:practical} constitute an extended review and sections~\ref{sec:multicritical} and~\ref{sec:WFfrompaper1} contain previously unpublished results.
Versions of \cite{Paper2,Paper3,Paper4} have appeared in a previous DPhil thesis from the University of Oxford \cite{MarkThesis}.

\end{originalitylong}

\begin{acknowledgementslong}
{
\baselineskip=23pt plus1pt

\noindent First and foremost, I would like to thank my supervisor Fernando Alday for his great support throughout my DPhil. Without his never-ending enthusiasm and excellent suggestions of research topics this thesis would not have been possible. Apart from our direct collaboration I have received a lot of help from his useful advice and through his impressive skills with Mathematica. I also thank Tomasz {\L}ukowski for direct guidance and collaboration during my first year and for many useful discussions. Moreover, I have had the great pleasure to work closely with Mark van Loon, and during the final year with Andreas Stergiou and Stefanos Kousvos.

Throughout my time in Oxford, I have developed a deepened friendship with all my fellow DPhil students in the Mathematical Physics group. I would especially like to thank my academic siblings Mark, Pietro and Carmen, and my office mates Mohamed, Hadleigh and Juan for many interesting discussions. I owe thanks to each of the professors in our research group, as well as to Andre Lukas and Paul Fendley. I also thank my many teaching colleagues and students, and acknowledge the financial support from the Marvin Bower scholarship and from the Clarendon Fund.

I would like to express my eternal gratitude to my parents for their continuous and reliable support, as well as to Ylva and Axel. Thanks also to Mormor, for inspiration and for hosting the 2019 Uvan\r{a} reading camp, and to Farmor and Farfar who sadly passed away during my time in Oxford. Thanks to Elo\"{i}se Hamilton for our almost daily coffee breaks, and to Mohamed Elmi for keeping friends despite sharing office, college and master's degree. My time at this university has been greatly enhanced by the Scandinavian Society, Balliol MCR and Boat Club, the university Orienteering and Triathlon clubs, and all the friends I have made through these. Finally, thanks to my Uppsala friends, in particular to Jakob Jonnerby, who once convinced me to apply to Oxford. It has been a great time.

\vspace{1ex}
\begin{flushright}
Oxford, 7 April 2020
\end{flushright}

}
\end{acknowledgementslong}

\setcounter{tocdepth}{2}
{
\baselineskip=23.5pt plus1pt
\tableofcontents  
}
\listoffigures 
\chapter*{List of abbreviations}
{
\renewcommand{\arraystretch}{1.65}
\hspace{1.75ex} \begin{tabular}{p{1.5cm}p{16cm}}
AdS & anti-de Sitter (spacetime)
\\
BPS & Bogomol'nyi--Prasad--Sommerfield 
\\
CFT & conformal field theory
\\
GFF & generalised free fields
\\
IR & infrared, long-distance, low-energy
\\
irrep & irreducible representation
\\
LSPT & large spin perturbation theory
\\
OPE & operator product expansion
\\
QCD & quantum chromodynamics
\\
QFT & quantum field theory
\\
RG & renormalisation group
\\
SCFT & superconformal field theory
\\
SYM & super Yang--Mills (theory)
\\
UV & ultraviolet, short-distance, high-energy
\\
WF & Wilson--Fisher
\end{tabular}
}

\end{romanpages}

\chapter{Introduction}\label{ch:intro}


The main enterprise of theoretical physics is to construct mathematical models for describing physical phenomena. These models are constructed from a supply of experimental data, and are judged based on their success in explaining previous observations and in particular on their ability to make predictions that can be confirmed by new experiments. This approach, whose enormous success can be exemplified with Maxwell's equations for electromagnetism in the 1860's and the theory of quantum mechanics in the 1920's, has remained successful into present days with the discovery of the Higgs boson in 2012 \cite{Aad:2012tfa,Chatrchyan:2012xdj} and gravitational waves in 2015 \cite{Abbott:2016blz}.

In parallel with the main line of development, a slightly different perspective emerged and gained increasing popularity in the study of fundamental physics. The idea is to identify some fundamental principles, and then explore the implications that follow from imposing mathematical consistency. One example is Dirac's attempt in 1928 to write down a linear equation of motion for the electron quantum field \cite{Dirac:1928hu}. He found that the only way to write a consistent equation was to formulate it in terms of matrices of size at least $4\times 4$. This in turn introduced negative energy solutions interpreted as positrons \cite{Weyl1929}, which were experimentally observed a few years later \cite{Anderson:1933mb}.

Another example is the study of statistics in quantum mechanics. Under spatial rotation by $2\pi$, the wave function picks up a phase $+1$ for bosons and $-1$ for fermions. Famously this corresponds to $\parrm{\pm1}$ being the pre-image of the identity in the universal cover of the rotation group: $\SU2$ over $\SO3$. In two dimensions the universal cover of the rotation group $\SO2$ is non-compact, $\R$, and it was noted in 1977 that this would allow for a new kind of quantum statistics \cite{Leinaas:1977fm}. The corresponding particles were dubbed anyons, and were shown to play a role in the fractional quantum Hall effect \cite{Arovas:1984qr}, discovered in 1982 \cite{Tsui:1982yy}. However, even without the experimental realisation, the discovery of anyons as a consistent theory is interesting on its own, as it is investigating the boundaries for what kind of physics could at all possibly exist.

Instead of thinking about statistics, we may study the implications of spacetime symmetry in a relativistic quantum field theory. It is believed that the maximal extension with bosonic generators of the Poincar\'e group of spacetime symmetries for interacting quantum field theories is the conformal group\footnote{It is clear that the conformal group is an extension of the Poincar\'e group. In \cite{Maldacena:2011jn} it was shown that for three-dimensional theories, the existence of a higher spin current makes the theory free. In part of this thesis we will look at theories which contain infinitely many weakly broken higher spin currents $\mathcal J_\ell$.}. Theories with spacetime symmetries given by the conformal group---the Poincar\'e group extended by scalings and translations of the infinity---are called conformal field theories (CFTs).

In this thesis we are broadly interested in questions like \emph{what possible models for physics are consistent with conformal symmetry?} Again, the case of two dimensions is special, and we will here focus on $d>2$ spacetime dimensions.

Physics with conformal invariance has great importance. Apart from a large number of specific conformal field theories, some of which we will discuss shortly, CFT was given a special role at the heart of quantum field theory (QFT) through Wilsonian renormalisation \cite{Wilson:1973jj,Wilson:1974mb,Wilson:1979qg}. In Wilson's approach, physics at different energies---or equivalently different length scales---are related through the renormalisation group (RG), 
and it has been observed that the scale-invariant fixed-points of the renormalisation group flow in fact happen to be conformal field theories\footnote{In two and four dimensions, it has been shown that for unitary theories scale invariance implies conformal invariance \cite{Polchinski:1987dy,Dymarsky:2013pqa}, but it is not known if this holds in generic dimensions \cite{Nakayama:2013is}.}. 
An important consequence of this fact is that theories with different microscopic descriptions might flow to the same CFT at long distances (IR). The short-distance (UV) theory does not even need to be a quantum field theory, but could for instance be a spin chain, which is a statistical system defined on a lattice. An example is the Ising spin chain \cite{Ising1925} which consists of a spin chain in $d$ dimensions whose Hamiltonian contains a nearest neighbour interaction and a coupling to an external magnetic field. At zero magnetic field, the system undergoes a second-order phase transition between an ordered, low-temperature phase and a disordered, high-temperature phase. At the transition, the system becomes scale-invariant and is described by a CFT: the \emph{Ising model CFT}\footnote{In the following, we will refer to the CFT as just the Ising model, and use the phrase Ising spin chain to describe the statistical system.}. 
In two dimensions, the Ising model was solved exactly \cite{Onsager:1943jn}, but, interestingly, there is no exact solution to the 3d Ising model to this date.

Another way to reach the Ising model is to start from a Lagrangian quantum field theory containing a single real scalar $\phi$ with $\phi^4$ interaction. In the context of statistical physics this is said to give a Landau--Ginzburg description of the Ising model. In three dimensions, for instance, this results in a ``long RG flow'' as depicted in figure~\ref{fig:operahouse}, where the spectrum of the 3d Ising model differs substantially from that of the free theory where the flow started. This viewpoint was systematically developed by the introduction of the $\epsilon$ expansion by Wilson and Fisher \cite{Wilson:1971dc}. They considered the RG flow between the free theory and the interacting theory (Ising model) in $d=4-\epsilon$ dimensions. For small $\epsilon$, both fixed-points can be described as a perturbation from the free $4d$ theory, illustrated by the ``short RG flow'' in figure~\ref{fig:operahouse}. In practice, quantities of interest are computed by Feynman diagrams and are subsequently evaluated at the point of vanishing beta function, called the Wilson--Fisher (WF) fixed-point. The results computed through the $\epsilon$ expansion are in general given by asymptotic series in $\epsilon$, but the evaluation of suitably truncated series gives good predictions also at finite $\epsilon$, for instance at $\epsilon=1$ corresponding to three dimensions. In chapter~\ref{ch:paper2} we will study the $\epsilon$ expansion from a CFT point of view, without referring to Feynman diagrams. Generalisations of the Wilson--Fisher fixed-point, called multicritical models, were soon found and can be described by $\phi^{2\theta}$ interactions near appropriate critical dimensions $d_{\mathrm c}(\theta)$. Each such theory requires tuning $\theta-1$ different relevant couplings. We indicate the multicritical theories in figure~\ref{fig:operahouse}.
\begin{figure}
  \centering
\includegraphics{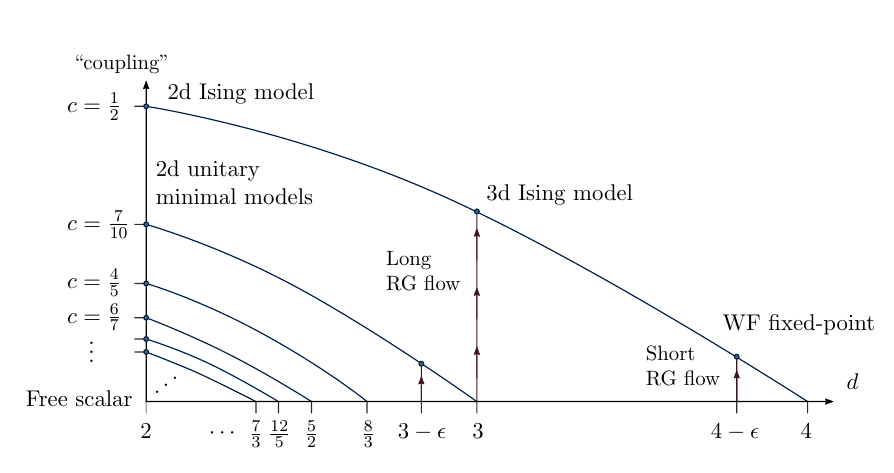}
\caption[The Wilson--Fisher fixed-point in $d=4-\epsilon$, and multicritical fixed-points in $d=d_{\mathrm c}(\theta)-\epsilon$.]{The Wilson--Fisher fixed-point in $d=4-\epsilon$, and multicritical fixed-points in $d=d_{\mathrm c}(\theta)-\epsilon$ with $d_{\mathrm c}(\theta)=\frac{2\theta}{\theta-1}$. Beyond the small $\epsilon$ limit the theories are reached from free theory by a long RG-flow tuning $\theta$ parameters. The multicritical fixed-points are connected to the unitary minimal models with (holomorphic) central charge $c=1-\frac6{(\theta+1)(\theta+2)}$.}
\label{fig:operahouse}
\end{figure}

The fact that several systems with different microscopic descriptions exhibit the same long-distance physics is referred to as universality. We say that systems with the same IR behaviour belongs to a common universality class. Typically, universality classes are characterised by the global symmetry group and the number of relevant singlet scalar operators. For instance, the Ising universality class with $\Z_2$ global symmetry contains, besides the Ising spin chain, some magnetic systems and some van der Waals gases---such as water---near the critical point of their phase diagrams. Universality classes with other global symmetry groups are common in second order phase transitions of certain materials, including structural phase transitions, and they also describe quantum critical phase transitions. 
\\

\noindent From the beginning, the renormalisation group played an important role also in high-energy physics, through the analysis of the strong force in deep inelastic scattering experiments in terms of asymptotic freedom. High-energy theorists started to develop conformal field theory as an independent subject. In an important paper by Ferrara, Gatto and Grillo in 1973 \cite{Ferrara1973}, representations of the conformal group were discussed and the operator product expansion (OPE) was analysed further in the CFT context, where it has a finite radius of convergence. Soon thereafter, Polyakov \cite{Polyakov1974} studied the implications of conformal symmetry on four-point correlators (as we will discuss shortly, two- and three-point correlators are completely fixed up to some theory-dependent constants called the CFT-data). The idea was to avoid any Lagrangian description of the theory and instead use the crossing equation to generate non-trivial equations for the CFT-data. 
The specific implementation of Polyakov, using OPE consistency for crossing-symmetric expressions, was recently revived using Mellin amplitudes to create the conformal bootstrap in Mellin space \cite{Gopakumar:2016wkt,Gopakumar:2018xqi}.

The advent of string theory directed interest towards two-dimensional CFTs, and a new version of the bootstrap appeared. In 1984 Belavin, Polyakov and A. Zamolodchikov studied the crossing equation for two-dimensional conformal field theories, with a particular focus on theories with central charge $c<1$ \cite{Belavin1984}. This was very successful and led to a complete classification of such theories, denoted \emph{minimal models}. The minimal models that satisfy unitarity, $0<c<1$, can be enumerated by an integer $\theta$, conjecturally connected to the Ising model and the multicritical theories as displayed in figure~\ref{fig:operahouse}. 

Let us explain the key ideas of the bootstrap programme in a bit more detail.
Unlike in conventional field theory, in this approach it proves useful to focus on operators rather than fields. Furthermore, the transformation properties of the correlators of these operators can be taken as axioms for the CFT. A CFT contains a distinguished set of operators called conformal primaries and the main observables are correlators of these primary operators. The OPE between two operators is convergent away from other operator insertions, which implies that we can reduce any $n$-point function to a sum over $(n-1)$-point functions,
\beq{\label{eq:OPEintro}
\big\langle
\contraction{}{\O_1}{(x)}{\O_2}
 \O_1(x)\O_2(0)
 \O_3(x_3)\cdots \O_n(x_n)\big\rangle
 =
 \sum _k c_{12k}\big\langle C(x,\de)\O_k(0)\O_3(x_3)\cdots \O_n(x_n)\big\rangle,
}
where $\de_\mu=\pp{x^\mu}$. The coefficients $c_{ijk}$ are theory-dependent OPE coefficients and $C(x,\de)$ are theory-independent functions depending only on the scaling dimensions and spins of the involved operators. Ultimately, any correlator can be reduced to a sum of two- or three-point functions, which are given in terms of the OPE coefficients and scaling dimensions in the theory, collectively referred to as the CFT-data.

From applying the OPE in two different ways within the four-point function, one can extract the \emph{crossing equation},
\beq{\label{eq:crossingintro}
\big\langle
\contraction[1.2ex]{\O(x_1}{)\vphantom{\rule{1pt}{2.7ex}}}{\O(x_2)\O(x_3}{)\vphantom{j}}
\contraction{}{\O(x_1}{)\vphantom{\rule{1pt}{2.7ex}}}{\O(x_2}
\contraction{\O(x_1)\vphantom{\rule{1pt}{2.7ex}}\O(x_2)}{\O(x_3}{)\vphantom{j}}{\O(x_4}
\O(x_1)\vphantom{\rule{1pt}{2.7ex}}\O(x_2)\O(x_3)\vphantom{j}\O(x_4)\big\rangle
=
\big\langle
\contraction[2.2ex]{}{\O(x_1}{)\O(x_2\vphantom{\rule{1pt}{2.7ex}})\O(x_3)}{\O(x_4}
\contraction{\O(x_1)}{\O(x_2}{\vphantom{\rule{1pt}{2.7ex}})}{\O(x_3}
\contraction[1.2ex]{\O(x_1)\O(x_2}{\vphantom{\rule{1pt}{2.7ex}}}{}{}{}
\O(x_1)\O(x_2\vphantom{\rule{1pt}{2.7ex}})\O(x_3)\O(x_4)
\big\rangle.
}
In this highly non-trivial equation the CFT-data enters in different ways in the left-hand and right-hand sides, referred to as the \emph{direct} and the \emph{crossed} channel, and as a functional equation it contains a vast amount of information. The goal of the conformal bootstrap is to use the crossing equation to harvest as many constraints as possible on the CFT-data, and ultimately to fix all involved quantities. 

The bootstrap was particularly powerful in the case of two dimensions due to enhanced symmetry from (global) conformal symmetry to Virasoro symmetry. This means that the conformal multiplets, which contain a primary operator and its descendants (constructed by action of $\de$), group into Virasoro multiplets. We illustrate this in the top left corner of figure~\ref{fig:sixcases}. For instance, the minimal models contain only a finite number of Virasoro multiplets, whose CFT-data could be completely determined. Conformal field theory in two dimensions has expanded to a large body of knowledge---an important result is the construction of the Wess--Zumino--Witten models \cite{Wess:1971yu,Witten:1983tw}---and is now established textbook material\footnote{The standard reference is \cite{DiFrancesco}, see also the lecture notes \cite{Ginsparg:1988ui} and other textbooks \cite{Sergei1995,Polchinski:1998rq,Schottenloher2008}. Attempts to rigorously axiomatise 2d CFT have been made, for instance by Moore and Seiberg \cite{Moore:1988qv} and by Segal \cite{Segal2002}.}. 

In higher dimensions the progress was slower. A set of conventions for higher-dimensional CFTs was given by Osborn and Petkou in 1993 \cite{Osborn:1993cr} and the conformal field theories behind the critical phenomena, in particular the critical $\OO N$ models, were studied from a CFT perspective in a series of papers \cite{Lang:1992zw,Lang:1992pp,Petkou:1994ad} identifying the set of conformal primaries and computing the central charges. The computation of critical exponents, which corresponds to a subset of the CFT-data, using the $\epsilon$ expansion was pushed further \cite{Kleinert2001}, and the collective knowledge about critical phenomena around the year 2000 was collected in \cite{Pelissetto:2000ek}.

One important motivation for increasing interest in CFT came through the AdS-CFT correspondence, or \emph{holography}, relating gravity in $(d+1)$-dimensional anti de~Sitter spacetime to strongly coupled conformal field theory on the $d$-dimensional asymptotic boundary \cite{Maldacena:1997re,Witten:1998qj,AGMOO}. The involved CFTs often have superconformal symmetry, combining conformal symmetry with supersymmetry, and the prime example is the 4d maximally supersymmetric Yang--Mills theory with gauge group $\SU N$ (\NN4 SYM). In the large $N$ limit it is dual to type IIB string theory in $\mathrm{AdS_5}\times S^5$, which at infinite coupling reduces to supergravity. Superconformal symmetry facilitates a variety of powerful methods such as integrability \cite{Beisert:2010jr} and supersymmetric localisation \cite{Pestun:2016zxk}. The intense activity within holography also led to important technical results, such as explicit results for conformal (and superconformal) blocks, which sum up the contribution to a four-point function from a given (super)conformal primary and its descendants. These results, many of which were obtained by Dolan and Osborn \cite{Dolan:2000ut,Dolan:2002zh,Dolan:2003hv}, are essential for what comes next, and for the computations in this thesis.
\\

\noindent In 2008 the ideas of conformal bootstrap were revived in higher dimensions---focussing initially on four dimensions---in the seminal work of Rattazzi, Rychkov, Tonni and Vichi \cite{Rattazzi2008}. The leading principle was to investigate the space of allowed CFT-data, and therefore the space of allowed conformal field theories, by using the mathematical consistency built into the crossing equation, without making use of any Lagrangians or perturbative limits. More precisely, the crossing equation was studied numerically in an expansion around a special kinematic configuration, and positivity of squares of real-valued OPE coefficients was used to rule out whole regions of CFT-data. This idea, which we refer to as the \emph{numerical conformal bootstrap}, has been refined and generated a wealth of results over the past decade, see \cite{Poland:2018epd} for a review, \cite{Poland:2016chs} for a brief summary and \cite{Chester:2019wfx} for a comprehensive and pedagogical introduction. Flagship results include the precise determination of the critical exponents in $\OO N$ models in three dimensions \cite{Kos:2016ysd}, where the results in the Ising \cite{ElShowk2012,Kos:2016ysd} and $\OO 2$ \cite{Chester:2019ifh} case are the most precise available by any method.

In this thesis we will focus on a parallel development, namely \emph{analytic conformal bootstrap}. We introduce the main objectives of this programme by figure~\ref{fig:sixcases}, where we illustrate the spectrum of a conformal field theory, given in terms of the set of primary operators and their scaling dimensions. Without any further assumptions, the axioms of CFT allow for arbitrary and independent values of the scaling dimensions of the various primary operators (up to certain unitarity bounds). This corresponds to the top centre part of figure~\ref{fig:sixcases}. In the cases of CFTs in two dimensions or supersymmetric CFTs in any dimension, the existence of additional symmetries induces an organisation of the conformal primaries into Virasoro multiplets or superconformal multiplets respectively. The scaling dimensions and OPE coefficients of all operators within in such an enlarged multiplet are all related, reducing the number of free parameters and facilitating a wider range of computational methods. 
\begin{figure}
\centering
\includegraphics{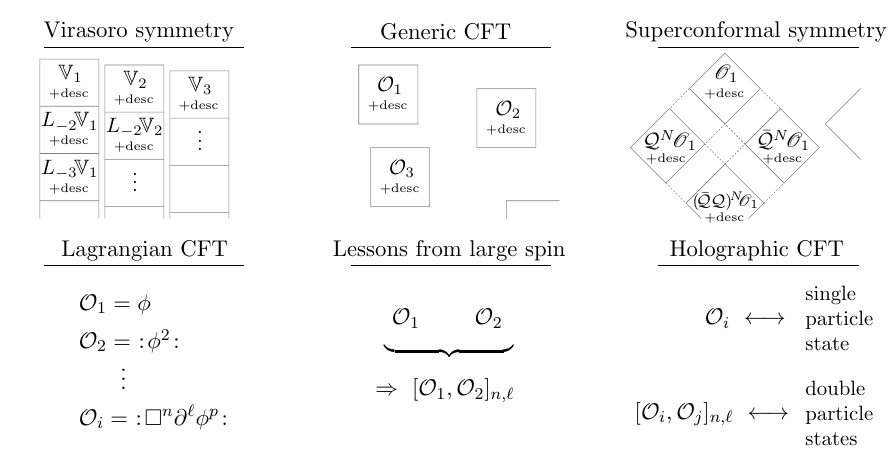}
\caption[Additional symmetries relate properties of different conformal primaries.]{In the cases of 2d CFTs or SCFTs, additional symmetries relate properties of different conformal primaries. In Lagrangian CFTs and holographic CFTs, there are direct methods for constructing and labelling operators. In the generic case, all operators can \emph{a priori} have independent CFT-data. The purpose of analytic bootstrap is to bring structure into this picture by studying operators with large spin.
}\label{fig:sixcases}
\end{figure}

While two-dimensional CFTs and supersymmetric CFTs are highly structured, one cannot \emph{a priori} infer much about the spectrum of a generic CFT. The goal of the analytic conformal bootstrap in higher dimensions is to overcome this gap. There are two main objectives: on the one hand, to make universal statements valid in any CFT, and on the other hand, to use the power of conformal invariance to deduce more properties of specific models.

A key concept in this quest is a \emph{twist family} of operators\footnote{By ``operator'' we here refer to a conformal primary operator. Descendant operators will be explicitly called descendants.}. This consists of a family of operators, parametrised by spin, with approximately equal value of the \emph{twist}, defined as the difference between scaling dimension and spin: $\tau_\ell=\Delta_\ell-\ell$. Such operators naturally occur in weakly coupled Lagrangian CFTs as well as in strongly coupled holographic CFTs. In the former case, twist families of operators are constructed from the fundamental fields, taking the form $\phi\square^n\de^{\{\mu_1}\cdots\de^{\mu_\ell\}}\phi$ with a twist of $2\Delta_\phi+2n+\gamma_\ell$, where $\gamma_\ell$ are small anomalous dimensions, $\square=\de^\mu\de_\mu$, and curly brackets denote symmetrisation and removal of traces. In the latter case, operators in a twist family have a natural definition as the operators dual to rotation modes of weakly interacting multiparticle states in AdS. 

In the examples studied, the operators in a twist family were observed to have collective properties. In Lagrangian theories, the anomalous dimensions could be parametrised in closed form in terms of the spin, and this played an important role in deep inelastic scattering in QCD, even beyond the strictly conformal limit. In this setting, Nachtmann's theorem \cite{Nachtmann:1973mr} further showed that the function $\gamma_\ell$ of the leading twist family takes a convex shape. Much later, important lessons were drawn in \cite{AldayMaldacena2007} about the large spin limit in conformal gauge theories, where it was shown that the logarithmic scaling of anomalous dimensions at large spin can be understood as the corresponding linear scaling with energy of a flux tube in an auxiliary theory.

In two papers from 2012 \cite{Fitzpatrick:2012yx,Komargodski:2012ek} it was independently proven, using crossing symmetry, that such twist families of operators must exist in any CFT in dimension $d>2$. This result is usually taken as the starting point for the analytic conformal bootstrap. The statement is that given any two operators $\O_1$ and $\O_2$ with twists $\tau_1$ and $\tau_2$, there must exist an infinite family of operators with twists $\tau_\ell$ approaching $\tau_\infty=\tau_1+\tau_2$ as $\ell\to\infty$. As we will describe in more detail in the next chapter, this follows from the presence of the identity operator in the crossed channel of the correlator $\expv{\O_1(x_1)\O_2(x_2)\O_2(x_3)\O_1(x_4)}$. Similarly, other operators in the crossed channel induce corrections to the twist, which means that we can non-perturbatively define anomalous dimensions by $\gamma_\ell=\tau_\ell-(\tau_1+\tau_2)$.

Subsequently, more systematics were developed for the analytic bootstrap. In \cite{AldayBissiLuk2015} it was shown that the anomalous dimensions $\gamma_\ell$, as well as corrections to OPE coefficients, naturally expand in inverse powers of the conformal spin, defined as $J^2=(\Delta_\ell+\ell)(\Delta_\ell+\ell-2)/4$. In \cite{AldayZhiboedov2015}, these principles were used to reproduce anomalous dimensions in a number of perturbative CFTs, and in \cite{Simmons-Duffin:2016wlq}, the methods, now dubbed the \emph{lightcone bootstrap}, were used to quantitatively explain a large part of the spectrum of the 3d Ising model, computed in the same paper. In 2016, a completely systematic framework named \emph{large spin perturbation theory} was introduced by Alday \cite{Alday2016,Alday2016b}. This framework facilitates significant progress in both generic and specific CFTs in the presence of a small expansion parameter, which may for instance be a small coupling constant, a dimensional $\epsilon$, or the inverse number of degrees of freedom. In this thesis we show that large spin perturbation theory not only elucidates the structure of many conformal field theories, but it is also powerful enough to generate new results beyond other methods.

The final ingredient to achieve this goal was given by Caron-Huot in the \emph{Lorentzian inversion formula} \cite{Caron-Huot2017}. It puts on firm grounds the empirical observations from all known examples that the functions $\gamma_\ell$ extend in exact form all the way down to some finite spin, typically $0$, $1$ or $2$. The anomalous dimensions and OPE coefficients, collectively the CFT-data, are given in terms of an integral over a compact domain of the \emph{double-discontinuity} of the correlator weighted against a kernel. The double-discontinuity restricts to terms containing enhanced singularities, which means that the CFT-data can be computed without knowing the full correlator. This can be phrased as a dispersion relation, meaning that the correlator can be reproduced from its double-discontinuity, up to contributions from spin $0$ or $1$ \cite{Carmi:2019cub}.
\\

\noindent The purpose of this thesis is to demonstrate how large spin perturbation theory can be turned into a powerful and systematic framework for studying perturbative conformal field theories, by which we mean CFTs equipped with any expansion parameter, not necessarily the coupling constant. This is achieved through a number of examples where the method is successfully applied to some of the most well-studied CFTs. The framework follows the analytic bootstrap approach, which means that it builds only on consistency conditions and on the axioms of conformal field theory, without any reference to Lagrangians or standard perturbation theory, and it does not make use of specific methods such as supersymmetric localisation. 

The thesis starts with a comprehensive review, which includes a practical guide to large spin perturbation theory, followed by a number of concrete examples. These examples are given as a demonstration of the method, but the results generated there are also contributions to the literature. We do not aim to cover all aspects of higher-dimensional CFTs and we refer instead to the excellent reviews on the subject \cite{Qualls:2015qjb,Rychkov:2016iqz,TASIBootstrap,Poland:2016chs,Poland:2018epd}. However, we do give the essential ingredients and present the ideas that lead up to work in this thesis. This is the purpose of chapter~\ref{ch:analyticstudy}, which finishes by outlining the method in terms of the following procedure:
\begin{enumerate}
\item Find operators that contribute at each order in the expansion parameter.
\item Compute their double-discontinuity in the crossed channel.
\item Find the corresponding corrections to the CFT-data using the Lorentzian inversion formula.
\item Where applicable, use consistency conditions to fix any undetermined constants and/or iterate the procedure.
\end{enumerate}
Chapter~\ref{ch:analyticstudy} also derives the precise version of the inversion formula used in large spin perturbation theory from the more general formula in \cite{Caron-Huot2017}, and it contains a presentation of some of the theories studied in detail in the later chapters.

Chapter~\ref{ch:practical} takes the form of a practical guide, giving more details on how to execute each of the steps given above. The presentation is encyclopaedic, and the purpose is to give a useful overview of the method, as an alternative to the often technical original publications. 

The following chapters contain explicit examples. In chapter~\ref{ch:paper2} we demonstrate how large spin perturbation theory combined with the Lorentzian inversion formula can be applied to the Wilson--Fisher fixed-point in the $d=4-\epsilon$ expansion, where new results are generated at order $\epsilon^4$, for instance for the central charge. In chapter~\ref{ch:paper1} we study conformal gauge theories and find the most general form of the order $g$ four-point function of a bilinear scalar operator. This reproduces known results in the \NN4 super Yang--Mills theory but applies to any theory satisfying a short list of assumptions. 

Chapter~\ref{ch:more} is divided into smaller sections, each giving yet another application of the framework but recycling some technical results from previous sections and chapters. While sections~\ref{sec:paper4short} and \ref{sec:paper5short}---which cover critical $\phi^4$ theories with $\OO N$ and general global symmetry---are based on work presented elsewhere, sections~\ref{sec:multicritical} and \ref{sec:WFfrompaper1} contain previously unpublished results. In section~\ref{sec:multicritical} we collectively study the multicritical theories described in figure~\ref{fig:operahouse} and derive new results for OPE coefficients, including the central charge. In section~\ref{sec:WFfrompaper1} we show that the results from chapter~\ref{ch:paper1} can be used to compute the order $\epsilon$ four-point function of the $\varphi^2$ operator in the Wilson--Fisher fixed-point.

We finish with a discussion in chapter~\ref{ch:disc}, where we summarise and give some outlook. This is followed by some appendices with technical details from the chapters described above.

\chapter{Analytic study of conformal field theories}\label{ch:analyticstudy}

In this chapter we review the necessary background for the analytic study of conformal field theories. After giving an overview of the fundamental definitions, we discuss Lorentzian kinematics and give some references on conformal blocks. We then introduce conventions regarding the operator content in CFTs and give three explicit examples. This is followed by a review of the developments within the analytic bootstrap, including the Lorentzian inversion formula. We give a precise derivation of the perturbative inversion formula, which plays a central role in the following chapters, and finish by outlining the principles of large spin perturbation theory.

\section{What is a CFT?}

A brief way of defining a conformal field theory is that it is a quantum field theory invariant under conformal symmetry. This leads to a description of CFTs built on the understanding of quantum field theory (QFT), where conformal symmetry is used to distil properties that are special to CFTs. One such property is that all fields are massless, i.e.\ the theory has no mass gap. However, while much of our understanding of QFT relies on the possibility of writing down a Lagrangian that describes a given theory, at least at weak coupling, it is possible to give a characterisation of conformal field theory that is independent of this construction. It is this perspective that we will take here. It leads to a more concrete, but at the same time more mathematically rigorous, definition of a conformal field theory.

We define a conformal field theory as a consistent set of operators together with correlation functions (\emph{correlators}) of these operators with appropriate transformation properties under the conformal group. A conformal transformation between subsets of Euclidean manifolds is an angle-preserving map
\beq{\label{eq:metrictransform}
x\mapsto x'=\psi(x),\quad (\psi^*g')_{\mu\nu}(x)=\Omega(x)^{2}g_{\mu\nu}(x),
}
where $\Omega(x)$ is a positive scalar function and $\psi^*g'$ denotes the pullback of the metric.

A formal approach, as in Segal's axiomatisation \cite{Segal2002}, is to view the CFT itself as a set of operators together with a framework (a set of functors) which assigns, to a given manifold, the set of correlators of its operators on that manifold. In this sense, the CFT is a tool that can be used to probe the geometry of the manifold. This philosophy is particularly useful in two dimensions, where any manifold is locally conformally flat. In this thesis, which is restricted to the case of $d>2$ dimensions, we focus our considerations on conformally flat manifolds, and therefore study flat space $\R^d$. After removing the origin, this is also conformally equivalent to the cylinder $\R\times S^{d-1}$ through a radial foliation, which implies that correlators on the cylinder are directly related to correlators on $\R^d$ \footnote{The perhaps most important manifold not conformally equivalent to $\R^d$ is $S^1\times \R^{d-1}$. Probing this geometry gives access to observables at finite temperature, where the length of the circle can be related to the inverse temperature. In \cite{Iliesiu:2018fao} a bootstrap analysis was developed for this geometry and in \cite{Iliesiu:2018zlz} observables for the Ising model at finite temperature were computed.}.
In addition, we will limit the set of observables to correlators of local operators. This excludes for instance Wilson loops, as well as some interesting non-local operators such as light-ray operators and shadow operators \cite{Kravchuk:2018htv}. 

The fundamentals of conformal field theories in flat $\R^d$ are well-documented, for instance briefly in \cite{DiFrancesco} and in more detail in some lecture notes \cite{Qualls:2015qjb,Rychkov:2016iqz,TASIBootstrap}\footnote{See also \cite{Simmons-DuffinPhys229} for some comprehensive but unfinished lecture notes.}. We will not repeat all details here, but instead just outline the main results.

The conformal group of Euclidean $\R^d$ is the Poincar\'e group extended by dilatations $x^\mu\mapsto\lambda x^\mu$ and special conformal transformations (translations of infinity), and it is isomorphic to the special orthogonal group $\SO{d+1,1}$. Local conformal operators, according to Mack's classification \cite{Mack:1975je}, are either primary operators (primaries) or descendant operators. A conformal primary may be defined as an operator which transforms locally under the conformal transformations~\eqref{eq:metrictransform}. For a scalar primary operator $\O$ this takes the form
\beq{\label{eq:primarytransformation}
\O(x')=\Omega(x)^{-\Delta_\O}\O(x),
}
which defines the \emph{scaling dimension} $\Delta_\O$. Primary operators also transform in irreducible representations (irreps) of the Lorentz group and in irreps of any potential global symmetry group. The transformation property~\eqref{eq:primarytransformation} implies that primaries inserted at the origin are annihilated by the generators $K_\mu$ of special conformal transformations \cite{Rychkov:2016iqz}. Descendant operators are generated by the action of the generator of momentum, $-iP_\mu=\de_\mu:=\pp {x^\mu}$, conjugate to $K_\mu$: $[P_\mu,K^\nu]=2i(D\ten\delta{_\mu^\nu}-\ten M{_\mu^\nu})$, where $D$ and $\ten M{_\mu^\nu}$ generate dilatation and Lorentz transformations respectively. For descendants, the transformation rule \eqref{eq:primarytransformation} holds only for constant dilatations, and it is corrected by derivatives in the case of more general transformations. The set of descendants generated from a given primary forms a conformal multiplet, and all properties of these operators are related to the corresponding primary. Therefore, we limit our considerations to primaries, and in what follows we refer to primary operators just as operators. 

From invariance under dilatation and special conformal transformations it follows that the two- and three-point correlation functions of scalar primaries take the form
\beqa{\label{eq:scalaroperatorsnormalisation}
\expv{\O_i(x_1)\O_j(x_2)}&=\frac{\delta_{ij}}{x_{12}^{2\Delta_i}}, 
\\
\expv{\O_i(x_1)\O_j(x_2)\O_k(x_3)}&=\frac{c_{\O_i\O_j\O_k}}{x_{12}^{\Delta_i+\Delta_j-\Delta_k}x_{23}^{\Delta_j+\Delta_k-\Delta_i}x_{13}^{\Delta_k+\Delta_i-\Delta_j}}, \quad x_{ij}=\sqrt{(x_i-x_j)^2},
}
where the set of OPE coefficients $c_{\O_i\O_j\O_k}$ and scaling dimensions $\Delta_{\O_i}$ forms the CFT-data and carries the dynamical information of the CFT\footnote{We have normalised scalar operators by choosing a diagonal basis \eqref{eq:scalaroperatorsnormalisation}. For spinning operators, there are multiple conventions in the literature. Our conventions will be clear from the normalisation of the conformal blocks below. Notice that in the presence of a global $\UU1$ symmetry it is customary to assign the non-vanishing two-point functions to charge conjugate pairs.}. 

The final essential ingredient needed to define a CFT is the state-operator correspondence. It implies that any quantum state $\ket\psi$ defined on a sphere $S^{d-1}$ in Euclidean $\R^d$ can be written as a linear combination of primary and descendant operators inserted at the centre $x_0^\mu$ of the sphere:
\beq{\label{eq:operatorstatecorrespondence}
\ket\psi=\sum_\O f_\O\, \O(x_0)\!\ket0.
}
If we take $\ket\psi$ to be the state $\O_1(x)\O_2(0)\!\ket0$ for primaries $\O_1$, $\O_2$, we get the operator product expansion (OPE) given in \eqref{eq:OPEintro} in the introduction. Importantly, in a CFT the OPE coefficients of descendants are related to those of the primary, and the coefficient functions $C(x,\de)$ in \eqref{eq:OPEintro} depend only on the quantum numbers of the involved primary operators.
In the case of scalar operators $\O_1$ and $\O_2$, the conformal primaries in the OPE must transform in a traceless symmetric representation of the Lorentz group, and can therefore be characterised by their scaling dimension $\Delta$ and spin $\ell$, the latter defined as the rank of the representation. We write $\O^{\mu_1\cdots\mu_\ell}$ and assume that the symmetrisation and removal of traces is understood.

\section{Lorentzian four-point functions}

In the previous section we saw that the two- and three-point functions of conformal primaries are completely fixed in terms of the CFT-data. The first correlator to carry non-trivial kinematics is therefore the four-point function. Moving from Euclidean $\R^d$ to Lorentzian $\R^{d-1,1}$ spacetime introduces an interesting kinematic limit for four-point functions, the \emph{lightcone limit}, where operators become collinear. In this section we describe this and other relevant limits for Lorentzian four-point functions.

Conformal symmetry can be used to map any four-point configuration onto a two-dimensional plane, which means that the spacetime dependence can be parametrised by two independent variables, called the conformal cross-ratios,
\beq{
u=z \zb=\frac{x_{12}^2x_{34}^2}{x_{13}^2x_{24}^2}, \qquad v=(1-z)(1-\zb)=\frac{x_{14}^2x_{23}^2}{x_{13}^2x_{24}^2}.
}
Throughout this thesis we will use $(u,v)$ and $(z,\zb)$ interchangeably. In slight abuse of notation, we will therefore write $\G(u,v)=\G(z,\zb)$ for the four-point function of identical external scalar operators $\phi$, as defined by
\beq{\label{eq:fourpointsetup}
\expv{\phi(x_1)\phi(x_2)\phi(x_3)\phi(x_4)}=\frac{1}{x_{12}^{2\Delta_\phi}x_{34}^{2\Delta_\phi}}\G(u,v).
}
Here we factored out a pair of two-point functions such that the contribution from the identity operator $\1$ in the pairwise OPEs $\phi(x_1)\times\phi(x_2)$ and $\phi(x_3)\times\phi(x_4)$ is just~$1$. In this notation, the crossing equation~\eqref{eq:crossingintro} takes the form\footnote{There is also another crossing equation which follows from exchanging the operators at $x_1$ and $x_2$. It takes the form $\mathcal{G}(u,v)=\mathcal{G}\left(\frac{u}{v},\frac{1}{v}\right)$ but it will not be important in this thesis.}
\beq{\label{eq:crossing}
\G(u,v)=\parr{\frac uv}^{\Delta_\phi}\G(v,u).
}

The OPE expansions within the four-point function can be organised as
\beq{\label{eq:CBexp}
\G(u,v)=\sum_{\O}c^2_{\phi\phi\O}G^{(d)}_{\Delta_\O,\ell_\O}(u,v),
}
where we have introduced \emph{the conformal blocks} $G^{(d)}_{\Delta,\ell}(u,v)$. In the OPE expansion of the four-point function they sum up the contributions from the primary $\O$ and all its descendants, and we refer to \eqref{eq:CBexp} as the \emph{conformal block expansion} of the correlator. The conformal blocks are theory-independent functions of the cross-ratios and depend only on the scaling dimension $\Delta$ and spin $\ell$ of the exchanged operators. The CFT-data enters the conformal block expansion through the parameters $\Delta_\O$ and $\ell_\O$~\footnote{Since we are considering only scalar external operators $\phi$, the exchanged operators transform in the traceless symmetric representations of the Lorentz group, uniquely labelled by an integer $\ell$ (corresponding to one-row Young tableaux of length $\ell$).} of the primaries, and through the squared OPE coefficients $c^2_{\phi\phi\O}$. 

Notice now the advantage of writing the crossing equation in the form~\eqref{eq:crossing}. The expressions in the direct channel (left-hand side) and the crossed channel (right-hand side) take the same functional form, where both sides have an expansion \eqref{eq:CBexp}, just with $u$ and $v$ exchanged. The conformal bootstrap aims to harvest as much information as possible about the CFT-data from this highly complicated equation.

\subsection{Lorentzian kinematics}

Let us take a closer look at the kinematics of the four-point function \eqref{eq:fourpointsetup} in Lorentzian signature. As described above, any configuration is conformally equivalent to one where the four operator insertion points $x_i$ are confined to a plane. Restricting to space-like separation, we can, up to permutation of the $x_i$, parametrise the plane by a time-like and a space-like direction $(x^0,x^1)$ and use additional conformal symmetry to place (gauge-fix) the operators at
\beq{\label{eq:lorentzpoints}
x_1=0, \quad x_2=(x_2^0,x_2^1,0,\ldots), \quad x_3=(0,1,0,\ldots), \quad x_4=\infty.
}
We define $z=x^1_2-x_2^0$ and $\zb=x_2^1+x_2^0$, by which the space-like separation corresponds to the values $z,\zb\in(0,1)$. This is illustrated in figure~\ref{fig:diamond}. In this region the conformal blocks are real-valued regular functions of $z$ and $\zb$ \cite{Rattazzi2008}.
\begin{figure}
\centering
\includegraphics{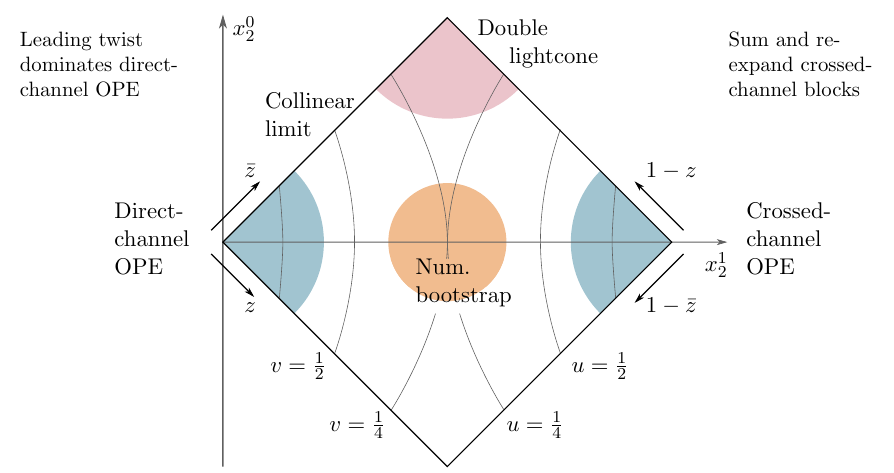}
\caption[A spacelike configuration in Lorentzian kinematics.]{A spacelike configuration in Lorentzian kinematics with operators $\O_1$, $\O_3$ and $\O_4$ at fixed positions and $\O_2$ confined in the $(x_2^1,x_2^0)$ plane. The cross-ratios and the relevant kinematic limits are indicated.}\label{fig:diamond}
\end{figure}
The Lorentzian configuration~\eqref{eq:lorentzpoints} can be reached by a Wick-rotation from Euclidean signature, where $(z,\zb)$ are complex and each other's conjugate. 

The convergence of the OPE in Euclidean signature is guaranteed by the operator-state correspondence \eqref{eq:operatorstatecorrespondence}, and carries over to Lorentzian signature, see e.g.\ \cite{Hartman:2016lgu}. In the limit $z,\zb\ll1$~\footnote{It is often necessary to separate the hierarchy between $z$ and $\zb$, which can be realised from the appearance of terms like $(z-\zb)^{-1}$ for instance in the four-dimensional conformal blocks \eqref{eq:CB4d} below. In these cases we assume that we have $z\ll\zb\ll1$.}, operators with the smallest scaling dimensions dominate the conformal block expansion, and we refer to this limit as the OPE limit. Conversely, the limit $1-z,1-\zb\ll1$ corresponds to the crossed-channel OPE limit.

In the OPE limit, the conformal blocks have the following expansion
\beq{
G_{\Delta,\ell}^{(d)}(u,v)\sim u^{\frac\tau2}(1-v)^\ell, \qquad\tau=\Delta-\ell,
}
which has two important implications. Firstly, it shows that the twist $\tau=\Delta-\ell$ is a useful label for operators since all operators with equal twist can be collected into a common $u$ power. Secondly, it means that given a correlator $\G(u,v)$ in closed form, one can expand both sides of \eqref{eq:CBexp} order by order in $u$ and $1-v$, and compute the involved OPE coefficients and scaling dimensions one by one. We refer to this procedure as performing the conformal block decomposition. This decomposition may equally well be done in the variables $(z,\zb)$ in the expansion $z\ll\zb\ll1$.

For the purpose of the numerical bootstrap, an expansion around the crossing symmetric point $z=\zb=\frac12$ is particularly useful, since it treats the direct and the crossed channel on an equal footing. However, for this thesis we will instead make use of an inherently Lorentzian regime, namely the lightcone limit. Since the notion of the lightcone limit sometimes is ambiguous in the literature, we will always use the two notions \emph{collinear limit} and \emph{double lightcone limit}, and for additional clarity we summarise our conventions for kinematic variables in table~\ref{tab:kinematicconventions} as well as in figure~\ref{fig:diamond}.

\begin{table}
\centering
\caption{Conventions for the kinematic variables, used throughout this thesis.}\label{tab:kinematicconventions}
{\small
\vspace{4pt}
\begin{tabular}{|lll|}\hline
&&\\[-2ex]
\multirow{2}{*}{Cross-ratios} & $u=\dfrac{x_{12}^2x_{34}^2}{x_{13}^2x_{24}^2}$ & $v=\dfrac{x_{14}^2x_{23}^2}{x_{13}^2x_{24}^2}$
\\[1.9ex]
&$u=z\zb$ & $v=(1-z)(1-\zb)$
\\\hline &&\\[-2ex]
Auxiliary cross-ratios &&$w=1-\zb,\qquad \xit = \dfrac{1-\zb}{\zb}$
\\[1.9ex]\hline &&\\[-2ex]
Harmonic superspace & $\alpha\alphab=\dfrac{y_{12}^2y_{34}^2}{y_{13}^2y_{24}^2}$ & $(1-\alpha)(1-\alphab)=\dfrac{y_{14}^2y_{23}^2}{y_{13}^2y_{24}^2}$
\\[1.9ex]
\hline
 \multirow{2}{*}{OPE limit} &$u\to 0$& $v\to 1$
\\  &$z\to 0$& $\zb\to 0$
\\
\hline
\multirow{2}{*}{Collinear limit}& $u\to 0$ & any $v$
\\  & $z\to 0$ & any $\zb$
\\\hline
\multirow{2}{*}{Double lightcone limit} & $u\to0$ & $v \to 0$
\\& $z\to0$ & $\zb \to 1$
\\
\hline
\multirow{2}{*}{Crossing} & $u\mapsto v$ & $v\mapsto u$
\\
 & $z\mapsto 1-\zb$ & $\zb\mapsto 1-z$
 \\\hline
\end{tabular}
}
\end{table}

The collinear limit is defined as $z\to0$ for any value of $\zb$. This corresponds to the point $x_2$ becoming null separated from $x_1$. Defining $x=x_{21}$ we see that this limit is characterised by the vanishing of the four-vector norm $x^2\to0$ in a limit where some of the components $x^\mu$ remain finite. The importance of this limit dates back to deep inelastic scattering experiments with hadrons, where an approximate conformal invariance was understood to control the operator product expansion, see e.g.\ \cite{PeskinSchroeder,Braun:2003rp}. In the collinear limit, the OPE is dominated by the operator with the lowest value of twist for each spin. This can be seen from the detailed form of the OPE,
\beq{\label{eq:OPEexplicit}
\O_2(x)\O_1(0)=\sum_{\O_\ell}c_{21\O}{(x^2)}^{\!\frac{\Delta_{\O}-\Delta_1-\Delta_2}2} \,\frac{x^{\mu_1}\cdots x^{\mu_{\ell}}-\text{traces}}{(x^2)^{\ell/2}}\parr{\O_{\mu_1\cdots\mu_{\ell}}(0)+\text{desc.}},
}
from which one can read off that the leading singularities as $x^2\to0$, keeping $x^\mu$ finite, scale as $\tau=\Delta_\O-\ell$.

The double lightcone limit is relevant when we discuss crossing in Lorentzian kinematics. It occurs when $x_2$ becomes collinear with both $x_1$ and $x_3$. The double lightcone limit is dominated by operators of large spin, and studying the crossing equation in this limit has become known as lightcone bootstrap, of which we can view large spin perturbation theory (LSPT) as a special case.

When dealing with expansions in the double lightcone limit we should take
\beq{\label{eq:lightconepreciselimit}
z\ll1-\zb\ll1.
}
This explicitly breaks the symmetry of the double lightcone towards the direct channel collinear limit. This means that in this thesis we will treat the direct channel and crossed channel differently. In the direct channel, we are always free to restrict ourselves to the leading twist family. When considering crossed-channel operators, however, we have to be careful. In the expansion of the crossed channel in the limit~\eqref{eq:lightconepreciselimit}, it is in general not enough to expand conformal blocks one by one. Instead, one needs to compute sums over twist families before taking $z\to0$. We indicate this in figure~\ref{fig:diamond}. 

Both the collinear and the double lightcone limit translate to corresponding limits for the cross-ratios $u$ and $v$, and all conventions for kinematics used in this thesis are collected in table~\ref{tab:kinematicconventions}.
From the conformal blocks in the collinear limit, \eqref{eq:collblocks} below, we can derive an approximate relation between $\zb$ and the spins that dominate the contribution to the four-point function. The dominant contributions come from spins $\ell$ of order
\beq{
\ell\sim\frac1{\sqrt{1-\zb}}\,,
}
and we give more detail on this in section~\ref{sec:lightconecrossing} and~\ref{sec:kernel.method}. 

Apart from the lightcone limit, which plays the central role in this thesis, there is another important intrinsically Lorentzian limit, denoted the \emph{Regge limit} \cite{Cornalba:2006xk,Costa:2012cb}. In a CFT four-point configuration with pairwise timelike separated operators at $x_4^0>x_1^0$ and $x_2^0>x_3^0$, the Regge limit arises when both pairs $(x_1,x_2)$ and $(x_3,x_4)$ approach null separation. In holographic theories, this corresponds to high-energy scattering in the dual AdS space. As explained in e.g.\ \cite{Kravchuk:2018htv}, in terms of the cross-rations this projects onto the OPE limit $z,\zb\to0$, but with $\zb$ evaluated on the second sheet after analytically continuing around the point $\zb=1$. The OPE is not convergent since the points $x_1$ and $x_2$ are not near each other. Instead each conformal block has a scaling that schematically looks like
\beq{\label{eq:reggescaling}
G^{(d)}_{\Delta,\ell}(z,\zb)\sim \sqrt{z\zb}^{1-\ell}.
}
However, any physical correlator is expected to be bounded in this limit\footnote{More precisely by a scaling of the form \eqref{eq:reggescaling}, where $\ell$ is taken to be the Regge/Pomeron intercept $\ell_0$, which has a value $\ell_0<2$ \cite{Kravchuk:2018htv}.}. We will not review the various applications of the Regge limit, but will refer back to the scaling \eqref{eq:reggescaling} when we discuss the Lorentzian inversion formula in section~\ref{sec:Inversionformula}.

\subsection{Conformal blockology}\label{sec:blockology}
From the discussion above, it should be clear that any method in conformal bootstrap will rely heavily on the conformal blocks $G^{(d)}_{\Delta,\ell}(u,v)$ appearing in the decomposition \eqref{eq:CBexp}. In general, these functions are not known in closed form, which means that we depend on various technologies for evaluating the blocks in certain expansions.

In the most general setting, the conformal blocks are functions of the cross-ratios, depending on the scaling dimensions of the exchanged and the external operators, and on the spacetime dimension $d$. In the case of identical external scalar operators $\phi$, the blocks are independent of $\Delta_\phi\,$\footnote{For generic external scalars they depend on the combinations $\Delta_1-\Delta_2$ and $\Delta_3-\Delta_4$.} and we reserve the notion $G^{(d)}_{\Delta,\ell}(u,v)$ for this case. 

By definition, the conformal block for a conformal primary operator of dimension $\Delta$ and spin $\ell$ sums up the contribution to the four-point function of that operator together with all its descendants. 
Since descendants of a given primary are related by the generator of translations, $\de_\mu$, all terms making up the conformal block have identical eigenvalues under the action of the Casimir operators of the conformal group. This leads to a set of differential equations satisfied by the blocks,
\beqa{\label{eq:casimireig2}
\mathcal C_2G^{(d)}_{\Delta,\ell}(u,v)&=\frac12\parr{\ell(\ell+d-2)+\Delta(\Delta-d)}G^{(d)}_{\Delta,\ell}(u,v),\\\label{eq:casimireig4}
\mathcal C_4G^{(d)}_{\Delta,\ell}(u,v)&=\ell(\ell+d-2)(\Delta-1)(\Delta-d+1)G^{(d)}_{\Delta,\ell}(u,v),
}
for the quadratic and quartic Casimir operators \cite{Dolan:2011dv}\footnote{Written in terms of generators of the Lorentzian conformal group $\SO{d,2}$, the Casimirs take the form $\mathcal C_2=J^{AB}J^{BA}$ and $\mathcal C_4=J^{AB}J^{BC}J^{CD}J^{DA}$ \cite{Ferrara1973}. Note that the Casimir eigenvalues satisfy a symmetry generated by $\{\ell\leftrightarrow 2-d-\ell,\Delta\leftrightarrow d-\Delta,\Delta\leftrightarrow 1-\ell\}$, corresponding to the dihedral group of eight elements \cite{Caron-Huot2017,Kravchuk:2018htv}.}
\beqa{\label{eq:quadraticCasimir}
\mathcal C_2&=D_z+D_\zb+(d-2)\frac{z\zb}{z-\zb}\parr{(1-z)\de_z-(1-\zb)\de_\zb},\\\label{eq:quarticCasimir}
\mathcal C_4&=\parr{\frac{z\zb}{z-\zb}}^{d-2}\parr{D_z-D_\zb}\parr{\frac{z\zb}{z-\zb}}^{2-d}\parr{D_z-D_\zb},
}
respectively, where  
\beq{
D_x=(1-x)x^2\de_{x}^2-x^2\de_x.
}
By solving the Casimir equations with appropriate boundary conditions, the conformal blocks in four dimensions were computed in a closed form by Dolan and Osborn in 2000 \cite{Dolan:2000ut}\footnote{A similar expression was also derived in two dimensions and through a recursion relation the blocks in all even dimensions can be generated \cite{Dolan:2003hv}. In two dimensions there is also the notion of Virasoro conformal blocks, summing up contributions from an entire Virasoro multiplet. They are much more complicated, but can be generated to arbitrary order \cite{Perlmutter:2015iya}.} 
\beq{\label{eq:CB4d}
G_{\tau,\ell}(z,\zb):=G^{(4)}_{\tau+\ell,\ell}(z,\zb)=\frac{z\zb}{z-\zb}\left(
k_{\frac\tau2+\ell}(z)k_{\frac\tau2-1}(\zb)-k_{\frac\tau2+\ell}(\zb)k_{\frac\tau2-1}(z)
\right),
}
where
\beq{
k_\beta(x)=x^\beta{_2F_1}(\beta,\beta;2\beta;x),
}
in which $_2F_1(a,b;c;x)$ is Gau\ss's hypergeometric function as defined in \eqref{eq:hyperdef} in appendix~\ref{app:identities}.
Dolan and Osborn also showed that the conformal block for a scalar operator in arbitrary dimension $d=2\mu$ is given by the infinite double-sum \cite{Dolan:2000ut}
\beq{
\label{eq:ScalarBlockAnyD}
G^{(d)}_{\Delta,0}(u,v) = u^{\frac\Delta2}\sum_{m,n=0}^\infty \frac{\left(\Delta/2\right)^2_m\left(\Delta/2\right)^2_{m+n}}{ \left( \Delta+1-\mu\right)_m \left(\Delta\right)_{2m+n}} \frac{u^m(1-v)^n}{m! \, n!},}
where $(a)_n=\frac{\Gamma(a+n)}{\Gamma(a)}$ is the Pochhammer symbol.

We have already stressed the importance of the collinear limit, and in fact, in this limit both the Casimir operators and the conformal blocks simplify dramatically. We are effectively left with one cross-ratio and the conformal group reduces to $\SL2\R$. Theories with this symmetry group are referred to as one-dimensional CFTs. The corresponding Casimir operator, called the collinear or $\SL2\R$ Casimir, takes the form\footnote{The form of the collinear Casimir, up to a constant shift, follows from acting with the $\mathcal C_2$ on an ansatz for the blocks given as a series expansion in $z$ starting at $z^{\frac\tau2}$, with coefficients as functions of $\zb$.}
\beq{\label{eq:Dbardeff}
\Dbar = D_\zb=(1-\zb)\zb^2\de^2_\zb-\zb^2\de_\zb,
}
and the conformal blocks expand as
\beq{\label{eq:collblocks}
G^{(d)}_{\Delta,\ell}(z,\zb)=z^{\frac{\Delta-\ell}2}k_{\frac{\Delta+\ell}{2}}(\zb)+O\parr{z^{\frac{\Delta-\ell}2+1}}.
}
We refer to $z^{\tau/2}k_\hb(\zb)$ as the \emph{collinear blocks} and $k_{\hb}(\zb)$ as the \emph{$SL(2,\R)$ blocks}. The collinear blocks, or equivalently the $\SL2\R$ blocks, have eigenvalue $J^2=\hb(\hb-1)$ under the Casimir action,
\beq{\label{eq:collcasrel}
\Dbar k_\hb(\zb)=\hb(\hb-1)k_\hb(\zb).
}

The expansion \eqref{eq:collblocks} suggests that in the collinear limit it is natural to introduce variables $h=\frac{\Delta-\ell}2$, $\hb=\frac{\Delta+\ell}2$, in analogy with two dimensions. In fact, $\hb$ will be very important in what follows, as we discuss in section~\ref{sec:largespinlightcone}. We will however not employ the notation $(h,\hb)$ for individual operators; instead we will use $\hb$ as an independent variable, parametrising operators with approximately equal value of $h$. We also keep the twist as a label rather than $h$, since $\tau=2h$ and the twist is more commonly used in the literature.

The subleading corrections in $z$ of \eqref{eq:collblocks} can be computed and at each order in $z$ they take the form of a finite sum of $\SL2\R$ blocks with shifted arguments, where the coefficients depend on $\Delta$, $\ell$ and $d$. In appendix~\ref{app:subcollinearblocks} we give more detail on this expansion. 
Notice that the explicit form of the collinear blocks is independent of the spacetime dimension $d$, following from the one-dimensional nature of the collinear expansion. Nevertheless, the subleading corrections do depend on $d$.

Finally, for completeness we give the collinear blocks in the case of non-identical external scalar operators. In this case we evaluate \eqref{eq:collblocks} under the replacement 
\beq{\label{eq:collineardifferentmod}
k_\hb(\zb)\rightsquigarrow \zb^{\hb}{_2F_1}\parr{\hb+\tfrac{\Delta_2-\Delta_1}2,\hb+\tfrac{\Delta_3-\Delta_4}2;2\hb;\zb}.
}

\subsection{Conserved currents and unitarity bounds}\label{sec:UBC}

A main object of interest is the spectrum of operators that appear in the OPE expansion of a four-point function. In section~\ref{sec:pertCFT} below we will discuss this in detail in the case of CFTs with a small expansion parameter, and we will look at a few explicit examples. Here, we instead make some universal statements valid in any CFT.

We have seen that the collinear limit emphasises the contribution from the leading twist family in the OPE. In a unitary CFT there is a minimal twist that such any spinning operator can admit \cite{Mack:1975je}
\beq{
\tau_{\O_\ell}\geqslant d-2, \qquad \ell>0.
}
Operators saturating this bound are referred to as conserved currents, since they are subject to a conservation equation $\de_{\mu_1}\O^{\mu_1\cdots\mu_\ell}=0$. Equivalently, we can view this as a shortening of the conformal multiplet, since this equation means that a subset of the possible descendant operators vanishes. For the conformal blocks this translates into a differential equation\cite{Alday2016b} 
\beq{\label{eq:boundspinning}
\dsat G^{(d)}_{d-2+\ell,\ell}(u,v)=0,\quad \dsat=(d-2) \left(z^2 \de_z-\zb^2 \de_\zb\right)+2 z \zb (\zb-z) \de_z\de_\zb,
}
which will be used in section~\ref{sec:paper4short}. For scalar operators the corresponding bound is 
\beq{\label{eq:boundscalar}
\Delta_\O\geqslant \frac{d-2}2, \qquad \ell=0, \ \O \neq\1.
}
The bounds \eqref{eq:boundscalar} and \eqref{eq:boundspinning} are saturated by the operators $\phi$ and $\phi\de^\ell\phi$ in the theory of a free scalar in $d$ dimensions.

The unitarity bounds are not the only way a Lorentzian CFT can fail to be unitary, and typically unitarity is broken by some OPE coefficients or scaling dimensions taking values off the real axis. In \cite{Hogervorst:2015akt} it was shown that the Ising model, described by the top curve in figure~\ref{fig:operahouse}, is in fact non-unitary away from any integer dimension.

A generic interacting conformal field theory contains only a finite number of conserved currents, namely the stress tensor $T\munu$ related to the generators of Poincar\'e invariance, and, where applicable, global symmetry currents $J_\mu$. In addition, supersymmetry adds further conserved currents. Correlators involving conserved currents satisfy conformal Ward identities, which introduce physically meaningful normalisation constants called central charges \cite{Petkou:1994ad}, see also \cite{Dolan:2000ut}.

The \emph{central charge}, $C_T$, determines the OPE coefficient with the stress tensor and is of the same order of magnitude as the number of degrees of freedom in the theory\footnote{However, $C_T$ is not a precise measure of the number of degrees of freedom of the theory, and it does not always decrease under RG-flow as in Cardy's $c$-theorem. However, in two dimensions, $C_T=2c$ has this role \cite{Zamolodchikov:1986gt}. In higher dimensions, the statements corresponding to the $c$ theorem are the $a$-theorem in four dimensions \cite{Komargodski:2011vj} and the conjectured $F$-theorem in three dimensions \cite{Jafferis:2011zi,Closset:2012vg}.}. In our conventions the stress tensor OPE coefficient takes the form
\beq{\label{eq:centralchargeOPE}
c_{\O_i\O_j T}=-\frac{d\Delta_{\O_i}}{d-1}\frac1{2\sqrt{C_T}}\delta_{ij}.
}
These conventions correspond to 
\beq{\label{eq:CTNscalars}
C_{T,\mathrm{free}}=\frac{Nd}{d-1}
}
for $N$ free scalars in $d$ dimensions.

The \emph{current central charge} $C_J$, related to the normalisation of global symmetry currents $J^\mu$, roughly corresponds to the amount of degrees of freedom charged under the corresponding symmetry. 
The exact normalisations for current central charges depend on conventions for the group generators and the normalisations of the adjoint representation. In this thesis we take the conventions such that for global $\OO N$ symmetry we have
\beq{
c^2_{\phi\phi J}=-\frac1{C_J}, \qquad C_{J,\mathrm{free}}=\frac{2}{d-2}.
}
The negative sign for the squared OPE coefficient is a consequence of our convention for the conformal blocks, which differs by a factor $(-2)^\ell$ from the conventions given in \cite{Dolan:2000ut}.

\subsection{Conformal bootstrap}

Before we move on to discuss conformal field theories with small expansion parameters, let us briefly review the developments within non-perturbative conformal bootstrap. The mainstream numerical approach relies on writing the crossing equation~\eqref{eq:crossing} as
\beq{\label{eq:bootstrapnumeq}
\sum_{\Delta,\ell}c^2_{\phi\phi\O_{\Delta,\ell}}F_{\Delta,\ell}(u,v)=0, \qquad F_{\Delta,\ell}(u,v)=v^{\Delta_\phi}G_{\Delta,\ell}(u,v)-u^{\Delta_\phi}G_{\Delta,\ell}(v,u).
}
The interpretation is that the left-hand side consists of a convex hull of vectors in an infinite-dimensional vector space spanned by the functions $F_{\Delta,\ell}(u,v)$. By acting with functionals $\mathcal F$ on \eqref{eq:bootstrapnumeq}, one derives strict bounds on the dimensions and spins of the exchanged operators. Typically, these functionals consist of acting with derivatives at the crossing symmetric point
\beq{\label{eq:derivativefunctional}
\mathcal F_{p,q}(F_{\Delta,\ell})=\left.\frac{\de^p}{\de z^p}\frac{\de^q}{\de \zb^q}F_{\Delta,\ell}(z,\zb)\right|_{z=\zb=\frac12}.
}
This idea was presented in 2008 by Rattazzi, Rychkov, Tonni and Vichi \cite{Rattazzi2008} and has since led to numerous applications and refinements, as reviewed in \cite{Poland:2018epd}. Important early results were the determination of 3d Ising exponents \cite{ElShowk2012}, including $c$ minimization \cite{El-Showk:2014dwa} and a set of universal bounds in 4d theories \cite{Poland:2011ey}. The framework was subsequently applied to supersymmetric theories \cite{Beem:2014zpa,Beem:2013qxa} and extended to systems of mixed correlators, the latter leading to high precision results for the 3d $\OO N$ models \cite{Kos:2016ysd}, and particularly high precision in the cases of Ising \cite{Kos:2016ysd} and $\OO 2$ \cite{Chester:2019ifh}. There have also been implementations for non-scalar external operators such as fermions \cite{Iliesiu:2015qra} and vector currents \cite{Reehorst:2019pzi}. Finally, let us mention the paper \cite{Cappelli:2018vir} which studies the Ising model in interpolating dimensions along the top curve of figure~\ref{fig:operahouse}, making interesting observations on the interplay between the large spin expansion (valid for $d>2$) and the 2d Virasoro symmetry.

Apart from the mainstream numerical bootstrap, other numerical techniques have been developed. Gliozzi \cite{Gliozzi:2013ysa} proposed a truncation method where the idea is to search for approximate solutions to crossing using only a small set of conformal primary operators. The method does not rely on the positivity of the squared OPE coefficients $c^2_{\phi\phi\O}$ and therefore also applies to non-unitary theories such as the Yang--Lee edge singularity.

More recently, analytic functionals have been developed, which replace the numeric functionals \eqref{eq:derivativefunctional}. By varying these functionals, one can get constraints on the spectrum for one-dimensional CFTs \cite{Mazac:2018mdx,Mazac:2018ycv}, generating an interesting relation to the problem of sphere packings \cite{Hartman:2019pcd}. Some generalisations to higher dimensions have also been made \cite{Paulos:2019gtx,Mazac:2019shk}.

\section{Perturbative structure of conformal field theories}
\label{sec:pertCFT}

So far, we have discussed the structure of the OPE and the conformal block decomposition in a generic conformal field theory. We now focus the discussion onto CFTs which admit a small expansion parameter $g$. We will from time to time refer to the expansion in $g$ as a perturbative expansion, but it does not need to be a coupling constant in the traditional, Lagrangian, sense. Indeed, for $g$ we can take the $\epsilon=d_{\mathrm c}-d$ in an $\epsilon$ expansion, $1/N$ in a planar expansion, or $1/\lambda$ for the 't Hooft coupling in strongly coupled holographic CFTs. However, we will assume that $g=0$ corresponds to \emph{twist degeneracy}, namely that all operators in a twist family has identical twist. We keep $g$ as a generic name and assume that all quantities in the theory of consideration admit expansions in powers of $g$. In general, such series might be asymptotic and may need to be complemented by non-perturbative corrections.

Let us focus on the contribution within the four-point function from a single twist family. We assume that there is one operator $\O_\ell$ for each even spin\footnote{The spin takes either even or odd values, depending on the transformation properties under the global symmetry group.}. This means that we can parametrise the CFT-data in terms of the spin
\beq{\label{eq:defCFTdata}
\Delta_\ell=\tau_0+\ell+\gamma_\ell ,\quad c^2_{\phi\phi\O_\ell}=a_\ell,
}
where the \emph{anomalous dimensions} $\gamma_\ell$ are of order $g$. Here we have introduced a \emph{reference twist} $\tau_0$. A natural choice of reference twist is $\tau_\infty$, chosen such that $\gamma_\ell\to0$ as $\ell\to\infty$, but other choices are also allowed as long as they are consistent with anomalous dimensions of order $g$.

In \eqref{eq:defCFTdata} we also introduced the notion $a_\ell$ to denote the squared OPE coefficients. As such, they are positive in unitary theories. By abuse of notation we will often refer to the $a_\ell$ as just the OPE coefficients. We assumed that both $a_\ell$ and $\gamma_\ell$ admit expansions in $g$, so we write
\beq{
a_\ell=g^{\alpha}\parr{a^{(0)}_\ell+ga^{(1)}_\ell+g^2a^{(2)}_\ell+\ldots}, \quad \gamma_\ell=g\gamma_\ell^{(1)}+g^2\gamma_\ell^{(2)}+\ldots,
}
where we have taken out a possible overall factor $g^{\alpha}$. Inserting this in the conformal block expansion \eqref{eq:CBexp} and expanding in the collinear limit gives
\beq{\label{eq:collineartwistfamily1}
\sum_\ell a_\ell G^{(d)}_{\Delta_\ell,\ell}(u,v)=z^{\frac{\tau_0}2}\sum_\ell a_\ell z^{\frac{\gamma_\ell}2}k_{\hb+\frac12\gamma_\ell}(\zb)+O(z^{\frac{\tau_0}2+1}),
}
where $\hb=\frac{\tau_0}2+\ell$. For reasons that will become clear in section~\ref{sec:conformalspin}, we refer to $\hb$ as the \emph{bare conformal spin}, often omitting the word ``bare''. Expanding each term in \eqref{eq:collineartwistfamily1} in powers of $g$, we can write the sum as
\beqa{\label{eq:ordergexpansion}
&\sum_\ell a_\ell G^{(d)}_{\Delta_\ell,\ell}(u,v)
\\
&\qquad=z^{\frac{\tau_0}2}g^{\alpha}\sum_\ell \parr{a^{(0)}_\ell+g \Big[ a^{(1)}_\ell+\frac 12a^{(0)}_\ell\gamma_\ell^{(1)}(\log z+\de_\hb)\Big]+O(g^2)}k_\hb(\zb)+O(z^{\frac{\tau_0}2+1}),
\nonumber
}
where evaluation at $\hb=\frac{\tau_0}2+\ell$ is understood. An important observation is that terms proportional to $\log z$ are multiplied by the leading order anomalous dimensions $\gamma^{(1)}_\ell$. Similarly, higher powers of $\log z$ will have leading terms corresponding to higher powers of the anomalous dimensions. At subleading orders in $g$, the contributions from anomalous dimensions and OPE coefficients are mixed and need to be resolved. In section~\ref{sec:one-loop.correlator.into.H} we describe a straight-forward way of resolving this at leading order in $g$, by introducing a shift in the OPE coefficients. In the majority of this thesis we will instead make use of the formula \eqref{eq:aellfromT}, or rather \eqref{eq:aellfromU}, which more transparently generalises to arbitrary orders.

\subsection{Generalised free field theory}
We discussed above a family of operators $\O_\ell$ parametrised by spin $\ell$. A natural place where such operators appear is in what is called the \emph{generalised free field} (GFF) theory, which is also known as ``mean field theory''. It is the theory of a single non-interacting field $\phi$ with arbitrary dimension $\Delta_\phi$. Correlators are computed via pair-wise Wick contractions, using the CFT two-point function \eqref{eq:scalaroperatorsnormalisation}
\beq{
\expv{\phi(x_1)\phi(x_2)}=\parr{\frac1{x_{12}}}^{2\Delta_\phi}.
}
This means that the four-point function of $\phi$ can be constructed from three contributions, corresponding to $s$-channel, $t$-channel and $u$-channel contractions. Normalising with respect to the $s$-channel, in agreement with \eqref{eq:fourpointsetup}, we get that the four-point function takes the form
\beq{
\G(u,v)=1+u^{\Delta_\phi}+\parr{\frac uv}^{\Delta_\phi}.
}
From this expression we can perform a conformal block decomposition. The powers of $u$ indicate that there must be operators of twist $\tau=2\Delta_\phi+2n$ for integer $n$, which we refer to as GFF operators and denote by $[\phi,\phi]_{n,\ell}$. We write the OPE schematically as
\beq{\label{eq:GFFOPEschem}
\phi\times\phi=\1+\sum_{n,\ell}[\phi,\phi]_{n,\ell},
}
and denote the corresponding OPE coefficients $a_{n,\ell}^{\mathrm{GFF}}$. These OPE coefficients, which were worked out in two and four dimensions in \cite{Heemskerk2009} and in full generality in \cite{Fitzpatrick:2011dm}, take the form
\beq{\label{eq:aGFF}
\hspace{-2pt}a^{\mathrm{GFF}}_{n,\ell}|_{\Delta_\phi}=\frac{2(\Delta_\phi+1-\mu)_n^2(\Delta_\phi)_{n+\ell}^2}{\ell !\, n!\, (\ell+\mu)_n(2\Delta_\phi+n+1-2\mu)_n(2\Delta_\phi+\ell+n-\mu)_n(2\Delta_\phi+2n+\ell-1)_\ell}
}
for $\mu=d/2$.

Despite having well-defined scaling dimensions, correlators and OPE, the generalised free field theory is not a local conformal field theory, since, unless $\phi$ is the free scalar ($\Delta_\phi=\mu-1$) the theory lacks a stress tensor. Such a theory is sometimes called a conformal theory. However, the GFF theory is often a useful tool in understanding CFTs. Firstly, its operator content is exactly dual to freely propagating fields in AdS, where the lack of stress tensor signals the lack of gravitational interaction. Secondly, the spectrum  of GFF operators and the OPE \eqref{eq:GFFOPEschem} is a useful starting point in describing spectra of many CFTs, as we will see in the examples below.

\subsection{Operators, labels and mixing}
\label{sec:mixing}

In order to discuss the spectra of actual CFTs we need to introduce some language to precisely describe the primary operators in a conformal block expansion. Even in the cases where we can construct primary operators from the fundamental fields of the Lagrangian, it is often too cumbersome to write down the explicit form of these operators, since this involves projecting away terms that are descendants of other primaries. We have already noted that operators come in twist families labelled by some $\tau_0$ and that the OPE expansions sometimes involve the whole tower of GFF operators $[\phi,\phi]_{n,\ell}$. We will build on these ideas to formulate some universal naming conventions which could be used to describe any perturbative CFT. The name of a primary operator should be as short as possible but still carry the essential information about that operator, such as its spin and its belonging to a twist family.  At the same time, our conventions need to be flexible enough in order to describe a variety of theories, which means that sometimes a given operator may be assigned several different names. We adopt the following conventions.

\begin{deff}
We use the following types of symbols to denote primary operators. Each symbol comes with a convention for the reference twist $\tau_0$ of the twist family the that the operator belongs to.
\begin{itemize}
\item Unique names for scalar operators, generically $\O_1$, $\O_2$ etc., where the scaling dimension is denoted $\Delta_{\O_1}$, $\Delta_{\O_2}$ etc. Some of these operators are referred to as fields, or fundamental fields, since in a Lagrangian description they correspond to fields integrated over in the path integral. In that case, we define the anomalous dimension of these operators as the difference between the scaling dimension and the canonical dimension: $\Delta_{\O}=\Delta_\O^{(0)}+\gamma_\O$. We often use the letter $\phi$ in the case of weakly coupled scalar fields with $\Delta^{(0)}=\mu-1$, but we will sometimes let $\phi$ denote a generic external operator without any assumptions about its scaling dimension.
\item Universal names for conserved currents, $T^{\mu\nu}$ and $J^\mu$, as well as weakly broken higher spin currents $\mathcal J_\ell$ with $\tau_0=d-2+O(g)$, i.e.\ near the unitarity bound.
\item Composite operators (or, in the terminology of \cite{Fitzpatrick:2011dm}, conglomerate operators) written as $\square^n\de^\ell\O_1^{k_1}\O_2^{k_2}\cdots$~\footnote{This should be read as the following: An operator constructed from $2n$ contracted and $\ell$ uncontracted gradients acting on $k_1$ operators $\O_1$ etc., distributed in such a way that it is not a descendant.}, with reference twist $\tau_0=2n+k_1\Delta_{\O_1}+k_2\Delta_{\O_2}+\ldots$. When there are exactly two fields involved we write the derivatives between the operators.
\item GFF operators $[\O_1,\O_2]_{n,\ell}$ with $\tau_0=\Delta_{\O_1}+\Delta_{\O_2}+2n$. In composite operator notation we would write $\O_1\square^n\de^\ell\O_2$.
\end{itemize}
The choice of $\tau_0$ is arbitrary up to terms of order $g$. We will often be explicit with the choice made for a given twist family, especially when using conventions which do not agree with the definitions above. The consequence of changing reference twist is simply an order $g$ redefinition of anomalous dimension.

In the presence of global symmetry, operators transform in irreducible representations of the global symmetry group. We then add a label $R$ denoting the irrep, and write the corresponding operators names as $\mathcal J_{R,\ell}$, $[\O_1,\O_2]_{R,n,\ell}$, $(\square^n\de^\ell\phi^k)_{R}$ etc.
\end{deff}
\vspace{2.5ex}

\noindent Let us now describe the important concept of operator mixing. The existence of mixing arises naturally from the following considerations. 
By the naming conventions above, we may parametrise all operators in a theory by their reference twist $\tau_0$ and spin $\ell$. Assuming this, and focussing on a given twist family, the conformal block decomposition of a correlator and a re-expansion in $g$ would generate a sum like \eqref{eq:ordergexpansion}
\beq{\label{eq:expwithoutmixing}
\sum_{\ell} a_\ell G^{(d)}_{\Delta_\ell,\ell}(u,v)=z^{\frac{\tau_0}2}\sum_\ell \parr{\! a_\ell\! +\! \frac 12a_\ell\gamma_\ell(\log z\! +\! \de_\hb)\! +\! \frac 18a_\ell\gamma_\ell^2(\log z\! +\! \de_\hb)^2\! +\! \ldots}k_\hb(\zb)\! +\! \ldots.
}
The conformal blocks now depend only on $\tau_0$ and $\ell$, and the expansion is blind to any additional information about the involved operators. In particular, there may exist $d_\ell$ different degenerate operators with equal $\tau_0$ and $\ell$. By our naming convention, such operators would share the name, say $\O_\ell$. To distinguish them we need to employ an additional label and write $\O_{\ell,i}$, $i=1,\ldots d_\ell$.
In the expansion above we define
\beq{\label{eq:mixingdef}
\expv{a_\ell\gamma_\ell^p}=\sum_{i=1}^d a_{\ell,i}\gamma_{\ell,i}^p\,,
}
by which \eqref{eq:expwithoutmixing} takes the form
\beq{
z^{\frac{\tau_0}2}\sum_\ell \parr{\expv{a_\ell}+\frac 12\expv{a_\ell\gamma_\ell}(\log z+\de_\hb)+\frac 18\expv{a_\ell\gamma_\ell^2}(\log z+\de_\hb)^2+\ldots}k_\hb(\zb)+\ldots.
}

Mixing has some severe consequences. For instance, without mixing, knowing $\langle {a^{(0)}_\ell} \rangle$ and $\langle{a^{(0)}_\ell\gamma^{(1)}_\ell}\rangle$ would give access to all $\langle{a^{(0)}_\ell(\gamma^{(1)}_\ell)^p}\rangle$, i.e.\ to the leading power of $\log z$ at all orders in perturbation. With degenerate operators, the mixing must be resolved before one can compute even the sum of anomalous dimensions squared. Resolving the mixing problem in a given theory requires knowledge of the individual anomalous dimensions and/or considerations of mixed correlators and it is, in general, a difficult task.

\subsection{Spectrum of the Wilson--Fisher model}\label{sec:WFspectrum}

As a first example of a fully interacting conformal field theory, we review the spectrum of the Wilson--Fisher (WF) model in $d=4-\epsilon$ dimensions \cite{Wilson:1971dc,Wilson:1973jj}. Here $\epsilon$ serves as the expansion parameter $g$. As discussed briefly in the introduction in connection to figure~\ref{fig:operahouse}, one can view this CFT as the IR fixed-point of a short RG flow starting from the theory with a free scalar field $\phi$ perturbed by a quartic interaction $\lambda\phi^4$. At the fixed-point, $\lambda$ takes a value of order $\epsilon$. Another point of view is that the $\epsilon$ expansion follows from a limit of a family of conformal field theories non-perturbatively defined in $d$ dimensions---the $d$-dimensional Ising model---which approaches the free theory as $d\to4$.
A more concrete description follows from studying the multiplet recombination induced by the equation of motion $\square\phi\propto\phi^3$. This equation generates $\phi^3$ as a descendant of $\phi$, and it was shown in \cite{Rychkov:2015naa} how this simple statement can be used to deduce several properties of the Wilson--Fisher fixed-point.

We focus on the operators that appear in the conformal block decomposition of the four-point function \beq{
\label{eq:WFcorrelator}
\G(u,v)=x_{12}^{2\Delta_\phi}x_{34}^{2\Delta_\phi}\expv{\phi(x_1)\phi(x_2)\phi(x_3)\phi(x_4)}.
}
Here $\Delta_\phi=\mu-1+\gamma_\phi=1-\frac\epsilon2+\epsilon^2\gamma_{\phi}^{(2)}+\ldots$, where $\gamma_\phi^{(2)}=\frac1{108}$ and where we have indicated the well-known fact that $\phi$ has no anomalous dimension at order $\epsilon$~\footnote{From a Lagrangian point of view, $\gamma_\phi^{(1)}=0$ corresponds to the fact that there is no one-loop field renormalisation.}. Other conformal primaries in the theory can be explicitly constructed from $\phi$ and $\de^\mu$, and are defined up to contribution from descendants. Due to the global $\Z_2$ symmetry $\phi\mapsto-\phi$, only $\Z_2$ even operators, constructed from an even number of fields, appear in the OPE.

The scaling dimensions of $\phi$, $\phi^2$ and $\phi^4$ can be computed from standard dimensional regularisation, where the coupling is evaluated at the fixed-point. The dimensions of the first two operators $\phi$ and $\phi^2$ are often presented in terms of a pair of critical exponents, such as $\eta$ and $\nu$ using the relations $\eta=2\Delta_\phi-d+2=2\gamma_\phi$ and $\nu^{-1}=d-\Delta_{\phi^2}$~\footnote{Other critical exponents for the Ising model can be related to $\eta$ and $\nu$ through scaling relations, see e.g.\ \cite{Pelissetto:2000ek}. The exception is the exponent $\omega$, defined through $\omega=\Delta_{\phi^4}-d$.}. They were computed to order $\epsilon^4$ soon after the WF model was proposed \cite{Brezin:1974eb} in order to generate estimates for the critical exponents of the 3d Ising model, and have since been computed to order $\epsilon^7$ \cite{Schnetz:2016fhy}\footnote{The results for the Ising exponents were not added to \cite{Schnetz:2016fhy} until after the $\epsilon^6$ results appeared in \cite{Kompaniets:2017yct}. I thank Erik Panzer for making me aware of \cite{Schnetz:2016fhy} and providing me with the explicit results for future reference.}.

At leading twist, the OPE $\phi\times\phi$ contains weakly broken currents $\mathcal J_\ell=\phi\de^\ell\phi$, with \beq{
\label{eq:gammaWFlitt}
\gamma_\ell=-\frac{\epsilon^2}{9\ell(\ell+1)}+O(\epsilon^3)} as derived in \cite{Wilson:1972cf,Wilson:1973jj}. In \cite{Derkachov:1997qv} they were computed to order $\epsilon^4$ and we provide an independent computation in chapter~\ref{ch:paper2} based on \cite{Paper2}.
\begin{figure}
\centering
\includegraphics[width=\textwidth]{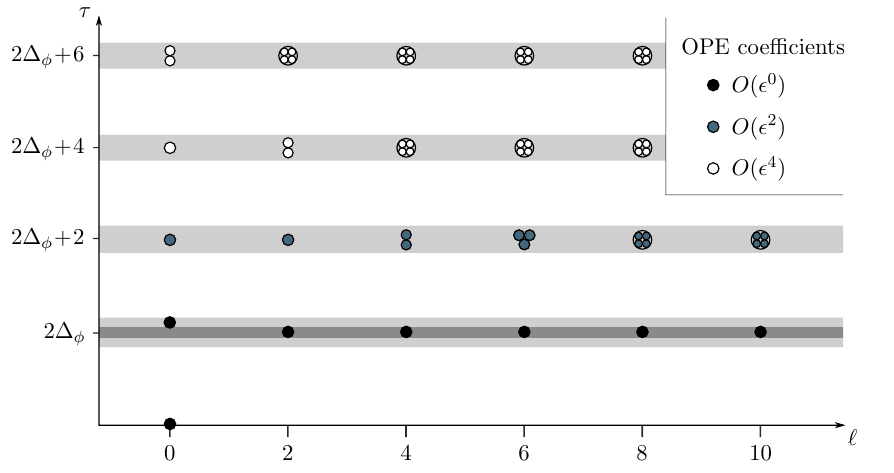}
\caption[Operators, labels and mixing in the Wilson--Fisher fixed-point in the $\epsilon$ expansion.]{Operators, labels and mixing in the Wilson--Fisher fixed-point in the $\epsilon$ expansion. We display the number of degenerate operators as one, two, three or more than three.}\label{fig:spectrumWF}
\end{figure}
At higher reference twist mixing occurs and in figure~\ref{fig:spectrumWF} we illustrate the spectrum of operators in the OPE decomposition of the four-point function \eqref{eq:WFcorrelator}. The identity operator $\Delta=0$, $\ell=0$ and the operators $\phi\de^\ell\phi$ are the only operators with (squared) OPE coefficients at leading order, as illustrated by the black dots. In the figure we have indicated the anomalous dimensions in terms of grey bands of width $\epsilon$ and $\epsilon^2$, centred around twists $2\Delta_\phi+2n$, $n\in\mathbb N$. Of the bilinear operators $\phi\de^\ell\phi$, the scalar ${\phi^2}$ is the only one that has an anomalous dimension at order $\epsilon$. The positions of the grey bands, as well as the corresponding ones for the theories we consider below, depend on which four-point function we study and will be very important when we develop the analytic bootstrap approach later. It is instructive to compare figure~\ref{fig:spectrumWF} with figure~1 of \cite{Simmons-Duffin:2016wlq}, which gives a similar display of the operator spectrum in the 3d Ising model as found by the numerical bootstrap.

Now we take a look at the operators of higher twist. Since $\Delta_\phi=1+O(\epsilon)$, and the $\Z_2$ symmetry enforces an even number of fields $\phi$, all operators in the $\phi\times\phi$ OPE will have twists of the form $\tau=2\Delta_\phi+2n+O(\epsilon)$ for $n=0,1,2,\ldots$. It is possible to compute the order $\epsilon$ anomalous dimension of an arbitrary operator of this kind \cite{Kehrein:1992fn}, and in \cite{Kehrein:1994ff} the spectrum was systematically investigated. The leading anomalous dimensions of arbitrary composite operators may also be computed using conformal perturbation theory \cite{Hogervorst:2015akt}\footnote{See appendix~C of \cite{Hogervorst:2015tka} for more details, and \cite{Liendo:2017wsn} for an alternative method based on the multiplet recombination. I thank M.~Hogervorst and P.~Liendo for detailed discussions on these two methods.}.

At $n=1$, the operators take the schematic form $\de^\ell\phi^4$, and are now subject to mixing. At spins $0$ and $2$ there is a unique operator, but at spin $4$ there are two different operators, with dimensions $\Delta_{(\de^4\phi^4)_1}=8-2\epsilon+\frac49\epsilon+O(\epsilon^2)$ and $\Delta_{(\de^4\phi^4)_2}=8-2\epsilon+\frac{15}{19}\epsilon+O(\epsilon^2)$. The degeneracy keeps growing for each subsequent spin and in figure~\ref{fig:spectrumWF} we only indicate the precise degeneracy $d_\ell$ when $d_\ell<4$.

At higher $n$ the situation is even more complicated, with mixing between operators of different number of fields, for instance $\phi^8$ and $\square^2\phi^4$~\footnote{Recall that we discuss mixing here in meaning of having equal reference twist, in the context of our discussion in section~\ref{sec:mixing}. In constructing the explicit form of the conformal primary operators, there is no order $\epsilon$ mixing between operators with different number of fields.}. The only non-degenerate point for $n\geqslant2$ is $n=2$, $\ell=0$, where the operator is $\phi^6$ with $\Delta_{\phi^6}=6+2\epsilon+O(\epsilon^2)$. This is the only point with $n\geqslant1$ where no operator of the form $\de^\ell\square^{n-1}\phi^4$ takes part in the mixing, a fact that will have an interesting consequence in section~\ref{sec:WFfrompaper1}.

\subsection[Spectrum of \texorpdfstring{$\mathcal N=4$}{N=4} SYM at weak coupling]{Spectrum of $\boldsymbol{\mathcal N=4}$ SYM at weak coupling}

\label{sec:NN4spectrum}

In preparation for chapter~\ref{ch:paper1}, we give a short description of the weak coupling spectrum of operators in the \NN4 supersymmetric Yang--Mills (SYM) theory in four dimensions. Due to its properties as a highly complex but still well-structured theory, the literature on the topic is vast and we will not be able to review it. Here we only give the minimum amount of information needed to use the theory as a test and prototype for the methods of chapter~\ref{ch:paper1}.

The \NN4 SYM theory is the maximally supersymmetric quantum field theory in four dimension. The field content consists of one vector multiplet in the adjoint representation of the gauge group, which we will take to be $\SU N$. The vector multiplet contains a gauge field $A_\mu \in \boldsymbol 1$, four Majorana spinors $\lambda^i\in \boldsymbol 4$ and six real scalars $\Phi^I\in\boldsymbol 6$, transforming in the indicated irreps of the R-symmetry $\SU4\cong \SO6$. The theory has an exactly marginal coupling $g_{\mathrm{YM}}$, and is thus conformal at all values of this coupling. Here we look at the weak coupling limit and define
\beq{
g=\frac{g_{\mathrm{YM}}^2N}{4\pi^2}
}
as our expansion parameter.

\begin{figure}
\centering
\includegraphics[width=\textwidth]{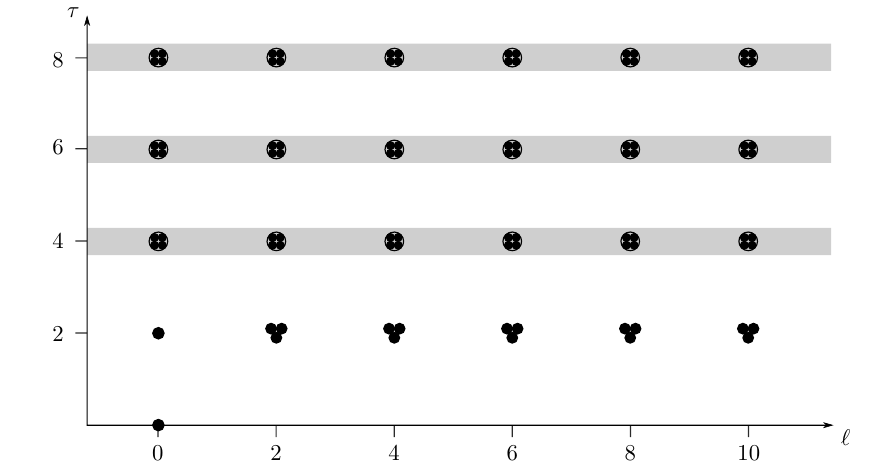}
\caption[Operators, labels and mixing in \NN4 SYM at weak coupling.]{Operators, labels and mixing in \NN4 SYM at weak coupling, displaying operators in the $\boldsymbol 1$ representation. At each even spin and near each even twist, there are operators with OPE coefficients at order $g_{\mathrm{YM}}^0$.} \label{fig:spectrumNN4}
\end{figure}

Conformal primary operators are constructed from gauge-invariant combinations of the fields, and transform in irreducible representations of the R-symmetry. Supersymmetry further groups the operators into supermultiplets, labelled by superconformal primaries. The conformal primaries are generated from the superconformal primaries by the supersymmetry generators and transform in Lorentz and R-symmetry irreps related to their superprimaries. In particular, the scaling dimensions are related and all superdescendants share the same anomalous dimensions. 

The theory contains a number of superconformal primary operators whose dimensions are protected by supersymmetry, typically on some BPS-bound. They are referred to as short multiplets due to various shortening conditions, which means that a fraction of the operator content of these multiplets is annihilated. In addition to the short multiplets, the theory contains long multiplets with unprotected scaling dimensions. A detailed presentation of the various supermultiplets can be found in \cite{Dolan:2002zh}.

We will be interested in four-point correlators of the simplest possible scalar operators, which are constructed from bilinears in $\Phi$, with a $\SU N$ trace to render them gauge-invariant. There are two such operators: the Konishi operator $\mathcal K=\Tr(\Phi^I\Phi^I)\in\boldsymbol 1$, and half-BPS operator $\O_{\boldsymbol{20'}}=\Tr(\Phi^{\{I}\Phi^{J\}})\in \boldsymbol{20'}$, where the latter is the rank two traceless symmetric representation of $\SO6$.

The Konishi operator is the superconformal primary of a long multiplet and has scaling dimension $\Delta_{\mathcal K}=2+3g+O(g^2)$. Its anomalous dimension is known to order $g^4$ \cite{Velizhanin:2009gv}, and non-perturbatively in the planar limit $N\to\infty$ \cite{Gromov:2009zb}. The operator $\O_{\boldsymbol{20'}}$ with $\Delta_{\O_{\boldsymbol{20'}}}=2$ is the superconformal primary of the short supermultiplet, which in addition contains amongst others the stress tensor and the R-symmetry currents. In the OPE decomposition of the Konishi four-point function, only R-symmetry singlets contribute, whereas the decomposition of the $\O_{\boldsymbol{20'}}$ four-point function contains operators in all $\SU4$ irreps in the tensor product
\beq{
\boldsymbol{20'}\otimes\boldsymbol{20'}=\boldsymbol{1}\oplus\boldsymbol{15}\oplus\boldsymbol{20'}\oplus\boldsymbol{84}\oplus\boldsymbol{105}\oplus \boldsymbol{175},
}
where we used the notation of \cite{Slansky:1981yr} for the  irreps. Since both correlators contain R-symmetry singlets, we will focus on them. In fact, the only unprotected superconformal primaries in the $\O_{\boldsymbol{20'}}$ four-point function are in the singlet representation, which means that the singlet representation contains all dynamical information of the perturbative correlator.

Figure~\ref{fig:spectrumNN4} contains a plot similar to figure~\ref{fig:spectrumWF}, displaying the singlet conformal primaries, where we have shaded regions within order $g$ from $\tau=4+2n$, $n\in \mathbb N$. At the leading twist, there are three conformal primaries at each even spin, denoted $T_\ell$, $\Sigma_\ell$ and $\Xi_\ell$. They have anomalous dimensions $\gamma=\Delta-(2+\ell)$ of the form
\beq{\label{eq:anomalousdimensionsLTSYM}
\gamma^{(1)}_{T_\ell}=2S_1(\ell-2),\qquad \gamma^{(1)}_{\Sigma_\ell}=2S_1(\ell)\qquad \gamma^{(1)}_{\Xi_\ell}=2S_1(\ell+2),
}
where $S_1(n)=\sum_{k=1}^{n}\frac1n$ denotes the harmonic numbers, defined in appendix~\ref{integrals}. These operators, called leading twist operators, or twist-2 operators, follow from diagonalisation of the one-loop perturbative anomalous dimension in the space of bilinears in $\Phi$, $\lambda$ and $F\munu$ respectively \cite{Kotikov:2001sc,Anselmi:1998ms}\footnote{More details can be found in \cite{Velizhanin:2014zla}, where the matrix elements of the one-loop dilatation operator are given. Notice, however, a typo in that paper; the proper form is $\gamma_{\lambda\lambda}^{(0)}=-4S_1(j)+8/j-8/(j+1)$.}. At $\ell=0$, the operator is non-degenerate: $\Xi_0=\mathcal K$. $T_2$ is the stress tensor. In fact, the operators belong to superconformal multiplets in groups of three, which explains why the anomalous dimensions are related to the universal function $\gamma_{\mathrm{univ.}}(\ell)=2S_1(\ell)$. In the four-point function of the Konishi operator they appear with an average given by

\beq{\label{eq:averagesKonishi}
\expv{a_\ell\gamma_\ell}=
\sum_{\O_\ell=T_\ell,\Sigma_\ell,\Xi_\ell}a_{\O_\ell}\gamma_{\O_\ell}
=\frac{2c\Gamma(\ell+1)^2}{\Gamma(2\ell+1)}\gamma_{\mathrm{univ.}}(\ell)+2c\, 3\delta_{\ell,0},
}	
where $c=\frac2{3(N^2-1)}$.

The operators just discussed constitute the leading twist family in \NN4 SYM, but let us emphasise that they are not double-twist operators. Since they have twist below the double twist of the external operator, they lie outside the grey bands displayed in figure~\ref{fig:spectrumNN4}.
At the double twist, as well as at higher twists, a large number of operators contribute, and to resolve the mixing problem is a difficult task. 
This was, however accomplished in the four-point function of $\O_{\boldsymbol{20'}}$ case in \cite{AldayBissi2017,Aprile:2017bgs,Aprile:2017xsp} and more generally in \cite{Aprile:2018efk}.

\subsection[Spectrum of the critical \texorpdfstring{$\OO N$}{O(N)} model]{Spectrum of the critical $\boldsymbol{\OO N}$ model}
\label{sec:ONspectrum}

As a final example, let us discuss the spectrum of the \emph{critical $O(N)$ model}, where $\OO N$ denotes the orthogonal group. This theory is a generalisation of the Ising model and admits a $4-\epsilon$ expansion with a similar Lagrangian $\lambda\phi^4\rightsquigarrow \lambda(\varphi^i\varphi^i)^2$, where $i$ runs from $1$ to $N$. Likewise, in three dimensions the critical model follows from a long RG flow from the theory of $N$ free scalars, and describes a range of interesting critical phenomena \cite{Pelissetto:2000ek}. However, it is possible to treat the number of fields $N$ as an additional parameter of the theory, and indeed many observables can be seen as analytic functions of $N$. This group parameter expansion  has been common practice for a long time and was recently put on more firm ground using Deligne categories \cite{Binder:2019zqc}. Thanks to the continuation in $N$, the theory admits various overlapping perturbative limits, displayed in figure~\ref{fig:ONlimits}, which we will now describe. 
\begin{figure}
\centering
\includegraphics{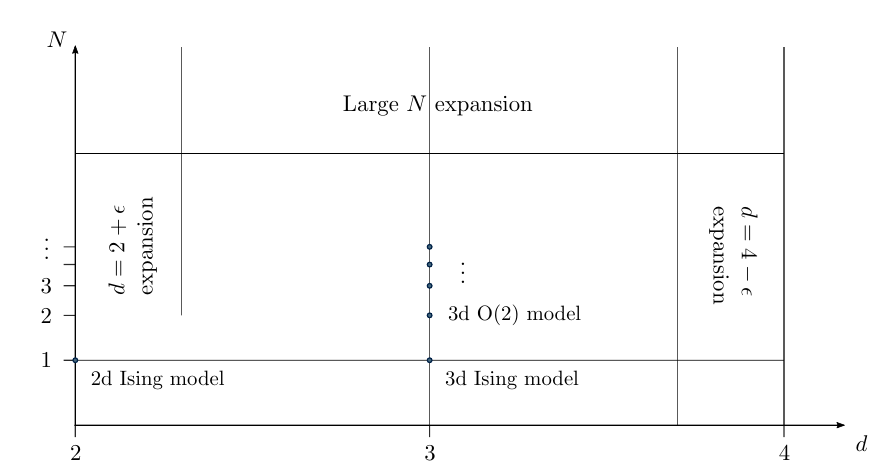}
\caption[Various limits for the critical $\OO N$ model.]{Various limits for the critical $\OO N$ model. The CFT-data agrees in the overlapping regions between two expansion limits.}\label{fig:ONlimits}
\end{figure}
Important for this thesis are the $4-\epsilon$ expansion and the large $N$ expansion, which we will discuss shortly. In addition, there is an expansion in $d=2+\epsilon$ dimensions, where the critical $\OO N$ model for $N>2$ is related to the UV fixed-point of a non-linear sigma model with target space $\OO N$, see e.g.\ chapter~31 of \cite{zinn2002quantum}. In that expansion, anomalous dimensions \cite{Giombi2016} and central charges \cite{Diab:2016spb} have been computed in a series in $\epsilon$. The behaviour near $N=d=2$ is not fully understood, and the limits $d\to2$ and $N\to2$ do not appear to commute \cite{Cardy:1980at}\footnote{I thank Slava Rychkov for mentioning this reference to me.}. The large $N$ expansion can be continued beyond $d=4$ to match a cubic model of $N+1$ fields in $d=6-\epsilon$ dimensions \cite{Fei:2014yja}, where perturbative CFT-data is known \cite{Giombi2016}. Unitarity in the five-dimensional theory a disputed topic, see \cite{Giombi:2019upv} for a recent discussion taking into account instanton contributions.

In the $4-\epsilon$ expansion, the spectrum of operators in the $\varphi$ four-point function is similar to the $N=1$ case described in section~\ref{sec:WFspectrum}. The $\varphi^i\times\varphi^j$ OPE contains three irreducible representations: singlet ($S$) and rank two traceless symmetric ($T$) and antisymmetric ($A$) tensors, where the latter is odd under $i\leftrightarrow j$ and therefore contains intermediate operators of odd rather than even spin. We focus on the singlet representation, which has the most interesting operator content. In the $\epsilon$ expansion, the spectrum of singlet operators looks similar to figure~\ref{fig:spectrumWF}, with the modification that the degeneracy of higher twist operators grows faster. However, at large $N$ the spectrum shows an interesting behaviour which we will now describe.

It has for long been understood how to develop a Lagrangian description for the critical $\OO N$ model at large $N$ and generic spacetime dimension $d=2\mu$, through the introduction of the Hubbard--Stratonovich auxiliary field $\sigma$, see e.g.\ \cite{Fei:2014yja} for a detailed discussion. This is accomplished by adding to the Lagrangian of $N$ free scalars $\varphi^i$ the interaction terms
\beq{\label{eq:HSextraterms}
S_{\mathrm I}=\int\mathrm d^dx\left(\frac1{2\sqrt N}\sigma\varphi^i\varphi^i-\frac{1}{4\lambda N}\sigma^2\right).
}
One can check that integrating out the field $\sigma$ gives back the usual $\lambda(\varphi^2)^2$ interaction. Alternatively, $\sigma$ can be promoted to a dynamical field and a perturbation theory can be developed with $1/\sqrt N$ as the effective coupling constant, where the second term becomes irrelevant in the IR. The large $N$ expansion of the $\OO N$ model can be used to generate approximate results at finite $N$, but it has also been conjectured to have a holographic dual given by type A Vasiliev theory $hs_4$ \cite{Vasiliev:1995dn} when limiting to correlators of $\OO N$ singlets.

In the spectrum, the operator $\varphi^2_S=\varphi^i\varphi^i$ gets replaced by $\sigma$, which has dimension $\Delta_\sigma=2+O(N^{-1})$. Generic operators are then constructed from $\sigma$, $\varphi^i$ (with $\Delta^{(0)}_\varphi=\mu-1$) and $\de^\mu$. In the overlap between the $4-\epsilon$ expansion and the large $N$ expansion, the operators in the two descriptions are in one-to-one correspondence with each other, and the scaling dimensions agree. For instance, for the first non-trivial singlet scalar we have
\beqa{
\Delta_{\varphi^2_S}&=2-\epsilon+\frac{N+2}{N+8}\epsilon+\frac{(N+2)(44+13N)}{2(N+8)^3}\epsilon^2+O(\epsilon^3), && \text{($\epsilon$ expansion),}
\\\label{eq:Deltasigma}
\Delta_{\sigma}&=2-\frac{4(\mu-1)(2\mu-1)}{2-\mu}\frac{\gamma_\varphi^{(1)}}N+O(N^{-2}), && \text{(large $N$ expansion).}
}
Inserting the literature value
\beq{\label{eq:gammaphival}
\gamma^{(1)}_{\varphi }= \frac{(\mu-2)\Gamma(2\mu-1)}{\Gamma(\mu+1)\Gamma(\mu)^2\Gamma(1-\mu)}
}
we can explicitly check that for $\mu=2-\frac\epsilon2$ both expressions expand to $2-\frac{6\epsilon}N+\frac{13\epsilon^2}{2N}+\ldots$. In table~\ref{tab:operatornames} we list a few operators and give their names in the different expansions, including their conventional names in the 3d Ising model.

In figure~\ref{fig:spectrumON} we display the large $N$ spectrum of $\OO N$ singlet operators, displaying bands corresponding to twists within $2\Delta_\varphi+2n+O(N^{-1})$ for $n\in\mathbb N$. It is clear that $\sigma$ is outside the first of these bands, whereas all the spinning weakly broken currents $\mathcal J_{S,\ell}=\varphi^i\de^\ell\varphi^i$ have twist $2(\mu-1)+O(N^{-1})$ and fall within the first band.

\begin{figure}
\centering
\includegraphics[width=\textwidth]{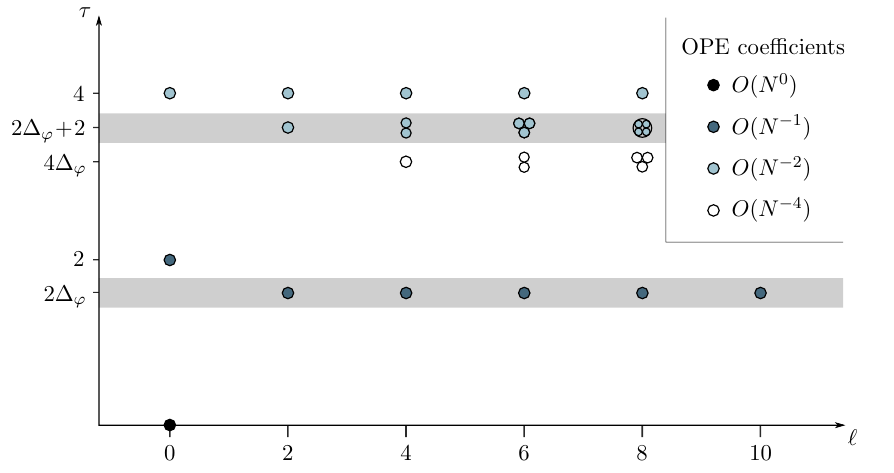}
\caption[Operators, labels and mixing in the critical $\OO N$ model at large $N$.]{Operators, labels and mixing in the critical $\OO N$ model at large $N$, displaying operators in the singlet representation.}\label{fig:spectrumON}
\end{figure}

It is more interesting to look at what happens near the next band, $n=1$. This corresponds to the $\OO N$ version of the $n=1$ band in figure~\ref{fig:spectrumWF}, where now the number of degenerate operators is $d^{\OO N}_\ell=1,2,4,6,8$ instead of $d^{N=1}_\ell=1,1,2,3,4$ for $\ell=0,2,4,6,8$. By using the techniques from \cite{Hogervorst:2015akt,Hogervorst:2015tka} to find the order $\epsilon$ anomalous dimension of a set of operators of the schematic form $\de^\ell\varphi^4_S=\de^\ell\varphi^i\varphi^i\varphi^j\varphi^j$, we can study the fate of these operators upon expanding at large $N$. It turns out that exactly one operator at each spin gets $\tau_\ell=4+O(N^{-1})$, which we interpret as the leading GFF operator $[\sigma,\sigma]_{0,\ell}$. Other (degenerate) operators have $\tau_\ell=2+2\Delta_\varphi+O(N^{-1})$ and we identify them as subleading GFF operators of the form $[\varphi,\varphi]_{S,1,\ell}$. Finally, the remaining (degenerate) operators have $\tau_\ell=4\Delta_\varphi+O(1/N)$, for which we use the notation $\de^\ell\varphi^4_S$.

\begin{table}
\centering
\caption{Some operators with low twist in the critical $\OO N$ model and their conventional names in the different regimes.}\label{tab:operatornames}
{\small
\vspace{4pt}
\renewcommand{\arraystretch}{1.3}
\begin{tabular}{|cclll|}\hline
Irrep & Spin & $\epsilon$ expansion & Large $N$ & 3d Ising ($N=1$)
\\\hline
$S$&$0$&$\1$&$\1$&$\1$
\\\hline
$V$&$0$&$\varphi^i$&$\varphi^i$&$\sigma$
\\\hline
$S$&$0$&$\varphi^2_S=\varphi^i\varphi^i$&$\sigma$&$\epsilon$
\\
$S$&$2$&$\mathcal J_{S,2}=T^{\mu\nu}$&$\mathcal J_{S,2}=T^{\mu\nu}$&$T^{\mu\nu}$
\\
$S$&$4$&$\mathcal J_{S,4}=\varphi^i\partial^4\varphi^i$&$\mathcal J_{S,4}=\varphi^i\partial^4\varphi^i$&$C^{\mu\nu\rho\sigma}$
\\
$S$&$\ell$ even&$\mathcal J_{S,\ell}=\varphi^i\partial^\ell\varphi^i$&$\mathcal J_{S,\ell}=\varphi^i\partial^\ell\varphi^i$&$\O^{\mu_1\cdots\mu_\ell}$
\\\hline
$T$&$\ell$ even&$\mathcal J_{T,\ell}=\varphi^{\{i}\partial^\ell\varphi^{j\}}$&$\mathcal J_{T,\ell}=\varphi^{\{i}\partial^\ell\varphi^{j\}}$& ---
\\\hline
$A$&$\ell$ odd&$\mathcal J_{A,\ell}=\varphi^{[i}\partial^\ell\varphi^{j]}$&$\mathcal J_{A,\ell}=\varphi^{[i}\partial^\ell\varphi^{j]}$& ---
\\\hline
$S$ & $0$ & $\varphi^4_S=(\varphi^i\varphi^i)^2$ & $[\sigma,\sigma]_{0,0}=\sigma^2$ & $\epsilon'$
\\\hline
\end{tabular}
}
\end{table}

Near the band corresponding to $n=2$ in figure~\ref{fig:spectrumWF} the situation is similar, where now operators in the large $N$ limit take one of the following twists: $6$, $2\Delta_\varphi+4$, $4\Delta_\varphi+4$ and $6\Delta_\varphi$. We have omitted it in figure~\ref{fig:spectrumON} to keep the figure less cluttered. Also at higher values of $n$ the situation is similar. The only operators at the bands with $n\geqslant1$ that will play a role in this thesis are the GFF operators $[\sigma,\sigma]_{n,\ell}$. 
In figure~\ref{fig:spectrumON} we have also indicated at what order in $1/N$ the various (squared) OPE coefficients enter within the $\varphi$ four-point function. The essential assumption needed is that $a_\sigma=c_{\varphi\varphi\sigma}^2$ is of order $1/N$. This will in fact imply that all operators $[\sigma,\sigma]_{n,\ell}$ have (squared) OPE coefficients at order $1/N^{2}$.

The $\OO N$ model is also a prototype for $\phi^4$ theories with various global symmetries, some of which admit expansions similar to the large $N$ here. Based on the results in this thesis and in \cite{Paper4}, we can treat all such theories in a unified way, which we describe in section~\ref{sec:paper5short}.

\section{Large spin and the lightcone}\label{sec:largespinlightcone}

The lightcone bootstrap, and therefore large spin perturbation theory, is developed from an interplay between large spin expansions of CFT-data and an expansion of the crossing equation near the double lightcone limit. In this section we will review the development of these ideas to be able to provide a full description of the method in the subsequent sections.

\subsection{Large spin expansion of CFT-data}
\label{sec:conformalspin}

Although the scaling dimension $\Delta$ and the spin $\ell$ are natural labels for primary operators based on the conformal algebra, we have seen that the collinear expansion of the conformal blocks \eqref{eq:collblocks} motivated the introduction of another pair of labels: the twist $\tau=\Delta-\ell$, and the variable $\hb=\frac{\Delta+\ell}2$. The introduction of twist implies that we can parametrise the operators in a given twist family by the spin $\ell$. In this section we will give further motivation for the linear change of variable to $\hb=\frac\tau2+\ell$.

In relation to the experiments on deep inelastic scattering of hadrons, anomalous dimensions were computed for leading twist operators in QCD in the early 1970's. There, as well as in some other theories, the anomalous dimension of operators within a twist family could we written as a function in spin; they were ``analytic in spin''. These functions $\gamma_\ell$ were observed to have some universal properties. For instance, based on high-energy bounds on the scattering cross-section, Nachtmann proved that $\gamma_\ell$ must be an upward convex function, referred to as Nachtmann's theorem or convexity \cite{Nachtmann:1973mr}.

Another empirically motivated result from that time is \emph{reciprocity}, formulated in the context of deep inelastic scattering by Gribov and Lipatov \cite{Gribov:1972rt}. It concerns the large spin expansion of $\gamma_\ell$, i.e.\ the potentially asymptotic expansion around the point $\ell=\infty$. In this limit, the spin dependence of $\gamma_\ell$ was found to come purely through the combination
\beq{
J^2=j(j+1)=\hb(\hb-1)=\Big(\frac{\Delta+\ell}2
\Big)\Big(\frac{\Delta+\ell}2-1\Big).
}
This combination is exactly the eigenvalue of the collinear Casimir \eqref{eq:Dbardeff}, and due to the equivalence on the level of complexified Lie algebras of $\SL2\R$ and the three-dimensional rotation group $\SO3$ it was later referred to as the \emph{conformal spin}\footnote{It is difficult to find the first use of the expression \emph{conformal spin}. The use dates back far, for instance in \cite{Balitsky:1987bk} it is used with clear reference to the collinear Casimir equation. I thank V.M. Braun and A.N. Manashov for discussions on this topic.}~\footnote{Similar to the $\SO3$ case, there is a slight abuse of notation, where both $j$ and $J^2$ are referred to as the conformal spin. Notice, however, that we are now considering a non-compact real form of the algebra, which means that $j$ is no longer restricted to integer values.}. 

The Gribov--Lipatov reciprocity was originally phrased in terms of splitting functions $P(x)$, which are dual to the anomalous dimensions through a Mellin transform\footnote{This is a peculiar use of the word Mellin transform, and it does not agree with the usual definition used in the context of Mellin amplitudes and Mellin space bootstrap. The Mellin transform \eqref{eq:mellinsplitting} is defined up to some regularisation of the $x\to1$ limit, see e.g.\ \cite{Vermaseren:1998uu} for a precise definition.}
\beq{\label{eq:mellinsplitting}
\gamma_j =- \int\limits_0^1 d x \, x^{j-1} P(x),
}
where $x$ is the Bjorken variable. Then reciprocity takes the form 
\beq{
P(x)=-x\, P\parr{\frac1x}.
}
A proof of the equivalence between the two statements can be found in \cite{Basso:2006nk}.
We illustrate the reciprocity principle by the anomalous dimension $\gamma_{\mathrm{univ.}}(\ell)=2S_1(\ell)$ in \NN4 SYM, where now $j=\hb+1=\ell$ to leading order. We have for the harmonic numbers 
\beq{\label{eq:S1expansion}
S_1(j)=\gamma_{\mathrm E}+\log\sqrt{j(j+1)}+\frac1{6j(j+1)}-\frac1{30(j(j+1))^2}+\frac{4}{315(j(j+1))^3}+\ldots,
}
where $\gamma_{\mathrm E}$ is the Euler--Mascheroni constant. The exact form of this and similar expansions at large spin is discussed in \cite{Albino2009}.

Reciprocity was initially thought to be broken at two-loop order in QCD, but the principle was restored by realising that the correct variable to use is the \emph{full conformal spin}, $\hb_{\mathrm f}=\tau_0/2+\ell+\gamma_\ell/2$ rather than the \emph{bare} counterpart $\hb_{\mathrm b}=\tau_0/2+\ell$ \cite{Dokshitzer:1995ev}. It can thus be phrased in the following way.
\begin{prop}\label{prop:reciprocitygamma}
Anomalous dimensions of operators in a twist family with approximate twist $\tau_0$ satisfy the equation
\beq{\label{eq:reciprocitygamma}
\gamma_\ell=\mathsf g\parr{\frac{\tau_0}2+\ell+\frac12\gamma_\ell},
}
where $\mathsf g(\bar{\mathsf h})$ has a large $\bar{\mathsf h}$ expansion symmetric under $\bar{\mathsf h}\leftrightarrow 1-\bar{\mathsf h}$.
\end{prop}
\vspace{2.5ex}

\noindent By the assumption $\gamma_\ell=O(g)$, the relation \eqref{eq:reciprocitygamma} can be studied order by order in $g$ and the expression for $\gamma_\ell$ beyond leading orders will contain derivatives of the function $\mathsf g$. Reciprocity therefore assumes that there exists an analytic continuation in spin making these derivatives well-defined.

The reciprocity relation was subsequently observed to hold in perturbative results at higher order, such as QCD and \NN4 SYM at three loops \cite{Moch:2004pa,Vogt:2004mw} and \NN4 SYM at seven loops in the planar ($N\to\infty$) limit \cite{Marboe:2016igj}. Indeed, reciprocity and the related principle of transcendentality was a leading organisational principle in this work \cite{Beisert:2006ez}. Reciprocity was also observed to persist recursively in other conformal field theories such as the critical $\OO N$ model, in the spirit of proposition~\ref{prop:reciprocitygamma}.

While traditional diagrammatic methods have generated results at high loop order for anomalous dimensions, OPE coefficients are much harder to compute. However, explicit results for correlators at loop order in \NN4 SYM generated OPE coefficients of spinning operators by direct conformal block decomposition \cite{Dolan:2004iy}, for instance in the $\O_{\boldsymbol{20'}}$ correlator at three-loops \cite{Eden:2012rr}. It was realised that the large spin expansion of OPE coefficients has similar properties to the anomalous dimensions \cite{AldayBissi2013}, and a combined reciprocity principle was proven in \cite{AldayBissiLuk2015} for any conformal field theory. We will re-derive this and give a precise statement in theorem~\ref{thm:reciprocity} in section~\ref{sec:reciprocityrevisited}.

\subsection{Lightcone limit and crossing}\label{sec:lightconecrossing}

The relation between the large spin limit of CFT-data and the double lightcone expansion of conformal four-point functions is the key ingredient in this thesis. In the discussion until this point, we have mostly focussed on the whole twist family and the collinear limit $z\to0$. Let us now specialise further and look at the double lightcone limit $z\to0$, $\zb\to1$. This limit emphasises the asymptotic behaviour at large spin of the CFT-data, which corresponds to expansions like \eqref{eq:S1expansion}.

In \cite{AldayMaldacena2007}, this limit was investigated for \NN4 SYM, where the anomalous dimensions of leading twist operators admit the particular expansion \eqref{eq:S1expansion}, which is dominated by the term $\log \ell$~\footnote{The prefactor of this leading logarithm agrees with the \emph{cusp anomalous dimension} and is known at four loops \cite{Henn:2019swt}, and non-perturbatively in the planar limit \cite{Beisert:2006ez}.}. In that paper, configurations corresponding to operators of large spin $\ell$ were analysed in terms of states in an auxiliary theory in $\mathrm{AdS}_{3}\times S^1$ consisting of two particles at a given separation distance $\chi=\log \ell$. In this picture, twists in \NN4 SYM correspond to energies in the auxiliary theory. For leading twist operators, which are single-trace, a flux tube connecting the operators gives rise to an energy linear in $\chi$, explaining the logarithmic scaling of $\gamma_\ell$. Double-trace operators, on the contrary, correspond to configurations of where the interaction energy decays as $E\sim e^{-\alpha\chi}$, where $\alpha$ is equal to the smallest twist in the CFT spectrum: $\alpha=\tau_{\mathrm{min}}$. This gives the generic scaling $\gamma_\ell\sim \ell^{-\tau_{\mathrm{min}}}$.

In two important papers from 2012 \cite{Fitzpatrick:2012yx,Komargodski:2012ek}, the observations from \cite{AldayMaldacena2007} were generalised to arbitrary CFTs and were proved using explicit computations in the double lightcone limit. In \cite{Komargodski:2012ek} connections were made between the picture of \cite{AldayMaldacena2007} and the older results from deep inelastic scattering and Nachtmann's theorem. In \cite{Fitzpatrick:2012yx} a more direct approach was taken, and the results were then related to physics in AdS, noting that in any CFT, even beyond the usual holographic limits, double-trace operators for sufficiently large spin can be interpreted as states which correspond to two disjoint ``blobs'' orbiting each other. The most important results in the two almost simultaneous papers were the same, and we review and prove two of them here.
\begin{prop}\label{prop:DTexist}
In any conformal field theory in $d>2$ dimensions, containing operators $\O_1$, $\O_2$ with twists $\tau_{1}$ and $\tau_2$, the value $\tau_\infty=\tau_1+\tau_2$ is an accumulation point in twist, i.e.\ there is a family of operators $\O_\ell$ where $\tau_\ell\to\tau_1+\tau_2$ as $\ell\to\infty$.
\end{prop}
\begin{proof}
Consider the mixed correlator $\G_{2112}(u,v)\sim\expv{\O_2(x_1)\O_1(x_2)\O_1(x_3)\O_2(x_4)}$. In the standard normalisation\footnote{More precisely, we use conventions such that $\G_{ijkl}(u,v)=x_{12}^{\Delta_i+\Delta_j}x_{34}^{\Delta_k+\Delta_l} x_{13}^{\Delta_k-\Delta_l}x_{24}^{\Delta_j-\Delta_i}\times x_{14}^{\Delta_i-\Delta_j-\Delta_k+\Delta_l}\expv{\O_i(x_1)\O_j(x_2)\O_k(x_3)\O_l(x_4)}$.}, crossing for this correlator reads
\beq{
\G_{2112}(u,v)=\frac{u^{\frac{\Delta_1+\Delta_2}2}}{v^{\Delta_1}}\G_{1122}(v,u).
}
In the direct channel (left-hand side), the collinear conformal blocks, using \eqref{eq:collineardifferentmod}, expand in the double lightcone limit as $\frac{-\Gamma(2\hb)}{\Gamma(\hb+\frac{\Delta_{12}}2)\Gamma(\hb-\frac{\Delta_{34}}2)}z^{\tau/2}\log(1-\zb)$ plus regular and higher order terms\footnote{For the case $\Delta_1=\Delta_2$ we give the complete expansion in \eqref{eq:gausshyperexp}.}. The crossed channel (right-hand side), contains the contribution $1$ from the identity operator, multiplied by the crossing factors. This leads to the equation
\beq{\label{eq:impossible}
-\sum_\O a_\O z^{\tau_\O/2} \frac{\Gamma(2\hb_\O)\log(1-\zb)}{\Gamma(\hb_\O+\frac{\Delta_{12}}2)\Gamma(\hb_\O-\frac{\Delta_{34}}2)}= \frac{z^{\frac{\Delta_1+\Delta_2}2}}{(1-\zb)^{\Delta_1}}+\text{reg.},
}
where we sum over all possible direct-channel operators. Since each term on the left-hand side only contain a $\log$ divergence, the power divergence on the right-hand side must arise from infinitely many terms. By further matching the correct $z$ dependence we find that we must have, for any interval $\tau_1+\tau_2\pm\delta$, an infinite number of operators $\O$ with $\tau_\O$ in that interval. This proves proposition~\ref{prop:DTexist}.
\end{proof}
\vspace{2.5ex}

\noindent The involved expansions around large spin were analysed quantitatively in \cite{Fitzpatrick:2012yx,Komargodski:2012ek} by approximating the sums \eqref{eq:impossible} over spin with an integral, which is valid up terms regular or at most logarithmically divergent in $w=1-\zb\to0$. We refer to this as the \emph{kernel method}, and provide more details in section~\ref{sec:kernel.method}. The leading $w$ divergence in the corresponding expansion can be computed by
\beq{\label{eq:kernelshort}
\sum_\hb \frac{2\Gamma(\hb)^2}{\Gamma(2\hb-1)}u\!\parr{\hb(\hb-1)} k_\hb(1-w)\sim\int\limits_0^\infty \df \hat J\ \frac{4\hat J}{w}\,  u\parr{\frac{\hat J^2}{w}} K_0(2\hat J), 
} where $u(J^2)$ denotes any additional spin dependence relative to the free theory OPE coefficients in four dimensions\footnote{We note how the ratios of Gamma functions cancel between \eqref{eq:impossible} (restricted to identical external operators) and \eqref{eq:kernelshort}, up to from a factor $2\hb-1$. That factor is in turn consumed by the Jacobian of the change of variables $\hb\rightsquigarrow \hat J=\sqrt{\hb(\hb-1)}/w$.} and $K_0$ is a modified Bessel function of the second kind. A direct application of the kernel method for the case $u(J^2)=1$ gives the sum $\frac1w$. Taking instead $u(J^2)\sim J^{2\alpha}$ we get a sum which generates a leading divergence of the form $\frac1{w^{1+\alpha}}$. We will use this result to prove the next proposition.

\begin{prop}\label{prop:DTgamma}
The double-twist operators $[\O_1,\O_2]_{0,\ell}$ according to proposition~\ref{prop:DTexist} have anomalous dimensions which have asymptotic behaviour at large $\ell$ of the form
\beq{\label{eq:gammafrom1212}
\gamma_\ell\sim -\frac{a_{\mathrm{min}}}{\ell^{\tau_{\mathrm{min}}}},
}
where $\tau_{\mathrm{min}}$ is the twist of the smallest twist operator $\O_{\mathrm{min}}\neq\1$ appearing in both OPEs $\O_1\times \O_1$ and $\O_2\times \O_2$, and $a_{\mathrm{min}}$ is the corresponding OPE coefficient $a_{\mathrm{min}}=c_{11\mathrm{min}}c_{22\mathrm{min}}$.
\end{prop}
\begin{proof}
For this proof we consider the divergence in $w$ introduced by the contributions from the operators $\1$ and $\O_\mathrm{min}$ appearing in the crossed channel. Including crossing factors these contributions take the form
\beq{
\1: \quad \frac{z^{\frac{\Delta_1+\Delta_2}2}}{w^{\Delta_1}}, \qquad \O_{\mathrm{min}}:\quad a_{\mathrm{min}}\frac{z^{\frac{\Delta_1+\Delta_2}2}w^{\frac{\tau_{\mathrm{min}}}2}}{w^{\Delta_1}}\parr{-\log z+\text{reg.}},
}
where $w=1-\zb$.
We match this with an expansion of the form of \eqref{eq:expwithoutmixing}: $\sum_\hb a_\hb(1+\frac12\gamma_\hb\log z)z^{\tau_0/2}k_\hb(\zb)$. We see that the anomalous dimensions correspond to the relative power $-a_{\mathrm{min}}w^{\frac{\tau_\mathrm{min}}2}$ between the terms $a_\hb$ and $a_\hb\gamma_\hb$, which translates exactly to the result \eqref{eq:gammafrom1212} using the kernel method. 
\end{proof}
\vspace{2.5ex}

\noindent The principles behind \cite{Fitzpatrick:2012yx,Komargodski:2012ek}, essentially the argument in the proof of proposition~\ref{prop:DTgamma}, were subsequently refined and extended to higher orders in the large spin expansion, providing understanding of which crossed-channel operators correspond to particular terms in the anomalous dimensions in various theories \cite{Alday:2015ota,AldayZhiboedov2015}. Collectively these methods became known as the \emph{lightcone bootstrap}, used in parallel with the more general \emph{analytic bootstrap}. Thanks to its universal assumptions, the lightcone bootstrap could be used for rigid derivations of facts valid in a wide range of theories. 
Starting from some considerations in \cite{Komargodski:2012ek}, this was used to rederive general properties of correlators in holographic CFTs, where the expansion parameter is $1/C_T$ \cite{Kaviraj:2015xsa}. If one further assumes that the only light operator corresponding to a single-particle state in $\mathrm{AdS}$ is the stress tensor, one gets a CFT definition of Einstein gravity. In \cite{Kulaxizi:2019tkd} CFT-data were derived for the double-twist operators in such a theory, the ``double stress tensors''. 
Another fruitful direction has been the relation to conformal collider physics \cite{Li:2015itl}, leading to a proof of the average null energy condition \cite{Hartman:2016lgu}. Finally, a demonstration of the lightcone bootstrap beyond any perturbative limit came in the elegant paper \cite{Simmons-Duffin:2016wlq}. There the CFT-data was computed for a large number of operators in the 3d Ising model using numerical bootstrap, and the spectrum was then analysed from the lightcone bootstrap. While the twist family $[\sigma,\sigma]_{0,\ell}$ was easily understood, the two families $[\epsilon,\epsilon]_{0,\ell}$ and $[\sigma,\sigma]_{1,\ell}$ have approximately equal $\tau_\infty$~\footnote{As indicated in table~\ref{tab:operatornames}, the operators $\sigma$ and $\epsilon$ in the 3d Ising model are identified with $\phi$ and $\phi^2$ in the $\epsilon$ expansion. The values $\Delta_\sigma= 0.5181489(10)$ and $\Delta_\epsilon={1.412625(10)}$ \cite{Kos:2016ysd} generate the two values $\tau_\infty=2.825$ and $\tau_\infty=3.036$.} and participate in a non-trivial non-perturbative mixing, which generates an eigenvalue repulsion of the two families at low spin.

Large spin perturbation theory (LSPT), proposed in 2016 in \cite{Alday2016} and demonstrated with a number of examples in \cite{Alday2016b}, builds on the lightcone bootstrap with the following additional ingredients. In LSPT, the anomalous dimensions and OPE coefficients are treated on the same footing, whereas previous work had focussed mostly on the former. Another feature is that the crossed-channel operators generating corrections to the CFT-data may be introduced in terms of an ansatz where no assumptions need to be made on for instance their anomalous dimensions. This introduces free parameters in the theory, which can be fixed at later stages through consistency conditions. These ingredients are tied together with a computational procedure of computing CFT-data from the crossed-channel operators, which we call an \emph{inversion procedure}. We will give a more concrete presentation of large spin perturbation theory at the end of this chapter, after we have introduced Caron-Huot's Lorentzian inversion formula \cite{Caron-Huot2017}, which provides one such inversion procedure.

The ideas generated from the analytic bootstrap and lightcone bootstrap have become a powerful tool for practical computations. This has become particularly useful in applications to holographic CFTs, in particular \NN4 SYM at strong coupling. Specifically, studying the boundary CFT at second order perturbation theory in the planar and strong coupling limit has generated results corresponding to loop supergravity and string corrections in $\mathrm{AdS}$ \cite{Aharony:2016dwx,AldayBissi2017,Aprile:2017bgs,Aprile:2017xsp,Aprile:2018efk,Caron-Huot:2018kta,Alday:2018pdi,Aprile:2019rep,Alday:2019nin}. The main obstacle that was overcome in these works was the resolution of mixing of degenerate operators, and it was shown on general grounds in \cite{Alday:2019qrf} that the growth in degeneracy is related to the number of extra dimensions in the dual gravity/string theory.

\section{The Lorentzian inversion formula}\label{sec:Inversionformula}

A major concern with the lightcone bootstrap, and indeed large spin perturbation theory, was the assumption, based on empirical observation, that CFT-data could be written as analytic functions of spin, however with spin zero often excluded. For instance, the reciprocity statement in proposition~\ref{prop:reciprocitygamma} relies on being able to differentiate the function $\mathsf g$. At best, the lightcone bootstrap could argue that the expansions around infinite spin correspond to the asymptotic behaviour. Even with a large spin expansion like \eqref{eq:S1expansion} known to all orders, it would not be certain that the anomalous dimension would take the precise value $S_1(\ell)$ for small or any finite value of $\ell$. 

The situation was greatly improved by a paper by Caron-Huot in 2017 with the title \emph{Analyticity in spin in conformal theories} \cite{Caron-Huot2017}. There it was not only shown that the CFT-data is analytic in spin, but an explicit integral formula was provided for performing the inversion procedures described above. With such a formula, one can directly check that asymptotic series like \eqref{eq:S1expansion}, with appropriate non-perturbative completions, indeed correspond to functions which give correct values at finite 
spin\footnote{We refer generically to the large spin expansion as the asymptotic behaviour of CFT-data. However, here we use asymptotic in the precise meaning of a series expansion with zero radius of convergence.}.

The inversion formula plays a central role in this thesis, we will devote this whole section to it. We start with an overview of its derivation, leaving the details to \cite{Caron-Huot2017}. Then we will extract from the general formula a specific, one-dimensional formula adopted for CFTs with a small expansion parameter. Since this is the main formula of the thesis, we give a detailed derivation keeping track of all factors. Finally, the Lorentzian inversion formula will allow us to rederive reciprocity and give a precise formulation thereof.

\subsection{Caron-Huot's inversion formula}

We begin by summarising the derivation of the Lorentzian inversion formula in \cite{Caron-Huot2017}. For simplicity we consider the case of external identical scalar operators. Details of the computation, as well as the extension to non-identical scalars, can be found in the original reference.

The starting point is the \emph{Euclidean inversion formula}, which follows as a property of harmonic analysis on the Euclidean conformal group $\SO{d+1,1}$ \cite{Dobrev:1977qv}, for a recent treatment see \cite{Karateev:2018oml}. The objects of study there are conformal partial waves, which form a basis for the space of Euclidean correlators. Each conformal partial wave $\Psi_{\Delta,\ell}(z,\zb)$ is a function labelled by an positive integer spin $\ell$ and a continuous dimension $\Delta$ taking values on the principal series $\Delta\in\frac d2+i \R$. The conformal partial wave can be constructed from the corresponding conformal block, together with the conformal block with the shadow dimension:
\begin{equation}
\Psi_{\Delta,\ell}=\frac12\parr{G^{(d)}_{\Delta,\ell}(z,\zb)+N^{(d)}_{\Delta,\ell}G^{(d)}_{d-\Delta,\ell}(z,\zb)},
\end{equation}
for some relative constant $N^{(d)}_{\Delta,\ell}$. The conformal partial waves satisfy an orthogonality relation $\expv{\Psi_{\Delta,\ell},\Psi_{\Delta',\ell'}}\sim \delta_{\ell,\ell'}\delta(-i(\Delta-\Delta'))$, where the inner product is given by a two-dimensional integration over the complex Euclidean $z$ plane with an appropriate measure factor. The Euclidean inversion formula is the corresponding Fourier transform for a Euclidean correlator and results in a function $C(\Delta,\ell)$ given by
\begin{equation}\label{eq:EuclInversion}
C(\Delta,\ell)\sim \int \df z\df \zb\, \mu(z,\zb) \overline\Psi_{\Delta,\ell}(z,\zb)\G_{\mathrm{Eucl.}}(z,\zb), 
\end{equation}
where $ \mu(z,\zb)$ is a measure factor. The function $C(\Delta,\ell)$ carries the dynamical information of the correlator; for each spin $\ell$ it has residues at physical operator dimensions $\Delta=\Delta_0$ and the residues are proportional to the OPE coefficients of the corresponding operators within the correlator $\G(z,\zb)$. We give the precise relation in \eqref{eq:Cdeltaell} below. By looking at the inverse transform 
\begin{equation}
\G_{\mathrm{Eucl.}}(z,\zb)\sim \sum_{\ell=0}^\infty\int\limits_{\Delta\in \frac d2+i\R}\frac{\df \Delta}{2\pi i}C(\Delta,\ell) \Psi_{\Delta,\ell}(z,\zb)
\end{equation}
one can reproduce the Euclidean correlator. By closing the $\Delta$ contour, evaluating the residues, and disentangling the contributions from the shadow blocks one recovers the usual conformal block decomposition \eqref{eq:CBexp} in the OPE limit.

The \emph{Lorentzian inversion formula} presented by Caron-Huot for identical scalar external operators $\phi$ takes the form \cite{Caron-Huot2017}
\begin{equation}\label{eq:invformulagen}
C(\Delta,\ell)=\left(
1\pm(-1)^\ell
\right)\frac{\kappa_{\Delta+\ell}}4 \int\limits_0^1\df z\int\limits_0^1\df \zb \mu(z,\zb)G^{(d)}_{d-1+\ell,1-d+\Delta}(z,\zb)\dDisc[\mathcal G(z,\zb)], 
\end{equation}
where we now keep track of all factors, given by $\mu(z,\zb)=|z-\zb|^{d-2}(z\zb)^{-d}$ and $\kappa_\beta=\frac{\Gamma(\beta/2)^4}{2\pi^2\Gamma(\beta)\Gamma(\beta-1)}$.
The kernel $G^{(d)}_{d-1+\ell,1-d+\Delta}(z,\zb)$ is functionally a conformal block, but it corresponds to a non-physical operator with scaling dimension $d-1+\ell$ and spin analytically continued to the value $\Delta+1-d$. This combination has the same eigenvalues \eqref{eq:casimireig2} and \eqref{eq:casimireig4} under the Casimir operators as the block for dimension $\Delta$ and spin $\ell$. The integration domain is now the spacelike Lorentzian kinematics, i.e.\ the square in figure~\ref{fig:diamond}. Finally, the $\pm$ sign is the same as the transformation of the correlator under $1\leftrightarrow2$.

The derivation of the Lorentzian inversion integral takes as a starting point the Euclidean formula, \eqref{eq:EuclInversion}, with the correct normalisation factors inserted. The idea is to analytically continue $z$ and $\zb$ to independent complex variables, and perform contour deformations. This requires dropping contributions from arcs at infinity, which turns out to be valid for $\ell>1$ and relies on analytic properties of the conformal partial wave and of the correlator. While the conformal partial waves have known analytic properties, the constraints from the correlator require physical input. Specifically, we require that the correlator belongs to a unitary CFT, and as such it is bounded in the Regge limit. More precisely, the correlator is more bounded than any individual block \eqref{eq:reggescaling} for $\ell>1$, which means that the contributions from operators with spin $\ell\geqslant2$ must all be related.

The result of the contour manipulations is a sum over four terms, where the correlator is evaluated at Lorentzian kinematics $z,\zb\in(0,1)$, but on different sheets in the complex $\zb$ plane. The terms combine into the double-discontinuity
\begin{equation}\label{eq:ddiscdef}
\dDisc[\G(z,\zb)]=\G(z,\zb) -\frac{1}{2} \G^\circlearrowleft(z,\zb)-\frac{1}{2}\G^\circlearrowright(z,\zb),
\end{equation}
defined as the correlator minus its two analytic continuations around $\zb=1$.
From a spacetime point of view, the double-discontinuity corresponds to the double commutator of the correlator
\begin{equation} \label{eq:doublecommutator}
\dDisc[\G(z,\zb)]=(z \zb)^{\Delta_\phi}\braccket0{[\phi(0,0),\phi(z,\zb)][\phi(1,1),\phi(\infty)]}0.
\end{equation}
The appearance of the double commutator is more obvious from the alternative derivation of the inversion formula given in \cite{Simmons-Duffin:2017nub}. Also there, the starting point is the Euclidean inversion formula \eqref{eq:EuclInversion}. The conformal partial wave is given a shadow representation, introducing a further integral over a point $x_5$. Under some partial gauge fixing, the integral variables become $x_3$ and $x_4$. Subsequent contour deformations move these points from the Euclidean configuration via a Wick rotation to their Lorentzian configuration. This results in four terms that combine into the double commutator \eqref{eq:doublecommutator}. Following the contour deformations in terms of the cross-ratios shows that for two of the terms, $\zb$ moves in its complex plane around branch cut at $\zb\geqslant1$ (in opposite directions), which produces the double-discontinuity \eqref{eq:ddiscdef}.

\subsection{The perturbative inversion formula}\label{sec:pertinversion}

We now derive a one-dimensional version of the inversion formula \eqref{eq:invformulagen}, which will be the main formula of this thesis. In particular, by focussing on a particular power $z^{\tau_0/2}$, the one-dimensional inversion formula will give the CFT-data corresponding to a twist family of reference twist $\tau_0$. We will present two equivalent versions, \eqref{eq:Tgenfdef} and \eqref{eq:Ugenfdef}, which are valid in perturbation theory, and give an explicit form of the CFT-data for a family of operators with that reference twist. 

We simplify the discussion by looking at the leading twist family, which dominates the small $z$ limit. Higher twist families are found by suitable projections, and we defer this to section~\ref{sec:projections}. We follow the manipulations of section~4 of \cite{Caron-Huot2017} and write \eqref{eq:invformulagen} as\footnote{We have limited the integral to $\zb>z$, at the expense of an extra factor of $2$, see \cite{Caron-Huot2017} for details.}
\beq{
C(\Delta,\ell)=\int\limits_0^1\frac{\df z}{2z}z^{\frac{\ell-\Delta}2}\int\limits_z^1\df \zb\, 2K_{\Delta,\ell}(z,\zb)\dDisc[\G(z,\zb)],
}
where we have factored out a potentially non-integer power of $z$ such that the remaining $z$ dependence can be expanded in a power series:
\beq{
K_{\Delta,\ell}(z,\zb)=\sum_{k=0}^\infty z^k K^{(k)}_{\Delta,\ell}(\zb).
}
This means that for each power $z^{\tau/2}$ in $\dDisc[\G(z,\zb)]$, the integral over $z$ results in a pole
\beq{\label{eq:Cdeltaell}
C(\Delta,\ell)\sim-\frac{a_\ell}{\Delta-(\tau+\ell)},
}
as well as poles from $k>0$. Taking the residue in $\Delta$ for fixed integer $\ell$ shows the existence of an operator with dimension $\tau+\ell$ and OPE coefficient $a_\ell$. 

The kernel contains the non-physical conformal block $G^{(d)}_{d-1+\ell,1-d+\Delta}(z,\zb)$, which we expand in the collinear limit, \eqref{eq:collblocks}. Changing variables to $h=\frac{\Delta-\ell}2$ and $\hb=\frac{\Delta+\ell}2$, this leads to an integral of the form
\beq{\label{eq:Chathhb}
\hat C(h,\hb)=\int\limits_0^1\frac{\df z}{z}z^{-h}\int\limits_0^1 \frac{\df \zb}{\zb^2}\kappa_{2\hb}k_\hb(\zb)\dDisc[\G(z,\zb)],
}
where $C(\Delta,\ell)=\frac12\hat C \parr{\tfrac{\Delta-\ell}2,\tfrac{\Delta+\ell}2}$ and we extended the limit of the inner integral to $0$. 

When using \eqref{eq:Chathhb} to read off the OPE coefficients, there will be an extra Jacobian factor induced by the change of variables. If we are interested in the OPE coefficient for a particular spin $\ell_0$, we integrate the residue of $C(\Delta,\ell)$ against a delta function
\beqa{\nonumber
a_{\ell_0}&=-\int \df \ell\oint \frac{\df \Delta}{2\pi i}C(\Delta,\ell)\delta(\ell-\ell_0)\\&=-\int \df \hb\oint \frac{\df h}{2\pi i}\hat C(h,\hb)\delta(\hb-h-\ell_0).\label{eq:aellfromcontours}
}
The locus $h_L$ of the pole in $h$ will depend on $\hb$, and evaluating the $\delta$ function means that we need to divide by the factor $\mathsf{Jac}=\pp\hb(\hb-h_L(\hb))$ evaluated at $\hb=h_L+\ell_0$:
\beq{\label{eq:aellfromresidue}
a_{\ell_0}=-\frac1{\mathsf{Jac}}\ \res_{h=h_L(\hb)}\left.\hat C(h,\hb)\right|_{\hb=h_L+\ell_0}.
}

Let us now specify to the case where we have a small expansion parameter $g$, which means that we can derive an explicit relation between the integral and the CFT-data. More precisely, we assume that the spectrum of the theory expands in a series in $g$, where $g=0$ corresponds to twist degeneracy, i.e.\ at $g=0$ all operators in a twist family has identical twist. With this assumption, we collect in the correlator $\G(z,\zb)$ all powers $z^h$ that are infinitesimally close to some value $z^{h_0}$, i.e.\ $h=h_0+h_1g+\ldots$, at the expense of introducing logarithms $z^h=z^{h_0}(1+gh_1\log z+\ldots)$. This defines a generating function
\beqa{\nonumber
\mathbf T(\log z,\hb)
&=T^{(0)}_\hb+\frac12T^{(1)}_\hb\log z+\frac18T^{(2)}_\hb\log^2 z+\ldots
\\&=2\kappa_{2\hb}\int\limits_0^1\frac{\df \zb}{\zb^2} \, k_\hb(\zb)\left.\dDisc[\G(z,\zb)]\right|_{z^{h_0}}, \label{eq:Tgenfdef}
}
where we have chosen the rational prefactors of $T^{(p)}_\hb$ as $2^{-p}/p!$. We will refer to \eqref{eq:Tgenfdef} as the \emph{perturbative inversion formula}.
The exact relation to the CFT-data is given by the following theorem, formulated in analogy with \cite{Alday:2017vkk}.
\begin{thm}\label{thm:mainT}
Study a correlator $\G(z,\zb)$ of identical scalar operators $\phi$ in an expansion in $g$, in a theory where $g=0$ corresponds to twist degeneracy.
If the double-discontinuity $\dDisc[\G(z,\zb)]$ of a correlator of identical scalars $\phi$, in an expansion in $g$, contains a leading power $z^{h_0}$, then the following holds. \begin{enumerate}
\item The OPE $\phi\times \phi$ contains an infinite family of operators $\O_\ell$ for $\ell=2,4,\ldots$, of twist $\tau=2h_0+\gamma_\ell$ with $\gamma_\ell\sim O(g)$.
\item The OPE coefficients $a_\ell=c_{\phi\phi\O_\ell}^2$ and the anomalous dimensions $\gamma_\ell$ of these operators are analytic functions of spin, given by the formula
\beq{\label{eq:aellfromT}
a_\ell (\gamma_\ell)^p=\left.T^{(p)}_\hb+\frac12 \de_\hb T^{(p+1)}_\hb+\frac18 \de^2_\hb T^{(p+2)}_\hb+\ldots\ \right|_{\hb=\hb_0}, \quad \hb_0=h_0+\ell,
}
for $T^{(p)}_\hb$ given by \eqref{eq:Tgenfdef}.
\end{enumerate}
In the case of mixing of operators within the twist $\tau=2h_0$ family, \eqref{eq:aellfromT} is modified by
\beq{
a_\ell (\gamma_\ell)^p\rightsquigarrow \expv{a_\ell \gamma_\ell^p}:=\sum_{i=1}^{d_\ell} a_{\ell,i}\gamma^p_{\ell,i},
}
where $a_{\ell,i}$ and $2h_0+\gamma_{\ell,i}$ denote the OPE coefficients and twists of the $d_\ell$ operators of equal spin and approximate twist $2h_0$. The statement is now that the functions $\expv{a_\ell \gamma_\ell^p}$ are analytic in spin.
\end{thm}

\begin{proof}
Performing the $z$ integral in \eqref{eq:Chathhb} with $\G(z,\zb)=\sum_p z^{h_0}G_p(\zb)\log^pz$ gives
\beqa{\nonumber
\hat C(h,\hb)
&=\sum_p\int\limits_0^1\frac{\df z}{z}z^{-h+h_0}\log^pz\int\limits_0^1 \frac{\df \zb}{\zb^2}2\kappa_{2\hb}k_\hb(\zb)\dDisc[G_p(\zb)]
\\&=-\sum_p\frac{
2^{-p}
}{(h-h_0)^{p+1}}T^{(p)}_\hb .
\label{eq:ChhinTs}
}
Assume now that each function $T^{(p)}_\hb$ admits an expansion in $g$ starting at order $g^p$~\footnote{It may be that the leading contribution is not at $g^0$ but at some $g^\alpha$. Such overall contribution can be factored out and the argument below holds.}. To make this dependence visible we will make the temporary replacement $T^{(p)}_\hb\rightsquigarrow  g^p\, T^{(p)}_\hb$, and omit terms higher order in $g$ at each $p$. 
This means that \eqref{eq:ChhinTs} takes the form
\beq{\label{eq:definitionoftheTs}
\hat C(h,\hb)=-\frac{T^{(0)}_\hb}{h-h_0}-\frac12\frac{gT^{(1)}_\hb}{(h-h_0)^2}-\frac14\frac{g^2T^{(2)}_\hb}{(h-h_0)^3}+\ldots.
}
Non-perturbatively, we expect only single-poles, which means that the presence of higher order poles must be a result of the expansion in $g$. Consider the function $\hat C(h,\hb)$ near the $d_\ell$ degenerate operators at spin $\ell$. Assuming twist degeneracy we expect the dependence
\beq{
\hat C(h,\hb)=-\sum_{i=1}^{d_\ell} \frac{a_i}{h-(h_0+\frac g2\gamma_i)},
}
where we have explicitly factored out $g$ and where we omit in $\gamma_i$ the terms higher order in $g$. Expanding this around small $g$ gives
\beq{\label{eq:fromexpandingmixing}
\hat C(h,\hb)\sim\sum_i \frac{a_i}{h-h_0}+\frac g2\frac{a_i\gamma_i}{(h-h_0)^2}+\frac{g^2}{4}\frac{a_i\gamma_i^2}{(h-h_0)^3}+\ldots,
}
matching the pole structure of \eqref{eq:definitionoftheTs}. In principle, we can now use \eqref{eq:aellfromresidue} to read off the OPE coefficients. This, however, requires computing the Jacobian factor, which becomes complicated at higher order in $g$. Instead we will make direct use of \eqref{eq:aellfromcontours} to extract $\expv{a_\ell}$. When evaluating the contour integrals in $h$, the higher order poles generate derivatives of the integrand:
\beqa{
&\expv{a_{\ell}}\\&=\!\int \!\df \hb\parr{\!T^{(0)}_\hb\delta(\hb-h-\ell)\!+\!\frac g2\de_h\!\parrk{T^{(1)}_\hb\delta(\hb-h-\ell)}\!+\!\frac12\!\parr{\frac g2}^2\!\de^2_h\!\parrk{T^{(2)}_\hb\delta(\hb-h-\ell)}\!+\!\dots\!}\!\!.\nonumber
}
When integrating against $\hb$, the delta function turns the derivatives into derivatives with respect to $\hb$, and we arrive at
\beq{
\expv{a_{\ell}}=\left.T^{(0)}_\hb+\frac g2 \de_\hb T^{(1)}_\hb+\frac{g^2}8 \de^2_\hb T^{(2)}_\hb+\ldots\right|_{h_0+\ell}.
}
This proves the $p=0$ case of \eqref{eq:aellfromT}. The case for higher $p$ can be shown by multiplying the integrand in \eqref{eq:aellfromcontours} by $(h-h_0)^p$, and performing the same contour integration. From \eqref{eq:fromexpandingmixing} we see that this now corresponds to extracting $\expv{a_\ell\gamma_\ell^p}$, and we get 
\beq{
\expv{a_{\ell}\gamma_{\ell}^p}=\left.g^pT^{(p)}_\hb+\frac {g^{p+1}}2 \de_\hb T^{(p+1)}_\hb+\frac{g^{p+2}}8 \de^2_\hb T^{(p+2)}_\hb+\ldots\right|_{h_0+\ell}
,}
which finishes our proof.
\end{proof}
\vspace{2.5ex}

\noindent For later convenience we define
\beq{\label{eq:Ugenfdef}
\mathbf U(\log z,\hb)
=\frac{\Gamma(\hb)^2}{\pi^2\Gamma(2\hb)}\int\limits_0^1\frac{\df \zb}{\zb^2} \, k_\hb(\zb)\left.\dDisc[\G(z,\zb)]\right|_{z^{h_0}} ,
}
where $\mathbf U(\log z,\hb)=U^{(0)}_\hb+\frac12U^{(1)}_\hb\log z+\frac18U^{(2)}_\hb\log^2 z+\ldots$. The CFT-data is now given by 
\beq{\label{eq:aellfromU}
A_\ell (\gamma_\ell)^p=\left.U^{(p)}_\hb+\frac12 \de_\hb U^{(p+1)}_\hb+\frac18 \de^2_\hb U^{(p+2)}_\hb+\ldots\right|_{\hb=\hb_0}, \quad \hb_0=h_0+\ell,
}
where $A_\ell$ are related to the usual OPE coefficients by
\beq{\label{eq:aArel}
a_\ell=\frac{\Gamma\parr{\frac{\Delta+\ell}2}^2}{\Gamma(\Delta+\ell-1)}A_\ell.
}
The normalisation of $\mathbf U(\log z,\hb)$ is defined such that the OPE coefficients of a free scalar field in four dimensions correspond to $A_\ell=2$~\footnote{In the original articles \cite{Paper2,Paper3,Paper4}, the normalisation of $U^{(p)}_\hb$ differs from here with a factor of $2\hb-1$.}. The functions $\mathbf T$ and $\mathbf U$ carry the same information, but in the following we find it useful to work with the $\mathbf U$.

\subsection{Reciprocity revisited}\label{sec:reciprocityrevisited}

As promised, let us now return to the statement about reciprocity, namely that CFT-data admit expansions around large spin organised in terms of integer powers of $J^2=\hb(\hb-1)$. From the discussion above, we have concluded that the CFT-data of a twist family can be described by the function $\mathbf U(\log z,\hb)$ computed from the perturbative inversion formula \eqref{eq:Ugenfdef}. We assume that the double-discontinuity $\left.\dDisc[\G(z,\zb)]\right|_{z^{h_0}} $ takes the form of a power series expansion in $(1-\zb)$, multiplied by an overall factor $(1-\zb)^\alpha$, however in general it can be a sum of several superimposed such series, potentially with logarithmic insertions. Since it is the $\zb\to1$ limit that is responsible for the large spin expansion, we can always re-expand this series in terms of $\frac{1-\zb}\zb$, giving
\beq{\label{eq:sumgeneraldisk}
\left.\dDisc[\G(z,\zb)]\right|_{z^{h_0}} =\sum_{k=0}^\infty c_k\parr{\frac{1-\zb}\zb}^{\alpha+k}.
}
In section~\ref{sec:invprocedures} we will explicitly show that integrating the terms in this sum against the kernel in \eqref{eq:Ugenfdef} gives the result
\beq{\label{eq:afterintegralterm}
\frac{\Gamma(\hb)^2}{\Gamma(2\hb)}\int\limits_0^1\frac{\df \zb}{\zb^2} \, k_\hb(\zb)\parr{\frac{1-\zb}\zb}^{\alpha+k}=\frac{\Gamma\parr{\hb-(\alpha+k+1)}\Gamma(\alpha+k+1)^2}{\Gamma\parr{\hb+(\alpha+k+1)}}.
}
Thus expanding the integral of the sum \eqref{eq:sumgeneraldisk} gives a sum of terms \eqref{eq:afterintegralterm}. Each such term expands for large $J=\sqrt{\hb(\hb-1)}$ as $J^{-2-2\alpha-2k}$ times integer powers of $J^{-2}$. We therefore conclude that the whole sum \eqref{eq:sumgeneraldisk} expands as $J^{-\nu}$ multiplied by integer powers of $J^{-2}$, where $\nu=2+2\alpha$, which may not be an even integer.

Allowing for logarithms and superimposed series we have in general
\beq{\label{eq:Uphbexpansion}
U^{(p)}_\hb=\sum_i\frac1{J^{\nu_i}}\sum_{k=0}^\infty\frac{u^{(p)}_{i,k}(\log J)}{J^{2k}}.
}
This will be used to derive the following precise version of the reciprocity principle, equivalent to \cite{AldayBissiLuk2015}.
\begin{thm}\label{thm:reciprocity}
For non-degenerate operators in a twist family, parametrised by $\ell$, the anomalous dimensions $\gamma_\ell$ and the OPE coefficients $a_\ell$ satisfy the recursive relations
\beqa{\label{eq:recip1}
\gamma_\ell&=\mathsf g\hspace{-1.5pt}\parr{\hb+\tfrac12\gamma_\ell},\\
A_\ell&=\parr{1-\frac12\mathsf g\hspace{1.5pt}'\hspace{-1.5pt}\parr{\hb+\tfrac12\gamma_\ell}}\mathsf A\hspace{-1.5pt}\parr{\hb+\tfrac12\gamma_\ell},
\label{eq:recip2}
}
with $a_\ell$ and $A_\ell$ are related by \eqref{eq:aArel}, where the functions $\mathsf A$, $\mathsf g$ have asymptotic expansions of the form
\beq{
\mathsf A(\bar{\mathsf h})=\sum_{i}\frac1{\mathsf J^{\alpha_i}}\sum_{k=0}^\infty \frac{a_{i,k}(\log \mathsf J)}{\mathsf J^{2k}}, \qquad
 \mathsf g(\bar{\mathsf h})=\sum_{i}\frac1{\mathsf J^{\beta_i}}\sum_{k=0}^\infty \frac{b_{i,k}(\log \mathsf J)}{\mathsf J^{2k}},
 \qquad \mathsf J^2=\bar{\mathsf h}(\bar{\mathsf h}-1).
}
\end{thm}

\begin{proofsketch}Notice that both the functions $U^{(p)}_\hb$ and the functions $\mathsf g$ and $\mathsf A$ admit the same kind of reciprocity-respecting expansions. However, any derivative of such functions will break this, since $\de_\hb J^2=2\hb-1=\sqrt{1+4J^2}$. We therefore need to check that these violating terms are exactly cancelled by the process of extracting the CFT-data from $U^{(p)}_\hb$ and re-packaging it in the form \eqref{eq:recip1} and \eqref{eq:recip2}

We need to perform this proof order by order in perturbation theory, using the fact that $\gamma_\ell$, and therefore $\mathsf g$, are of order $g$. The leading dependence of $\mathsf g $ and $\mathsf A$ is given by $\gamma_\ell $ and $A_\ell$ respectively. They are in turn related to the functions $U^{(p)}_\hb$ at leading order in \eqref{eq:aellfromU}, which is free from derivatives with respect to $\hb$: $A_\ell=U^{(0)}_\hb$ and $\gamma_\ell=U^{(1)}_\hb/U^{(0)}_\hb$. Hence the correct expansions of $\mathsf g $ and $\mathsf A$ at leading order follow directly from the expansions \eqref{eq:Uphbexpansion} of the $U^{(p)}_\hb$.

At subleading order in the expansion parameter $g$, the derivatives with respect to $\hb$ in \eqref{eq:aellfromU} induce terms that break the $J^2$ expansion. By carefully following the propagation of all terms one can check that these terms cancel if and only if one assumes that the operators are non-degenerate. This is because we have to impose relations like ${U^{(2)}_\hb}\big/{U^{(0)}_\hb}=\gamma_\ell^2=\parr{{U^{(1)}_\hb}\big/{U^{(0)}_\hb}}^2$ which are not true for operators with mixing.
\end{proofsketch}

\section{Large spin perturbation theory}\label{sec:LSPTtwo}

Large spin perturbation theory aims to produce perturbative results in conformal field theories by using the crossing equation and inversion procedures for CFT-data. These results are either specific for a given model, or generic for classes of CFTs satisfying some stipulated assumptions. This is achieved through an initial ansatz of crossed-channel operators generating the entire double-discontinuity of the four-point function at a given order in perturbation theory, and through a systematic inversion procedure. These steps are supplemented by imposing consistency conditions and may be iterated at higher orders in the perturbation.

The results of large spin perturbation theory consist of a set of CFT-data, or alternatively of an explicit expression for the correlator. These are essentially equivalent; given the correlator, the CFT-data is found by a conformal block decomposition, and given the CFT-data the correlator can be reconstructed by explicitly summing conformal blocks, often referred to as \emph{resummation}.
The discovery of the Lorentzian inversion formula adds a new dimension to this equivalence, as depicted in figure~\ref{fig:inversionflowchart}, where we now note that the double-discontinuity of the correlator is equivalent to the function $\mathbf U(\log z,\hb)$ through the Lorentzian inversion formula\footnote{The reverse arrow corresponds to the kernel method.}. This leads to the a commuting diagram, where the central rectangle of figure~\ref{fig:inversionflowchart} conveys the picture that the whole correlator $\G(u,v)$ is essentially determined by its double-discontinity $\dDisc[\G(u,v)]$. This fact was formulated in \cite{Carmi:2019cub} in terms of a dispersion relation. The only ambiguities come from terms at low spin, which are beyond the range of validity of the Lorentzian inversion formula (spin $0$ and potentially spin $1$). 
\begin{figure}
  \centering
\includegraphics{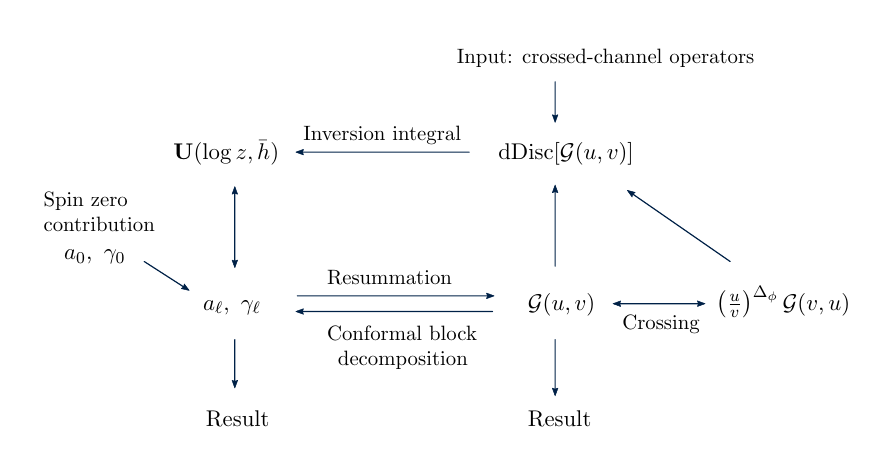}
\caption[A flowchart describing the method of large spin perturbation theory.]{A flowchart describing the method of large spin perturbation theory.
The input is CFT-data for a small set of crossed-channel operators. The output is results for the CFT-data or various twist families, or by resummation expressions for the correlator. Using crossing, the twist families contribute in the crossed channel and the process can be iterated.
}\label{fig:inversionflowchart}. 
\end{figure}

It is now clear why large spin perturbation theory turns out to be an effective method. At each order in perturbation theory, the entire double-discontinuity can be generated from just a small subset of crossed-channel operators. The reason is that the double-twist operators themselves have suppressed double-discontinuities in the crossed-channel. This can be realised by considering a crossed-channel operator with twist $2\Delta_\phi+2n+\gamma_{n,\ell}$ in the double lightcone limit. By making the same expansion as in the proof of proposition~\ref{prop:DTgamma}, and taking the double-discontinuity, we get the term
\beq{\label{eq:ddiscofdoubletwistinit}
\dDisc\left[\frac{\zb^{\Delta_\phi}}{(1-\zb)^{\Delta_\phi}}(1-\zb)^{\frac12\parr{2\Delta_\phi+2n+\gamma_{n,\ell}}}\right]\sim\zb^{\Delta_\phi}\frac{\pi^2}{2}\parr{\gamma_{n,\ell}}^2,
}
where we used that $\dDisc[\log(1-\zb)^2]=4 \pi^2$. Hence we can see that the first non-zero double-discontinuity appears at an order suppressed by the squared anomalous dimension. In the theories that we reviewed in section~\ref{sec:pertCFT}, we marked out these operators in the respective spectra by the grey bands in the figures~\ref{fig:spectrumWF}, \ref{fig:spectrumNN4} and \ref{fig:spectrumON}.

Our strategy will thus be as follows. Work at a given order in perturbation theory and identify which operators have a non-zero double-discontinuity in the crossed channel. Then create an ansatz for the double-discontinuity generated by these operators---in a specific theory one may want to use additional information about these operators, in a generic theory this introduces some undetermined constants. Following through the inversion procedure gives the CFT-data of twist families at this order. Next one can proceed to higher orders. New operators may appear, which expand the ansatz. Eventually the double-twist operators themselves will also appear but their contribution can be derived from results at lower order through crossing. This induces an iterative procedure, cycling through the diagram in figure~\ref{fig:inversionflowchart} multiple times.

We have presented the Lorentzian inversion integral as the prototype way of recovering the CFT-data from the correlator. However, there are other inversion procedures as well, such as those used in the original papers on large spin perturbation theory \cite{Alday2016,Alday2016b}. In this thesis we will use such alternative procedures in chapter~\ref{ch:paper1}. Before the role of the double-discontinuity was made clear, terms in the double lightcone limit were classified as either regular or singular, where singular terms referred to those that cannot be constructed from a finite sum of direct-channel blocks. In chapter~\ref{ch:paper1} we refer to these terms as having an \emph{enhanced singularity}. These terms are exactly those which develop a power-divergence in $\zb\to1$ after repeated action by the collinear Casimir \cite{AldayZhiboedov2015,Simmons-Duffin:2016wlq}. By constructing building blocks, called \emph{twist conformal blocks} or \emph{H-functions}, as sums of conformal blocks modulated by powers of $J^{-2}$, the enhanced singularities can be matched between a correlator and the corresponding CFT-data, turning the inversion into an algebraic problem \cite{AldayZhiboedov2015,Alday2016b}.

After giving a practical guide to large spin perturbation theory in the next chapter, we demonstrate the power of the method in chapter~\ref{ch:paper2}, where we apply it to the Wilson--Fisher fixed-point and derive results up to order $\epsilon^4$ \cite{Paper2}. We comment on the generalisation to $\OO N$ symmetry \cite{Paper3}. Then we show in chapter~\ref{ch:paper1} how the method facilitates the computation of the most general four-point function at order $g\sim g_{\mathrm{YM}}^2$ of a scalar of dimension $2+O(g)$ in a four-dimensional conformal gauge theory \cite{Paper1}. We give further applications in chapter~\ref{ch:more}: The $\OO N$ model at large $N$ \cite{Paper4}, general $\phi^4$ theories in both an $\epsilon$ expansion and a large $N$ expansion \cite{Paper5}, multicritical theories and an adaptation of chapter~\ref{ch:paper1} to the Wilson--Fisher model.

\chapter{A practical guide to large spin perturbation theory}
\label{ch:practical}

The previous chapter contained background material leading up to a formulation of large spin perturbation theory in section~\ref{sec:LSPTtwo}, where the main idea was presented in the diagram of figure~\ref{fig:inversionflowchart}. Large spin perturbation theory is a systematic framework for studying perturbative conformal field theories and the procedure applies to a wide range of theories. Anyone who wants to apply it to a new theory with a new set of assumptions will follow through the diagram by executing the steps listed at the end of chapter~\ref{ch:intro}. 

In the later chapters of this thesis we will give complete examples of applying large spin perturbation theory to specific cases. However, heading straight into these examples would obscure the many common features that emerge only after studying several different theories. The purpose of the present chapter is therefore to highlight these general aspects in order to give more information about each part of the procedure outlined above. This includes introducing useful notation and giving some specific statements in terms of some propositions and standard inversions.

At the centre of the diagram in figure~\ref{fig:inversionflowchart} sits the \emph{the perturbative inversion formula}. Although other inversion procedures exist, it is the main tool of this thesis and we repeat it here:
\beq{\label{eq:masterformula}
\mathbf U(\log z,\hb)=\frac{\Gamma(\hb)^2}{\pi^2\Gamma(2\hb)}\int\limits_0^1\frac{\df \zb}{\zb^2}  k_\hb(\zb)\left.\dDisc[\G(z,\zb)]\right|_{z^{h_0}},
}
where $\tau_0=2h_0$ is the reference twist of the twist family under consideration. To appreciate how the formula works in practice, we give some concrete computational examples in the first section of this chapter. In the subsequent sections we then follow the steps of chapter~\ref{ch:intro}. In section~\ref{corrstwists} we give some generic statements about the structure of the OPE, both in the direct and the crossed channel. In section~\ref{sec:fromcrossedchannel} we focus on how to compute the double-discontinuities that arise from the crossed-channel operators. In section~\ref{sec:invprocedures} we survey the most useful ways of executing the inversion integral \eqref{eq:masterformula} and give some concrete examples of inversions. We finish the chapter with section~\ref{sec:applicationsofLSPT} containing a literature review of applications of large spin perturbation theory to date.

\section{Invitation: sums and inversions}\label{sec:invitation}

The inversion formula \eqref{eq:masterformula} is the main tool for performing the inversion procedure that plays the central role in large spin perturbation theory. In this section we will give some concrete examples of how the inversion procedures work in practice. The examples we consider here will be used later in the thesis, typically for leading order computations. At higher order in $g$, more complicated functions will appear and to explicitly perform the inversion procedure will require a variety of methods, explained later on.

The central square of figure~\ref{fig:inversionflowchart} represents the computational machinery of large spin perturbation theory. Working in the collinear limit, the $z$ dependence decouples and we are in practice left with sums and inversions of $\SL2\R$ blocks. This is essentially a one-dimensional problem, where the CFT-data is parametrised by spin $\ell$, or equivalently by $\hb=\tau_0/2+\ell$, and the correlator is a function of $\zb$. The CFT-data of operators with spin $\ell>0$ is represented by $U(\hb)$, which is computed by the inversion integral
\beq{\label{eq:sl2rinversion}
 U(\hb)=\frac{\Gamma(\hb)^2}{\pi^2\Gamma(2\hb)}\int\limits_0^1\frac{\df \zb}{\zb^2}  k_\hb(\zb)\dDisc[G(\zb)],
}
where the double-discontinuity is still taken around $\zb=1$. 
We will now give some explicit examples of the second line of figure~\ref{fig:inversionflowchart}, namely resummation and conformal block decomposition. For simplicity we assume that we are working with operators on the unitarity bound in four dimensions, which means that the $\SL2\R$ block decomposition corresponding to \eqref{eq:sl2rinversion} is
\beq{\label{eq:sl2rsum}
\sum_{\ell=\ell_0,\ell_0+2,\ldots} \frac{\Gamma(\ell+1)^2}{\Gamma(2\ell+1)}U(\ell+1) k_{\ell+1}(\zb)=G(\zb).
}
In the free four-dimensional theory, it is natural to begin the sum at $\ell_0=0$. However, since $J^2=\ell(\ell+1)$ becomes zero for $\ell=0$ we take $\ell_0=2$. The difference would be the spin zero $\SL2\R$ block: $k_1(\zb)=-\log(1-\zb)$.

\subsection[Elementary sums of $\SL2\R$ blocks]{Elementary sums of $\boldsymbol{\SL2{\mathrm{R}}}$ blocks}\label{sec:elementarysums}
Let us start with the simplest possible sum, where $U(\hb)$ is a constant. For later convenience, we choose the constant to be $2$. In this case, the sum \eqref{eq:sl2rsum} can be performed directly with computer algebra software like Mathematica \cite{Mathematica}, by using the following manipulations. 
First we use a convenient integral representation for the $\SL2\R$ block in the integral kernel,
\beq{\label{eq:HyperInt1}
k_\hb(\zb)=\frac{\Gamma(2\hb)}{\Gamma(\hb)^2}\zb^\hb\int\limits_0^1\frac{\df t}{t(1-t)}\parr{\frac{t(1-t)}{1-t \zb}}^{\hb}.
}
Then the sum over $\ell=2,4,\ldots$ can be performed to give a rational function in $t$ and $\zb$. Finally, integrating over $t$ gives the result
\beq{\label{eq:sumnumber1}
\sum_{\ell=2,4,\ldots}\frac{\Gamma(\ell+1)^2}{\Gamma(2\ell+1)} 2k_{\ell+1}(\zb)=\frac{1}{1-\zb}+\underbrace{\zb-1+2\log(1-\zb)}_{\text{regular}},
}
where we have marked the terms that are regular in the limit $\zb\to1$. By regular we mean terms which have no double-discontinuity in this limit, corresponding to terms which have a regular series expansion at $\zb\to1$, or a series expansion multiplied by a single factor of $\log(1-\zb)$. Notice that we could have absorbed the logarithm by extending the sum to include the $\ell=0$ block $k_1(\zb)=-\log(1-\zb)$.

Another explicit example of a sum is the case where $U(\hb)=2/J^2$, for $J^2=\hb(\hb-1)$. We have
\beq{\label{eq:sumnumber2}
\sum_{\ell=2,4,\ldots}\frac{\Gamma(\ell+1)^2}{\Gamma(2\ell+1)} \frac{2}{\ell(\ell+1)}k_{\ell+1}(\zb)=\frac12\log^2(1-\zb)+\underbrace{2\,\mathrm{Li}_2(\zb)+2\log(1-\zb)}_{\text{regular}},
}
where $\mathrm{Li}_p(x)$ denotes the polylogarithm.
This sum is in fact much harder to find than \eqref{eq:sumnumber1} by explicit computations. Indeed, very few sums of conformal blocks can be computed directly, which means that typically \emph{resummation} is a more difficult task than \emph{conformal block decomposition}\footnote{However, the conformal block decomposition gives only the OPE coefficients, here corresponding to $U(\ell+1)$, one by one in $\ell$, and it may be a non-trivial task to deduce the closed form. In practice, these tasks are often accomplished by a combination of educated guessing, the use of functions such as Mathematica's \texttt{FindSequenceFunction} \cite{Mathematica} and searches in the Online encyclopedia of integer sequences \cite{OEIS}.}. A practical way of performing the sum \eqref{eq:sumnumber2} is therefore to make use of the conformal block decomposition. It turns out that in some cases, sums of $\SL2\R$ blocks $k_\hb(\zb)$ organise according to a transcendentality principle, and in general take the form of rational functions of $\zb$ multiplied by polylogarithms. From the result it is clear that an ansatz of polylogarithms of maximal combined degree $2$ would be enough to perform the sum \eqref{eq:sumnumber2} (recall that $\mathrm{Li}_1(\zb)=-\log(1-\zb)$).

We finish by giving a couple of examples of sums, computable in the same way, where $U(\hb)$ involve the harmonic numbers:
\beqa{\label{eq:sumnumber3}
\sum_{\ell=2,4,\ldots}\frac{2\Gamma(\ell+1)^2}{\Gamma(2\ell+1)} S_1(\ell) k_{\ell+1}(\zb)&=-\frac{\log(1-\zb)}{2(1-\zb)}+\underbrace{\frac12(1+\zb)\log(1-\zb)}_{\text{regular}},
\\\label{eq:sumnumber4}
\sum_{\ell=2,4,\ldots}\frac{2\Gamma(\ell+1)^2}{\Gamma(2\ell+1)} \frac{S_1(\ell)}{\ell(\ell+1)} k_{\ell+1}(\zb)&=-\frac{\log^3(1-\zb)}{12}.
}
Here the latter sum has no regular part.

\subsection[Elementary inversions for $\SL2\R$ blocks]{Elementary inversions for $\boldsymbol{\SL2{\mathrm{R}}}$ blocks}\label{sec:elementaryinversions}
In the explicit sums given above, we have indicated the regular parts, which have no double-discontinuity. Let us now show how the functions $U(\hb)$ used to produce these sums can be recovered from the inversion integral, using only the enhanced singular part.

It is easiest to start with the terms involving no negative powers of $1-\zb$, and we begin by analysing the sum in \eqref{eq:sumnumber2}. We let $G(\zb)=\frac12\log^2({1-\zb})+2\, \mathrm{Li}_2(\zb)+2\log(1-\zb)$ and start by computing the double-discontinuity. As discussed above, only the first term has a double-discontinuity, and a direct use of the definition \eqref{eq:ddiscdef} gives
\beq{
\dDisc[\log^2(1-\zb)]=\log^2(1-\zb)-\frac12\parr{\log(1-\zb)+2\pi i}^2-\frac12\parr{\log(1-\zb)-2\pi i}^2=4\pi^2,
}
which implies that $\dDisc[G(\zb)]=2\pi^2$. This means that $U(\hb)$ is given by the integral
\beq{\label{eq:inversion2comp}
U(\hb)=\frac{\Gamma(\hb)^2}{\pi^2\Gamma(2\hb)}\int\limits_0^1\frac{\df\, \zb}{\zb^2}  k_\hb(\zb)2\pi^2
=
2\int\limits_{[0,1]^2}\frac{\df t \df \zb}{t(1-t)\zb^2}\parr{\frac{t(1-t)\zb}{1-t \zb}}^\hb,
}
where we used the integral representation~\eqref{eq:HyperInt1} for $k_\hb$. Evaluating first the $\zb$ integral and then the $t$ integral gives $U(\hb)=\frac2{\hb(\hb-1)}$, in exact agreement with \eqref{eq:sumnumber2}. To systematise the notation it is useful to write the result of the inversion as
\beq{\label{sec:INVdef}
\INV[G(\zb)]=\frac{\Gamma(\hb)^2}{\pi^2\Gamma(2\hb)}\int\limits_0^1\frac{\df \zb}{\zb^2}  k_\hb(\zb)\dDisc[G(\zb)].
}
In this notation we have shown that 
\beq{\label{eq:inversion2INV}
\INV[\log^2(1-\zb)]=\frac{4}{J^2}.
}
Using $\dDisc[\log^3(1-\zb)]=12\pi^2\log(1-\zb)$ we can use the same integral representation as in \eqref{eq:inversion2comp} to show that
\beq{\label{eq:inversion3INV}
\INV\parrk{\frac{\log^3(1-\zb)}{12}}=-\frac{2S_1(\hb-1)}{\hb(\hb-1)}.
}

We now turn to the inversion problem corresponding to the sum \eqref{eq:sumnumber1}. The inversion of negative integer powers of $1-\zb$ requires a regularisation, and we therefore start by considering the inversion of the general power $G(\zb)=\parr{\frac{\zb}{1-\zb}}^p$. Since we have
\beq{
\dDisc\parrk{\parr{\frac{\zb}{1-\zb}}^p\,}=2\sin^2(\pi p)\parr{\frac{\zb}{1-\zb}}^p,
}
we get, again using the integral representation \eqref{eq:HyperInt1},
\beq{
U(\hb)=2\sin^2(\pi p)\frac{\Gamma(1-p)^2\Gamma(\hb+p-1)}{\pi^2\Gamma(\hb-p+1)}.
}
We see that this expression vanishes for $p=0,-1,-2,\ldots$ which corresponds exactly to the cases where $G(\zb)$ becomes regular. In the limit $p\to1$ the pole at $\Gamma(0)$ cancels with the zero at $\frac1\pi\sin \pi$, and we recover $U(\hb)=2$, in agreement with \eqref{eq:sumnumber1}. This can also we written
\beq{\label{eq:inversion1INV}
\INV\parrk{\frac1{1-\zb}}=2.
}
We save the inversion corresponding to the sum~\eqref{eq:sumnumber3} until we have discussed the $\SL2\R$ Casimir operator.

\subsection{Inversion integral and the Casimir}\label{sec:tastercasimir}

A very important tool in computing inversions is the $\SL2\R$ Casimir operator $\Dbar=(1-\zb)\zb^2\de^2_\zb-\zb^2\de_\zb$, introduced in section~\ref{sec:blockology}, which on the $\SL2\R$ blocks has eigenvalue $J^2=\hb(\hb-1)$:
\beq{
\Dbar\, k_\hb(\zb)= J^2\,k_\hb(\zb).
}
Acting with $\Dbar$ on a sum of \eqref{eq:sl2rsum} for a given $U(\hb)$ will therefore give the corresponding sum with $U(\hb)$ replaced by $\hb(\hb-1)U(\hb)$. The same holds for the inversion integral, and we get the useful equation
\beq{\label{eq:DbarJcommute}
\INV\parrk{\Dbar \,G(\zb)}=J^2\INV\parrk{G(\zb)}.
}
This can be used to demonstrate the last of the four sums discussed in section~\ref{sec:elementarysums}. We note first that $\Dbar\log^3({1-\zb})=\frac{\log(1-\zb)}{1-\zb}$. Then the relation \eqref{eq:DbarJcommute} combined with the inversion \eqref{eq:inversion3INV} gives that
\beq{
\INV\parrk{\frac{\log(1-\zb)}{1-\zb}}=J^2\INV\parrk{\frac12\log^3(1-\zb)}=-4S_1(\hb-1).
}
We summarise the four elementary inversions discussed in this section in table~\ref{tab:initialinversions}. In appendix~\ref{integrals} we collect more results of this kind, useful for inversions near four dimensions and in particular for chapter~\ref{ch:paper2}.

\begin{table}
\centering
\renewcommand{\arraystretch}{1.25}
\caption{Some elementary inversion results, where $w=1-\zb$ and $J^2=\hb(\hb-1)$, showing pairs $G(\zb)$ and $U(\hb)=\INV[G(\zb)]$.}\label{tab:initialinversions}
{\small
\vspace{4pt}
\begin{tabular}{|c|c|c|c|}\hline
${G(\zb)}$ & ${U(\hb)}$&${\Dbar G(\zb)}$ & ${J^2U(\hb)}$  
\\\hline
  $\log^2 w$ & $\dfrac{4}{J^2}$
 &
 $\dfrac2w+\text{reg.}$ & $4$
 \rule{0pt}{3.5ex} \rule[-3.5ex]{0pt}{0pt} 
 \\
 $\log^3 w$ & $-\dfrac{24S_1(\hb-1)}{J^2}$ 
 &
$\dfrac{6\log w}w+\text{reg.}$ & $-24 S_1(\hb-1)$
 \rule{0pt}{3.5ex} \rule[-3.5ex]{0pt}{0pt} 
 \\\hline
\end{tabular}
}
\end{table}

In fact, the relation \eqref{eq:DbarJcommute} between the Casimir $\Dbar$ and its eigenvalue $J^2$ can be used to derive the exact form of the $\SL2\R$ inversion integral \eqref{sec:INVdef}. We give this argument in section~\ref{sec:derivationCollinearFormula}. However, such derivation relies on the assumption that the CFT-data of the underlying theory is analytic in spin down to some finite value. The extraction of the perturbative inversion formula from Caron-Huot's general formula, which we worked out in section~\ref{sec:pertinversion}, is therefore necessary to establish analyticity and to determine the limit $\ell_0$ below which we can not trust the result.

\subsection{An algebraic method}
\label{sec:Hfunctionsmaintext}

The discussion so far shows that there is a direct correspondence between the functions $U(\hb)$ and the enhanced divergent part of $G(\zb)$. This implies that by matching appropriate terms on both sides we can turn the inversion problem into an algebraic problem. This programme was initiated in \cite{AldayZhiboedov2015} and was the main method used in the original papers on large spin perturbation theory \cite{Alday2016,Alday2016b}. In chapter~\ref{ch:paper1} based on \cite{Paper1}, we develop this idea further. We define what we call H-functions, which are a special case of the more general twist conformal blocks introduced in \cite{Alday2016}. The H-functions are sums of conformal blocks modulated by a specified function of $J^2$ \footnote{The H-functions used in chapter~\ref{ch:paper1} have an additional factor $2c\zb$ compared to here.}
\beq{\label{eq:def.Hbarpre}
\overline H^{(m,\log^n)}(\zb)=
\sum_{\hb} \frac{\Gamma(\hb)^2}{\Gamma(2\hb-1)} \frac{\log^nJ}{J^{2m}}\, k_{\hb}(\zb)+\text{reg.}
}
By expanding the CFT-data as
\beq{
U(\hb)=\sum_{m,n}A_{(m,\log^n)}\frac{\log^nJ}{J^{2m}}.
}
we can write the sum~\eqref{eq:sl2rsum} as
\beq{
\sum_{m,n}A_{(m,\log^n)}\overline H^{(m,\log^n)}(\zb)=G(\zb)+\text{reg.}
}
The H-functions can be computed by various techniques, but once they have been found the inversion procedure can be turned into a simple algebraic problem involving solving systems of linear equations. We give further details of this method, including an explicit toy example, in section~\ref{sec:twist.conformal.blocks}.

\subsection{The free 4d scalar}\label{sec:free4dscalar}
We finish this section by a concrete example, namely the free scalar field theory in four dimensions. Using Wick contractions the four-point function of the field $\phi$ takes the form
\beq{\label{eq:freefieldcorrelator}
\G(u,v)=1+\frac uv+u.
}
We will now demonstrate that this correlator can be determined completely from its double-discontinuity, which is the middle term $\frac uv$. This term corresponds precisely to the exchange of the identity operator in the crossed channel, which demonstrates the machinery of large spin perturbation theory as prescribed in figure~\ref{fig:inversionflowchart}: The identity operator generates a double-discontinuity, inverting this produces the CFT-data of direct-channel operators which resums into the correlator~\eqref{eq:freefieldcorrelator}, constituting our result.

Expanding $\dDisc[\G(u,v)]$ in the collinear limit we get
\beq{\label{eq:freeddisc}
\dDisc[\G(u,v)]=\dDisc\parrk{\frac{z}{1-\zb}}+O(z^2).
}
From the power $z$ we see that the double-discontinuity must correspond to operators of reference twist $\tau_0=2$, and from the lack of dependence on $\log z$ we see that $U^{(p)}_\hb=0$ for $p=1,2,\ldots$, meaning that the corresponding operators have no anomalous dimension.
The inversion \eqref{eq:inversion1INV} gives immediately that
\beq{
U^{(0)}_\hb=2,
}
from which we derive the four-dimensional free field OPE coefficients using \eqref{eq:aellfromU} and 
\eqref{eq:aArel}:
\beq{\label{eq:OPEfree}
a_\ell=\frac{2\Gamma(\ell+1)^2}{\Gamma(2\ell+1)}.
}
The next step is to perform the resummation of the correlator. The four-dimensional conformal blocks were given in \eqref{eq:CB4d} and for $\tau=2$ they take the particularly simple form
\beq{
G_{2,\ell}(z,\zb)=\frac{z\zb}{z-\zb}\parr{k_{\ell+1}(z)-k_{\ell+1}(\zb)}.
}
This means that we can use the sum \eqref{eq:sumnumber1} to compute\footnote{In the free theory, the formula \eqref{eq:OPEfree} analytically continues to spin zero. In principle, this OPE coefficient could take any other value, and the ambiguity at spin zero must be checked by independent methods such as a direct analysis of crossing.}
\beq{
\sum_{\ell=0,2,\ldots}\frac{2\Gamma(\ell+1)^2}{\Gamma(2\ell+1)}G_{2,\ell}(z,\zb)=\frac{z \zb}{(1-z)(1-\zb)}+z \zb.
}
Adding the direct-channel term $1$ corresponding to the identity operator, we have reconstructed the correlator \eqref{eq:freefieldcorrelator}. We shall also check subleading powers in $z$ omitted in \eqref{eq:freeddisc}. A careful analysis shows that they correspond to subleading contributions of the $\tau=2$ operators, which means that no further operators need to be considered.

We conclude this section by discussing how to generate corrections to the free theory. From \eqref{eq:ddiscofdoubletwistinit} it is clear that any new double-discontinuity must arise at order $g^2$ for some coupling $g$. This means that no operator in the leading twist family will receive anomalous dimensions until order $g^2$. The only exception is at spin zero, where analyticity in spin does not hold. We can therefore define $g=\gamma_{\phi^2}$ and conclude that all CFT-data at order $g^2$ will be depending on this constant. The leading contribution to $U^{(1)}_\hb$ from the operator $\phi^2$ is proportional to $-g^2\dDisc[\log^2({1-\zb})]$. Using the inversion \eqref{eq:inversion2INV} we get $U^{(1)}_\hb\sim -g^2/J^2$, and ultimately we get
\beq{
\gamma_\ell=-\frac{g^2}{\ell(\ell+1)}+O(g^3).
}
Noting that $\gamma_{\phi^2}=\frac\epsilon3+O(\epsilon^2)$ this agrees precisely with the result \eqref{eq:gammaWFlitt} quoted in section~\ref{sec:WFspectrum}. In chapter~\ref{ch:paper2} we continue reconstructing the Wilson--Fisher model from large spin perturbation theory and ultimately compute all CFT-data of the spinning operators to order $\epsilon^4$. Further explicit results in the Wilson--Fisher model can be found in \cite{Bissi:2019kkx}, where the whole correlator at order $\epsilon^2$ is given.

\section{Correlators and twist families}\label{corrstwists}
The first step in an application of large spin perturbation theory is to specify a given conformal field theory and a correlator to study. Large spin perturbation theory will then generate CFT-data for twist families in this correlator. Of course, the details depend on the specific choice of theory and correlator, but there are some universal features that we will describe here.

\subsection{Direct channel structure}\label{sec:directchannelstructure}

We limit ourselves to the simplest case and consider the four-point function of identical external operators. The generic content of twist families appearing in such a correlator is described by the following three propositions, however any specific theory may of course contain other twist families as well. In these propositions, we allow $\phi$ to have generic dimensions, not necessarily close to the unitarity bound.

\begin{prop}\label{prop:doubletwistsexist}
The $\phi\times\phi$ OPE contains operators $\phi\de^\ell\phi=[\phi,\phi]_{0,\ell}$ with $\tau_\ell=2\Delta_\phi+\gamma_\ell$.
\end{prop}
\noindent This is essentially the statement we proved in proposition~\ref{prop:DTexist}. In addition, unless we are in the free theory, we have GFF operators with twists $\tau_{n,\ell}=2\Delta_\phi+2n+\gamma_{n,\ell}$ for all positive $n$. If $\phi$ is near the unitarity bound, i.e.\ $\Delta_\phi=\mu-1+O(g)$ with $\mu=d/2$, the OPE coefficients of the GFF operators for $n\geqslant1$ are suppressed with a factor $g$ compared to $n=0$. This can be seen from the explicit expression \eqref{eq:GFFOPEschem}.

\begin{prop}\label{prop:currentsexist} 
In theories where the expansion parameter $g$ corresponds to a coupling constant, the $\phi\times\phi$ OPE contains weakly broken conserved currents $\mathcal J_\ell$ with $\tau_\ell=d-2+\tilde\gamma_\ell$, where in the non-degenerate case we have $\tilde\gamma_2=0$. 
\end{prop}
\noindent Proposition~\ref{prop:currentsexist} can be proved by the following argument. Analyticity in spin means that any operator with spin $\ell>\ell_0=1$ must be member of a twist family. Since the stress tensor always appears in the OPE of any two identical operators, with a non-zero OPE coefficient given by \eqref{eq:centralchargeOPE}, it must be a member of the leading twist family. If $\tilde\gamma_\ell$ is order $g$, proposition~\ref{prop:currentsexist} follows. However, in expansions around strong coupling the argument may break down. The reason is that at each order in the strong coupling expansion, the limit $\ell_0$ of analyticity may be shifted upwards to another small integer \cite{Alday:2017vkk}. Non-perturbatively, as well as in a weak coupling expansion, the limit $\ell_0=1$ holds. In section~5.1 of \cite{Caron-Huot2017}, this is discussed further in the context of a CFT dual to Einstein gravity, where it is only non-perturbatively that the stress tensor belongs to a twist family.

\begin{prop}\label{prop:moredoublesexist}
Assume that the $\phi\times\phi$ OPE contains an operator $\O$ with OPE coefficient $a_\O=c^2_{\phi\phi\O}$, where $\Delta_\O\neq 2\Delta_\phi+ O(g)$ and $\Delta_\O\neq \Delta_\phi+m+ O(g)$ for $m\in \Z$. Then the $\phi\times\phi$ OPE contains the operators $[\O,\O]_{n,\ell}$ with OPE coefficients at order $a_\O^2$.
\end{prop}

\begin{proof}
We follow the approach of appendix~E of \cite{Paper4} and consider the mixed correlator $\expv{\phi(x_1)\phi(x_2)\O(x_3)\O(x_4)}$. The crossing equation reads
\beq{
\G_{\phi\phi\O\O}(u,v)=\frac{u^{\Delta_\phi}}{v^{\frac{\Delta_\phi+\Delta_\O}2}}\G_{\O\phi\phi\O}(v,u).
}
By assumption $c_{\phi\O\phi}=c_{\phi\phi\O} $ is non-zero, which means that the crossed-channel OPE contains the operator $\phi$. It contributes to the double-discontinuity with a term proportional to its conformal block, which in the mixed correlator takes the form \cite{Dolan:2000ut}
\beq{
G^{(d)}_{\Delta,0|\Delta_i}=v^{\frac\Delta2}\sum_{m,n=0}^\infty \frac{\big(\frac{\Delta+\Delta_{12}}2\big)_m\big(\frac{\Delta-\Delta_{34}}2\big)_m\big(\frac{\Delta-\Delta_{12}}2\big)_{m+n}\big(\frac{\Delta+\Delta_{34}}2\big)_{m+n}}{(\Delta)_{2m+n}(\Delta+1-\mu)_m}\frac{v^m(1-u)^n}{m!n!}
}
with $\Delta_{ij}=\Delta_i-\Delta_j$ representing the crossed-channel external operator dimensions $\Delta_1=\Delta_4=\Delta_\O$, $\Delta_2=\Delta_3=\Delta_\phi$. We focus on the leading contribution to the CFT-data, which comes from the $m=0$ term. For this term, the sum over $n$ can be computed and gives, together with the crossing factor,
\beq{
\frac{u^{\Delta_\phi}}{v^{\frac{\Delta_\phi+\Delta_\O}2}}G^{(d)}_{\Delta,0|\Delta_i}=\frac{u^{\Delta_\phi}}{v^{\frac{\Delta_\O}2}}{_2F_1}\parr{\Delta_\phi-\frac{\Delta_\O}2,\Delta_\phi-\frac{\Delta_\O}2;\Delta_\phi;1-u}+O(v^{-\frac{\Delta_\O}2+1}).
}
Using \eqref{eq:generic2F1expansion} to expand the hypergeometric for small $u$, which is equivalent to small $z$, we get two contributions,
\beq{
\frac{u^{\Delta_\phi}}{v^{\frac{\Delta_\phi+\Delta_\O}2}}G^{(d)}_{\Delta,0|\Delta_i}=C_1\parr{ \frac{u^{\Delta_\phi}}{v^{\frac{\Delta_\O}2}}+O(u)}+C_2\parr{  \frac{u^{\Delta_\O}}{v^{\frac{\Delta_\O}2}}+O(u)},
}
for some constants $C_1$ and $C_2$ depending only on the involved operator dimensions. $C_2$ is regular and non-zero as long as the assumptions in the proposition are satisfied. This term signals the existence of operators $[\O,\O]_{0,\ell}$ in the direct channel, with OPE coefficients 
\beq{
c_{\phi\phi[\O,\O]_{0,\ell}}c_{\O\O[\O,\O]_{0,\ell}}\sim c_{\phi\O\phi}^2=a_\O
.}
Using that $c_{\O\O[\O,\O]_{0,\ell}}$ are of order $1$ by proposition~\ref{prop:doubletwistsexist}, we conclude that $c_{\phi\phi[\O,\O]_{0,\ell}}^2\sim a_\O^2$. The case $n>0$ follows by projections to higher twist and is valid as long as $\Delta_\O\neq \mu-1+O(g)$.
\end{proof}

\subsection{Spin zero}\label{sec:spinzerogendisc}
As we reviewed in section~\ref{sec:Inversionformula}, an important step in the derivation of the Lorentzian inversion formula \cite{Caron-Huot2017} is the contour deformations from the Euclidean integration domain $\C$ to the Lorentzian region $z,\zb\in[0,1]$, a manipulation that is valid for spin $\ell>1$. This means that the analytic expressions for CFT-data that result from the inversion formula may not correctly reproduce the CFT-data for operators at spin $\ell=0$ or $\ell=1$. Indeed, it is the case in some examples that the spin zero operator of a given twist family explicitly breaks the formula for e.g.\ anomalous dimensions. On the other hand, there are also many examples where the spin zero operator appears to obey the generic spin formula. A systematic determination of the conditions under which this happens remains an open problem. We summarise here a number of observations made in connection with the theories studied in this thesis.

As a first attempt to analyse the situation, we can study the convergence properties of the inversion integral from the one-dimensional perturbative inversion formula \eqref{eq:masterformula}. The limit $\zb\to0$ in the integration domain induces a pole at $\hb=1$. This has two implications: Firstly, the evaluation of the CFT-data at the spin corresponding to value of $\hb\approx1$ may not be defined, or least needs to be suitably regularised. Secondly, for any spin corresponding to $\hb<1$, the CFT-data has to be evaluated in its analytic continuation beyond the first pole, which is beyond the region of convergence of the inversion integral.
 
From the relation $\hb=\frac{\tau_0}2+\ell$ it is clear that the pole at $\hb=1$ in many cases affects only the leading twist family. This is for instance the case in \NN4 SYM, where we observed in section~\ref{sec:NN4spectrum} that the leading twist anomalous dimension in the average \eqref{eq:averagesKonishi} has a finite support solution at spin zero. In chapter~\ref{ch:paper1} we will assume that the CFT-data in the higher twist families can be extended to spin zero with no ambiguity. For instance, at the subleading twist, $\tau_0=4$, the pole is at $\ell=-1$ which is beyond the physical values of spin. However, there exist solutions to crossing which have finite support in spin at all twists. They were first constructed in \cite{Heemskerk2009} and correspond via holography to higher derivative interactions in AdS. For such theories, the Regge bounds used in establishing the inversion formula must be suitably modified.

The critical and multicritical models offer another venue for exploring spin zero. In the multicritical theories, i.e.\ the $\lambda\phi^{2\theta}$ theories for $\theta\geqslant3$ as displayed in figure~\ref{fig:operahouse}, the operator $\phi^2$ can be included in the family of weakly broken currents $\phi\de^\ell\phi$, and the CFT-data correctly extends to spin zero. Since $\hb=\mu-1+\ell$ and $\mu\leq\frac32$, this leads to an evaluation to the left of the pole at $\hb=1$, but the result is still consistent with the literature. We give more details of this in section~\ref{sec:multicritical}.

In the $\phi^4$ case with $\OO N$ symmetry, the situation is more subtle, but it turns out that in both the $\epsilon$ expansion and in large $N$ expansion, the CFT-data of broken currents can be extended to spin zero in a suitable way. In the $\epsilon$ expansion, spin zero appears at the pole $\hb=1$, but by shifting from the bare conformal spin $\hb$ to the full conformal spin $\hb_{\mathrm f}=\frac{\Delta+\ell}2$, the pole can be resolved at the expense of a factor of $\epsilon$. This means that the spin zero operator $\phi^2$ has an anomalous dimension of one order lower in $\epsilon$ than the broken currents, which is in agreement with the literature. We discuss this further in section~\ref{sec:mathingconditions} based on \cite{Paper2}. In the large $N$ expansion, the behaviour at spin zero depends on whether the representation contains an auxiliary field at spin zero. In the traceless symmetric, $T$, representation the CFT-data trivially extends to spin zero as in the multicritical models. In the singlet, $S$, representation, the scaling dimension extended to spin zero satisfies instead a shadow relation with respect to the auxiliary field $\sigma$:
\beq{
\Delta_{S,0}=d-\Delta_\sigma.
}
This relation, trivial at infinite $N$, appears to holds also in perturbation theory in $1/N$. We use this to analyse $\phi^4$ theories in section~\ref{sec:paper4short} for $\OO N$ case and section~\ref{sec:paper5short} in a generalised large $N$ expansion for generic global symmetry.

Finally, in \cite{Carmi:2019cub} it was discussed how the inversion formula can be improved by suitably subtracting the terms that determine the limit of convergence $\ell_0$. Of course, also in this approach the spin zero operators need to be added by hand, but the discussion may be useful in explaining why the operators at spin zero in the cases described above do inherit properties from the corresponding twist family.

\subsection{Global symmetries and crossing}\label{sec:globalsymmetries}

In the case of global symmetry, the formalism introduced here can be easily modified by introducing an extra label $R$ representing the irreducible representations involved in a given correlator. We typically consider the correlator of $\phi^I$, where $I$ is an index for the vector representation $V$ of a global symmetry group. Then the correlator can be projected onto the irreps $R$ in the tensor product $ V\otimes V$ by introducing tensor structures  $\mathsf T_R^{IJKL}$. This means that we write the correlator as
\beq{\label{eq:fourpointsglobal}
\expv{\phi^I(x_1)\phi^J(x_2)\phi^K(x_3)\phi^L(x_4)}=\frac{1}{x_{12}^{2\Delta_\phi}x_{34}^{2\Delta_\phi}}\sum_{R\in V\otimes V} \mathsf T_R^{IJKL} \G_R(u,v),
}
where each function $\G_R(u,v)$ has a conformal block decomposition of the form \eqref{eq:CBexp},
\beq{\label{eq:CBexpglobal}
\G_R(u,v)=\sum_{\O_R}c^2_{\phi\phi\O_R}G^{(d)}_{\Delta_{\O_R},\ell_{\O_R}}(u,v).
}
The representations $R$ will have different parity transformations under $x_1\leftrightarrow x_2$, and if the tensor structure $\mathsf T_R^{IJKL}$ is even (odd) under $I\leftrightarrow J$, the operators in the sum \eqref{eq:CBexpglobal} have even (odd) spin. The crossing equation can be written on the form
\beq{\label{eq:crossingwithglobalsym}
\G_R(u,v)=\parr{\frac uv}^{\Delta_\phi}\sum_{\rp\,\in V\otimes V}M_{R \rp}\G_{\rp}(v,u),
}
where the exact form of the matrix $M_{R\rp}$ has to be worked out from the tensor structures $\mathsf T_R^{IJKL}$ for a given symmetry group. We give the matrix for the $\OO N$ case in \eqref{eq:crossingmatrixON} in section~\ref{sec:paper4short}.

The types of operators described above now exist in various different representations. Double-twist operators $[\phi,\phi]_{R,n,\ell}$ according to proposition~\ref{prop:doubletwistsexist} exist in all representations $R\in V\otimes V$. In addition, for operators $\mathcal R_1$ and $\mathcal R_2$ in irreps $R_1$ and $R_2$ respectively, proposition~\ref{prop:moredoublesexist} generalises to operators $[\mathcal R_1,\mathcal R_2]_{R,n,\ell}$ in all representations $R\in R_1\otimes R_2 \cap V\otimes V$.

The twist families containing conserved currents are more interesting. Proposition~\ref{prop:currentsexist} only applies to the singlet $S$ representation. If this representation is normalised such that the contribution from the identity operator is $1$, then the OPE coefficient of the stress tensor at spin two in \eqref{eq:fourpointsglobal} has exactly the same relation to the central charge as in \eqref{eq:centralchargeOPE} above, namely 
\beq{\label{eq:centralchargesinglet}
\Delta_{S,2}=d,\quad a_{S,2}=\frac{d^2\Delta_\phi^2}{4(d-1)^2C_T}.
}
If the global symmetry is continuous, there are conserved currents in one or several of the odd representations, with $\Delta_{R,1}=d-1$, and corresponding current central charges related to the normalisations of the irreps.

\subsection{Further aspects}

\subsubsection{Mixing}

Mixing of operators within a twist family is a major hurdle for the analytic bootstrap and it has two serious consequences. One consequence affects the goal of large spin perturbation theory, namely to derive explicit results for the CFT-data. The existence of mixing means that one can, in general, only access averages such as $\expv{a\gamma}$ in \eqref{eq:mixingdef}. Since these averages are defined within the specific correlator, they are not very meaningful observables of the theory. On the other hand, these averages are precisely the building blocks needed to compute the mentioned correlator, which hence can be computed without resolving the mixing.

The second consequence happens when large spin perturbation theory is iterated to subleading orders in $g$. For instance, the contribution from the double-twist operators themselves is proportional to
\beq{\label{eq:mixedforg2}
\expv{(a\gamma^2)_{\ell}}=\sum_{i=1}^{d_{\ell}}a_{\ell,i}\gamma^2_{\ell,i}.
}
If $d_{\ell}>1$, knowing the averages $\expv a$ and $\expv{a\gamma}$ does not lead to \eqref{eq:mixedforg2} without knowing the individual anomalous dimensions.

As we saw in section~\ref{sec:NN4spectrum}, the mixing of leading twist operators in \NN4 SYM can be explicitly resolved, since the individual anomalous dimensions are known and related to the universal function $\gamma_{\mathrm{univ.}}(\ell)$. In section~\ref{sec:lightconecrossing}, we mentioned that in addition, for correlators of half-BPS operators in the planar expansion, the mixing has been resolved also for higher-twist operators. There, the structure of individual anomalous dimensions has a lot of symmetry, and has in fact been explained in terms of a, potentially accidental, ten-dimensional conformal symmetry \cite{Caron-Huot:2018kta}. 

In chapter~\ref{ch:paper2} we encounter the same difficulty, in this case involving mixing within the operators $\de^\ell\phi^4$ in the Wilson--Fisher fixed-point. In that case, a transcendentality principle facilitates an ansatz for the sum of the twist family, consistent with the non-degenerate cases at spins $\ell=0$ and $\ell=2$.

\subsubsection{Projections onto higher twist families}
\label{sec:projections}

So far, we have been concerned with extracting CFT-data for the leading twist family in a given CFT. Specifically, the perturbative inversion formula \eqref{eq:masterformula} involves a projection to the power $z^{\Delta_\phi}$, which corresponds to studying CFT-data of operators with reference twist $\tau_0=2\Delta_\phi$. By instead projecting onto another power, say $z^{h_0}$ one can extract CFT-data for a family with reference twist $\tau_0=2h_0$. Everything that we have described so far translates to the general case when $h_0$ and $\Delta_\phi$ are not related by an integer, i.e.\ when $h_0\neq \Delta_\phi+n+O(g)$ for $n\in \Z$. We will now outline what happens when $h_0$ and $\Delta_\phi$ are related by an integer.

Let us assume that the inversion problem has been studied at reference twist $\tau_0=2h_0$, which could be $2\Delta_\phi$, and that we are interested in operators in a subleading twist family with $\tau_0=2h_0+2n$ for some positive integer $n$. To this end we define a new function 
\beqa{\nonumber
\mathbf T_{\parrm{h_0}}(z,\log z,\hb)&=\sum_{n=0}^\infty z^{h_0+n}2\kappa_\hb\int\limits_0^1\frac{\df \zb}{\zb^2} k_\hb(\zb)\left.\dDisc[\G(z,\zb)]\right|_{z^{h_0+n}} 
\\
&=\sum_{n=0}^\infty z^{h_0+n}\mathbf T_{h_0+n}(\log z,\hb),
}
computed from the double-discontinuity of all powers of $z$ related to $z^{h_0}$ by an integer multiple. In this notation $\mathbf T(\log z,\hb)$ of \eqref{eq:Tgenfdef} corresponds to the $n=0$ term for $h_0=\Delta_\phi$.

Assume that we are interested in the reference twist $2h_0+2$ and that we have computed both $\mathbf T_{h_0}(\log z,\hb)$ and $\mathbf T_{h_0+1}(\log z,\hb)$. In the direct channel at the power $z^{h_0+1}$ we have both contributions from subleading collinear blocks of the twist $2h_0$ operators, and contributions from the new operators at twist $2h_0+2$. Using the form of the subcollinear blocks, given explicitly in 
appendix~\ref{app:subcollinearblocks}, 
we get
\beq{\label{eq:VUsum}
\sum_\hb \Big(\mathbf S(\log z,\hb)k_{\hb}(\zb)+\sum_{i=-1}^1 c_{1,i}(h_0+\tfrac\gamma2,\hb)\mathbf T_{h_0}(\log z,\hb) k_{\hb+i}(\zb) \Big)= \sum_\hb \mathbf T_{h_0+1}(\log z,\hb)k_{\hb}(\zb),
}
where $\mathbf S(\log z,\hb)$ denotes the contribution due to new primary operators at twist $2h_0+2$, and $\gamma$ is the anomalous dimensions at twist $2h_0$. The trick to extract $\mathbf S(\log z,\hb)$ is to use a linear change in variables in the sum such that all terms in \eqref{eq:VUsum} multiply $k_\hb(\zb)$. Ignoring regular terms, this re-writing means that we can read off the equation
\beq{
\mathbf S(\log z,\hb)=\mathbf T_{h_0+1}(\log z,\hb)-\sum_{i=-1}^1c_{1,i}(h_0+\tfrac\gamma2,\hb-i)\mathbf T_{h_0}(\log z,\hb-i).
}
The change in variables in the sum is allowed, because the difference between a sum and its shifted version corresponds to a single conformal block and therefore contains no enhanced divergence. The procedure outlined here generalises to higher order powers $z^{h_0+n}$, where corrections from subcollinear blocks of all lower twist families need to be projected away.

\section{Contributions from crossed-channel operators}
\label{sec:fromcrossedchannel}

In this section we give some more details on how to compute the double-discontinuity of a correlator in a given theory from crossed-channel operators. As outlined in figure~\ref{fig:inversionflowchart} in section~\ref{sec:LSPTtwo}, at any given order in the expansion parameter $g$, the double-discontinuity will be computed by considering operators appearing in the crossed-channel conformal block decomposition. Let us now refer back to figure~\ref{fig:diamond}, displaying the kinematic limits in Lorentzian signature. The conformal blocks of crossed-channel operators naturally expand in the crossed-channel OPE limit, i.e.\ for small $1-z$ and small $w=1-\zb$. The small $w$ expansion is convenient, since it can be carried over to the double-lightcone limit where it generates a large $J^2$ expansion for the CFT-data. The small $1-z$ expansion, however, is problematic, since the perturbative inversion formula \eqref{eq:masterformula} requires expanding in small $z$ and extracting the coefficient of a given power. This requires a summation and re-expansion of powers of $1-z$.

In summary, our strategy is therefore as follows:
\begin{enumerate}
\item At each order in $g$, identify which operators contribute with a non-zero double-discontinuity at that order.
\item Find expressions for the conformal blocks in the crossed-channel OPE limit for these operators.
\item Compute full sums over powers $(1-z)^k$, and then re-expand the result in small~$z$.
\item Select an appropriate power $z^{h_0}$, and construct the corresponding function $G(\zb)$ to invert, either in closed form or as a series in $w=1-\zb$.
\end{enumerate}
The headings below give some details about each step. Importantly, due to the re-expansion in $z$, the contribution to the double-discontinuity from operators in a twist family must be treated collectively. Since the crossed-channel large spin limit exactly corresponds to the direct-channel small $z$ limit, performing the sum over spins can introduce new non-trivial behaviour at small $z$ not exhibited by a single crossed-channel block. We will therefore separate the case of contribution from individual operators, which will be predominantly scalars, and from entire twist families.

\subsection{Crossed channel structure in perturbation theory}\label{sec:crossedstructure}
Recall that a crossed-channel operator $\O$ with twist $\tau_\O$ and OPE coefficient $a_\O=c^2_{\phi\phi\O}$ appears in the inversion integral with a prefactor $a_\O\sin^2(\tau_\O/2-\Delta_\phi)$. This followed from the discussion in section~\ref{sec:LSPTtwo} leading up to \eqref{eq:ddiscofdoubletwistinit}. From the linear expansion around the zeros of the sine function, the proposition below follows.
\begin{prop}\label{prop:firstappearance}
If $\tau_\O=2\Delta_\phi+2n+\kappa g^\delta+O(g^{\delta+1})$, with $n\in \Z$ and we allow for the case $\delta=0$, and if $a_\O = \tilde\kappa g^\alpha+O(g^{\alpha+1})$, then the operator $\O$ has the first non-vanishing double-discontinuity at order $g^a$, where
\begin{enumerate}
\item $a=2\delta+\alpha$ if $n\geqslant0$,
\item $a=\alpha$ if $n<0$.
\end{enumerate}
\end{prop}
\vspace{2.5ex}

\noindent Proposition~\ref{prop:firstappearance} implies that the order at which an operator contributes is completely determined by the order of its OPE coefficient and the distance to the nearest double-twist dimension. Step~1 of the strategy is therefore to analyse the leading contribution of the OPE coefficients of particular operators in the theory, i.e.\ to identify $\alpha$ above. In general, this task requires some knowledge of or assumptions about the theory of consideration and here we give some general guiding principles.

For the contribution from crossed-channel twist families, the considerations in section~\ref{sec:directchannelstructure} apply. Double-twist operators and GFF operators $[\O,\O]_{n,\ell}$ appear with suppressions of order $\gamma^2$ and $a_\O^2$ respectively. Therefore the first operators (different from $\1$) to appear are typically either scalars, or twist families with reference twist $\tau_0<2\Delta_\phi+O(g)$. If a scalar operator is the leading contribution, its OPE coefficient becomes an effective coupling constant of the perturbative expansion. This is for instance the case in the $\OO N$ model at large $N$ in section~\ref{sec:paper4short}. In the case of a twist family below the double-twist threshold, one has to introduce an appropriate crossing symmetric ansatz for its contribution. This is the philosophy of our approach in chapter~\ref{ch:paper1}.

\subsubsection{Heuristic diagrammatic method}\label{sec:heuristicmethod}

If one studies a CFT with a Lagrangian description, information about $\alpha$ of the (squared) OPE coefficient of a given operator can be extracted from a heuristic diagrammatic method, similar to the cuts in Witten diagrams in AdS as described in 
\cite{Fitzpatrick:2011dm}. The idea relies on drawing position space Feynman diagrams contributing to the four-point function such that it is possible to make a cut through the diagram corresponding to the operator under consideration. From this the following rule can be formulated: \emph{If there is a cut through $k_1$ lines of $\O_1$, $k_2$ lines of $\O_2$ etc.\ in a diagram at order $g^\alpha$, then the operator $\square^n\de^\ell\O_1^{k_1}\O_2^{k_2}\cdots$ contributes in the OPE at order $g^{\alpha}$.}
There are two exceptions to this rule. If the dimensions of the $\O_i$ are near the scalar unitarity bound, operators with $n\geqslant1$ are further suppressed. If all field lines join to a single point at both sides of the cut, only the scalar $ O_1^{k_1}\O_2^{k_2}\cdots$ operator contributes at that order. We give an explicit example in figure~\ref{fig:multidiagram} in the case of multicritical theories.

\subsection{Individual operators}
The contribution from individual operators is given by considering the crossed-channel conformal block and re-expanding it in the direct-channel small $z$ limit. Of course, in even integer dimensions any conformal block can be evaluated exactly, by which the re-expansion in small $z$ is a trivial task. In generic dimension, explicit forms of the conformal blocks are not known except in specific cases.

The single most important case is the scalar, where the conformal block takes the form \eqref{eq:ScalarBlockAnyD} given in section~\ref{sec:blockology}. Taking this expression with $u$ and $v$ interchanged and performing the sum over $m$, the small $z$ limit can be computed. The sum over $n$ subsequently gives the result
\beqa{\nonumber
v^{-\Delta/2}\left.G^{(d)}_{\Delta,0}(v,u)\right|_{\text{small $z$}}&=-\frac{\Gamma(\Delta)}{\Gamma(\frac\Delta2)^2}\Big[
\pp a\,{_2F_1}\left.\left(\tfrac\Delta2+a,\tfrac\Delta2;\Delta+1-\mu;1-\zb\right)\right|_{a=0}
\\&\qquad+\left(2S_1(\tfrac\Delta2-1)+\log(z\zb)\right){_2F_1}\left(\tfrac\Delta2,\tfrac\Delta2;\Delta+1-\mu;1-\zb\right)
\Big],\label{eq:genericcrossedblock}
}
The details of these manipulations can be found in \cite{Paper4}. In section~\ref{eq:inversionandcasimir} we will make use of this expression evaluated at $\Delta=2$. In expansions around four dimensions, the hypergeometric functions reduce to polylogarithms similar to those encountered in section~\ref{sec:invitation}.

The double light-cone expansion of crossed-channel blocks for general spinning operators was considered in \cite{Li:2019dix}, and in particular, explicit expressions were derived for operators at the unitarity bound. We reproduce here the stress tensor case,
\beqa{\nonumber
v^{-(\mu-1)}\left.G^{(d)}_{d,2}(v,u)\right|_{\text{small $z$}}=\zb^{1-\mu}&\frac{\Gamma(2\mu+2)}{\Gamma(\mu+1)^2}\Big[ \log\parr{\frac \zb z}-2S_1(\mu)
\\&+ \frac{(4\mu-2)(1-\zb)}{\mu(\mu+1)}{_2F_1}\left(1,2;\mu+2;1-\zb\right)\Big],
}
where $\mu=d/2$.

\subsection{Families of operators}\label{sec:familiesofoperators}
Let us emphasise again that large spin perturbation theory requires that the contribution from operators in the same twist family be summed up before the re-expansion in small $z$ is performed. This is in general a formidable task requiring a variety of techniques. The general form of such a sum is given by the following proposition:

\begin{prop}\label{prop:crossedsums}
Let $\O_\ell$ be a family of crossed-channel operators. Then the sum over the corresponding crossed-channel blocks with GFF OPE coefficients \eqref{eq:aGFF}, modulated by $J^{-\kappa}$, has the small $z$ expansion which takes the form
\beq{\label{eq:sumwithkappas}
\sum_\ell a^{\mathrm{GFF}}_{0,\ell}|_\Delta \frac{1}{J^\kappa} G^{(d)}_{2\Delta+\ell,\ell}(1-\zb,1-z)=z^{-\Delta+\frac\kappa2}F_1(z,\zb)+F_2(z,\zb) .
}
In this expression, each of the $F_i(z,\zb)$ expands in non-negative integer powers of $z$ and $\log z$: $F_i(z,\zb)=\sum_{j,k}z^j\log^kz\, f_{ijk}(\zb)$.
\end{prop}
\vspace{2.5ex}

\noindent The first term, $F_1(z,\zb)$, follows from the kernel method, \eqref{eq:kernelshort}, combined with the fact that the case $\kappa=0$ reproduces the correct dependence $z^{-\Delta}$ for the GFF theory. The second term exists in each conformal block (see \eqref{eq:gausshyperexp}) and can therefore not be excluded from the sum\footnote{When considering mixed correlators, $F_2(z,\zb)$ gets an additional contribution proportional to $z^{\frac{\Delta_1-\Delta_2-\Delta_3+\Delta_4}2}$, which exists in each collinear block \eqref{eq:collineardifferentmod}, expanded using \eqref{eq:generic2F1expansion}.}.

Proposition~\ref{prop:crossedsums} is very useful for determining what direct-channel families a crossed-channel family of operators gives rise to. Consider for instance the contribution from the leading twist singlet operators $\mathcal J_{S,\ell}$ in the $\OO N$ model. In section~\ref{sec:paper4short} we show that their anomalous dimension takes the form
\beq{\label{eq:gammaSres}
\gamma_{S,\ell} = -\frac{2\gamma_{\varphi}^{(1)}}{N}\left(  \frac{ (\mu -1) \mu}{J^2} +  \gamma_{\varphi}^{(1)} \frac{\pi  \csc (\pi  \mu ) \Gamma (\mu +1)^2 \Gamma (\ell+1)}{J^2(\mu -2) \Gamma (\ell+2 \mu -3)}   \right)+O(N^{-2}),
}
where $J^2=(\mu-1+\ell)(\mu-2+\ell)$ and where the second term expands as $J^{-2(\mu-1)}$. The operators $\mathcal J_{S,\ell}$ contribute with a leading order double-discontinuity that arises from a sum of anomalous dimensions squared. This gives rise to three terms of the form \eqref{eq:sumwithkappas}, with $\kappa=4$, $\kappa=2\mu$ and $\kappa=4\mu-4$ respectively. Let us identify what direct-channel singlet operators this corresponds to. We need to multiply by the factor $z^{\Delta_\varphi}M_{SS}$ from crossing, where $\Delta_\varphi=\mu-1+O(N^{-1})$ and where $M_{SS}=\frac1N$ is the matrix element in the $\OO N$ model crossing matrix \eqref{eq:crossingmatrixON}. Since to leading order $a_{S,\ell}=\frac1N a_{0,\ell}^{\mathrm{GFF}}|_{\Delta=\mu-1}$, we get that the contribution must happen at order $1/N^4$. Multiplying the prefactor of $F_1(z,\zb)$ in \eqref{eq:sumwithkappas} by $z^{\mu-1}$ from crossing, we get contributions of the form $z^{\kappa/2}$, which correspond to twists $\kappa$ for the values given above. We identify the three values of $\kappa$ with twist families $[\sigma,\sigma]_{0,\ell}$, $[\varphi,\varphi]_{S,1,\ell}$ and $(\de^\ell\varphi^4)_S$ respectively. Finally, the term corresponding to $F_2(z,\zb)$ contributes to the operators $[\varphi,\varphi]_{S,0,\ell}$ themselves. All of this is consistent with figure~\ref{fig:spectrumON}.

Proposition~\ref{prop:crossedsums} has an important practical consequence for the organisation of the inversion problem. In the sum over a twist family, it enables the sum to be computed using subcollinear blocks in the crossed channel. The reason is that the crossed-channel collinear limit $\zb\to1$ generates a series in $w=1-\zb$. Under the inversion integral, this in turn produces a series in large $J^2$. This means that the inversion can proceed without finding the explicit form of the sums \eqref{eq:sumwithkappas}, and if the resulting large $J^2$ series can be matched to an explicit function, the goal is achieved.

Various refinements of the kernel method can be used to compute relevant limits of the function $F_1(z,\zb)$ of proposition~\ref{prop:crossedsums}. This played an important role in \cite{Simmons-Duffin:2016wlq} and is described there. 
The term $F_2(z,\zb)$, however, is in general harder to extract. 
In \cite{Paper4}, as we will return to in section~\ref{sec:paper4short}, this was achieved by the method of twist conformal blocks, combined with an additional differential equation on the unitarity bound.

A twist conformal block is defined as a sum over a single twist family: $\mathcal H^{(0)}(z,\zb)=\sum_\ell a_\ell G^{(d)}_{\tau_0+\ell,\ell}(z,\zb)$ \cite{Alday2016}. This generalises to the \emph{level $m$ twist conformal block} defined by
\beq{
\mathcal H^{(m)}(z,\zb)=\sum_\ell \frac{a_\ell}{J^{2m}} G^{(d)}_{\tau_0+\ell,\ell}(z,\zb).
}
The definition is similar to the H-functions introduced in section~\ref{sec:Hfunctionsmaintext}, but we now consider complete functions of both $z$ and $\zb$ rather than just the singular part. By making a shift in the quadratic Casimir $\mathcal C_2$, given in \eqref{eq:quadraticCasimir}, such that the eigenvalue is the conformal spin $J^2$, the twist conformal blocks at different levels are related by a differential equation
\beq{\label{eq:casimirrelationforTCB}
{\cas}^m\mathcal H^{(m)}(z,\zb)=\mathcal H^{(0)}(z,\zb), \quad \cas=\mathcal C_2-\frac{\tau_0(\tau_0+2-2d)}4.
}
If $\mathcal H^{(0)}(z,\zb)$ is known, the first few functions $\mathcal H^{(m)}(z,\zb)$ may be computed by solving the relation \eqref{eq:casimirrelationforTCB} in specific limits.

\section{Inversion procedures}\label{sec:invprocedures}

After identifying the operators that contribute at a given order and computing their double-discontinuities, the job is in principle done. The corresponding CFT-data is packaged in the function $\mathbf U(\log z,\hb)$ which is given by the inversion integral \eqref{eq:masterformula}. In practice, however, the extraction of the CFT-data requires the integral to be explicitly computed, which in general is a difficult task. In particular, it is desirable to extract $\mathbf U(\log z,\hb)$ as a closed-form expression in $\hb$, especially since the actual CFT-data is given in terms of derivatives of this function through \eqref{eq:aellfromU}.

The formulation of the inversion procedure as a one-dimensional integral over a compact domain has an important advantage compared to earlier and alternative procedures. Any candidate result for a specific inversion problem, regardless how it was extracted, can and should be checked by a direct numerical integration. This is done by high-precision evaluation of the integral for a number of finite values of $\hb$, not necessarily related to integer spin. This is especially important when the result has been derived using an indirect method such as a large spin expansion.

\subsection{Direct evaluation}\label{sec:directinversion}
We start by giving some examples of how the inversion integral can be solved by direct evaluation. The first is the result quoted already in section~\ref{sec:reciprocityrevisited} and discussed again in section~\ref{sec:elementaryinversions}, namely the inversion of a power $\xit^{-p}$ where $\xit=(1-\zb)/\zb$. Since the crossing factor $\parr{\frac uv}^{\Delta_\phi}$ takes exactly this form in the collinear limit for $p=\Delta_\phi$, this result corresponds to inverting the identity operator. We summarise it in the following way:
\begin{invtool}\label{inv:identity}
The identity operator $\1$ appearing in the crossed-channel OPE in the $\phi$ four-point function generates CFT-data given by $\mathbf U(\log z,\hb)=
\AA[\Delta_\phi](\hb)$, where we define
\beq{\label{eq:InversionIdentity}
\AA[\Delta_\phi](\hb)=\frac{2\Gamma(\hb+\Delta_\phi-1)}{\Gamma(\Delta_\phi)^2\Gamma(\hb-\Delta_\phi+1)}.
}
\end{invtool}\vspace{2.5ex}

\noindent We show this by first noting that $\dDisc[\xit^{-\Delta_\phi}]=2\sin^2(\pi\Delta_\phi)\xit^{-\Delta_\phi}$. Then, the integral representation \eqref{eq:HyperInt1} gives that
\beq{\label{eq:invXiDphi}
\INV[\xit^{-\Delta_\phi}]=\frac{2\sin^2(\pi\Delta_\phi)}{\pi^2}\int\limits_{[0,1]^2}\frac{\df t \df \zb}{t(1-t)}\parr{\frac{t(1-t)}{1-t \zb}}^\hb\frac{\zb^{\Delta_\phi+\hb-2}}{(1-\zb)^{\Delta_\phi}}.
}
The $\zb$ integral can be computed using the integral \eqref{eq:BetaIntegralmod} in appendix~\ref{app:identities}. The resulting hypergeometric function collapses since $_2F_1(a,b;b;x)=(1-x)^{-a}$, by which the $t$ integral takes the form of the Euler integral of first kind (Beta function), \eqref{eq:BetaIntegralDef}. Finally, using the identity \eqref{eq:subSinGamma} we replace the factors of $\sin(\pi \Delta_\phi)$ by Gamma functions and arrive at the result $\AA[\Delta_\phi](\hb)$. We may write this result on a general form as
\beq{\label{eq:Ahatminusp}
\INV[\xi^p]=\AA[-p](\hb).
}
The simplicity of the result \eqref{eq:Ahatminusp} can be compared to the situation for a general factor $(1-\zb)^p\zb^{-q}$, where the result is \cite{Li:2019dix}
\beq{\label{eq:inversionpq}
\INV\parrk{\frac{(1-\zb)^p}{\zb^q}}=\frac{2\Gamma(\hb-q-1)}{\Gamma(-p)^2\Gamma(\hb+p-q)}\frac{\Gamma(\hb)^2}{\Gamma(2\hb)\Gamma(1+p)}\,{ _3F_2}\!\left(\!\left. {\hb,\hb,\hb-q-1} ~\atop~{\!\!\!\!2\hb,\hb+p-q} \right| 1\right).
}
For $p=q$ this reduces to \eqref{eq:Ahatminusp} using the identity \eqref{eq:Buhringsrelation}.

\subsection{Large conformal spin expansions}
While direct evaluation of the inversion integral is limited to cases where suitable integral identities exist, the result for the inversion of a single power of $\xit=\frac{1-\zb}\zb$ can be used to generate a large $J^2$ expansion for the inversion of any function $G(\zb)$. First, we require that $G(\zb)$ admits an expansion in powers of $w=1-\zb$. This is a natural expansion for crossed-channel conformal blocks, since it corresponds to the crossed-channel OPE limit. Order by order, this expansion can be converted into a series expansion in $\xit=\frac{w}{1-w}$. Then inversion~\ref{inv:identity}, or equivalently \eqref{eq:inversionpq}, gives
\beq{\label{eq:sumexpansionbeforeinversion}
\INV\Big[\sum_{p=0}^\infty c_p\xit^{p-\alpha}\Big]=\sum_{p=0}^\infty c_p\AA[\alpha-p].
}
Since term $\AA[\alpha-p]$ expands as
\beq{
\AA[\alpha-p]=\frac{2}{\Gamma(\alpha-p)^2}\parr{\frac{1}{J^2}}^{1-\alpha+p}\parr{1+\frac{(p - \alpha) (1 + p - \alpha) (2 + p -\alpha)}{3J^2}+\ldots},
}
any truncated sum of the form \eqref{eq:sumexpansionbeforeinversion} will generate the same number of terms in the large $J$ expansion.

This method applies also to expansions of $G(\zb)$ which contain logarithmic terms. For instance, a term $\xit^p\log \xit$ can be generated from applying a derivative of the exponent. Using this, we get
\beq{
\INV[\xit^p\log\xit]=-\de_a\left.\AA[a-p]\right|_{a=0}.
}
The method described here therefore generates a large $J$ expansion for the functions $U^{(p)}_\hb$ that takes the form \eqref{eq:Uphbexpansion}. However, the ultimate goal is find these functions in a closed form. In a number of situations, this can be achieved by comparing with expansions of known functions. In theories near four dimensions, for instance, CFT-data typically takes the form of rational functions in $J^2$, multiplied by the harmonic number $S_1(\hb-1)$ and its generalisations. In practice, this is done by creating an ansatz consisting of suitable functions and matching this with the expansion created through \eqref{eq:sumexpansionbeforeinversion}.

One example of this is the inversion of the contribution from the scalar bilinear $\phi^2$ in the $\epsilon$ expansion. 
\begin{invtool}\label{inv:scalarEps}
In the $d=4-\epsilon$ expansion, the bilinear scalar
$\Delta_{\phi^2}=2\Delta_\phi+\gamma$ with OPE coefficient
$c^2_{\phi\phi\phi^2}$, assuming $\gamma=\gamma^{(1)}\epsilon+\gamma^{(2)}\epsilon^2+\ldots$, has the following inversion expanded to order $\epsilon^3$
\begin{align}\nonumber
\mathbf U(\log z,\hb) &= \, \frac{c^2_{\phi\phi\phi^2}}2 \gamma^2\frac{1}{J^2}\left(-1-\gamma+\epsilon+\gamma S_1(\hb-1)\right)\log z\\&+\frac{c^2_{\phi\phi\phi^2}}2 \gamma^2\frac{1}{J^4}\left(
-1+(J^2\zeta_2+1)\epsilon+(S_1(\hb-1)-J^2\zeta_2-1)\gamma
\right),\label{eq:scalarEps}
\end{align}
where $\zeta_n$ denote Riemann's zeta function.
\end{invtool}\vspace{2.5ex}
 
 \noindent 
We derive this in section~\ref{sec:orderepsthree}, starting from the scalar conformal block \eqref{eq:ScalarBlockAnyD}, putting $d=4-\epsilon$, $\Delta=2\Delta_\phi+\gamma$ and $\Delta_\phi=1-\frac\epsilon2+O(\epsilon^2)$. In that case the sums defining the scalar conformal block can be explicitly computed and generate the type of polylogarithms encountered in section~\ref{sec:invitation}. Expanding to order $\epsilon^3$ we get 
\beqa{\nonumber
\dDisc & \parrk{\parr{\frac uv}^{\Delta_\phi} \left.c_{\phi\phi\phi^2}G^{(4-\epsilon)}_{2\Delta_\phi+\gamma,0}(v,u)\right|_{z^{\Delta_\phi}}
}
\\&=\dDisc\Big[{c_{\phi\phi\phi^2}}\log^2(1-\zb)\Big(\Big(
-\frac{\gamma^2}8-\frac {\gamma^3}{48}\log(1-\zb)+\frac{\epsilon\gamma^2-\gamma^3}{8}
\Big)\log z
\nonumber\\&\quad+\frac{\gamma^2\log \zb}{8}+\frac{(\epsilon\gamma^2-\gamma^3)(\zeta_2-\log \zb)}{8}+\frac{\gamma^3}{48}(6\,\mathrm{Li}_2(1-\zb)+\log(1-\zb)\log \zb) \Big)\Big].
}
The techniques described above are then used to generate a large $J$ series for the inversion of this expression, which can be matched with a suitable ansatz of terms of the form $J^{-2p}$ and  $J^{-2p}S_1(\hb-1)$. The result is the expression~\eqref{eq:scalarEps} quoted in inversion~\ref{inv:scalarEps}. Alternatively, the large $J$ series can be generated by the H-function method. Of course, the results for the various inversions of single terms, for instance $\INV[\log^2(1-\zb)\log \zb]=-4/J^4$, can be recorded for later use by extending table~\ref{tab:initialinversions}. We give such an extension in appendix~\ref{integrals}, containing all inversions needed for studying the $\epsilon$ expansion to order $\epsilon^4$.

We stress again the importance of checking the inversion results by numerical integration. The large $J^2$ expansions generated by the series method are often asymptotic, meaning that they have zero radius of convergence around $J^{-2}=0$. For instance, this is the case with the expansion \eqref{eq:S1expansion} of the harmonic number $S_1(\hb-1)$ and similar functions. It is therefore possible that the true result of the inversion integral contains additional terms, exponentially suppressed as $J\to\infty$. Fortunately, this is not the case for the harmonic number and its generalisations, where the expansions generated by the series \eqref{eq:sumexpansionbeforeinversion} agree with the standard large spin expansions given in e.g.\ \cite{Albino2009}. When extracting CFT-data at finite spin, one should always use the closed form expression and not the large $J$ expansion. For instance, in \cite{Albayrak:2019gnz} it was shown that taking into account the finite spin corrections in the case of the the 3d Ising model improved the precision of the computations in \cite{Simmons-Duffin:2016wlq}, which were derived using the large spin asymptotics.

\subsection[Inversion and the $\SL2\R$ Casimir]{Inversion and the $\boldsymbol{\SL2{\mathrm R}}$ Casimir}\label{eq:inversionandcasimir}

In section~\ref{sec:tastercasimir} we demonstrated how the $\SL2\R$ Casimir operator $\Dbar=(1-\zb)\zb^2\de^2_\zb-\zb^2\de_\zb$ was used to relate the inversion $\INV[ G(\zb)]$ to the inversion $\INV[\Dbar G(\zb)]$ by a simple division by $J^2$. This principle can be very useful in proving the exact form of some inversions. 

To give another example, consider the inversion of a scalar operator of dimension $\Delta=2$ in the correlator of external operators of dimension $\Delta_\phi=\mu-1$, where $\mu=d/2$. In \eqref{eq:genericcrossedblock} we gave the double lightcone expansion of a single crossed-channel conformal block. Multiplying by $\parr{\frac{\zb}{1-\zb}}^{\mu-1}$ and specialising to $\Delta=2$ we need to invert
\beqa{
\nonumber
G(\log z,\zb)=
-\left(\frac{\zb}{1-\zb}\right)^{\mu-1}(1-\zb)\big[&(\log z+\log \zb)\,{_2F_1}\left(1,1;3-\mu;1-\zb\right)\\&+2\pp a\left.  {_2F_1}\left(1+a,1;3-\mu;1-\zb\right)\right|_{a=0}\big].
\label{eq:toinvertsigmatree}
}
While integrating these hypergeometric functions against the $k_\hb(\zb)$ appearing in the inversion integral appears to be a difficult task, the situation simplifies drastically when acting on this function by $\Dbar$:
\beq{
\Dbar G(\log z,\zb)=-\left(\frac{\zb}{1-\zb}\right)^{\mu-1}\!\!\left((\mu-2)^2(\log z\!+\!\log \zb)+(1-\zb){_2F_1}\left(1,1;3-\mu;1-\zb\right)\right).
}
It is now trivial to invert the term proportional to $\log z$, since it is just a pure power of $\xit=\frac{1-\zb}\zb$. The result for $U^{(1)}_{\hb}$ is therefore
\beq{
2\frac1{J^2}\INV[-(\mu-2)^2\xit^{1-\mu}]=-2(\mu-2)^2\frac{\AA[\mu-1](\hb)}{J^2}.
}
Also the inversion of the last term follows straightforwardly, since it takes the same form as the original $\log z$ term. By acting once again by the Casimir $\Dbar$ it is clear that it inverts to $(\mu-2)^2\AA[\mu-1](\hb)/J^4$. Finally, the term proportional to $\log \zb$ can be inverted by expanding the inversion integrand in powers of $\zb$ where $k_\hb(\zb)$ is regular, and inverting term by term. The result is given in terms of the combination\footnote{This combination of harmonic numbers is closely related to the function $\AA[\alpha](\hb)$ noting that
\beq{\nonumber
\de_\alpha \AA[\alpha](\hb)=\AA[\alpha](\hb)\parr{-2S_1(\alpha-1)+S_1(\hb-2+\alpha)+S_1(\hb-\alpha)}.
}
Moreover, $\SS[\alpha](\hb)$ has a large $J$ expansion that is free from terms $\log J$. This is in agreement with the fact that the function $G(\log z,\zb)$ has no terms scaling as $\log(1-\zb)$ in the limit $\zb\to1$.
}
\beq{
\SS[\alpha](\hb)=2S_1(\hb-1)-S_1(\hb-2+\alpha)-S_1(\hb-\alpha).
}
We can summarise our findings, first derived in \cite{Paper4}, in the following way.
\begin{invtool}\label{inv:scalarN}
The contribution from a scalar $\O$ with $\Delta_\O=2$ in the $\phi$ four-point function, where $\Delta_\phi=\mu-1$, in generic spacetime dimension $d=2\mu$ takes the form
\beq{
\mathbf U(\log z,\hb)= (\mu-2)^2c^2_{\phi\phi\O}\frac{\AA[\mu-1](\hb)}{J^2}\left(
-\log z+\SS[\mu-1](\hb)-\frac1{J^2}
\right).
}
\end{invtool}

\section{Applications of large spin perturbation theory}\label{sec:applicationsofLSPT}

We finish the practical guide to large spin perturbation theory by reviewing the numerous applications of the framework that have appeared in the literature. We limit ourselves to the work in the direct spirit of \cite{Alday2016} and its companion paper \cite{Alday2016b}, with or without the inversion integral, and we do not aim to cover the whole range of analytic bootstrap work that we briefly summarised in section~\ref{sec:lightconecrossing}.

In \cite{Alday2016b} the leading order implications of large spin perturbation theory were studied in a variety of examples.  This included leading corrections to anomalous dimensions in the $\OO N$ model at order $\epsilon^2$ and order $1/N$, as well as the leading order corrections to dimensions and OPE coefficients in a generic conformal guage theory.

The $\epsilon$ expansion is particularly suitable for large spin perturbation theory, which can be realised by studying figure~\ref{fig:spectrumWF} in section~\ref{sec:WFspectrum} in connection with the considerations in proposition~\ref{prop:firstappearance} above. Since all crossed-channel operators, except the identity, have twists of the form $2\Delta_\phi+2n+O(\epsilon)$, illustrated by the grey bands in figure~\ref{fig:spectrumWF}, their contributions to the double-discontinuity are suppressed by at least an order $\epsilon^2$. In addition, the weakly broken currents are not corrected until order $\epsilon^2$, and the higher twist operators have OPE coefficients of order $\epsilon^2$ or higher, and therefore the vast majority of operators do not contribute until order $\epsilon^4$. This means that the whole double-discontinuity up to order $\epsilon^3$ is generated from the identity operator $\1$ and the bilinear scalar $\phi^2$. All CFT-data to this order therefore follow from a direct application of inversions~\ref{inv:identity} and \ref{inv:scalarEps}. As we describe in detail in chapter~\ref{ch:paper2}, based on \cite{Paper2}, the whole double-discontinuity at order $\epsilon^4$ can also be computed in terms of the CFT-data at lower orders, by an iterative procedure in the spirit of figure~\ref{fig:inversionflowchart}. This consists of two contributions: the weakly broken currents themselves and operators of approximate twist four (the $n=1$ case in the discussion in section~\ref{sec:WFspectrum}). The latter contribution is found through an ansatz based on transcendentality, and some input from the literature is needed to fix some coefficients.

In \cite{Paper3}, which we do not have room to reproduce in this thesis, the problem at order $\epsilon^4$ was revisited in the $\OO N$ symmetric case, and by using the projections to subleading twists in the spirit of our section~\ref{sec:projections}, all dependence on literature values was circumvented. The resulting OPE coefficients of broken currents $\mathcal J_{R,\ell}$ at order $\epsilon^4$ were all new results, as well as the scaling dimensions in the rank two representations $T$ and $A$. From the OPE coefficients, new results at order $\epsilon^4$ for the central charges $C_T$ and $C_J$ were computed, which we give at the end of chapter~\ref{ch:paper2}.

A different approach is needed to study conformal gauge theories, where the simplest operator in the spectrum is a bilinear scalar of dimension $\Delta_\O=2+O(g)$. The perturbative structure of the four-point of such an operator was determined already in \cite{AldayBissi2013}, and was revisited in \cite{Alday2016b}. In the double lightcone limit, the most general expression at order $g$ takes the form
\beq{\label{eq:ansatzforgaugetheory}
\G^{(1)}(u,v)=\frac uv\parr{a_{11}\log u\log v+a_{10}\log u+a_{01}\log v+a_{00}}+O(u^2),
}
where crossing relates $a_{10}$ and $a_{01}$ through the external anomalous dimensions $g\gamma_{\mathrm{ext}}=\Delta_\O-2$. In chapter~\ref{ch:paper1}, based on \cite{Paper1}, we will show that this expression, complemented by a contribution at spin zero, generates the entire double-discontinuity of the correlator and can be completed to an explicit expression for the whole correlator. To understand why this is possible we refer back to figure~\ref{fig:spectrumNN4} in section~\ref{sec:NN4spectrum} for the special case of \NN4 SYM; the structure of the spectrum in a general conformal gauge theory is similar. From this figure it is clear that the identity operator together with the leading twist operators generate the entire double-discontinuity to order $g$. All other operators are suppressed by proposition~\ref{prop:firstappearance}, as indicated by the grey bands in figure~\ref{fig:spectrumNN4}. Since the CFT-data for leading twist operators can be extracted from the ansatz \eqref{eq:ansatzforgaugetheory}, the four constants $a_{ij}$ together with the anomalous dimension at spin zero generate a five-dimensional solution space for all CFT-data entering the correlator, and therefore for the most general form of the four-point function.

It is natural to extend this to next order, where the higher-twist operators themselves contribute to the double-discontinuity. Due to the complicated mixing of operators, this has not yet been achieved. However, as we mentioned already in section~\ref{sec:lightconecrossing}, more progress has been made in the planar limit at strong coupling. There, the expansion is typically phrased in holographic language and written as
\beq{
\G(u,v)=\G_{\mathrm{disc.}}(u,v)+\frac1{N^2}\G_{\mathrm{tree}}(u,v)+\frac1{N^4}\G_{\mathrm{loop}}(u,v)+O(N^{-6}),
}
where the subscripts refer to disconnected, tree-level and one-loop diagrams in supergravity. By studying systems of tree-level supergravity correlators of half-BPS operators $\O_p$ in the traceless symmetric $[0,p,0]$ representations of the R-symmetry $\SU4$, the mixing problem was resolved, and subsequently the loop supergravity correlator could be determined. This was first done for the $[0,2,0]={\boldsymbol{20'}}$ case \cite{AldayBissi2017,Aprile:2017bgs,Aprile:2017xsp} and later for general half-BPS operators \cite{Aprile:2018efk,Aprile:2019rep,Alday:2019nin}, as well as with string theory corrections \cite{Alday:2018pdi}, the latter corresponding to $1/\lambda$ for the `t Hooft coupling $\lambda=g^2_{\mathrm{YM}}N$.

CFTs in three dimensions were studied in \cite{Aharony:2018npf}, which considered CFTs with weakly broken higher spin symmetry and gauge group $\SU N$ for large $N$, the main example being Chern--Simons theories coupled to a fundamental complex scalar or a Dirac fermion. The object of study was the four-point function of the smallest-dimension scalar $J_0$, with dimension $\Delta_{0}=1$ or $\Delta_0=2$ depending on theory. Similar to weakly coupled 4d \NN4 SYM, the OPE contains broken currents $\mathcal J_\ell$ which generate the double-discontinuity at order $1/N$, as well as GFF operators $[J_0,J_0]_{n,\ell}$. Also here, mixing amongst the higher twist operators prevents a full determination at order $1/N^2$, which can only be determined in the case where there is no mixing. The general case at order $1/N^2$ remains an open problem.

In \cite{Paper4}, which we summarise in section~\ref{sec:paper4short}, we studied the critical $\OO N$ model at large $N$ based on some initial considerations in \cite{AldayZhiboedov2015} and \cite{Alday2016b}. Referring to figure~\ref{fig:spectrumON} in section~\ref{sec:ONspectrum}, we see that the leading double-discontinuities are generated by the identity operator $\1$ and the auxiliary field $\sigma$. At subleading orders, the contribution from the broken currents $\mathcal J_{R,\ell}$ needs to be computed, as well as contributions from the GFF operators $[\sigma,\sigma]_{n,\ell}$. As we review in section~\ref{sec:paper5short}, based on \cite{Paper5}, the leading computations in the large $N$ expansion, as well as the $\epsilon$ expansion, generalise to $\phi^4$ theories with other global symmetries. 

The tools developed for the $\OO N$ model can in fact be used to study multicritical $\phi^{2\theta}$ theories near their critical dimensions $d_{\mathrm c}(\theta)$. We describe how to do this in section~\ref{sec:multicritical}. Also cubic theories in $6-\epsilon$ dimensions have been studied using large spin perturbation theory in \cite{Alday2016b} and \cite{Goncalves:2018nlv}.

The inversion formula has been used to reproduce results in the heavy-light bootstrap \cite{Li:2019zba}. The purpose is to study universal properties of CFTs with holographic interpretation, where the expansion parameter is the inverse of the number of degrees of freedom, or equivalently $1/C_T$. The operators are divided into light $\mathcal L$ and heavy $\mathcal H$, where $\Delta_{\mathcal L}=O(1)$ and $\Delta_{\mathcal H}=O(C_T)$, and the starting point of the bootstrap analysis is the mixed correlator $\expv{\mathcal H\mathcal L\mathcal L\mathcal H}$. The contributions from the identity operator $\1$ and the stress tensor $T\munu$ in the crossed-channel generate double-twist operators $[\mathcal H,\mathcal L]_{n,\ell}$ with anomalous dimensions of the order $\gamma_{n,\ell}\sim J^{-(d-2)}/C_T$. The next step is to look at crossing for the correlator $\expv{\mathcal L\mathcal L\mathcal H\mathcal H}$, where the minimal set of direct-channel operators are $\1$, $T\munu$ and the \emph{double-stress tensors} $[T,T]_{0,\ell}$. In \cite{Li:2019zba} it was shown how the OPE coefficients for the double-stress tensors can be computed from the large spin perturbation theory, matching with the results of \cite{Kulaxizi:2019tkd}. More precisely, the kernel method can be used to determine the crossed-channel contribution from $[\mathcal H,\mathcal L]_{n,\ell}$. By proposition~\ref{prop:crossedsums} this gives the power $z^{d-2}$ which exactly matches the reference twist of the double stress tensors. The OPE coefficients are proportional to $1/C_T^2$ by proposition~\ref{prop:firstappearance} and finally the twists of $[\mathcal H,\mathcal L]_{n,\ell}$ correspond to the correct asymptotic spin dependence of the double stress tensor OPE coefficients\footnote{As mentioned in \cite{Li:2019zba}, an important assumption in deriving this result is that the operators $[\mathcal H,\mathcal L]_{n,\ell}$ are non-degenerate.}.

Some further applications and generalisations of large spin perturbation theory have been made. In \cite{Loon2017}, scalar correlators in fermionic theories were considered, with specific applications to the Gross--Neveu model in $2+\epsilon$ dimensions and the Gross--Neveu--Yukawa model in $4-\epsilon$ dimensions. Interestingly, the former case admits an all twist result that is very similar to our results in chapter~\ref{ch:paper1}, where the spacetime dimensionality is $2$ rather than $4$. In \cite{Elkhidir:2017iov} conformal blocks, crossing equation and large spin expansion were developed for the $\langle{\hspace{0.5pt}\phi\hspace{1pt}\psi\hspace{1pt}\phi\hspace{1pt}\bar\psi\hspace{1pt}}\rangle$ correlator for a scalar $\phi$ and a fermion $\psi$ in four dimensions. Similarly, in \cite{Albayrak:2019gnz} some initial considerations were made for the fermion four-point function in three dimensions. Alternative versions of the Lorentzian inversion formula have also been derived in the case of defect CFTs \cite{Lemos:2017vnx,Liendo:2019jpu} and CFTs at finite temperature \cite{Iliesiu:2018fao}, where the latter was used in \cite{Iliesiu:2018zlz} to study the thermal 3d Ising model.

In connection to the lightcone bootstrap, one might attempt to determine in full generality the exact contribution from a crossed-channel operator to a given direct-channel operator. This is known as the \emph{crossing kernel}, or $6j$ symbol \cite{Sleight:2018epi,Sleight:2018ryu,Liu:2018jhs,Cardona:2018qrt}. 
The results derived in these references are non-perturbative and do not simply translate to large spin perturbation theory. Often they are phrased as a contribution to $\gamma_\ell$, but in the language of this thesis, the results rather match $U^{(1)}_{\hb}/\AA[\Delta_\phi](\hb)$, which, contrary to $\gamma_\ell$, is additive in crossed-channel contributions. As we stressed in section~\ref{sec:fromcrossedchannel}, in a perturbative setting, the contribution from twist families cannot be computed by inverting operators one by one, so crossing kernels are not enough to perform large spin perturbation theory beyond leading order in perturbation.

\chapter[Wilson--Fisher model in the \texorpdfstring{$\boldsymbol \epsilon$}{epsilon} expansion]{Wilson--Fisher model in the $\boldsymbol \epsilon$ expansion}
\label{ch:paper2}

\section{Introduction}

In this chapter we will apply the method of large spin perturbation theory to the Wilson--Fisher (WF) model in $d=4-\epsilon$ dimensions. In \cite{Alday2016b} results were obtained for the anomalous dimensions of weakly broken currents to the first non-trivial order in $\epsilon$.
In a series of papers \cite{Gopakumar:2016wkt,Gopakumar:2016cpb, Dey:2017fab,Dey:2016mcs} a proposal has been put forward for an alternative method to compute CFT-data analytically. In this approach one uses Mellin space and crossing symmetry is built in. Consistency with the OPE then constrains the CFT-data. This method has been applied to the WF model in the $\epsilon$ expansion leading to impressive results. More precisely, the CFT-data for weakly broken currents has been obtained to cubic order in $\epsilon$. The purpose of this chapter is first to show how these results can be recovered from the perspective of large spin perturbation theory using the Lorentzian inversion formula. To cubic order the relevant divergences of the correlator arise, via crossing symmetry, from just two operators in the crossed channel: the identity operator and the bilinear scalar operator. This makes our derivation very simple: in the present framework it essentially involves a first-order computation. The simplicity of our method is also manifest when dealing with the $\mathrm O(N)$ model where the results to cubic order follow straightforwardly from those for $N=1$. A remarkable feature of our computation is that the convergence properties of the inversion integral allow to extrapolate the results down to spin zero. Conservation of the stress tensor together with a matching condition for spin zero lead to two non-trivial constraints, that allow to fix not only the dimension of the external operator but also the dimension of the scalar operator $\phi^2$. We then move on to the computation at fourth order. In this case the divergences of the correlator are more involved and arise from infinite towers of operators with arbitrarily large spins. The computation is complicated by the appearance of new operators in the OPE at quadratic order. A remarkable feature of these operators, together with intuition from perturbation theory, makes it possible to guess their contribution to the divergence, and hence to determine the CFT-data of weakly broken currents to fourth order. The results for the anomalous dimensions agree with those in the literature, computed by Feynman techniques, while the OPE coefficients are a new result. From the latter we deduce the central charge of the WF model to fourth order in the $\epsilon$ expansion:
\begin{equation}\label{eq:CTintrores}
\frac{C_T}{C_{T,\text{free}}}= 1 - \frac{5}{324} \epsilon^2 - \frac{233}{8748}\epsilon^3 - \left(\frac{100651}{3779136}-\frac{55}{2916}\zeta_3\right) \epsilon^4+\ldots,
\end{equation}
where we stress the fact that the contribution proportional to $\epsilon^4$ is also negative. 
 
 This chapter is organised as follows. The computation up to cubic order is presented in section \ref{thirdorder}. After introducing the basic ingredients we explain the connection between the inversion formula and large spin perturbation theory. Since we are dealing with leading twist operators, the inversion problem for $\SL2{\mathbb{R}}$ suffices, and we give a quick derivation of the $\SL2\R$ inversion formula. Then we proceed to obtain the CFT-data for leading twist operators, up to this order, from the double-discontinuity of the correlator. We also show how to generalise these results to the $\mathrm O(N)$ model. In section \ref{fourthorder} we tackle the problem to fourth order and give the full answer for the anomalous dimensions and OPE coefficients of leading twist operators. We finish with some conclusions. Appendix \ref{integrals} contains a database of the necessary inversion integrals to compute the CFT-data at hand, while appendix \ref{ddisc} contains expressions for double discontinuities at fourth order.   

\section[Lorentzian OPE inversion in the \texorpdfstring{$\epsilon$}{epsilon} expansion]{Lorentzian OPE inversion in the $\boldsymbol{\epsilon}$ expansion}
\label{thirdorder}

\subsection{Generalities}

Consider the four-point correlator of a scalar field $\phi$ in a $d$-dimensional CFT
\begin{equation}
\langle \phi(x_1) \phi(x_2)  \phi(x_3)  \phi(x_4) \rangle = \frac{{\cal G}(z,\bar z)}{x_{12}^{2\Delta_\phi}x_{34}^{2\Delta_\phi}}.
\end{equation}
It admits a decomposition in conformal blocks, which in the direct channel decomposition reads
\begin{equation}
{\cal G}(z,\bar z) = \sum_{\Delta,\ell} a_{\Delta,\ell} G^{(d)}_{\Delta,\ell}(z,\bar z),
\end{equation}
where $G^{(d)}_{\Delta,\ell}(z,\bar z)$ are the $d-$dimensional conformal blocks defined in section~\ref{sec:blockology}. We assume that there is a free point where the correlator reduces to that of generalised free fields (GFF)
\begin{equation}
{\cal G}^{(0)}(z,\bar z) = 1+ (z \bar z)^{\Delta_\phi} + \left(\frac{z \bar z}{(1-z)(1-\bar z)}\right)^{\Delta_\phi}.
\end{equation}
The intermediate operators are the identity and towers of bilinear operators of twist $2\Delta_\phi +2n$ and spin $\ell$. We will be interested in leading twist operators with $n=0$. In this case the GFF OPE coefficients \eqref{eq:aGFF} reduce to
\begin{equation}
a^{(0)}_{\ell} = \frac{2 \left((\Delta_\phi )_\ell\right){}^2}{\ell! (\ell+2 \Delta_\phi -1)_\ell}.
\end{equation}
 As we show below, these OPE coefficients are fixed by the structure of divergences of the correlator. Next we consider perturbations by a small parameter $g$. This introduces a correction to the scaling dimensions and OPE coefficients of the leading-twist operators
\begin{eqnarray}
\Delta_\ell &=& 2\Delta_\phi + \ell+ \gamma^{(1)}_\ell g + \ldots\\
a_{\ell}  &=& a^{(0)}_{\ell} +a^{(1)}_{\ell} g+ \ldots .
\end{eqnarray}
We will assume that at this order no new operators appear in the OPE $\phi \times  \phi$. From the analysis of \cite{Alday2016} it follows that the only solutions consistent with crossing symmetry have finite support in the spin. For generic $\Delta_\phi$ these solutions can be constructed following \cite{Heemskerk2009}. For the present chapter we will be interested in the case $\Delta_\phi=\frac{d-2}{2}$ at leading order in $g$. In this case it was proven in \cite{Alday2016b} that crossing symmetry admits a non-trivial solution only around $d=4$, with support for spin zero. We {\it define} the coupling constant $g$ as the anomalous dimension of the bilinear operator with spin zero
\begin{eqnarray}\label{eq:deltazero}
\Delta_0 = 2\Delta_\phi + g.
\end{eqnarray}
All other quantities will be computed in terms of this coupling constant. In \cite{Alday2016b} it was also shown that $\Delta_\phi$ can receive corrections only at order $g^2$. Note that the dimensionality of space-time can differ from four by at most something of order $g$, so that $d=4-\epsilon$ with $g \sim \epsilon$. The correction to the OPE coefficients can be found through an extension of the analysis of \cite{Alday2016b}. Again, the corresponding solution has support only for spin zero and one finds $a_{0}  =a^{(0)}_{0}(1- g +\ldots)$. In summary, for spin two and higher the corrections start at order $g^2$
\begin{eqnarray}
\Delta_\ell &=& 2\Delta_\phi + \ell+  \gamma^{(2)}_\ell g^2 + \ldots,\qquad\ell=2,4,\ldots,\nonumber\\
a_{\ell}  &=& a^{(0)}_{\ell} +  a^{(2)}_{\ell} g^2 + \ldots,\qquad\ell=2,4,\ldots ,
\end{eqnarray}
and the same is true for the external operator
\begin{eqnarray}
\Delta_\phi= \frac{d-2}{2} + \gamma^{(2)}_{\phi} g^2 + \ldots.
\end{eqnarray}
We would like to find the corrections consistent with crossing symmetry. Our method relies on the fact that the double-discontinuity of the correlator contains all the relevant physical information. Let us explain this in more detail.
\subsection{From large spin perturbation theory to an inversion formula}
\label{sec:derivationCollinearFormula}
Consider a basis of $\SL2\R$ conformal blocks $f_{\Delta,\ell}(\bar z)$. We find it convenient to introduce the following normalisation
\begin{equation}
f_{\Delta,\ell}(\bar z) = r_{\frac{\Delta+\ell}{2}} k_{\frac{\Delta+\ell}{2}}(\bar z),\qquad r_h= \frac{\Gamma(h)^2}{\Gamma(2h-1)},
\end{equation}
with $k_h(\bar z)=\bar z^h {_2F_1}(h,h,2h,\bar z)$. We are interested in solving the following inversion problem: find $A_\ell$ such that
\begin{equation}
\label{SL2R}
\sum_{\substack{\Delta=2\Delta_\phi+\ell,\\\ell=0,2,\ldots}} A_\ell f_{\Delta,\ell}(\bar z) = G(\bar z),
\end{equation}
for a given $G(\bar z)$ containing an enhanced singularity as $\bar z \to 1$. By enhanced singularity we mean a contribution which becomes power-law divergent upon applying the Casimir operator a finite number of times, and as such it cannot be obtained by a finite number of conformal blocks. This is equivalent to saying that $G(\bar z)$ contains a double-discontinuity. For a correlator the double-discontinuity is defined in \eqref{eq:ddiscdef}
\begin{equation}
 \dDisc [G(\bar z)] \equiv G(\bar z) -\frac{1}{2} G^\circlearrowleft(\bar z)-\frac{1}{2} G^\circlearrowright(\bar z).
\end{equation}
An algorithm to find $A_\ell$ as a series in $1/\ell$ to all orders was developed in \cite{Alday2016}. The idea is the following. First recall that the $SL(2,\mathbb{R})$ conformal blocks are eigenfunctions of a quadratic Casimir operator \eqref{eq:collcasrel}
\begin{equation}
\Dbar f_{\Delta,\ell}(\bar z)  = J^2 f_{\Delta,\ell}(\bar z) ,
\end{equation}
where $\Dbar =\bar z^2 \bar \partial(1-\bar z)\bar \partial$ and $J^2=\frac{1}{4}(\Delta+\ell)(\Delta+\ell-2)$. We then assume that $A_\ell \equiv A(J)$ admits an expansion in inverse powers of the conformal spin
\begin{equation}
A(J) = \sum_m \frac{a_m}{J^{2m}}
\end{equation}
and define the following family of functions
\begin{equation}
h^{(m)}(\bar z)= \sum_{\substack{\Delta=2\Delta_\phi+\ell,\\\ell=0,2,\ldots}}  \frac{f_{\Delta,\ell}(\bar z)}{J^{2m}} .
\end{equation}
From the explicit form of the blocks we can compute
\begin{equation}
h^{(0)}(\bar z)= \sum_{\substack{\Delta=2\Delta_\phi+\ell,\\\ell=0,2,\ldots}} f_{\Delta,\ell}(\bar z) = \frac12 \frac{\bar z}{1-\bar z}   +\, \text{regular},
\end{equation}
where the regular terms do depend on $\Delta_\phi$ but are not important for us. The sequence of functions $h^{(m)}(\bar z)$ can then be generated by the inverse action of the Casimir 
\begin{equation}
\label{recursion}
\Dbar h^{(m+1)}(\bar z) = h^{(m)}(\bar z).
\end{equation}
The inversion problem (\ref{SL2R}) then amounts to decomposing $G(\bar z)$ in the basis of functions $h^{(m)}(\bar z)$. The precise range of $m$ depends on the specific form of  $G(\bar z)$. The recursion (\ref{recursion}) can be used to systematically construct the functions $h^{(m)}(\bar z)$ and hence find the coefficients $a_m$. More specifically, one matches the double-discontinuity on both sides of (\ref{SL2R}). To make contact with the inversion formula of \cite{Caron-Huot2017} assume there exists a family of projectors $K^{(m)}(\bar z)$ such that
\begin{equation}
\int_0^1 \df\bar z K^{(m)}(\bar z) \dDisc \left[ h^{(n)}(\bar z) \right]= \delta^{mn}.
\end{equation}
Having the projectors $K^{(m)}(\bar z)$ we can write
\begin{equation}
A(J)= \int_0^1 \df\bar z K(\bar z,J) \dDisc \left[G(\bar z)\right],
\end{equation}
where 
\begin{equation}
K(\bar z,J) = \sum_m \frac{K^{(m)}(\bar z)}{J^{2m}}.
\end{equation}
As will be clear momentarily, the precise form of these projectors will not be necessary. Acting on both sides of (\ref{SL2R}) with the Casimir operator $\Dbar$ and integrating by parts we obtain
\begin{equation}
\left( {\Dbar\,}^\dagger-J^2\right) K(\bar z,J) =0
\end{equation}
where we have assumed the absence of boundary terms and ${\Dbar\,}^\dagger = \bar \partial (1-\bar z) \bar \partial \bar z^2$. Introducing the notation $J^2=\bar h(\bar h-1)$ we find two independent solutions related by $\bar h \leftrightarrow 1-\bar h$. We will be interested in the one regular for positive $\bar h$. Requiring the inversion formula to give $A(J)=1$ for $G(\bar z)=h^{(0)}(\bar z)$ fixes the overall normalisation. We find it convenient to use the integral representation \eqref{eq:HyperInt1} which leads to the following result
\begin{equation}
A(\bar h)= \frac{1}{\pi^2} \int_0^1 \df t \df\bar z \frac{\bar z^{\bar h-2}(t(1-t))^{\bar h-1}}{(1-t \bar z)^{\bar h}}  \dDisc \left[G(\bar z)\right].
\end{equation}
Integrating over $t$ leads to the inversion formula \eqref{eq:Ugenfdef}. For all the inversions needed in this chapter it will be convenient to integrate first over $\bar z$.

While this discussion is not a rigorous derivation of the inversion formula, it explains its relation to large spin perturbation theory in the original approach of \cite{Alday2016,Alday2016b}. In appendix \ref{integrals} we give several results relevant for our computations below. In all cases the integral is convergent in the region $\bar h>1$. For our application below this means the integral converges and is expected to give the right answer for $\ell > 0$. Below we will discuss the case $\ell=0$ in more detail. 

\subsection[Inverting discontinuities in the \texorpdfstring{$\epsilon$}{epsilon} expansion]{Inverting discontinuities in the $\boldsymbol\epsilon$ expansion}
\label{sec:orderepsthree}
Let us return to the correlator introduced at the beginning of this section. We will use the inversion formula to compute the CFT-data of leading twist operators in an expansion to cubic order in $\epsilon$ (or rather $g$). Crossing symmetry implies
\begin{equation}
\label{crossinglt}
\sum_{\substack{\Delta=\tau_\ell+\ell,\\\ell=0,2,\ldots}} A_\ell z^{\tau_\ell/2} f_{\Delta,\ell}(\bar z) = z^{\Delta_\phi} \left. \left( \frac{\bar z}{1-\bar z}\right)^{\Delta_\phi} {\cal G}(1-\bar z,1-z)\right|_{\text{small $z$}},
\end{equation}
where the sum runs over leading twist operators with $\tau_\ell=2\Delta_\phi+g^2 \gamma_\ell^{(2)}+\ldots$ and the OPE coefficients are related to $A_\ell$ by $a_\ell= A_\ell\, r_{\frac{\tau_\ell}{2}+\ell}$. According to our discussion above, the CFT-data appearing on the left-hand side of (\ref{crossinglt}) can be recovered from the double-discontinuities of the right-hand side. Up to cubic order in $g$ those are straightforward to compute, as they only arise from the identity operator and the bilinear operator of spin zero, so that
\begin{equation}
\sum_{\substack{\Delta=\tau_\ell+\ell,\\\ell=0,2,\ldots}}\!\! A_\ell z^{\tau_\ell/2} f_{\Delta,\ell}(\bar z) = z^{\Delta_\phi}\left. \left(\frac{\bar z}{1-\bar z}\right)^{\Delta_\phi}\left(1 + a_0  G^{(4-\epsilon)}_{\Delta_0,0}(1-\bar z,1-z) + \text{regular}\right)\right|_{\text{small $z$}}\!\!,
\end{equation}
where we remind that $\Delta_0=2\Delta_\phi+g$ we defined in \eqref{eq:deltazero}. The regular terms do not contribute to the double-discontinuity to the order we are considering. The $d$-dimensional conformal block for a scalar exchange between two identical scalar operators was given in \eqref{eq:ScalarBlockAnyD}
\begin{equation}
G^{(d)}_{\Delta,0}(1-\bar z,1-z) = \sum_{m,n=0} \frac{\left(\Delta/2\right)^2_m\left(\Delta/2\right)^2_{m+n}}{m! n! \left( \Delta+1-d/2\right)_m \left(\Delta\right)_{2m+n}} \left[(1-z)(1-\bar z)\right]^{m+\frac \Delta2}(1-z \bar z)^n.
\end{equation}
Note that in order to extract the small $z$ dependence the sum over $n$ has to be performed. Expanding the right-hand side of (\ref{crossinglt}) in powers of $g$ up to cubic order and keeping only terms that contribute to the double-discontinuity we obtain
\begin{align}\label{eq:scalarblocktoorderepscubed}
&\sum_{\substack{\Delta=\tau_\ell+\ell,\\\ell=0,2,\ldots}} A_\ell z^{\tau_\ell/2} f_{\Delta,\ell}(\bar z) =z^{\Delta_\phi}  \left( \frac{\bar z}{1-\bar z}\right)^{\Delta_\phi} + \\
&\quad+ z^{\Delta_\phi} \bar z^{\Delta_\phi} a_0 \left( \frac{g^2}{8} \log^2(1-\bar z) \left(1+ \epsilon \partial_\epsilon+g \partial_\Delta \right) + \frac{g^3}{48} \log^3(1-\bar z)\right) g^{(4d)}_{2,0}(1-\bar z,1-z), \nonumber
\end{align}
where $a_0=2(1-g+\ldots)$ and 
\begin{align}
 g^{(4d)}_{2,0}(1-\bar z,1-z) &= \frac{\log \bar z - \log z}{\bar z}, \nonumber\\
 \partial_\epsilon g^{(4d)}_{2,0}(1-\bar z,1-z) &= \frac{(\log \bar z-\log z)(\log \bar z-2)+2\zeta_2}{2 \bar z},  \\
  \partial_\Delta g^{(4d)}_{2,0}(1-\bar z,1-z) &= \frac{\text{Li}_2(1-\bar z) + \log \bar z-\log z - \zeta_2}{\bar z}, \nonumber
\end{align}
and only the small $z$ limit has been considered. We would like to recover the CFT-data for leading twist operators from these singularities. This data admits the following decomposition
\begin{align}
A_\ell = A^{(0)}_\ell + g^2 A^{(2)}_\ell+\ldots,\nonumber\\
\tau_\ell = 2\Delta_\phi + g^2 \gamma^{(2)}_\ell+\ldots, 
\end{align}
where 
\begin{equation}
A^{(0)}_\ell = \frac{2 \Gamma \left(\bar h+\Delta _{\phi }-1\right)}{\Gamma \left(\Delta _{\phi }\right){}^2 \Gamma \left(\bar h-\Delta _{\phi }+1\right)},\qquad\bar h = \ell+\Delta_\phi,
\end{equation}
i.e.\ $A_\ell=\AA[\Delta_\phi](\Delta_\phi+\ell)$ by the notion of \eqref{eq:InversionIdentity}.
In order to apply the inversion procedure to this order we introduce
\begin{eqnarray}
 A_\ell  &=&U^{(0)}_{\bar h} + \frac{1}{2} \partial_{\bar h}  U^{(1)}_{\bar h}, \nonumber\\
 A_\ell  \gamma_\ell &=& U^{(1)}_{\bar h},  
\end{eqnarray}
where we have made clear that the natural variable in which to express $U^{(0)}_{\bar h},U^{(1)}_{\bar h}$ is $\bar h = \ell+ \Delta_{\phi}$ as opposed to $\ell$. These combinations are the ones that preserve the reciprocity principle proven in \cite{AldayBissiLuk2015}\footnote{For the present computation we find it convenient to work with this ``bare'' $\bar h$ as opposed to the ``full'' one, given by $\bar h_{\mathrm f} =\frac{\Delta_{\ell}+\ell}{2}$. The standard reciprocity principle for the CFT-data is usually expressed in terms of the full conformal spin $\bar h_{\mathrm f}(\bar h_{\mathrm f}-1)$. Note that $\bar h_{\mathrm f} $ and $\bar h$ coincide to leading order.}:
\begin{eqnarray}
U^{(0)}_{\bar h} =   \sum \frac{u^{(0)}_m}{J^{2m}},\qquad U^{(1)}_{\bar h} =   \sum \frac{u^{(1)}_m}{J^{2m}},
\end{eqnarray}
where in principle these expansions could contain both even and odd powers of $1/J$ as well as logarithmic insertions. In terms of these expansions we obtain
\begin{align}
&\sum_m z^{\Delta_\phi} \left( u_m^{(0)} + \frac{1}{2} \log z u_m^{(1)} \right)h^{(m)}(\bar z) =z^{\Delta_\phi}  \left( \frac{\bar z}{1-\bar z}\right)^{\Delta_\phi} + \\
&\quad+ z^{\Delta_\phi} \bar z^{\Delta_\phi} a_0 \left( \frac{g^2}{8} \log^2(1-\bar z) \left(1+ \epsilon \partial_\epsilon+g \partial_\Delta \right) + \frac{g^3}{48} \log^3(1-\bar z)\right) g^{(4d)}_{2,0}(1-\bar z,1-z). \nonumber
\end{align}
This has exactly the form of the inversion problem discussed above. With the inversion formulas given in appendix \ref{integrals} we find
\begin{align}
U^{(0)}_{\bar h} &=\AA[\Delta_\phi](\hb)+\left( -\frac{ g^2}{(\bar h-1)^2 \bar h^2}   + \frac{ \zeta _2 (\bar h-1) \bar h+1}{(\bar h-1)^2 \bar h^2} g^2 \epsilon -\frac{\zeta _2 (\bar h-1) \bar h-S_1}{(\bar h-1)^2 \bar h^2} g^3\right) + \ldots, \nonumber \\
U^{(1)}_{\bar h} &=- \frac{2}{(\bar h-1) \bar h} g^2 +\frac{2 }{(\bar h-1) \bar h} g^2 \epsilon + \frac{2  S_1}{(\bar h-1) \bar h} g^3+\ldots,
\end{align}
where $S_{1}$ denotes the harmonic number with argument $\bar h-1$. These results encode the full CFT-data for leading twist operators to cubic order. They translate easily into the standard anomalous dimensions and OPE coefficients and agree exactly with those obtained previously in \cite{Gopakumar:2016cpb}. The explicit results, including order $\epsilon^4$ and for $\OO N$ symmetry are available in the ancillary data of the Arxiv submission of \cite{Paper3}.

\subsection{Matching conditions at low spin}\label{sec:mathingconditions}

Let us write the result we have just obtained for the anomalous dimensions in terms of the full $\bar h_{\mathrm f}$, defined as $\bar h_{\mathrm f}= \ell+ \Delta_\phi+ \frac{1}{2} \gamma_{\ell}$. We obtain
\begin{equation}
\label{deltaell}
\Delta_\ell =2 \Delta_\phi+ \ell - \frac{g^2}{(\bar h_{\mathrm f}-1) \bar h_{\mathrm f}} + \frac{g^2 \epsilon+(g^3-g^2\epsilon)S_1}{(\bar h_{\mathrm f}-1) \bar h_{\mathrm f}} + \ldots
\end{equation}
These results followed only from crossing symmetry of a single correlator and the inversion procedure used in this work shows that they basically follow from a one-loop computation (since squares of anomalous dimensions will generate discontinuities only at quartic order). 
We now impose two further matching conditions at low values of the spin
\begin{eqnarray}
\Delta_2&=&d,\\
\Delta_0 &=& 2\Delta_\phi+g.
\end{eqnarray}
The first condition is implied by the existence of a conserved stress tensor and fixes the dimension of the external operator
 \begin{equation}
 \Delta_\phi = 1-\frac{1}{2}\epsilon + \frac{1}{12} g^2 -\frac{1}{8} g^3 + \frac{11}{144} g^2 \epsilon + \ldots.
 \end{equation}
The second condition arises from the requirement that the inversion results can be extrapolated down to spin zero\footnote{We would like to thank Aninda Sinha for suggesting this idea.}. For $\epsilon,g \neq 0$, in order to reach $\ell=0$ we need to continue $\Delta_\ell$ to the left of the pole at $\bar h_{\mathrm f}=1$. We will assume the standard continuation across a pole, i.e.\ that the expression (\ref{deltaell}) remains valid also in this region. This is summarised in Figure~\ref{fig:plot}. 
\begin{figure}
\centering
\includegraphics{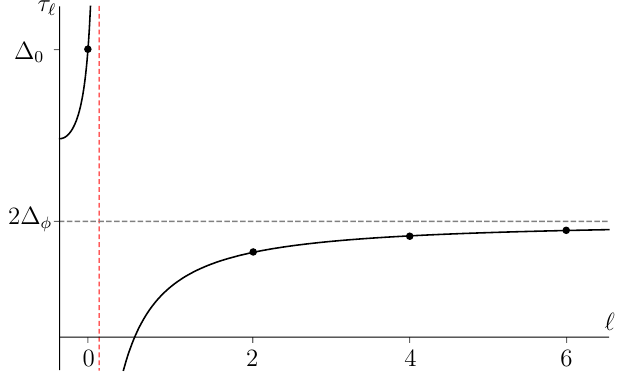}
\caption[Schematic graph of $\tau_\ell$ in the $\epsilon$ expansion.]{
Schematic graph of $\tau_\ell$. As we move from spin two to spin zero we move to the left of the pole at $\bar h_{\mathrm f}=1$, denoted by a red line. Note the change of sign in the correction. Assuming the standard continuation in (\ref{deltaell}), we reproduce the correct dimension on both sides. 
}\label{fig:plot}
\end{figure}
Note that in the $\epsilon$ expansion $\bar h_{\mathrm f}-1 \sim \epsilon$, so that the limit is somewhat subtle. To leading order we obtain the following relation
 \begin{equation}\label{eq:betalikeequation}
-g \epsilon + 3 g^2=0.
 \end{equation}
This equation has two solutions. One corresponds to the free theory with $g=0$ and the other corresponds to
 \begin{equation}\label{eq:epsilongrel}
g = \frac{1}{3}\epsilon + \ldots,
\end{equation}
fixing the relation between $g$ and $\epsilon$. Plugging this into the expression for $\Delta_\phi$ we obtain
 \begin{equation}
 \Delta_\phi = 1-\frac{1}{2}\epsilon + \frac{1}{108} \epsilon^2+ \ldots,
 \end{equation}
which exactly agrees with the well-known value for the WF model! The order $g^4$ results obtained in the next section allow us to go one order further, and find the relation
 \begin{equation}
g = \frac{1}{3}\epsilon + \frac{8}{81}\epsilon^2+ \ldots.
\end{equation}
This fixes the relation between $g$ and $\epsilon$, and therefore all the quantities entering the problem. 

\subsection[$\mathrm O(N)$ model]{$\boldsymbol{\mathrm O(N)}$ model}\label{sec:ONwithinPaper2}
The method used in this chapter generalises to the $\mathrm O(N)$ model immediately. Let us consider the WF model with $N$ scalar fields $\varphi^i$ with global $\mathrm O(N)$ symmetry in $d=4-\epsilon$ dimensions. We can now consider the four-point correlator of the fundamental field $\varphi^i$. Intermediate operators decompose into the singlet ($S$), symmetric traceless ($T$) and anti-symmetric ($A$) representations of $\mathrm O(N)$. It is convenient to write the crossing equations as
\beqa{
\label{crossON}
f_S(z,\bar z) &= \frac{1}{N}f_S(1-\bar z,1-z) + \frac{N^2+N-2}{2N^2} f_T(1-\bar z,1-z) +\frac{1\!-\!N}{2N} f_A(1-\bar z,1-z),\nonumber\\
f_T(z,\bar z) &=f_S(1-\bar z,1-z)+ \frac{N-2}{2N} f_T(1-\bar z,1-z) +\frac{1}{2} f_A(1-\bar z,1-z),\\
f_A(z,\bar z) &=-f_S(1-\bar z,1-z) + \frac{2+N}{2N} f_T(1-\bar z,1-z) +\frac{1}{2} f_A(1-\bar z,1-z), \nonumber
}
where $f_{R}(z,\bar z) =((1-z)(1-\bar z))^{\Delta_\varphi}  {\cal G}_{R}(z,\bar z)$. The crossing equations at leading order have been analysed in \cite{Alday2016b} with the methods of large spin perturbation theory. Again, at leading order the fundamental field does not acquire any corrections while
\begin{equation}
\gamma^{(1)}_{\varphi^2_S} = g =: g_S,\qquad\gamma^{(1)}_{\varphi^2_T} = \frac{2}{2+N}g + \ldots =: g_T.
\end{equation}
In order to reconstruct the CFT-data from double discontinuities we note that these arise from the identity operator, present in the singlet representation, and the bilinear operators in the singlet and traceless-symmetric representations, which acquire an anomalous dimension at order $g$. By looking at the double-discontinuity of the identity operator on the right-hand side of the crossing equations (\ref{crossON}) we see that at leading order the OPE coefficients of the $T$ and $A$ representations are exactly as before, up to a sign for $A$, while those of the $S$ representation have an extra factor of $1/N$. 
 \begin{equation}
 A^{(0)}_{T/A,\ell} = \pm \AA[\Delta_\phi](\hb),\qquad A^{(0)}_{S,\ell} = \frac{1}{N}\AA[\Delta_\phi](\hb),
 \end{equation}
where $\ell$ is even for the symmetric-traceless and singlet representations and odd for the anti-symmetric representation. A careful analysis of the crossing conditions also determines the corrections to order $g$ of the OPE coefficients for the spin zero operators:
\begin{equation}\label{eq:alpha1correctionp2}
a_{R,0} = a^{(0)}_{R,0} (1-g_{R} + \ldots),
\end{equation}
which in fact holds for $\phi^4$ theories in any global symmetry group. 
By looking at the crossing equations (\ref{crossON}) and comparing them with our computation for the $N=1$ case, it is then straightforward to write down the result for  $U^{(1)}_{\bar h}=A_\ell \, \gamma_\ell$ for each representation. We obtain
\begin{align}
U^{(1)}_{S,\bar h} = &\frac{2}{N^2J^2}\left( -g_S^2 +g_S^2 \epsilon +  g_S^3S_1+ \ldots\right)  + \frac2{J^2}\frac{N^2+N-2}{2N^2}\left( -g_T^2 + g_T^2 \epsilon +  g_T^3S_1+ \ldots\right)  ,
 \nonumber\\
U^{(1)}_{T,\bar h} = &\frac{2}{NJ^2}\left( -g_S^2 +g_S^2 \epsilon +  g_S^3S_1+ \ldots\right) 
+\frac{2}{J^2} \frac{N-2}{2N}\left( -g_T^2 + g_T^2 \epsilon +  g_T^3S_1+ \ldots\right)  \nonumber\\
U^{(1)}_{A,\bar h} = &\frac{2}{NJ^2}\left( -g_S^2 +g_S^2 \epsilon +  g_S^3S_1+ \ldots\right) - \frac2{J^2}\frac{2+N}{2N}\left( -g_T^2 + g_T^2 \epsilon +  g_T^3S_1+ \ldots\right)   ,
\end{align}
where $J^2=\bar h(\bar h-1)$ and as before the harmonic number $S_1$ is evaluated at $\bar h-1$. Similar expressions for $U^{(0)}_{R,\hb}$ can be obtained in exactly the same way. All the results are in full agreement with those obtained in \cite{Dey:2016mcs,Manashov:2017xtt} after substituting the literature values \beq{g_S=\frac{2+N}{8+N}\epsilon+6\frac{(N+2)(N+3)}{(N+8)^3} \epsilon^2+\ldots, \quad g_T= \frac{2}{8+N}\epsilon+\frac{36+4N-N^2}{(N+8)^3}\epsilon^2+\ldots.\label{eq:ONgslitt}
}

\section{Results to fourth order}
\label{fourthorder}
\subsection{New operators at second order}
Before proceeding to solve the crossing constraints to higher order, we would like to make the following crucial observation. At order $g^2$ new intermediate operators are expected to appear, which are of the schematic form $\phi^2 \Box^n \partial_{\mu_1} \ldots \partial_{\mu_\ell} \phi^2$ and have twist $\tau=4+2n$ and spin $\ell$. These operators are expected to acquire an anomalous dimension to order $\epsilon$. Hence, they generate a double-discontinuity, proportional to the square of their anomalous dimension, to order $g^4$. Furthermore, these operators are highly degenerate in perturbation theory, so that computing this double-discontinuity would require solving a mixing problem. The statement that the CFT-data can be reconstructed from the double-discontinuities of the correlator is not restricted to leading twist operators and the projection methods to higher twist, described in section~\ref{sec:projections}, can be used to find the leading OPE coefficients of these operators\footnote{Since we are near four dimensions, the problem simplifies and one can use the four-dimensional conformal blocks instead of the subcollinear blocks. The details are explicitly worked out in \cite{Paper3}.}. The steps are very similar to the ones above, and to second order in $g$ we find
\begin{equation}
\label{twist4ope}
a_{4+2n,\ell}= 
     \begin{cases}
       \frac{\Gamma (\ell+2)^2}{\Gamma (2 \ell+3)} \frac{\ell^2+3 \ell+8}{12 (\ell+1) (\ell+2)} g^2 + \ldots &\quad\text{for $n=0$,}\\ \\
       {\cal O}(g^4) &\quad\text{for $n \neq 0$.}\\     
       \end{cases}
\end{equation}
This is a somewhat surprising result:  only operators with approximate twist four appear at this order\footnote{As a byproduct, this result justifies an ansatz made in \cite{Liendo:2012hy}, where the vanishing of OPE coefficients involving operators with $n \neq 0$ was assumed. 
}. As we will see, this constrains the possible structure of double-discontinuities at fourth order and it will allow us to solve the problem completely. Given the convergence of the inversion integrals we expect these results to be valid down to spin zero. 

\subsection{Solving the inversion problem at fourth order}
The contribution arising from leading twist operators in a perturbative $\epsilon$ expansion can be encoded as follows
\begin{align}
\left. {\cal G}(z,\bar z) \right|_{\text{small $z$}} = \sum_m z^{\Delta_\phi} &\left( u_m^{(0)} + \frac{1}{2} \log z \, u_m^{(1)}+ \frac{1}{8} \log^2 z\,  u_m^{(2)} + \ldots \right)h^{(m)}(\bar z),
\end{align}
where $u_m^{(p)} \sim g^{2p}$ for small $g$. As before, the $u_m^{(p)}$ are the coefficients in the large $J$ expansions of $U^{(p)}_{\bar h}$, whose relation to the usual OPE data is given by \eqref{eq:aellfromU},
\begin{align}
\label{datafromU}
A_\ell \left( \gamma_\ell  \right)^p = U^{(p)}_{\bar h} + \frac{1}{2} \partial_{\bar h} U^{(p+1)}_{\bar h}+ \frac{1}{8} \partial^2_{\bar h} U^{(p+2)}_{\bar h}+\ldots.
\end{align}
To order $g^4$ the double-discontinuity of the correlator arises from four distinct contributions, so that
\begin{align}
\left. {\cal G}(z,\bar z) \right|_{\text{small $z$}} = z^{\Delta_\phi} \left( \left( \frac{\bar z}{1-\bar z}\right)^{\Delta_\phi}+ I_{\phi^2} +I_{2} + I_{4} + \, \text{regular} \right).
\end{align}
$I_{\phi^2}$ denotes the contribution from the scalar bilinear operator. To cubic order it was given in the previous section. It is straightforward to compute it to fourth order and the result is given in appendix \ref{ddisc}. $I_{2}$ denotes the contribution arising from leading twist operators of spin two and higher: the square of their anomalous dimension generates a double-discontinuity at fourth order. Since these operators are non-degenerate, this contribution can be readily computed and it is given in  appendix \ref{ddisc}. As already mentioned, a direct computation of $I_{4}$ would require solving a mixing problem, for instance by considering more general correlators\footnote{The contribution from twist-four operators to the anomalous dimension of leading twist operators starts at order $1/\ell^4$, see \cite{AldayMaldacena2007}, so that the leading terms in a $1/\ell$ expansion can still be computed without its knowledge. This was done in \cite{Dey:2017oim} by applying directly the methods of \cite{AldayZhiboedov2015} for isolated operators. Since there is an accumulation point at twist two, one should be careful. In principle the correct procedure from the large spin perspective would be to compute the double-discontinuity due to the tower of twist-two operators and then compute the anomalous dimensions from there. The procedure of \cite{Dey:2017oim} is justified since the resulting series are convergent.}. However, note that at fourth order $I_{4}$ involves four-dimensional conformal blocks evaluated at reference twist four. This implies the following structure
\begin{equation}
I_{4}= \left( \log z g(\bar z) - \log \bar z g(z) \right) \log^2(1-\bar z),
\end{equation}
where $g(\bar z)$ arises from a sum over twist-four operators 
\begin{equation}
\label{gsum}
g(\bar z) = \frac{1}{8}\sum_{\ell=0,2,\ldots} \!\eta_\ell\, k_{2+\ell}(1-\bar z)
\end{equation}
for some $\eta_\ell$ equal to the weighted average, over degenerate operators, of the square anomalous dimensions $\eta_\ell =\langle a_{4,\ell} \gamma^2_{4,\ell}\rangle=\sum_i a_{4,\ell,i} \gamma^2_{4,\ell,i}$. As such it is regular around $\bar z=1$. Furthermore, the structure of the OPE to this order implies the following expansion around $z=0$~\footnote{Specifically, note that in equation (3.4), on the left-hand side any higher powers $\log^k z$ would have to be generated by higher powers $\gamma_{2,\ell}^k$ of anomalous dimensions, which contribute only at order $g^{2k}$ and higher.}
\begin{equation}
\label{garound0}
g(z) = \alpha_0 \log^2 z + \alpha_1 \log z+ \alpha_2 + \ldots.
\end{equation}
We will now discuss how to fix  $U^{(0)}_{\bar h},U^{(1)}_{\bar h},U^{(2)}_{\bar h}$ to quartic order. Before we proceed, note that the term $\log z g(\bar z)$ in $I_4$ will only contribute to $U^{(1)}_{\bar h}$. Hence $U^{(0)}_{\bar h}$ and $U^{(2)}_{\bar h}$ only require minimal information about $g(z)$, namely only its limit as $z \to 0$. As a result, they could be fully determined in terms of $\alpha_0$  and $\alpha_2$, even without any knowledge of twist-four operators. We will be able to do even better than this. 

Let us start with $U^{(2)}_{\bar h}$. From the expressions in appendix \ref{ddisc}, it follows that $I_{\phi^2}$ and $I_2$ do not contribute to $U^{(2)}_{\bar h}$, as they do not contain a $\log^2z$ piece. The whole contribution arises then from $I_{4}$ and is proportional to $-\alpha_0\log \bar z \log^2(1-\bar z)$. From the results in appendix \ref{integrals} this immediately gives
\begin{equation}
U^{(2)}_{\bar h}=-8 \alpha_0 \frac{4(1-2\bar h)}{\bar h^2(1-\bar h)^2} g^4,
\end{equation}
which exactly agrees with $A_\ell(\gamma_\ell)^2$ to order $g^4$ provided $\alpha_0=1/16$. 

To compute $U^{(1)}_{\bar h}$ one needs to know $g(\zb)$. The full results for double discontinuities up to cubic order as well as the double discontinuities in appendix \ref{ddisc} suggest that perturbative results for the present correlator organise themselves in pure transcendental functions with discontinuities around $\bar z=0$ and regular around $\bar z =1$. Furthermore, the degree of these functions increases with the perturbative order in a prescribed way\footnote{More precisely, up to this order we will assume that the answer can be written as combinations of polylogarithms of $\bar z$ and $1-\bar z$, without rational functions in front, such that the total degree increases linearly with the loop order. This structure is very familiar in other perturbative contexts.}. If this principle holds then we expect $g(\zb)$ to be given by a linear combination of the following building blocks
\begin{equation}\label{eq:basisfuncs}
\{\log^2 \bar z,\, \mathrm{Li}_2(1-\bar z),\,\log^3 \bar z, \,\log \bar z \,\mathrm{Li}_2(1-\bar z) ,\,\mathrm{Li}_3(1-\bar z),\,\mathrm{Li}_3\left(\frac{\bar z-1}{\bar z}\right)  \}.
\end{equation}
These blocks form a basis of functions as described above. Any other function with the same features can be related to combinations of these by identities for polylogarithms such as \eqref{eq:polyLogid2} and \eqref{eq:polyLogid3}. The fact that $g(\bar z)$ arises from twist-four operators in the dual channel, constrains the possibilities. Furthermore, consistency with (\ref{gsum})  and (\ref{garound0}) leads us to the following result
\begin{equation}
\label{gzb}
g(\bar z) =  \frac{1}{16} \log^2 \bar z +\alpha \left(-\frac{1}{6} \log^3 \bar z-\frac{2}{3} \log z \,\text{Li}_2(1-\bar z)+  \text{Li}_3(1-\bar z)+\text{Li}_3\left(\frac{\bar z-1}{\bar z}\right)\right),
\end{equation}
with a single undetermined coefficient.  We would like to stress that this expression can be systematically tested as an expansion around $\bar z = 1$.  Since $k_{2+\ell}(1-\bar z) \sim (1-\bar z)^{2+\ell}$, to any given order in $(1-\bar z)$ only a finite number of operators contribute and the mixing problem is finite. For instance,  twist-four operators with spin zero and two are non-degenerate. 
The anomalous dimensions for these operators are known, see section~\ref{sec:WFspectrum}, and in the conventions used here they take the form $\gamma_{4,0}=3 g+\ldots$ and $\gamma_{4,2}=4/3 g+\ldots$~\footnote{Alternatively, these anomalous dimensions can be computed from the discontinuities of the correlator at cubic order by a projection from the leading twist family \cite{Paper3}.}. 
From (\ref{twist4ope}) we can also read off $a_{4,0}=g^2/6+\ldots $ and  $a_{4,2}=g^2/160 + \ldots$. These values are exactly consistent with the expression for $g(\bar z)$ up to fifth order in $(1-\bar z)$ and furthermore fix $\alpha=-3/2$. With this we find
\begin{equation}
g(z) = \frac{1}{16} \log^2 z -\frac{1}{2}\zeta_2 \log z-\frac{3}{2} \zeta_3 + \ldots,\qquad\text{around $z=0$}.
\end{equation}
We have now all the ingredients to compute $U^{(0)}_{\bar h}$ and $U^{(1)}_{\bar h}$ to fourth order. Using the inversion formulae in appendix \ref{integrals} we find
\begin{align}
U^{(1)}_{\bar h} =   \frac{-2}{J^2} g^2
+  \frac{2\left(3+S_1\right)}{J^2}g^3&
+  \frac{1}{6 J^2} \Bigg(
\frac{6}{J^4}  + \frac{7+48 S_{-2}}{J^2}  
 \nonumber\\&
  -9 \zeta _2-6 S_1^2-36 S_1-12 S_{-2}-58\Bigg)g^4+ \ldots
\end{align}
and 
\begin{align}
U^{(0)}_{\bar h} &= \AA[\Delta_\phi](\hb)
+ \frac{ -2}{J^4} g^2
+  \frac{1}{J^2}\left(\frac{3+S_1}{J^2} + 2  \zeta_2 \right)g^3
  + \frac{1}{12 J^2}\Bigg(\frac{2}{J^4}-106\zeta _2
 \nonumber\\&\quad
-\frac{56+3 \zeta _2+72 \zeta _3+6 S_1^2+36 S_1-12 S_{-2}}{J^2} +72 \zeta _3-24 \zeta _2 S_1-54 S_3 \Bigg)g^4 + \ldots,
\end{align}
where the argument of all nested sums, defined in appendix \ref{integrals}, is $\bar h-1$. In these expressions we have traded the dependence on $\epsilon$ in favour of $g$. The CFT-data can then be recovered from (\ref{datafromU}). In particular
\begin{equation}
\gamma_\ell = \frac{U^{(1)}_{\bar h} + \frac{1}{2} \partial_{\bar h} U^{(2)}_{\bar h} + \ldots}{U^{(0)}_{\bar h} + \frac{1}{2} \partial_{\bar h} U^{(1)}_{\bar h} + \ldots},
\end{equation}
and the result can be seen to exactly agree with that obtained in \cite{Derkachov:1997pf}\footnote{We would like to thank the authors of \cite{Dey:2017oim} for making us aware of a typo in  \cite{Derkachov:1997pf}.}. In order to fix $\Delta_\phi$ and $g(\epsilon)$ to this order one could proceed exactly as before: $\Delta_\phi$ follows again from conservation of the stress tensor while $g(\epsilon)$ follows from the matching condition at spin zero. However, the later result to cubic order would require going to higher orders in our computation. Instead, we will take a shortcut and assume the known value of the dimension of the fundamental field $\Delta_\phi=1-\frac{\epsilon}{2}+\frac{\epsilon^2}{108}+\frac{109}{11664}\epsilon^3+(\frac{7217}{1259712}-\frac{2}{243}\zeta_3)\epsilon^4+\ldots$. This together with the conservation of the stress tensor gives the relation between $g$ and $\epsilon$:
\begin{equation}
g=\frac{\epsilon}{3}+\frac{8}{81}\epsilon^2+\left(\frac{305}{8748}-\frac{4}{27}\zeta_3 \right)\epsilon^3+\ldots.
\end{equation}
Let us emphasise however, that the first two orders follow completely from our results, without any additional input, and also the next term could be in principle computed in our formalism if extended to order $\epsilon^6$. The result for the OPE coefficients is completely new. The most interesting quantity that can be extracted from them is the central charge, related to the OPE coefficient for $\ell=2$. In terms of $\epsilon$ we find exactly the fourth order result \eqref{eq:CTintrores} quoted at the beginning of this chapter.
The result to cubic order exactly reproduces what was found in \cite{Gopakumar:2016cpb}. The result to fourth order is new. Setting $\epsilon=1$ we observe that this new contribution gets us closer to the highly precise numerical result for the 3d Ising model found in \cite{El-Showk:2014dwa,Kos:2016ysd}. 

\subsection[$\OO N$ model at order \texorpdfstring{$\epsilon^4$}{epsilon**4}]{$\boldsymbol{\OO N}$ model at order $\boldsymbol{\epsilon^4}$}\label{sec:ONe4}
Before we conclude, let us summarise briefly the results of \cite{Paper3}, which considered the $\OO N$ model. Also there, the new operators to contribute at order $\epsilon^4$ were families of weakly broken currents and of operators of approximate twist four, but now in all three representations $S$, $T$ and $A$. The contributions from twist-four operators required an ansatz similar to \eqref{eq:basisfuncs}, but now the anomalous dimensions of non-degenerate operators (spin zero in $S$ and $T$, and spin one in $A$) was computed using a projection at order $\epsilon^3$ along the lines of section~\ref{sec:projections}. 
For completeness, we give here the results from \cite{Paper3} for the central charges in the critical $\OO N$ model in the $4-\epsilon$ expansion:
\begin{align} \label{eq:CentralCharge}
\frac{C_T}{C_{T,\mathrm{free}}} &= 1
-\frac{5 (N+2) }{12 (N+8)^2} \epsilon ^2 
-\frac{(N+2) \left(7 N^2+382 N+1708\right)}{36 (N+8)^4} \epsilon ^3
\nonumber\\
&\quad-\frac{(N+2) \left(65 N^4+5998 N^3+309036 N^2+2396800N+5440832\right)}{1728 (N+8)^6}  \epsilon^4
\nonumber\\
&\quad+\frac{(N+2) \left(2 N^3+43 N^2+922 N+3488\right)\zeta_3}{12 (N+8)^5}  \epsilon^4+O(\epsilon^5),
\end{align}
and
\begin{align}\label{eq:JCharge}
\frac{C_J}{C_{J,\mathrm{free}}}&=1
-\frac{3 (N+2)}{4 (N+8)^2} \epsilon^2
-\frac{(N+2) \left(N^2+132 N+632\right) }{8 (N+8)^4}\epsilon^3
\nonumber\\
&\quad + \frac{(N+2) \left(11 N^4+246 N^3-13124 N^2-126976 N-310976\right)}{64 (N+8)^6} \epsilon ^4
\nonumber\\
&\quad +\frac{(N+2) \left(7 N^2+442 N+1792\right)\zeta_3}{4 (N+8)^5}\epsilon ^4+O(\epsilon^5).
\end{align}

\section{Conclusions} 
We have used analytic bootstrap techniques to derive the anomalous dimensions and OPE coefficients of bilinear operators (weakly broken currents) in the WF model in $d=4-\epsilon$ dimensions, to fourth order in the $\epsilon$ expansion. To cubic order the computation is essentially straightforward, since the double-discontinuity arises solely from the identity operator and the bilinear scalar. This simplicity is also manifest in the results of the $\mathrm O(N)$ model, and in section~\ref{sec:anyglobalsymmetryeps} we will generalise this to any global symmetry. At fourth order the situation is much more interesting, since two towers of high spin operators, of twist two and four respectively, contribute to the discontinuity. The contribution from twist two operators can be readily computed, while the structure of perturbation theory, together with the explicit form of four-dimensional conformal blocks, allows to make a proposal for the double-discontinuity due to twist-four operators. This proposal can be systematically tested order by order in powers of $(1-\bar z)$, by solving a finite order mixing problem. This satisfies all possible consistency conditions and is compatible with features of perturbation theory from other CFTs. With this result, we have found the CFT-data to fourth order. Two further constraints, namely conservation of the stress tensor, together with a continuation to spin zero, allowed to fix the anomalous dimensions of both the scalar operator $\phi^2$ as well as the dimension of the external operator. 

There are several interesting open problems. A remarkable feature of our computation is the apparent analyticity down to spin zero. This allowed us to reproduce constraints analogous to those of a vanishing beta function. It would be interesting to understand the systematics of this to higher orders, and even non-perturbatively. It would also be interesting to understand the structure of double-discontinuities to higher orders in the $\epsilon$ expansion. Up to fourth order we have observed that the functions that appear have pure transcendentality.  It is tantalising to propose that this persists to higher orders which would greatly simplify the computation of CFT-data. The extension to order $\epsilon^4$ requires detailed knowledge of the operator content of the theory in question, such as the degeneracy of the operators at approximate twist four. As mentioned above, this has been done for the $\OO N$ model in \cite{Paper3}, but not in theories with other global symmetry groups.

A natural direction would be to extend these results to higher orders. At order $\epsilon^5$ the same operators contribute as at order $\epsilon^4$. The challenge, again, is to find the contribution from the twist-four operators. It would also be interesting to consider analytic constraints arising from mixed correlators. In the present case one could consider correlators of the fundamental field and the bilinear scalar. The crossing constraints for such a system are expected to be stronger than the ones considered in this chapter.

\chapter{Weakly coupled gauge theories}
\label{ch:paper1}

\section{Introduction}

In this chapter we apply this method to weakly coupled conformal field theories in four space-time dimensions. We study four-point correlation functions
\begin{equation}
\mathcal{G}(x)=(x_{12}^2x_{34}^2)^{\Delta_{\mathcal O}}\langle \mathcal{O}(x_1)\mathcal{O}(x_2)\mathcal{O}(x_3)\mathcal{O}(x_4)\rangle
\end{equation}
 of identical operators built out of fundamental scalar fields of the theory in the small coupling $g$ expansion. Here, $\Delta_{\mathcal{O}}$ is the conformal dimension of the operator $\mathcal{O}$ and $x_{ij}$ denotes the distance between two space-time points. A prototypical example of such theory is $\mathcal{N}=4$ SYM. In order to focus our attention we will discuss two very particular scalar operators in $\mathcal{N}=4$ SYM: the Konishi operator $\mathcal{K}$ and the half-BPS operator $\mathcal{O}_{{\boldsymbol{20'}}}$ in the $[0,2,0]$ representation of the $\SU 4 $ R-symmetry. Both of them are the simplest gauge invariant scalar operators and have the schematic form $\mathcal{O}=\mbox{Tr}(\Phi^2)$, where $\Phi$ is a fundamental scalar field of the theory. The methods developed here will however apply to a large class of conformal field theories satisfying a set of assumptions spelled out at the end of this section.

In the following we study four-point correlation functions in the perturbation theory around vanishing coupling constant $g=0$,
\begin{equation}
\mathcal{G}(x)=\mathcal{G}^{(0)}(x)+g\, \mathcal{G}^{(1)}(x)+\ldots.
\end{equation}
 The leading-order answers $\mathcal{G}^{(0)}(x)$ can be found by directly performing Wick contractions and depend on a single parameter related to the central charge of the theory. In this chapter we focus most of our attention on the one-loop function $\mathcal{G}^{(1)}(x)$ and find its general form using only conformal symmetry, crossing symmetry and the structure of the operator product expansion (OPE). In the two cases that we study we find a family of crossing-symmetric solutions which depend on a small number of free parameters. The most transcendental part of the answer is given by the scalar box function times a rational function. These have to be supplemented by lower transcendental functions. We find the explicit form of these functions without referring to Feynman diagram calculations. In particular, we will avoid introducing any regularisation or any redundancies fundamentally bound to the Feynman approach. In order to find a particular four-point correlator we supplement our general solution with a few explicit values of the CFT-data for operators with small classical conformal dimension and spin. 

Our method will be based on only a few assumptions:
\begin{itemize}
\item We study unitary weakly coupled conformal gauge theories in four dimensions. In particular, unitarity implies that the operators in the OPE expansion satisfy the unitarity bound and have non-negative (squared) OPE coefficient with $\mathcal{O}=\mbox{Tr}(\Phi^2)$. Moreover, the fact that we study gauge theories implies that the fundamental field $\Phi$ is not part of the spectrum, and therefore the correlator of $\mathcal{O}$ provides the strongest constraint on the CFT-data. 
\item We assume that infinite towers of operators parametrised by spin $\ell$ have a regular expansion of the CFT-data at large spin, i.e. the CFT-data can be written as a Taylor expansion of $\frac{1}{\ell}$ with possible $\log \ell$ insertions. 
\end{itemize}
Furthermore we will use the following properties of conformal field theories:
\begin{itemize}
\item We use the fact that four-point correlation functions are crossing symmetric.
\item We use the knowledge of the OPE structure. Furthermore, we rely on an explicit form of the conformal blocks in four dimensions and the superconformal blocks for the half-BPS operators $\mathcal{O}_{{\boldsymbol{20'}}}$ in $\mathcal{N}=4$ SYM.
\end{itemize}
It was already found in \cite{Heemskerk2009,Alday:2014tsa} that there exists a class of crossing symmetric solutions which correspond to CFT-data that is truncated in spin. In particular, the instanton solutions are of this type, as shown in \cite{Arutyunov:2000im}. Our analysis extends these results by including also solutions unbounded in spin. Since crossing at one loop in perturbation theory is a linear problem, we can treat these two types of solutions separately and focus only on the latter.

The chapter is organised as follows: in section \ref{sec:four.points} we collect basic information about four-point correlation functions and their properties. In section \ref{sec:twist.conformal.blocks} we introduce the notion of twist conformal blocks and H-functions and study their properties. In section \ref{sec:FindingNemo} we use H-functions to find a family of solutions to the conformal bootstrap equation and in particular recover the known form of the four-point correlator of Konishi operators. In section \ref{sec:super.case} we repeat the analysis from the previous two sections in the case of the correlation function of four half-BPS operators $\mathcal{O}_{{\boldsymbol{20'}}}$ in $\mathcal{N}=4$ SYM. We end the chapter with conclusions and outlook and supplement it with a few appendices containing the more technical ingredients of our results. 

\section{Four-point correlators}
\label{sec:four.points}

In this section we collect all relevant information about four-point correlation functions of operators that we will study in the rest of this chapter. In the first part we describe four-point correlators of four identical scalar operators with classical dimension $\Delta_0=2$. This is relevant for the Konishi operator in $\mathcal{N}=4$ SYM, which is of the form 
\begin{equation}
\mathcal{K}(x)=\mbox{Tr}(\Phi^I(x) \Phi^I(x)),
\end{equation}
where $I$ is the $\SO 6$ R-symmetry index. We study the correlation function of four Konishi operators using the ordinary conformal block decomposition in four dimensions \cite{Dolan:2000ut}. 

In the second part we study the $\mathcal{N}=4$ SYM half-BPS operator in the $[0,2,0]={\boldsymbol{20'}}$ representation of the $\SU4$ R-symmetry
\begin{equation}
\mathcal{O}_{{\boldsymbol{20'}}}(x,y)=y_I\,y_J\,\mbox{Tr}(\Phi^I(x)\Phi^J(x))\,,
\end{equation}
where we have introduced an auxiliary six-dimensional complex null vector $y_I$, namely $y\cdot y\equiv y_Iy^I=0$. In order to properly accommodate for a non-trivial R-symmetry structure of the correlation function of four half-BPS operators we employ the superconformal blocks introduced in \cite{Dolan:2004iy}.

\subsection{Conformal block decomposition for Konishi operators}
\label{sec:bosonic.four.points}
First, let us consider the case relevant for the Konishi operator $\mathcal{K}$, namely a scalar operator with the conformal dimension 
\begin{equation}
\Delta_\mathcal{K}=2+\sum_{i=1}^\infty \gamma^{(i)}_{\mathcal{K}}\, g^i\,.
\end{equation}
The crossing equations for the four-point correlator of Konishi operators are
\begin{equation}\label{crossing.symmetry}
\mathcal{G}(u,v)=\mathcal{G}\left(\frac{u}{v},\frac{1}{v}\right)\,,\qquad v^{\Delta_{\mathcal{K}}}\mathcal{G}(u,v)=u^{\Delta_{\mathcal{K}}}\mathcal{G}(v,u)\,.
\end{equation}
In the following, we will solve these equations and study their solutions as perturbations in the double lightcone limit. While the first equation in \eqref{crossing.symmetry} can easily be expanded using the conformal block decomposition, the second equation has to be treated more carefully. In order to do that we will need to employ the twist conformal blocks introduced in \cite{Alday2016}. We refer to the second equation in \eqref{crossing.symmetry} as the {\it conformal bootstrap equation}. 

The conformally invariant function $\mathcal{G}(u,v)$ entering \eqref{crossing.symmetry} admits a decomposition into conformal blocks obtained by considering the OPE expansion in the limit $x_1\to x_2$
\begin{equation}\label{eq:BlockDecomposition}
\mathcal G(u,v)=\sum_{\tau,\ell,i}a_{\tau,\ell,i}\,G_{\tau,\ell}(u,v).
\end{equation}
Here the sum runs over all conformal primaries of twist $\tau=\Delta-\ell$, where $\Delta$ is the conformal dimension, and even spin $\ell$ present in the OPE decomposition of two Konishi operators 
and the index $i=1,\ldots,d_{\tau_0,\ell}$ runs over a possible additional degeneracy in the spectrum of operators with a given twist and spin.
We denote the square of OPE coefficients by $a_{\tau,\ell,i}=c^2_{\mathcal{K}\mathcal{K}\mathcal{O}_{\tau,\ell,i}}$. 
The conformal blocks $G_{\tau,\ell}(u,v)$, which resum contributions coming from all descendants of a given conformal primary operator, can be found explicitly for four dimensions \cite{Dolan:2000ut}. For even spins they take the following form \eqref{eq:CB4d}:
\begin{equation}\label{eq:ConformalBlock}
G_{\tau,\ell}(z,\zb)=\frac{z \zb}{z-\zb}\left(k_{\frac\tau2+\ell}(z)k_{\frac \tau2-1}(\zb)-k_{\frac\tau2+\ell}(\zb)k_{\frac \tau2-1}(z)\right).
\end{equation}
It is easy to check that each conformal block satisfies the first equation in \eqref{crossing.symmetry}.

On the other hand, in perturbative conformal gauge theories the four-point function admits a small coupling expansion
\begin{equation}\label{four.point.weak}
\mathcal{G}(u,v)=\mathcal{G}^{(0)}(u,v)+g \,\mathcal{G}^{(1)}(u,v)+\ldots,
\end{equation}
where $g$ is the gauge coupling.
The tree-level term can be directly evaluated using Wick contractions in the free theory as in figure~\ref{fig:Free}
\begin{figure}[t!]
\centering
\includegraphics{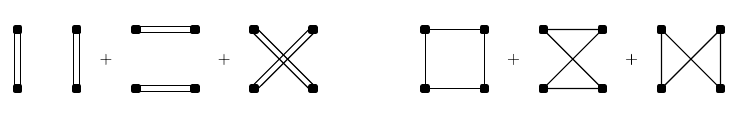}
\caption[Wick contractions for the tree-level calculation in conformal gauge theories.]{Wick contractions relevant for the tree-level calculation.}
\label{fig:Free}
\end{figure}
\noindent and renders
\begin{equation}\label{tree.level}
\mathcal G^{(0)}(u,v)=\left(1+u^2+\frac{u^2}{v^2}\right)+c\left(u+\frac uv+\frac{u^2}v\right),
\end{equation}
where $c$ is a theory-dependent constant which for example for $\mathcal{N}=4$ SYM with gauge group $\SU N$ is proportional to the inverse of the central charge, 
$c\sim (N^2-1)^{-1}$.
Performing the conformal block decomposition we find that for each reference twist $\tau_0=2,4,6,\ldots$ there exists an infinite tower of operators contributing to the sum in \eqref{eq:BlockDecomposition}, labelled by spin $\ell$ and degeneracy index $i$. This twist degeneracy will be partially lifted in the next sections, when we include perturbative corrections to the four-point correlator. Using the conformal block decomposition of \eqref{tree.level} we can compute the tree-level structure constants, i.e.\ the OPE coefficients. They are non-zero only for even spins $\ell$ and take the form
\begin{align}\label{eq:structureconstants}
\langle a^{(0)}_{\tau_0,\ell}\rangle=\begin{cases} 2c\,\frac{\Gamma(\ell+\frac{\tau_0}{2})^2}{\Gamma(2\ell+\tau_0-1)}\,,&\tau_0=2\,, \vspace{0.2cm}
\\
2\frac{\Gamma(\frac{\tau_0}{2}-1)^2\Gamma(\frac{\tau_0}{2}+\ell)^2}{\Gamma(\tau_0-3)\Gamma(\tau_0+2\ell-1)}\left(c\,(-1)^{\frac{\tau_0}{2}}+(\tau_0+\ell-2)(\ell+1)\right)\,,& \tau_0>2\,,\end{cases}
\end{align}
where we have introduced an average of structure constants over operators with the same reference twist and spin, $\langle a^{(0)}_{\tau_0,\ell}\rangle:=\sum_i a^{(0)}_{\tau_0,\ell,i}$. Notice that from the correlator \eqref{tree.level} alone it is not possible to calculate individual structure constants by this procedure. 

In the following sections we will find the most general one-loop correction to \eqref{tree.level} using the conformal symmetry, crossing symmetry and the structure of the OPE. In particular, we will compute an explicit form of the perturbative corrections to the structure constants $\langle a^{(0)}_{\tau_0,\ell}\rangle \to\langle a^{(0)}_{\tau_0,\ell}\rangle+g\langle a^{(1)}_{\tau_0,\ell}\rangle$ as well as to the twists $\tau_0\to\tau_0+g \,\frac{\langle a^{(0)}_{\tau_0,\ell} \gamma^{(1)}_{\tau_0,\ell}\rangle}{\langle a^{(0)}_{\tau_0,\ell}\rangle}$. The knowledge of results for individual operators $\mathcal{O}_{\tau,\ell,i}$ will not be necessary to find the complete four-point correlator at one loop, they will become relevant only at the two-loop order. We will comment on this matter in the outlook of this chapter.

\subsection{Superconformal block decomposition for half-BPS operators}
\label{sec:super.four.points}
As the second example, we consider the four-point correlation function of four half-BPS operators $\mathcal{O}_{{\boldsymbol{20'}}}$ in $\mathcal{N}=4$ SYM, which are protected and their dimension is $\Delta_{\mathcal{O}_{{\boldsymbol{20'}}}}=2$. The four-point correlation function of such operators decomposes into the following two contributions
\begin{equation}\label{eq:super.four.point}
\langle \mathcal{O}_{{\boldsymbol{20'}}}(x_1,y_1)\mathcal{O}_{{\boldsymbol{20'}}}(x_2,y_2)\mathcal{O}_{{\boldsymbol{20'}}}(x_3,y_3)\mathcal{O}_{{\boldsymbol{20'}}}(x_4,y_4)\rangle=\mathcal{G}_{\mathrm{free}}(x,y)+\mathcal{G}_{\mathrm{pert}}(x,y),
\end{equation}
where $\mathcal{G}_{\mathrm{pert}}(x,y)$ vanishes when $g\to0$. The part $\mathcal{G}_{\mathrm{free}}(x,y)$ corresponds to the free theory and is a rational function of space time and R-symmetry coordinates. Again, it can be evaluated directly by Wick contractions and it boils down to the same set of graphs as in figure~\ref{fig:Free}. It renders
\begin{equation}
\mathcal{G}_{\mathrm{free}}(x,y)=d_{12}^2d_{34}^2+d_{13}^2d_{24}^2+d_{14}^2d_{23}^2+\tilde c\,\big(d_{12}d_{23}d_{34}d_{14}+d_{12}d_{24}d_{34}d_{13}+d_{13}d_{24}d_{23}d_{14}\big),
\end{equation}
where the superpropagator $d_{ij}$ is given by
\begin{equation}
d_{ij}=\frac{y_{ij}^2}{x_{ij}^2},\qquad y_{ij}=y_i\cdot y_j,
\end{equation}
and $\tilde c$ is a theory-dependent constant which for $\SU N$ $\mathcal N=4$ SYM again depends only on the central charge $\tilde c\sim (N^2-1)^{-1}$.

From the superconformal Ward identities \cite{Nirschl:2004pa}, the interacting part of the four-point function can be written in a factorised form 
\begin{equation}
\mathcal{G}_\mathrm{pert}(z,\zb,\alpha,\alphab)=d_{12}^2d_{34}^2 \frac{(z-\alpha)(z-\alphab)(\zb-\alpha)(\zb-\alphab)}{(\alpha\, \alphab)^2}\mathcal{H}(u,v)\,,
\end{equation}
where we have introduced a set of cross-ratios for the R-symmetry coordinates
\begin{equation}
\alpha \alphab=\frac{y_{12}^2y_{34}^2}{y_{13}^2y_{24}^2},\qquad (1-\alpha)(1-\alphab)=\frac{y_{14}^2y_{23}^2}{y_{13}^2y_{24}^2}\,.
\end{equation}
Similar to the four-point function of Konishi operators, crossing symmetry implies that the function $\mathcal{H}(u,v)$ satisfies the two equations
\begin{equation}\label{eq:superbootstrap}
 \mathcal{H}(u,v)=\frac{1}{v^2}\mathcal{H}\left(\frac{u}{v},\frac{1}{v}\right)\,,\qquad v^{2}\mathcal{H}(u,v)=u^{2}\mathcal{H}(v,u)\,,
\end{equation} 
where in the second equation we used explicitly the fact that $\Delta_{\mathcal{O}_{{\boldsymbol{20'}}}}=2$.

On the other hand, the four-point correlation function~\eqref{eq:super.four.point} admits a superconformal block decomposition, see e.g.\ \cite{Doobary:2015gia} 
\begin{equation}
\mathcal{G}_{\mathrm{free}}(x,y)+\mathcal{G}_{\mathrm{pert}}(x,y)=d_{12}^2d_{34}^2 \sum_{\mathcal{R},i}A_{\mathcal{R},i}\,\mathcal{S}_{\mathcal{R}}(x,y),
\end{equation}
where the sum runs over all superconformal primary operators appearing in the OPE expansion of two half-BPS operators
\begin{equation}\label{eq:super.OPE}
\mathcal{O}_{{\boldsymbol{20'}}}\times \mathcal{O}_{{\boldsymbol{20'}}}\sim\sum_{\mathcal{R},i}C_{\mathcal{O}_{{\boldsymbol{20'}}}\mathcal{O}_{{\boldsymbol{20'}}}\mathcal{O}_{\mathcal{R},i}}\left(\mathcal{O}_{\mathcal{R},i}+\ldots\right).
\end{equation}
Superconformal primaries in \eqref{eq:super.OPE} are labelled by their twist $\tau=\Delta-\ell$, spin $\ell$ and a representation of the $\SU4$ R-symmetry of $\mathcal{N}=4$ SYM, which we collectively denote by $\mathcal{R}$. Again, we also introduced the label $i$  which takes care of a possible additional degeneracy of operators with the same twist, spin and the R-symmetry label. Importantly, the superconformal blocks do not depend on the label $i$. An explicit description of superconformal multiplets and an explicit form of the superconformal blocks $\mathcal{S}_{\mathcal{R}}$ can be found in the appendix~\ref{app:superblocks}. As it is summarised there, we distinguish three types of supermultiplets in \eqref{eq:super.OPE}: half-BPS, quarter-BPS and long supermultiplets. All half-BPS and most quarter-BPS supermultiplets have their conformal dimensions and structure constants protected by supersymmetry. Then, their two-point and three-point correlation functions are completely determined by the free part $\mathcal{G}_\mathrm{free}$. They will therefore not contribute to the interacting part $\mathcal{H}(u,v)$ of the four-point correlation function. The only exception are quarter-BPS supermultiplets at the unitarity bound. They can combine in the interacting theory to form a long, non-protected supermultiplet \cite{Dolan:2002zh,Heslop:2003xu}. This is exactly the case for the twist-two operators. Together with the other long supermultiplets they form a complete non-protected spectrum of operators present in the intermediate channel. Since we want to find the one-loop correction to $\mathcal{H}(u,v)$, we will in the following be interested only in the non-protected part of the spectrum. 

We can perform a superconformal block decomposition of the leading contribution $\mathcal{G}_{\mathrm{free}}(x,y)$ to the four-point function and get structure constants for all non-protected multiplets
\begin{equation}\label{eq:structure.constants.Free}
 \langle A_{\tau_0,\ell}^{(0)}\rangle=\begin{cases} 2\tilde c\,\frac{\Gamma(\ell+\frac{\tau_0}{2}+2)^2}{\Gamma(2\ell+\tau_0+3)}\,,&\tau_0=2\,,
\\2\frac{\Gamma(\frac{\tau_0}{2}+1)^2\Gamma(\frac{\tau_0}{2}+\ell+2)^2}{\Gamma(\tau_0+1)\Gamma(\tau_0+2\ell+3)}\left(\tilde c\,(-1)^{\frac{\tau_0}{2}}+(\tau_0+\ell+2)(\ell+1)\right),& \tau_0=4,6,8,\ldots\,.\end{cases}
\end{equation}
It is interesting to notice that $\langle A^{(0)}_{\tau_0,\ell}\rangle=\langle a^{(0)}_{\tau_0,\ell}\rangle\big|_{c\to\tilde c,\tau_0\to\tau_0+4}$.

Furthermore, using the explicit form of superconformal blocks \eqref{eq:super.long} and \eqref{eq:twist2.recomb} for non-protected multiplets, the interacting part of the four-point correlation function can be expanded as
\begin{equation}\label{eq:superconformal.decomposition}
\mathcal H(u,v)=\sum_{\tau,\ell}\langle A_{\tau,\ell}\rangle u^{-2} \,G_{\tau+4,\ell}(z,\zb),
\end{equation}
where $G_{\tau,\ell}(z,\zb)$ is exactly the same conformal block as in \eqref{eq:ConformalBlock} in section \ref{sec:bosonic.four.points}. We notice that both leading-order structure constants $\langle A^{(0)}_{\tau_0,\ell}\rangle$ and superconformal blocks for non-protected supermultiplets are related to the Konishi case by shifting $\tau_0\to\tau_0+4$. For this reason, the one-loop calculation for the four-point correlator of half-BPS operators is analogous to a similar analysis for four Konishi operators, after this shift is implemented at the level of twist conformal blocks. 

\section{Twist conformal blocks}
\label{sec:twist.conformal.blocks}
In this section we describe twist conformal blocks and their generalisations introduced in \cite{Alday2016} and use them to rewrite the conformal block decomposition of four-point correlation functions from the previous section. We focus in this section exclusively on the case of four Konishi operators, leaving the half-BPS case to section~\ref{sec:super.case}. We start by defining twist conformal blocks relevant for the tree-level correlators and then define their generalisations with spin-dependent insertions that will be relevant for the perturbative expansion around the tree-level solution. 

\subsection{Twist conformal blocks}
A motivation to study twist conformal blocks is the observation that in perturbation theory there exists, for each even number $\tau_0=2,4,6,\ldots$, an infinite family of operators $\mathcal{O}_{\tau_0,\ell,i}$, $\ell=0,2,4,\ldots$, $i=1,\ldots,d_{\tau_0,\ell}$, with the reference twist equal to $\tau_0$:
\begin{equation}
\tau=\tau_0+O(g)\,.
\end{equation}
Therefore, at tree-level we have an infinite twist degeneracy which is lifted only when we turn on the coupling constant. In particular, it motivates us to resum contributions coming from all intermediate operators with the same reference twist $\tau_0$. In this case, the leading order four-point correlator \eqref{tree.level} can be decomposed as
\begin{equation}
\mathcal G^{(0)}(u,v)=\!\sum_{\tau_0=2,4,\ldots}\! H_{\tau_0}(u,v),
\end{equation}
where we have defined {\it twist conformal blocks} 
\begin{equation}\label{eq:Hfun.def}
H_{\tau_0}(u,v)=\sum_{\ell=0}^\infty \langle a^{(0)}_{\tau_0,\ell}\rangle G_{\tau_0,\ell}(u,v),
\end{equation}
with $\langle a^{(0)}_{\tau_0,\ell}\rangle$  given in \eqref{eq:structureconstants}. The sum in \eqref{eq:Hfun.def} can be  performed for any $\tau_0$ using the explicit form of conformal blocks. For example for $\tau_0=2$ it renders
\begin{equation}\label{eq.H2expl}
H_2(u,v)=c\frac{u}{v}+c\, u .
\end{equation}
For higher twists, the explicit form of $H_{\tau_0}(u,v)$ is more involved and we will not present it here. However, in all subsequent calculations we will need only their power divergent part as $v\to 0$. Such divergent parts can be easily calculated and written in a closed form as we will show below.

\subsection{H-functions}
In order to study perturbative corrections to the tree-level correlation function $\mathcal{G}^{(0)}(u,v)$ we need to generalise the notion of twist conformal blocks. In particular, when the coupling constant $g$ is not zero, the twist degeneracy we observed at the tree-level is lifted and each $\mathcal O_{\tau_0,\ell,i}$ gets individual corrections to their twists and structure constants,
\begin{align}
\tau_{\tau_0,\ell,i}&=\tau_0+g\,\gamma^{(1)}_{\tau_0,\ell,i}+\mathcal O(g^2),
\\
a_{\tau_0,\ell,i}&=a^{(0)}_{\tau_0,\ell,i}+g \,a^{(1)}_{\tau_0,\ell,i}+\mathcal O(g^2).
\end{align}
Here $\gamma^{(1)}_{\tau_0,\ell,i}$ is the one-loop anomalous dimension of $\mathcal O_{\tau_0,\ell,i}$ and $a^{(1)}_{\tau_0,\ell,i}$ is the one-loop correction to the structure constants. In the conformal block decomposition, these corrections will introduce an additional dependence on the spin and will modify the sum in the definition of the twist conformal blocks. Therefore, we will need to calculate sums of the form
\begin{equation}\label{eq:sum.with.insertions}
\sum_{\ell=0}^\infty \langle a^{(0)}_{\tau_0,\ell}\rangle  \kappa_{\tau_0}(\ell)\,G_{\tau_0,\ell}(u,v) ,
\end{equation}
where $\kappa_{\tau_0}(\ell)$ stands for the spin dependence coming from either the anomalous dimensions or the OPE coefficients. 
In particular, these insertions can be of two kinds: unbounded in spin $\ell$ or truncated contributions with finite support in $\ell$. The truncated contributions do not affect the enhanced divergent part of correlator and we will postpone their study to the following section. On the other hand, for the insertions unbounded  in spin the sum \eqref{eq:sum.with.insertions} can be calculated as an expansion around the infinite value of spin. In particular, in the unbounded case $\kappa_{\tau_0}(\ell)$ can be expanded in inverse powers of the conformal spin $J_{\tau_0}^2=(\frac{\tau_0}{2}+\ell)(\frac{\tau_0}{2}+\ell-1)$:
\begin{equation}
\kappa_{\tau_0}(\ell)=\sum_{m=0}^\infty \left( \frac{C_{(m)}}{J_{\tau_0}^{2m}}+\frac{C_{(m,\log)}}{J_{\tau_0}^{2m}}\log J_{\tau_0}+\ldots\right),
\end{equation}
as was shown in \cite{AldayBissiLuk2015}.
Then, in order to study perturbation theory beyond the tree-level, we consider a set of functions \cite{Alday2016}
\begin{equation}\label{eq:HfunctionsDef}
H^{(m,\log^n)}_{\tau_0}(u,v)=\sum_\ell \langle a_{\tau_0,l}^{(0)}\rangle \frac{(\log J_{\tau_0})^n}{J_{\tau_0}^{2m}}G_{\tau_0,l}(u,v),
\end{equation}
which we will refer to as {\em H-functions}. The H-functions describe contributions from an infinite sum of conformal blocks with spin-dependent insertions. In the case $m=n=0$ the H-functions $H_{\tau_0}^{(0)}(u,v)$ coincide with the twist conformal blocks. Importantly, the functions \eqref{eq:HfunctionsDef} satisfy the following recursion relation 
\begin{equation}\label{eq:Hrec}
H_{\tau_0}^{(m,\log^n)}(u,v)=\mathcal C H_{\tau_0}^{(m+1,\log^n)}(u,v),
\end{equation}
where we defined the shifted quadratic Casimir \eqref{eq:casimirrelationforTCB}
\begin{equation}\label{eq:fullCasimir}
\mathcal C=D_z+D_\zb+2\frac{z \zb}{z-\zb}\left((1-z)\partial_z-(1-\zb)\partial_\zb\right)-\frac{\tau_0(\tau_0-6)}4,
\end{equation}
with $D_x=(1-x)x^2\partial^2_x-x^2\partial_x$ .  The relation \eqref{eq:Hrec} can be easily proven by noticing that each individual conformal block $G_{\tau_0,\ell}(z,\zb)$ is an eigenvector of the Casimir operator $\mathcal{C}$ with the eigenvalue $J_{\tau_0}^2$.

\subsection{Enhanced divergences}\label{sec:enhanced.div}
In the following we will not need an explicit form of the functions $H^{(m,\log^n)}_{\tau_0}(u,v)$ but only their enhanced divergent part as $v\to0$. Expanding \eqref{eq:ConformalBlock} in this limit, one can notice that the conformal blocks behave as a logarithm $G_{\tau_0,\ell}(u,v)\sim \log(v)$ for $v\to0$. By enhanced divergence we will mean terms which cannot be written as a finite sum of conformal blocks. There are two kinds of enhanced divergences we will encounter: inverse powers of $v$, and functions with higher powers of the logarithm, that is functions of the form $p(v)\log^n v$, $n>1$, where $p(v)$ is regular for $v\to0$. 
As was shown in \cite{AldayBissiLuk2015}, the power divergent part of $H_{\tau_0}(u,v)$ is completely determined by operators with large spin $\ell$. In order to compute this divergent part it is therefore sufficient to study the tail of the sum in \eqref{eq:Hfun.def}. As explained in the following section such computations can be done explicitly. For example, at $\tau_0=2$ it renders
\begin{equation}\label{eq:H2.div}
H^{(0)}_2(u,v)=c\frac{u}{v}+O(v^0).
\end{equation}
One notices that the power divergence agrees with the explicit calculation in \eqref{eq.H2expl}. Moreover, the finite term $O(v^0)$ will not be necessary in the following sections.

Throughout the chapter we will often be interested in comparing only the enhanced divergent part of various functions. For this reason we introduce a notation
\begin{equation}
f(u,v)\doteq g(u,v) \quad \mathrm{if} \quad f(u,v)=g(u,v)+p(u,v)+q(u,v)\log v , 
\end{equation}
where $p$ and $q$ are polynomials in $v$ with coefficients that are functions of $u$.
This is to say that $f(u,v)$ and $g(u,v)$ are equal up to ``regular terms'', by which we mean contributions which can come from a finite number of conformal blocks. In particular, regular terms can contain a single power of $\log v$ but no higher powers of the logarithm nor inverse powers of $v$. 

\subsection{Computing H-functions}
We now describe how to construct the power divergent part of the H-functions that we will need in the subsequent calculations. First, we describe how to use the kernel method, motivated by \cite{AldayMaldacena2007} and systematically developed in \cite{Komargodski:2012ek,Fitzpatrick:2012yx}. This method, however, becomes inefficient very fast. For this reason we explain how to use an alternative method based on the recursion relation \eqref{eq:Hrec}. We start by focusing on the case of operators with twist $\tau_0=2$, and later on describe how H-functions for higher twists arise naturally from the twist-two case.

\subsubsection{Factorisation}
We are only interested in the terms with a power divergence as $v\to 0$. In the following, it will be more convenient to use the coordinates $(z,\zb)$ instead of the cross-ratios $(u,v)$. In these coordinates we are interested in the limit $\zb\to1$. 
Using the definition \eqref{eq:Hfun.def} and the explicit form of conformal blocks, any power divergent contributions to twist conformal blocks must arise from an infinite sum over spins. Moreover, they can only come from the second part of the conformal block \eqref{eq:ConformalBlock}. Then the part of the twist conformal blocks with a power divergence as $\zb\to 1$ can be written as
\begin{equation}\label{eq:factorised.CBlocks}
\frac{z\,\zb}{\zb-z}k_{\frac{\tau_0}{2}-1}(z)\sum_{\ell=0}^\infty \langle a^{(0)}_{\tau_0,\ell}\rangle k_{\frac{\tau_0}{2}+\ell}(\zb).
\end{equation} 
Similar reasoning can be applied to all H-functions defined in \eqref{eq:HfunctionsDef}. For this reason the power divergent part of the H-functions takes a factorised form
\begin{equation}\label{eq:factorised.form}
H_{\tau_0}^{(m,\log^n)}(z,\zb)\doteq\frac{z}{\zb-z}k_{\frac {\tau_0}2-1}(z) \overline H_{\tau_0}^{(m,\log^n)}(\zb),
\end{equation}
where we have defined the functions 
\begin{equation}\label{eq:def.Hbar}
\overline H^{(m,\log^n)}_{\tau_0}(\zb)=\zb
\sum_{\ell=0}^\infty \langle a^{(0)}_{\tau_0,\ell}\rangle \frac{\log^nJ_{\tau_0}}{J_{\tau_0}^{2m}}\, k_{\frac{\tau_0}{2}+\ell}(\zb).
\end{equation}

We notice now that the action of the quadratic Casimir \eqref{eq:fullCasimir} simplifies significantly when applied only to the divergent part of the H-functions
\begin{equation}\label{eq:casimiraction.red}
\mathcal C H_{\tau_0}^{(m,\log^n)}(z,\zb)\doteq\frac{z}{\zb-z}k_{\frac{\tau_0}{2}-1}(z)\, \overline{\mathcal D} \,\overline H_{\tau_0}^{(m,\log^n)}(\zb),
\end{equation}
where
\begin{equation}\label{eq:Casimir.reduced}
\overline{\mathcal D}= (2-\zb)(1-\zb\partial_\zb)+\zb^2(1-\zb)\partial_\zb^2 = \zb\, \Dbar\, \zb^{-1}
\end{equation}
where $\Dbar$ is the $\SL2\R$ Casimir defined in \eqref{eq:Dbardeff}.
Additionally, due to \eqref{eq:casimiraction.red}, the recursion \eqref{eq:Hrec} implies a similar recursion relation for $\overline H_{\tau_0}^{(m,\log^n)}(\zb)$, taking the form
\begin{equation}\label{eq:rec.simple}
\overline{H}_{\tau_0}^{(m,\log^n)}(\zb)=\overline{\mathcal{D}}\, \overline{H}_{\tau_0}^{(m+1,\log^n)}(\zb)\,.
\end{equation}

It is important to notice that the operator $\overline{\mathcal{D}}$ maps regular terms to regular terms and therefore does not introduce any enhanced divergence while acting on finite sums of conformal blocks. More generally, for polynomial functions $p(\zb)$ it acts as
\begin{equation}
\overline{\mathcal{D}} (p(\zb)\log(1-\zb)^n)=\frac{n(n-1)\,\zb\, p(\zb)\log(1-\zb)^{n-2}}{1-\zb}+O((1-\zb)^0).
\end{equation}
It is clear that for $n=0,1$ no enhanced divergence is produced when acting with $\overline{\mathcal{D}}$. On the other hand, expressions with higher powers of the logarithm, namely $n>1$, will always produce terms with negative powers of $1-\zb$ after we act on them with $\overline{\mathcal{D}}$ a finite number of times. This property explains why we refer to such terms as enhanced divergent.

\subsubsection{Derivation of H-functions: kernel method}\label{sec:kernel.method}
Let us now focus on finding the power divergent part of the functions $\overline H_{\tau_0}^{(m,\log^n)}(\zb)$. In principle, this is possible for any $m$ and $n$. However, in order to solve the one-loop problem we will see that it is sufficient to focus on $\overline H_{\tau_0}^{(m,\log^n)}(\zb)$ for $n=0,1$ and $m\leqslant 0$. Since we want to compute just the power divergent part of these functions we only need to consider the tail of the sum over spins in \eqref{eq:def.Hbar}. In this limit the sum is well-approximated by an integral which can be explicitly computed using the method described in \cite{Komargodski:2012ek,Fitzpatrick:2012yx}, see also the appendix A of \cite{AldayZhiboedov2015}. This method allows to capture all power divergences, namely all terms of the form $\sim \frac{1}{(1-\zb)^k}$ for $k>0$. 

Let us start by considering the twist conformal block $\overline{H}_{2}^{(0)}(\zb)$ and compute
\begin{equation}
 \zb\sum_\ell \langle a_{2,\ell}^{(0)}\rangle\,k_{\ell+1}(\zb)=\sum_\ell 2c\frac{\Gamma(\ell+1)^2}{\Gamma(2\ell+1)}\, \zb^{\ell+2}\,{_2F_1}(\ell+1,\ell+1,2\ell+2;\zb).
\end{equation}
The divergent contributions come from large spins of order $\ell\sim\frac{1}{\sqrt{\varepsilon}}$, where we have introduced the notation $\varepsilon=1-\zb$ in order to simplify the following formulae. Therefore, we can define $\ell=\frac{p}{\sqrt{\varepsilon}}$ and convert the sum over $\ell$ into the integral $\frac12\int \frac{dp}{\sqrt\varepsilon}$. We also replace the hypergeometric function by its integral representation
\begin{equation}
{_2F_1}(a,b;c;x)=\frac{\Gamma(c)}{\Gamma(b)\Gamma(c-b)}\int\limits_0^1d t\frac{t^{b-1}(1-t)^{c-b-1}}{(1-x\, t)^a} .
\end{equation}
Consecutively, we perform the change of variables
\begin{equation}
\frac{p}{\sqrt{\varepsilon}}\left(\frac{p}{\sqrt{\varepsilon}}+1\right)=\frac{\hat J^2}{\varepsilon},\qquad t=1-w\sqrt{\varepsilon}\,.
\end{equation}
 The integration limits of the $w$ integral can safely be extended to $[0,\infty)$ since this does not add any power divergent term. 
Implementing these changes of variables gives the result
\begin{equation}
 \zb\sum_\ell \langle a_{2,\ell}^{(0)}\rangle\,k_{\ell+1}(\zb)\rightarrow(1-\varepsilon)\,c \int_0^\infty d{\hat J}\, \mathcal K_2({\hat J},\varepsilon),
\end{equation}
where we have defined the integral kernel
\begin{equation}
\mathcal K_2(j,\varepsilon)=\int_0^\infty dw\,\frac{-2\hat J}{w\,\varepsilon(w\sqrt\varepsilon-1)}\left(\frac{w(1-\varepsilon)(1-w\sqrt\varepsilon)}{w+\sqrt\varepsilon-w\varepsilon)}\right)^{\frac12\left(1+\sqrt{1+\frac{4\hat J^2}\varepsilon}\right)}.
\end{equation}
Expanding $\mathcal K_2(\hat J,\varepsilon)$ in powers of $\varepsilon$ we get
\begin{equation}\label{eq:expanding.kernel}
\mathcal K_2(\hat J,\varepsilon)=4\hat J K_0(2j)\frac1\varepsilon-\frac43\left(\hat J K_0(2\hat J)+(1+2\hat J^2)K_1(2\hat J)\right)+\ldots,
\end{equation}
where $K_n(x)$ are the modified Bessel functions of the second kind.

In particular, this method allows us to find
\begin{equation}\label{eq:kernel.H2}
\overline H^{(0)}_2(\zb)\doteq\frac{1}{1-\zb}c\int_0^\infty dj\, 4jK_0(2j)= \frac1{1-\zb}c+O((1-\zb)^0),
\end{equation}
which is exactly the previously mentioned result \eqref{eq:H2.div}. Importantly, it agrees up to regular terms with the direct calculation \eqref{eq.H2expl}. Let us emphasise that for the twist conformal block $\overline{H}^{(0)}_2(\zb)$ there are no additional enhanced divergences beyond the power divergence, namely there are no terms with $\log^n(1-\zb)$ for $n>1$. This statement will become crucial when we use the recursion relation method in the following section.

More generally, using this method we can find all negative powers of $\varepsilon=1-\zb$ of the H-functions with $m\leqslant 0$ by modifying the integrand with suitable insertions
\begin{equation}\label{eq:kernel.Hbar}
\overline{H}_{2}^{(m,\log^n)}(\zb)\doteq
(1-\varepsilon)\,c
\int_0^\infty dj\, \mathcal K_2(j,\varepsilon)\left(\frac\varepsilon{j^2}\right)^m\log^n\left(\frac{j}{\sqrt\varepsilon}\right).
\end{equation}
For example for $m=0$, $n=1$ we find after an explicit calculation
\begin{align}\nonumber
\overline H^{(0,\log)}_2(\zb)&\doteq\frac{1}{1-\zb}c\int_0^\infty d\hat J\, 4\hat JK_0(2\hat J)\left(\log \hat J-\frac12\log(1-\zb)\right)\\
&\doteq -\frac{\gamma_E}{1-\zb}c-\frac{\log(1-\zb)}{2(1-\zb)}c+O((1-\zb)^0),
\end{align}
where $\gamma_E$ is Euler's constant.

By studying the $\varepsilon$-dependence in \eqref{eq:kernel.Hbar} we also immediately find a general schematic form of the power divergent part of $\overline H_2^{(m,\log^n)}(\zb)$ for $m\leqslant 0$,
\begin{equation}\label{eq:form.of.Hfunctions}
\overline{H}_{2}^{(m,\log^n)}(\zb)\doteq \sum_{i=0}^{-m}\sum_{j=0}^{n}k_{i,j}^{(m,\log^n)}\frac{\log^{j}(1-\zb)}{(1-\zb)^{-m-i+1}}c,
\end{equation}
where all coefficient $k_{i,j}^{(m,\log^n)}$ in principle can be  calculated from \eqref{eq:kernel.Hbar}. This quickly becomes very tedious and for this reason we present a different approach in the following section.

\subsubsection{Derivation of H-functions: recursion relation method}
We will now move to a more efficient approach, where we derive the H-functions $\overline H_2^{(m,\log^n)}(\zb)$ using the recursion relation \eqref{eq:rec.simple}. 
From \eqref{eq.H2expl} the complete enhanced divergent part of twist conformal block for $\tau_0=2$ is $\overline H_2^{(0)}(\zb)\doteq \frac{c}{1-\zb}$. The recursion relation \eqref{eq:rec.simple} immediately allows us to find all divergent parts for all H-functions $\overline H_{2}^{(m)}(\zb)$ with $m<0$ by simply using 
\begin{equation}\label{eq:recminusm}
\overline{H}^{(m)}_{2}(\zb)=\overline{\mathcal{D}}^{-m} \overline{H}_{2}^{(0)}(\zb)\,,\qquad \mathrm{for}\quad m<0\,.
\end{equation} 
Also for positive $m$ we could in principle find the enhanced divergent part of the H-functions by solving differential equations \eqref{eq:rec.simple}. This becomes tedious very quickly and moreover we would need to introduce two constants of integration every time we increase $m$. However, as we already pointed out, we will not need H-functions with positive $m$ at all. Left to construct are therefore the H-functions with logarithmic insertions. As described in the appendix A.4 of \cite{Alday2016b}, these are given by differentiating the $\overline H^{(m)}_2(\zb)$ with respect to the parameter $m$:
\begin{equation}\label{eq:Hlog}
\overline H_2^{(m,\log^n)}(\zb)=-\frac12\frac\partial{\partial m} \overline H_2^{(m,\log^{n-1})}(\zb).
\end{equation}
We will only need to consider the case $n=1$, although the computation for $n>1$ is analogous. In order to find $\overline H^{(0,\log)}_2(\zb)$, we need to analytically continue $\overline H^{(m)}_2(\zb)$ with respect to the parameter $m$ and then take the derivative. The most general form of the enhanced divergent parts of $\overline H^{(m)}_2(\zb)$ for $m\leqslant 0$ is given by \eqref{eq:form.of.Hfunctions},
\begin{equation}\label{eq:Hm.exp}
\overline{H}_2^{(m)}(\zb)\doteq\sum_{i=0}^{-m} \frac{k^{(m)}_i}{(1-\zb)^{-m-i+1}}c,
\end{equation}
where all coefficients $k_i^{(m)}$ can be found explicitly from \eqref{eq:recminusm}. In particular, it allows us to derive a recursion relation for the coefficients $k_i^{(m)}$. For example for $k^{(m)}_{0}$ we get
\begin{equation}
k_0^{(m)}=m^2\,k_0^{(m+1)},
\end{equation}
which together with the initial condition $k_0^{(0)}=1$ coming from $\overline H_2^{(0)}(\zb)\doteq c\, (1-\zb)^{-1}$ allows us to find the general form
\begin{equation}
k_{0}^{(m)}=\Gamma(-m+1)^2 \,,\qquad \mathrm{for}\quad m\leqslant 0.
\end{equation}
Proceeding to subleading terms, and using as boundary conditions the explicit values of $k_i^{(-i)}$ for $i>0$ that can be calculated directly from \eqref{eq:recminusm}, one can find all expansion terms in \eqref{eq:Hm.exp}. We present few first terms below
\begin{align}\label{eq:Hm.exp.expl}
\overline H_2^{(m)}(\zb)&\doteq\frac{\Gamma(-m+1)^2c}{(1-\zb)^{-m+1}}+\frac{m(2m^2-6m+1)}{3}\frac{\Gamma(-m)^2c}{(1-\zb)^{-m}}+\nonumber\\&\quad+\frac{(m-1)m(m+1)(20m^3-54 m^2-35m+36)}{90}\frac{\Gamma(-m-1)^2c}{(1-\zb)^{-m-1}}+\ldots.
\end{align}
For all $m\leqslant 0$ this expansion is valid up to the order $(1-\zb)^{-1}$. Now, all expressions in \eqref{eq:Hm.exp.expl} are meromorphic functions and can be analytically continued to any value of $m$. Taking the derivative with respect to $m$, as in \eqref{eq:Hlog}, we obtain the divergent part of $\overline H_2^{(m,\log)}(\zb)$
\begin{align}\label{eq:Hlog.expl}
\overline H_2^{(m,\log)}(\zb)\doteq&-\frac{1}{2}\frac{\Gamma(-m+1)^2c}{(1-\zb)^{-m+1}}\left(\log(1-\zb)-2S_1(-m)+2\gamma_E\right)+\ldots\,,
\end{align}
where $S_k(N)=\sum_{i=1}^N \frac{1}{i^k}$ are harmonic sums. Again, for given $m\leqslant0$, this expansion is valid up to the order $(1-\zb)^{-1}$.

There exists a very compact way to encode all negative powers of $1-\zb$ in the functions $\overline{H}^{(m,\log)}_2(\zb)$ for $m\leqslant0$ by constructing the complete enhanced divergent part of $\overline H^{(0,\log)}_2(\zb)$. In order to do that we start with a general ansatz
\begin{equation}
\overline H_2^{(0,\log)}(\zb)=\frac{e_{\log}}{1-\zb}c\log(1-\zb)+\frac{e_{-1}}{1-\zb}c+\sum_{i=0}^\infty e_i(1-\zb)^ic\log^2(1-\zb).
\end{equation}
We can fix the coefficients $e_{i}$ and $e_{\log}$ by using the relation
\begin{equation}\label{eq:reclog}
\overline{H}^{(m,\log)}_{\tau_0}(\zb)=\overline{\mathcal{D}}^{-m} \overline{H}_{\tau_{0}}^{(0,\log)}(\zb)\,,\qquad \mathrm{for}\quad m<0,
\end{equation}
and comparing it with the previously obtained expansion \eqref{eq:Hlog.expl}. 
This allows us to find
\begin{equation}\label{eq:H0log}
\overline H_2^{(0,\log)}(\zb)=-\frac12\frac{\log(1-\zb)}{1-\zb}c-\frac{\gamma_E}{1-\zb}c+\left(-\frac1{12}+\frac{1-\zb}{10}-\frac{5(1-\zb)^2}{504}+\ldots\right)c\log^2(1-\zb).
\end{equation}
With this method arbitrarily many terms multiplying $\log^2(1-\zb)$ can be computed if we use~\eqref{eq:reclog} for a sufficiently large $-m$. We refer the reader to the appendix~\ref{app:H0log} where we have collected more orders of this expansion. Now, using the explicit form of $\overline{H}_2^{(0,\log)}(\zb)$ in \eqref{eq:H0log} we can easily find all negative powers of $\overline{H}_2^{(m,\log)}(\zb)$ for $m\leqslant 0$ by applying the formula \eqref{eq:reclog}. A similar analysis can be done also for $\overline H^{(m,\log^n)}_2(\zb)$ for $n>1$, however we will not need these functions in solving the one-loop problem.

 \subsubsection{Higher twist H-functions}

We end this section by describing how to compute the H-functions $\overline{H}_{\tau_0}^{(m,\log^n)}(\zb)$ for $\tau_0>2$. First of all, notice that the tree-level structure constants for higher twists \eqref{eq:structureconstants} can be nicely written using the tree-level structure constants for twist-two operators
\begin{equation}
\langle a^{(0)}_{\tau_0,\ell}\rangle=\frac{\Gamma(\frac{\tau_0}{2}-1)^2}{\Gamma(\tau_0-3)}\frac{1}{c}\left(c\,(-1)^{\frac{\tau_0}{2}}-(\tfrac{\tau_0}{2}-2)(\tfrac{\tau_0}{2}-1)+J_{\tau_0}^2\right)\langle a^{(0)}_{2,\ell+\frac{\tau_0}{2}-1}\rangle,
\end{equation}
where again $J_{\tau_0}^2=\left(\frac{\tau_0}{2}+\ell\right)\left(\frac{\tau_0}{2}+\ell-1\right)$.
When we plug this into the definition of twist conformal blocks for higher twist and perform a change of variables $j=\ell+\frac{\tau_0}{2}-1$ we get
\begin{equation}
\overline{H}^{(0)}_{\tau_0}(\zb)=\zb\frac{\Gamma(\frac{\tau_0}{2}-1)^2}{\Gamma(\tau_0-3)}\sum_{j=\frac{\tau_0}{2}-1}^\infty \frac{1}{c}\left(c\,(-1)^{\frac{\tau_0}{2}}-(\tfrac{\tau_0}{2}-2)(\tfrac{\tau_0}{2}-1)+(J_{2})^2\right)\langle a^{(0)}_{2,j}\rangle k_{j+1}(\zb).
\end{equation}
where $(J_2)^2=j(j+1)$.
 In the limit $\zb\to1$ the sum over $j$ can be replaced by a sum from zero to infinity since the difference is a regular term. This leads to
\begin{equation}
\overline H^{(0)}_{\tau_0}(\zb)\doteq\frac{\Gamma(\frac{\tau_0}{2}-1)^2}{\Gamma(\tau_0-3)}\frac{1}{c}\left(\left(c\,(-1)^{\frac{\tau_0}{2}}-(\tfrac{\tau_0}{2}-2)(\tfrac{\tau_0}{2}-1)\right)\overline H^{(0)}_2(\zb)+\overline H_2^{(-1)}(\zb)\right).
\end{equation}
This allows us to rewrite the twist conformal blocks for higher twists in terms of functions we have already constructed. Similar analysis can be performed for all H-functions leading to the explicit form for higher-twists 
\begin{equation}\label{eq:computing.higher.twist.H}
\overline H^{(m,\log^n)}_{\tau_0}(\zb)\doteq\frac{\Gamma(\frac{\tau_0}{2}-1)^2}{\Gamma(\tau_0-3)}\frac{1}{c}\!\left(\!\left(c\,(-1)^{\frac{\tau_0}{2}}-(\tfrac{\tau_0}{2}-2)(\tfrac{\tau_0}{2}-1)\right)\!\overline H^{(m,\log^n)}_2(\zb)+\overline H^{(m-1,\log^n)}_2(\zb)\right)\!.
\end{equation}
To summarise, all H-functions relevant for the one-loop problem can be constructed using just two functions: $\overline H^{(0)}_2(\zb)$ and $\overline H^{(0,\log)}_2(\zb)$ whose explicit form can be found in \eqref{eq:H2.div} and \eqref{eq:H0log}, respectively.

\subsection{Decomposing one-loop correlator into H-functions}
\label{sec:one-loop.correlator.into.H}
Knowing the explicit form of the H-functions, we focus now on the one-loop four-point correlation function $\mathcal{G}^{(1)}(z,\zb)$ and expand its power divergent part in terms of the H-functions. By doing this we focus only on  contributions to anomalous dimensions and structure constants unbounded in spin $\ell$. Later on we will also include terms which are truncated in spin. The latter do not interfere with our analysis of the power divergent part of the correlator. 

For each operator present in the intermediate channel we expand their conformal dimension and structure constants as follows
\begin{align}
\tau_i&=\tau_0+g\, \gamma^{(1)}_{\tau_0,\ell,i}+O(g^2),\\ a_{\tau_i,\ell,i}&=a^{(0)}_{\tau_0,\ell,i}+g\, a_{\tau_0,\ell,i}^{(1)}+O(g^2).
\end{align}
Then the four-point correlation function $\mathcal{G}(z,\zb)$, up to the order $g$, can be written as
\begin{align}\label{eq:G0G1exp}
&\mathcal{G}^{(0)}(z,\zb)+g\, \mathcal{G}^{(1)}(z,\zb)\nonumber\\
&=\sum\limits_{\tau_0,\ell,i}\left(a^{(0)}_{\tau_0,\ell,i}+g\, a^{(1)}_{\tau_0,\ell,i}\right)\left(G_{\tau_0,\ell}(z,\zb)+g\,\gamma^{(1)}_{\tau_0,\ell,i}\left(\frac\partial{\partial\tau} G_{\tau,\ell}(z,\zb)\right)\big|_{\tau\to\tau_0}\right)
\\
&=\sum\limits_{\tau_0,\ell}\langle a_{\tau_0,\ell}^{(0)}\rangle G_{\tau_0,l}(z,\zb)+g\,\sum\limits_{\tau_0,\ell}\left(\langle a_{\tau_0,\ell}^{(1)}\rangle G_{\tau_0,\ell}(z,\zb)+\langle a_{\tau_0,\ell}^{(0)}\gamma^{(1)}_{\tau_0,\ell}\rangle\left(\frac\partial{\partial\tau} G_{\tau,\ell}(z,\zb)\right)\big|_{\tau\to\tau_0}\right)\nonumber,
\end{align}
where we have again defined the averages $\langle f_{\tau_0,\ell}\rangle=\sum_i f_{\tau_0,\ell,i}$. 

In the last line of \eqref{eq:G0G1exp} the derivative with respect to twist $\tau$ is understood as a partial derivative of a function of two variables: $\tau$ and $\ell$. It turns out that our further analysis simplifies significantly if we instead use the variables $(\tilde\tau,\tilde\ell)$ defined as
\begin{equation}
 \left(\tilde\tau,\,\tilde\ell\,\right)=\left(\tau,\ell+\frac{\tau}{2}\right).
\end{equation} 
Up to a constant factor, this is equivalent to a change of variables to $(h,\hb)$ as discussed at the end of section~\ref{sec:blockology}.
Then the partial derivatives in the new variables can be related to the partial derivatives with respect to the twist and spin as
\begin{equation}
\frac\partial{\partial\tau}=\frac\partial{\partial\tilde\tau}+\frac{1}{2}\frac\partial{\partial\tilde\ell}\,,\qquad\frac\partial{\partial\ell}=\frac\partial{\partial\tilde\ell}\,.
\end{equation}
In particular, it implies that $\partial_{\tilde\tau}k_{\frac\tau2+\ell}(\zb)=0$. 
We can now rewrite the derivative in the last line of \eqref{eq:G0G1exp} as
\begin{align}
\sum\limits_{\tau_0,\ell}\left(\langle a_{\tau_0,\ell}^{(0)}\gamma^{(1)}_{\tau_0,\ell}\rangle\left(\frac\partial{\partial\tilde\tau} G_{\tau,\ell}(z,\zb)\right)\big|_{\tau\to\tau_0}+\frac{1}{2}\langle a_{\tau_0,\ell}^{(0)}\gamma^{(1)}_{\tau_0,\ell}\rangle\left(\frac\partial{\partial\tilde\ell} G_{\tau_0,\ell}(z,\zb)\right)\right)\nonumber\\
\doteq \sum\limits_{\tau_0,\ell}\left(\langle a_{\tau_0,\ell}^{(0)}\gamma^{(1)}_{\tau_0,\ell}\rangle\left(\frac\partial{\partial\tilde\tau} G_{\tau,\ell}(z,\zb)\right)\big|_{\tau\to\tau_0}-\frac{1}{2}\frac{\partial}{\partial \tilde\ell}\left(\langle a_{\tau_0,\ell}^{(0)}\gamma^{(1)}_{\tau_0,\ell}\rangle\right)G_{\tau_0,\ell}(z,\zb)\right),
\end{align}
where in the second line we dropped a total derivative with respect to $\tilde\ell$, which is a regular term. Finally, we can rewrite the divergent part of $\mathcal{G}^{(1)}(z,\zb)$ as
\begin{equation}\label{eq:G1.almostH}
\mathcal{G}^{(1)}(z,\zb)\doteq\frac{ z\zb}{\zb-z}\sum\limits_{\tau_0,\ell}\langle a^{(0)}_{\tau_0,\ell}\rangle\left(\overline{\langle \hat\alpha_{\tau_0,\ell}\rangle} k_{\frac{\tau_0}{2}-1}(z)+\overline{\langle\gamma_{\tau_0,\ell}\rangle}\left(\frac\partial{\partial\tau} k_{\frac{\tau}{2}-1}(z)\right)\big|_{\tau\to\tau_0} \right)k_{\frac{\tau_0}{2}+\ell}(\zb),
\end{equation}
where we used the factorisation~\eqref{eq:factorised.CBlocks} of the divergent parts of the conformal blocks and introduced
\begin{align}\label{eq:gammatau}
\overline{\langle  \gamma_{{\tau_0},\ell}\rangle} &:= \frac{\langle  a^{(0)}_{{\tau_0},\ell}\gamma^{(1)}_{{\tau_0},\ell}\rangle}{\langle  a^{(0)}_{{\tau_0},\ell}\rangle},\\
\overline{\langle \hat \alpha_{{\tau_0},\ell}\rangle}
\langle a_{{\tau_0},\ell}^{(0)}\rangle \nonumber
&:=
 \langle a_{{\tau_0},\ell}^{(1)}\rangle-\frac{1}{2}\frac\partial{\partial\ell}\left(\langle a^{(0)}_{{\tau_0},\ell}\gamma^{(1)}_{{\tau_0},\ell}\rangle\right)\\&\ =\langle a_{{\tau_0},\ell}^{(1)}\rangle
 -
 \frac{1}{2}\langle a^{(0)}_{{\tau_0},\ell}\rangle\frac\partial{\partial\ell}\overline{\langle \gamma_{{\tau_0},\ell}\rangle}
 -
 \frac{1}{2}\frac\partial{\partial\ell}\langle a^{(0)}_{{\tau_0},\ell}\rangle\overline{\langle \gamma_{{\tau_0},\ell}\rangle}.\label{eq:modified.structure.constant}
\end{align}
One can recognise the last formula in \eqref{eq:modified.structure.constant} as the one-loop perturbative expansion of $\hat a_{\tau_0,\ell}$ introduced in \cite{Alday2016b}.

In weakly coupled CFTs at one loop, both the anomalous dimensions $\overline{\langle  \gamma_{{\tau_0},\ell}\rangle}$ and the modified structure constants $\overline{\langle \hat \alpha_{{\tau_0},\ell}\rangle}$ depend on spin as a single logarithm $\log \ell$ at large $\ell$. Therefore, in order to use the H-functions to constrain the unbounded parts of the CFT-data we expand the modified structure constants $\overline{\langle\hat\alpha_{\tau_0,\ell}\rangle}$ and anomalous dimensions $\overline{\langle\gamma_{\tau_0,\ell}\rangle}$ in the following way \cite{AldayBissiLuk2015}:
\begin{align}\label{eq:OPEexpansion}
\overline{\langle \hat \alpha_{\tau_0,\ell}\rangle} &=\sum_{m=0}^\infty \frac{A_{\tau_0,(m,\log)}}{J^{2m}_{\tau_0}}\log J_{\tau_0}+\sum_{m=0}^\infty \frac{A_{\tau_0,(m)}}{J_{\tau_0}^{2m}},\quad\\\label{eq:gammaexpansion}
\overline{\langle  \gamma_{\tau_0,\ell}\rangle} &=\sum_{m=0}^\infty \frac{B_{\tau_0,(m,\log)}}{J^{2m}_{\tau_0}}\log J_{\tau_0}+\sum_{m=0}^\infty \frac{B_{\tau_0,(m)}}{J_{\tau_0}^{2m}}.
\end{align}

Inserting the expansions \eqref{eq:OPEexpansion} and \eqref{eq:gammaexpansion} into \eqref{eq:G1.almostH} we can finally rewrite the divergent part of the one-loop correlator in terms of H-functions
\begin{align}\label{eq:oneloop.in.Hfunctions}
\mathcal G^{(1)}(z,\zb)\doteq \sum_{\tau_0}\frac{ z}{\zb-z}\sum_\rho\left( A_{\tau_0,\rho}\,k_{\frac{\tau_0}{2}-1}(z)+B_{\tau_0,\rho}\,\left(\frac\partial{\partial\tau} k_{\frac{\tau}{2}-1}(z)\right)\big|_{\tau\to\tau_0}\right) \overline H_{\tau_0}^{\,\rho}(\zb),
\end{align}
where $\rho=(m,\log)$ or $\rho=(m)$, $m=0,1,2,\ldots$ and we have used the definition of H-functions \eqref{eq:def.Hbar}. This is the most important formula of this section and in the following we will use it to completely fix the form of $\mathcal{G}^{(1)}(z,\zb)$. 

\subsection{Using H-functions: toy example}
\label{sec:method}

We present a simple example of how to use H-functions to extract the asymptotic spin dependence of CFT-data given a particular function with power divergences. In order to simplify our discussion we focus here only on the $\zb$ dependence. In analogy with the actual computations in the next section, we will assume that the sum of H-functions produces a divergent expression containing a constant term and a term proportional to $\log(1-\zb)$:
\begin{equation}\label{eq:toy.equation}
\sum_{m=0}^\infty\sum_{n=0}^1 C_{(m,\log^n)}\overline H_2^{(m,\log^n)}(\zb)\doteq\frac{\lambda_1\log(1-\zb)+\lambda_0}{1-\zb}c.
\end{equation}
We will work iteratively and fix coefficients $C_{(m,\log^n)}$ by repeatedly applying the Casimir operator \eqref{eq:Casimir.reduced} on both sides of \eqref{eq:toy.equation} and keeping only power divergent terms. As a first step let us analyse the power divergent terms of \eqref{eq:toy.equation} itself. In this case only two terms in the sum on the left hand side are power divergent as $\zb\to1$, namely $\overline H^{(0)}_2(\zb)$ and $\overline H^{(0,\log)}_2(\zb)$. Therefore we get
\begin{equation}
C_{(0)}\frac{1}{1-\zb}c+C_{(0,\log)}\left( -\frac{\gamma_E}{1-\zb}-\frac{\log(1-\zb)}{2(1-\zb)}\right)c =
\frac{\lambda_1\log(1-\zb)+\lambda_0}{1-\zb}c,
\end{equation}
where we used the explicit form of $\overline H^{(0)}_2(\zb)$ and  $\overline H^{(0,\log)}_2(\zb)$. Solving this equation we get
\begin{equation}\label{eq:toy.zero.order.sols}
C_{(0)}=\lambda_0-2\lambda_1\gamma_E,\qquad C_{(0,\log)}=-2\lambda_1.
\end{equation}
To compute higher coefficients we act with the Casimir $\overline{\mathcal D}$ on both sides of \eqref{eq:toy.equation} and again compare power divergent terms. On the left hand side, using the recurrence \eqref{eq:rec.simple}, the Casimir brings the previously undetermined coefficients $C_{(1)}$ and $C_{(1,\log)}$ into the problem. This renders
\begin{equation}
\sum_{m=0}^1 \sum_{n=0}^1 C_{(m,\log^n)}\overline H^{(m-1,\log^n)}_2(\zb) \doteq\overline{\mathcal D}\left(\frac{\lambda_1\log(1-\zb)+\lambda_0}{1-\zb}c \right).
\end{equation}
Using the explicit form of the H-functions
\begin{align}
\overline H^{(-1)}_2(\zb)&=\overline{\mathcal{D}}\,\overline{H}^{(0)}_2(\zb)\doteq\frac{1}{(1-\zb)^2}c-
\frac{3}{1-\zb}c,\\\!\!
\overline H^{(-1,\log)}_2(\zb)&=\overline{\mathcal{D}}\,\overline{H}^{(0,\log)}_2(\zb)\doteq
\frac{2-2\gamma_E-\log(1-\zb)}{2(1-\zb)^2}c
+
\frac{18\gamma_E-19+9\log(1-\zb)}{6(1-\zb)}c,\!
\end{align}
and plugging in the solutions \eqref{eq:toy.zero.order.sols}, the term proportional to $(1-\zb)^{-2}$ vanishes, and the term proportional to $(1-\zb)$ provides
\begin{equation}
C_{(1)}=-\frac{\lambda_1}{3},\qquad C_{(1,\log)}=0.
\end{equation}
We can continue in this fashion, and determine the coefficients $C_{(m)}$ and $C_{(m,\log)}$ after acting $m$ times with the Casimir $\overline{\mathcal D}$. The results for $m=1,2,\ldots$ are
\begin{equation}
C_{(m)}=-2\lambda_1\left\{\frac{1}{6},\,-\frac{1}{30},\,\frac{4}{315},\,-\frac{1}{105},\,\ldots\right\},\qquad C_{(m,\log)}=0.
\end{equation}
We identify the $C_{(m)}$ together with $C_{(0,\log)}$ as coefficients in the large $\ell$ expansion \eqref{eq:S1expansion} of the harmonic sum  $S_1(\ell)$ expanded in inverse powers of $J^2=\ell(\ell+1)$. They therefore describe a function
\begin{equation}
\sum_{m,n} {C_{(m,\log^n)}}\frac{\log^n J}{J^{2m}}=\lambda_0-2\lambda_1S_1(\ell).
\end{equation}
This computation proves the following relation, which can also be shown by explicit computation,
\begin{equation}\label{eq:toy.example.starting.point}
\sum_\ell \langle a_{2,\ell}^{(0)}\rangle\, \zb\,k_{\ell+1}(\zb)( \lambda_0-2\lambda_1S_1(\ell)) \doteq \frac{\lambda_1\log(1-\zb)+\lambda_0}{1-\zb}c.
\end{equation}

In the following we will apply this method to more complicated functions, but the general idea will stay exactly the same.

\section{Four-point correlator from H-functions}
\label{sec:FindingNemo}

In this section we use the H-functions to construct the one-loop correction to the four-point function of four identical scalar operators. Again, we think of the correlator of four Konishi operators as our example, but the method applies to a large family of scalar operators. 

\subsection{The strategy}
We remind the reader that the four-point correlation function in weakly coupled gauge theories admits an expansion in the coupling constant $g$ of the form
\begin{equation}
\mathcal{G}(u,v)=\mathcal{G}^{(0)}(u,v)+g \,\mathcal{G}^{(1)}(u,v)+\ldots.
\end{equation}
The contributions to the one-loop correlator $\mathcal{G}^{(1)}(u,v)$ come from two different sources. First of all, there are infinite towers of operators for which the CFT-data can be expanded as a power series at large spin $\ell$, with possible $\log \ell$ insertions. Such towers of operators necessarily produce power divergent contributions to the correlator and we can study them using the H-functions. Secondly, there are terms in the four-point correlator which after performing the conformal block decomposition render CFT-data that is truncated in spin. Such terms are always regular as $v\to0$. Importantly, these two kinds of contributions are partially interchanged under crossing. 
In fact, the interchange is such that all contributions from infinite towers, at any twists, are completely determined by the twist-two operators. 
Therefore we will start our analysis from general ansatz for the twist-two operators, and then use the crossing symmetry and the H-function method to extend the ansatz to a full solution for the one-loop four-point correlator. 
In the process we will assume that there are no truncated solutions of the form found in \cite{Heemskerk2009}. 

Our strategy to find the one-loop correlation function is the following:
\begin{itemize}
 \item Using the explicit form of conformal blocks \eqref{eq:ConformalBlock} and the bootstrap equation \eqref{crossing.symmetry} we find a general form of the power divergent part of $\mathcal{G}^{(1)}(u,v)$ in the limit $v\to0$. We show using crossing symmetry that this is fully described by operators at leading twist, namely $\tau_0=2$. Subsequently, we use the H-function method to constrain the form of the contributions from infinite towers of leading twist operators. Supplementing this with terms truncated in spin we arrive at the most general leading twist contribution to the correlator $\mathcal G_{\mathrm{L.T.}}(u,v)\sim uf(\log u,v)$, where $f(\log u,v)$ is expressed to all orders in $v$ in terms of a finite number of unknowns.
\item Crossing symmetry maps $u f(\log u,v)$ to the power divergent part of the complete four-point correlator. This allows us to use the H-function method to find the large spin expansion of the CFT-data for all twists, which can be resummed to  closed-form functions of spin. Plugging this result back to the conformal block expansion we find the complete form of the four-point correlator in terms of a finite number of unknowns.
\item   As a final step we check that such obtained function satisfy all necessary constraints. In particular, consistency with the bootstrap equation reduces the number of unknowns to just four. 
\end{itemize}

\subsection{The ansatz}\label{sec:ansatz}
We focus first on the most general form of the power divergent terms in the limit $v\to 0$ and show that the bootstrap equation implies that all such contributions are encoded by the twist-two operators. 

Let us start by writing down an explicit form of the bootstrap equation in the perturbative expansion 
 \begin{equation}\label{eq:crossing.Konishi}
 v^{2+g \,\gamma_{\mathrm{ext}}}(\mathcal{G}^{(0)}(u,v)+g\, \mathcal{G}^{(1)}(u,v))=u^{2+g\, \gamma_{\mathrm{ext}}}(\mathcal{G}^{(0)}(v,u)+g\, \mathcal{G}^{(1)}(v,u)),
\end{equation}
where $\gamma_{\mathrm{ext}}$ is the one-loop anomalous dimension of the external operators, which we at the moment will keep unspecified.
The one-loop part of this equation can be written in the form
\begin{equation}\label{eq:crossing.Konishi.simpl}
\tilde{\mathcal G}^{(1)}(u,v)=\frac{u^2}{v^2} \tilde{\mathcal G}^{(1)}(v,u),
\end{equation}
where for convenience we defined $\tilde{\mathcal{G}}^{(1)}(u,v)=\mathcal{G}^{(1)}(u,v)+\gamma_\mathrm{ext}\log v\,\mathcal{G}^{(0)}(u,v)$.
Both functions $\mathcal{G}^{(0)}(u,v)$ and $\mathcal{G}^{(1)}(u,v)$ can be expanded in conformal blocks. Let us then look at the expansion of a single conformal block in the small $g$ limit,
\begin{equation}\label{eq:leadingconfblock}
G_{\tau,\ell}(u,v)= G_{\tau_0,\ell}(u,v)+g\,(\partial_{\tau} G_{\tau,\ell}(u,v))|_{\tau=\tau_0}+O(g^2).
\end{equation} 
From the explicit form of the conformal blocks we notice that at one loop there is a contribution proportional to $\log u$ in this expansion but no higher powers of the logarithm. We also notice that in the small $u$ limit we have $G_{\tau_0,\ell}(u,v)\sim u^{\tau_0/2}$. Thus the first non-trivial part of $\tilde{\mathcal{G}}^{(1)}(u,v)$ at small $u$ comes exclusively from the twist-two operators and is of the form
\begin{equation}\label{eq:leadingG}
\tilde{\mathcal{G}}^{(1)}(u,v)= \gamma_{\mathrm{ext}}\log v+ u\left(Q^{(1)}(v,\log v)\log u+Q^{(2)}(v,\log v)\right)+O(u^2),
\end{equation}
where the first trivial term comes from the identity operator contribution to $\mathcal{G}^{(0)}(u,v)$ and $Q^{(i)}(v,\log v)$ are arbitrary functions.
The bootstrap equation \eqref{eq:crossing.Konishi.simpl} used for \eqref{eq:leadingG} now gives
 \begin{equation}\label{eq:leadingG2}
 \tilde{\mathcal{G}}^{(1)}(u,v)=\gamma_{\mathrm{ext}}\frac{u^2}{v^2}\log u +\frac{u^2}{v}\left(Q^{(1)}(u,\log u)\log v+Q^{(2)}(u,\log u)\right)+O(v^0).
 \end{equation}
We notice in particular that, when crossed, \eqref{eq:leadingG} produces a power divergence for $v\to0$. It is easy to see that also the opposite statement is true: any divergent part of $\tilde{\mathcal{G}}^{(1)}(u,v)$ is mapped to the first two leading $u$ powers under crossing. Finally, by comparing the formulae \eqref{eq:leadingG} and \eqref{eq:leadingG2} we conclude that we must have $Q^{(i)}(u,\log u)\sim \frac{1}{u}+\ldots$.

Since the term proportional to $u^0$ is completely determined by the tree-level, we will focus here on the term proportional to $u$. Therefore, we start our analysis by considering the most general ansatz for twist-two operators. There are two distinguished terms: the contributions containing a power divergent part at $v\to0$, and contributions truncated in the spin. From the discussion above, we conclude that the former takes the form
\begin{equation}\label{Ginf.leading}
 \mathcal{G}_{\mathrm{inf,L.T.}}^{(1)}(u,v)\sim\frac{u}{v} \left(\alpha_{11}\log u \log v+\alpha_{10}\log u+\alpha_{01}\log v+\alpha_{00}\right)c+\ldots\,.
\end{equation}
where $\alpha_{00},\alpha_{10},\alpha_{01},\alpha_{11}$ are arbitrary constants and we introduced an explicit dependence on $c$ for later convenience. In the subsequent part of this section, we will use the H-function method to extend this to all subleading orders in $v$.
 
For the truncated contributions, let us take $L$ such that
\begin{equation}
\begin{cases}
\langle a^{(1)}_{2,\ell}\rangle=\langle a^{(1)}_{2,\ell}\rangle_{\mathrm{inf}}+\langle a^{(0)}_{2,\ell}\rangle\mu_{\ell}\,,\\
\overline{\langle \gamma_{2,\ell} \rangle}=\overline{\langle \gamma_{2,\ell} \rangle}_{\mathrm{inf}}+\nu_\ell\,,
\end{cases}
\qquad 
\ell=0,2,\ldots,L\,,
\end{equation}
and that for spins $\ell>L$ we have only contributions from infinite towers of operators.
 In this case the truncated part of the one-loop answer is given by
\begin{equation}\label{eq.Gfin.leading}
\mathcal{G}^{(1)}_{\mathrm{trunc,L.T.}}(u,v)=\sum_{\ell=0}^L\langle a^{(0)}_{2,\ell}\rangle\left( \mu_\ell \,G_{2,\ell}(u,v)+\nu_\ell \left(\partial_{\tau} G_{\tau,\ell}(u,v)\right)\big|_{\tau\to2}  \right).
\end{equation}

Let us go back to the term containing a divergence as $v\to 0$ in \eqref{Ginf.leading}. It originates purely from an infinite tower of twist-two operators and can be expanded using H-functions as in \eqref{eq:oneloop.in.Hfunctions}:
\begin{align}\nonumber
&\frac{z}{1-\zb}\left( \alpha_{11}\log z \log(1-\zb)+\alpha_{10}\log z +\alpha_{01}\log(1-\zb) +\alpha_{00}\right)c\doteq\\
&\hspace{7cm}\doteq z \sum_\rho\left(A_{2,\rho}+\frac{1}{2}B_{2,\rho}\log z \right) \overline H_2^{\rho}(\zb)\,,
\end{align}  
 where $A_{2,\rho}$ and $B_{2,\rho}$ are large-$J$ expansion coefficients, as in \eqref{eq:OPEexpansion} and \eqref{eq:gammaexpansion}, of the modified structure constants and anomalous dimensions, respectively, with $\rho=(m,\log^n)$ for $n=0,1$ and $m=0,1,\ldots$. Using the H-function method described in section \ref{sec:method} we find 
 \begin{align}
 A_{2,(0,\log)}&=-2\alpha_{01},&  A_{2,(0)}&=-2\alpha_{01}\gamma_E+\alpha_{00},& A_{2,(m)} &=-2\alpha_{01}\left\{  
 \frac{1}6,\, \frac{-1}{30},\, \frac{4}{315},\,\ldots
 \right\},
 \\ 
 B_{2,(0,\log)}&=-4\alpha_{11}, &  B_{2,(0)}&=-4\alpha_{11}\gamma_E+2\alpha_{10}, &B_{2,(m)}&=-4\alpha_{11}\left\{ \frac{1}6,\, \frac{-1}{30},\, \frac{4}{315},\,\ldots \right\}.
 \end{align}
From these values we can find an explicit form of the anomalous dimension and one-loop structure constants coming from an infinite tower of twist-two operators:
 \begin{align}\label{eq:twisttwogammaa}
\overline{ \langle \gamma_{2,\ell}\rangle}_{\mathrm{inf}} &=-4\alpha_{11}\,S_1(\ell)+2\alpha_{10}\,,\\\label{eq:twisttwoalphaa}
\overline{ \langle \hat{\alpha}_{2,\ell}\rangle}_{\mathrm{inf}}&=-2\alpha_{01}\,S_1(\ell)+\alpha_{00}\,.
 \end{align}

In the next step we will take the results \eqref{eq:twisttwogammaa}, \eqref{eq:twisttwoalphaa} and plug them into the conformal block expansion \eqref{eq:BlockDecomposition}. We can perform a resummation of the complete leading $z$ expansion of the four-point correlator $\mathcal{G}^{(1)}_{\mathrm{inf,L.T.}}(u,v)$ and arrive at
\begin{equation}\label{eq:Ginfleading}
\mathcal{G}^{(1)}_{\mathrm{inf,L.T.}}(u,v)=z\zb \left(\alpha_{11} F_{11}( z,\zb)+\alpha_{10} F_{10}(z,\zb)+\alpha_{01} F_{01}( z,\zb)+\alpha_{00} F_{00}(z,\zb)\right),
\end{equation}  
where
\begin{align}
F_{11}(z,\zb)&=c\,\frac{\zb}{1-\zb}\log(1-\zb)\log(z\zb)+2c\left(\frac{\zb}{1-\zb} \mathrm{Li}_2(\zb)-\frac{2-\zb}{1-\zb}\zeta_2 \right),
\\
F_{10}(z,\zb)&=c\left( \frac1{1-\zb}+1 \right)\log(z\zb)-c\log(1-\zb)\,,
\\
F_{01}(z,\zb)&=c\left( \frac1{1-\zb}-1 \right)\log(1-\zb)\,,
\\
F_{00}( z,\zb)&=c\left( \frac1{1-\zb}+1 \right).
\end{align}
It is easy to confirm that the power divergent part of \eqref{eq:Ginfleading} indeed equals \eqref{Ginf.leading}. We emphasise that the expansion \eqref{eq:Ginfleading} is valid only at the leading order in $z\to0$ but is exact to all orders in $\zb$.

We add together \eqref{eq.Gfin.leading} and \eqref{eq:Ginfleading} to get the most general form of the one-loop correlator at the leading order in $u\to0$ expansion
\begin{equation}\label{eq:G.leading.twist}
\mathcal G^{(1)}_{\mathrm{L.T.}}(u,v)=\mathcal G^{(1)}_{\mathrm{inf,L.T.}}(u,v)+\mathcal G^{(1)}_{\mathrm{trunc,L.T.}}(u,v).
\end{equation}
This answer depends on $2L+4$ unspecified coefficients and concludes the first step in our strategy.

\subsection{Higher twist operators}
In the next step we will use the complete form of the leading twist four-point function $G^{(1)}_{\mathrm{L.T.}}(u,v)$ together with the crossing equation to study implications for higher twist operators. As we already have pointed out, the term proportional to $ u$ are, apart from the trivial contribution from the identity operator, the only ones which can produce power divergent terms after the crossing. It implies that after we apply the crossing symmetry to the function \eqref{eq:G.leading.twist} we get the complete power divergence of the full one-loop answer. 

In order to make our results more transparent, let us assume at the moment that $L=0$, namely only spin $\ell=0$ contributes to the truncated ansatz \eqref{eq.Gfin.leading}. We will come back to the general case later. Let us look again at the crossing equation \eqref{eq:crossing.Konishi} at order $g$, which gives the following equation for the one-loop correlation function:
\begin{equation}\label{eq:crossing.for.HT}
\mathcal G^{(1)}(u,v)=\frac{u^2}{v^2} \mathcal G^{(1)}(v,u)+\gamma_{\mathrm{ext}}\mathcal G^{(0)}(u,v)\left(\log u-\log v\right).
\end{equation}
From our previous computations, on the right hand side we know explicitly all power divergent contributions in the limit $v\to0$. First of all, we can expand \eqref{eq:crossing.for.HT} at leading $v\to0$ and $u\to0$ to get
 \begin{align}\nonumber
 \mathcal{G}^{(1)}(z,\zb)\sim\frac{z }{(1-\zb)} \left(\alpha_{11}\log z \log (1-\zb)+(\alpha_{01}+\gamma_{\mathrm{ext}})\log z\right.\\ \left.+(\alpha_{10}-\gamma_{\mathrm{ext}})\log (1-\zb)+\alpha_{00}\right)c+\ldots.
 \end{align}
 Comparing it with \eqref{Ginf.leading} we find the constraint 
 \begin{equation}
 \alpha_{01}=\alpha_{10}-\gamma_{\mathrm{ext}}.
 \end{equation}
 After substituting this into \eqref{eq:crossing.for.HT} we notice that the divergent part of $\mathcal{G}^{(1)}(u,v)$ depends on the anomalous dimension of external operator $\gamma_{\mathrm{ext}}$ and the five parameters $(\alpha_{11},\alpha_{10},\alpha_{00},\mu_0,\nu_0)$. We use this function to find the unbounded CFT-data for higher twist operators by solving \eqref{eq:oneloop.in.Hfunctions}. Applying the method explained in section \ref{sec:method} we can compute as many coefficients $A_{\tau_0,(m,\log^k)}$ and $B_{\tau_0,(m,\log^k)}$ as necessary. Similar to the case of twist-two operators, we plug it back to \eqref{eq:OPEexpansion}, \eqref{eq:gammaexpansion} and we are able to perform the sum to find an explicit form of the CFT-data coming from infinite towers of operators as a function of spin. The result for the anomalous dimensions is 
\begin{align}\label{eq:gamma.HT}
\overline{\langle\gamma_{\tau_0,\ell}\rangle}&=\frac{c}{P_{\tau_0,\ell}}\Big(
4\alpha_{11}\eta\left[S_1(\tfrac{\tau_0}2-2)+S_1(\tfrac{\tau_0}2+\ell-1)+\tfrac12\delta_{\tau_0,4}\right]-4\eta\,\alpha_{10}+2\eta\,\gamma_{\mathrm{ext}}
\nonumber\\
&\qquad\qquad-4\mu_0-4\nu_0\left[S_1(\tfrac{\tau_0}2-2)-S_1(\tfrac{\tau_0}2+\ell-1)+1\right]
\Big)
+2\gamma_{\mathrm{ext}},\quad \mathrm{for}\,\,\tau_0>2,
\end{align}
where $\eta=(-1)^{\frac{\tau_0}2}$ and $P_{\tau_0,\ell}=c\, \eta+(\tau_0+\ell -2)(\ell+1)$ is the factor that appears in the tree-level structure constants~\eqref{eq:structureconstants}.
The result for $\overline{\langle\hat\alpha_{\tau_0,\ell}\rangle}
$ is more involved and we present here only its schematic form
\begin{equation}\label{eq:structure.constants.sum}
\overline{\langle\hat\alpha_{\tau_0,\ell}\rangle}=\alpha_{11}\overline{\langle\hat\alpha_{\tau_0,\ell}\rangle}_{11}+\alpha_{10}\overline{\langle\hat\alpha_{\tau_0,\ell}\rangle}_{10}+\alpha_{00}\overline{\langle\hat\alpha_{\tau_0,\ell}\rangle}_{00}+\gamma_{\mathrm{ext}}\overline{\langle\hat\alpha_{\tau_0,\ell}\rangle}_{\mathrm{ext}}+\mu_{0}\overline{\langle\hat\alpha_{\tau_0,\ell}\rangle}_{\mu_0}+\nu_{0}\overline{\langle\hat\alpha_{\tau_0,\ell}\rangle}_{\nu_0}.
\end{equation}
The explicit expressions for $\overline{\langle\hat\alpha_{\tau_0,\ell}\rangle}_{i}$ can be found in the appendix~\ref{app:results.nonsuper}. In order to get the one-loop structure constants $\langle a^{(1)}_{\tau_0,\ell}\rangle$ one again needs to use the formula \eqref{eq:modified.structure.constant}. We can observe that the explicit results for $\overline{\langle\gamma_{\tau_0,\ell}\rangle}$ and $\langle a^{(1)}_{\tau_0,\ell}\rangle$ at higher twist satisfy an interesting symmetry, where formally exchanging $\frac{\tau_0}2-1$ by $\frac{\tau_0}2+\ell$ gives the same expression up to a sign. We also note that the conformal blocks have the same symmetry, but it is not clear to us if this carries any meaning.

\subsection{Complete one-loop resummation}
In the previous section we found the CFT-data for all twists and spins. We can now insert it into the conformal block expansion \eqref{eq:BlockDecomposition} and reproduce the full one-loop correlation function.  After we do that we need to check the obtained function indeed satisfies the bootstrap equation \eqref{eq:crossing.Konishi}. We have performed this calculation explicitly and have found that the crossing relation for such obtained function implies one more constraint on the parameters of our ansatz, namely
\begin{equation}
\mu_0=-\nu_0.
\end{equation}
Implementing this constraint we end up with the function
\begin{align}\nonumber
&\mathcal G^{(1)}(u,v)\\&=\alpha_{11}\mathcal G_{11}(u,v)+\alpha_{10}\mathcal G_{10}(u,v)+(\alpha_{00}-2\,\zeta_2\,\alpha_{11})\mathcal G_{00}(u,v)+\nu_0\mathcal G_{\nu_0}(u,v)+\gamma_{\mathrm{ext}}\mathcal G_{\mathrm{ext}}(u,v),\label{eq:oneloop.linear.combination}
\end{align}
where the individual functions are given by
\begin{align}
\mathcal G_{11}(u,v)&=c\frac{u(1+u^2+v^2-2u-2v-2uv)}v\Phi(u,v),
\\
\mathcal G_{10}(u,v)&=c\frac{u\left((1+v-2u)\log u+(1+u-2v)\log v\right)}{v},
\\
\mathcal G_{00}(u,v)&=c\frac uv(1+u+v),
\\
\mathcal G_{\nu_0}(u,v)&=-c\frac{u(u+v+uv)}v\Phi(u,v),
\\
\mathcal G_{\mathrm{ext}}(u,v)&=\left({u^2}+\frac{u^2}{v^2}+c\frac{2u^2}v \right)\log u+\left(cu-\frac{u^2}{v^2}-c\frac uv-c\frac{u^2}v\right)\log v.
\end{align}
Here we introduced the usual box function \cite{Usyukina:1992jd}
\begin{equation} \label{eq:boxfunction}
\Phi(u,v)=\frac{\log \left(\frac{1-z}{1-\zb}\right) \log
   \left(z\, \zb\right)+2\left( \mathrm{Li}_2(z)- \mathrm{Li}_2(\zb)\right)}{z-\zb}.
\end{equation}
Notice that we may interpret the contribution $\mathcal{G}_{00}(u,v)$ in \eqref{eq:oneloop.linear.combination} as a one-loop renormalisation of the constant $c$. We also emphasise that the solution $\mathcal{G}_{\nu_0}(u,v)$, which produces truncated CFT-data for leading twist, does not belong to the family of truncated solutions found in \cite{Heemskerk2009} since it contributes to all spins for $\tau_0>2$.

Let us now come back to a general ansatz for the truncated solution with $L>0$. We can repeat all the calculations we performed in this section and we find that the solution is even more constrained than in the $L=0$ case. Working with the general ansatz we find that there is no new solution to the bootstrap equation for higher truncated spins. Namely, we find
\begin{equation}
\mu_\ell=0,\qquad \nu_\ell=0,\qquad \mathrm{for}\,\,\ell=2,4,\ldots,L.
\end{equation}
Notice that this agrees with the range $\ell>1$ of analyticity in spin from the inversion integral \cite{Caron-Huot2017}.
 
\subsection{Comparing with Konishi}
In the previous section we have found the most general one-loop four-point correlator of four identical scalars with classical dimension $\Delta_0=2$. In this section we will find the values for all the constants which selects the Konishi solution from the family \eqref{eq:oneloop.linear.combination}. The best case scenario would be to use the properties of conformal field theories to do that. One additional piece of information which we could use is the fact that the CFT-data for the stress tensor, which is present in the OPE of two Konishi operators, are known. It is, however, often difficult to access this information since the stress tensor is not the only operator with twist $\tau_0=2$ and spin $\ell=2$ present in the OPE of two Konishi operators. For that reason we are not able to fix the Konishi four-point correlator directly from conformal symmetry and  we will need to refer to some explicit results of direct perturbative calculations which can be found in the literature. 

In particular, we start by noticing that the Konishi operator is the only operator of twist $\tau_0=2$ and spin $\ell=0$ in the OPE of two Konishi operators. For that reason the average $\langle a^{(1)}_{2,0}\rangle=
a^{(1)}_{\mathcal{K}\mathcal{K}\mathcal{K}}:=
2\,c^{(0)}_{\mathcal{K}\mathcal{K}\mathcal{K}}c^{(1)}_{\mathcal{K}\mathcal{K}\mathcal{K}}$ is the one-loop structure constant of three Konishi operators and $\overline{\langle\gamma_{2,0}\rangle}=\gamma^{(1)}_{\mathcal{K}}$ is the one-loop anomalous dimension of Konishi operator. These can be extracted from the results in \cite{Bianchi:2001cm} and in the normalisation we use in this chapter they take the values
\begin{equation}\label{eq:Konishi.res}
\gamma_{\mathrm{ext}}=\overline{\langle\gamma_{2,0}\rangle}=3,\qquad \langle a^{(1)}_{2,0}\rangle=-18c.
\end{equation}
for $c=\frac2{3(N^2-1)}$.
Moreover, the averages of leading twist anomalous dimensions for all spins can be calculated using the results from \cite{Kotikov:2001sc}, given in \eqref{eq:averagesKonishi} in section~\ref{sec:NN4spectrum},
\begin{equation}\label{eq:gamma.an.res}
\overline{\langle \gamma_{2,\ell} \rangle}=2S_1(\ell),\qquad \mathrm{for} \quad\ell>0.
\end{equation}
In fact, the first two values of \eqref{eq:gamma.an.res}, together with \eqref{eq:Konishi.res}, are enough to fix all the constants and we get
\begin{equation}
\alpha_{11}=-\frac12,\quad\alpha_{10}=0,\quad \alpha_{00}=-6-\zeta_2,\quad \nu_0=3,\quad\gamma_{\mathrm{ext}}=3.
\end{equation}
Substituting this in \eqref{eq:oneloop.linear.combination} we find
\begin{align}\nonumber
\mathcal{G}^{(1)}_{\mathcal{K}\mathcal{K}\mathcal{K}\mathcal{K}}(u,v)&=-\frac{c\, u}{2v}\left(1+4u+4v+4uv+u^2+v^2\right)\Phi(u,v)-6\frac{c\, u}{v}\left(1+u+v\right)\\\label{eq:Konishi.final}
&\quad+\frac{3u}{v}\left(\frac{u}{v}+uv+2cu\right)\log u+\frac{3u}{v}\left(-\frac{u}{v}-c-cu+cv \right)\log v,
\end{align}
which exactly agrees with the result in \cite{Bianchi:2001cm}. We have therefore shown that the one-loop four-point correlation function of four Konishi operators belongs to our family of solutions, and we have found the explicit values of the constants describing this solution.

\section{The superconformal case}
\label{sec:super.case}
In this section we will focus on the four-point function of half-BPS operators. We  follow very closely the logic from the previous section and adapt it to the case of superconformal block expansion. Following the observations in section~\ref{sec:super.four.points}, the computations in this case are very similar and here we will only highlight the differences and the results. 

The most relevant difference compared to the Konishi case is that the conformal blocks  take a different form, we need to replace the ordinary blocks by superconformal blocks. From \eqref{eq:superconformal.decomposition} it boils down to the replacement
\begin{equation}
G_{\tau,\ell}(u,v)\to u^{-2}G_{\tau+4,\ell}(u,v).
\end{equation}
Importantly, the superconformal blocks are eigenvectors of the shifted quadratic Casimir operator of the superconformal group
\begin{equation}
\mathcal C_S( u^{-2}G_{\tau_0+4,\ell}(u,v))=\mathcal J_{\tau_0}^2u^{-2}G_{\tau_0+4,\ell}(u,v).
\end{equation}
Here we have defined
\begin{equation}\label{eq:supercasimir}
\mathcal C_S=u^{-2}\mathcal C u^2+\frac{\tau_0(\tau_0-6)}4-\frac{(\tau_0+4)(\tau_0-2)}4=u^{-2}\mathcal C u^2-2\tau_0+2,
\end{equation}
so that the eigenvalue is 
\begin{equation}
\mathcal J_{\tau_0}^2=J^2_{\tau_0+4}=\left(\frac{\tau_0}2+\ell+1\right)\left(\frac{\tau_0}2+\ell+2\right).
\end{equation}
Led by these observations we define H-functions in the supersymmetric case to be
\begin{equation}
\boldsymbol{H}_{\tau_0}^{(m,\log^n)}(u,v)=\sum_{\ell=0}\langle A^{(0)}_{\tau_0,\ell}\rangle\frac{(\log \mathcal J_{\tau_0})^n}{\mathcal J_{\tau_0}^{2m}}u^{-2}G_{\tau_0+4,\ell}(u,v),
\end{equation}
where $\langle A^{(0)}_{\tau_0,\ell}\rangle$ are the structure constants \eqref{eq:structure.constants.Free}.
The H-functions satisfy again a recursion relation
\begin{equation}\label{eq:recursion.super.H}
\boldsymbol{H}_{\tau_0}^{(m,\log^n)}(u,v)=\mathcal C_S\boldsymbol{H}_{\tau_0}^{(m+1,\log^n)}(u,v).
\end{equation}
Following similar arguments to the ones presented in section~\ref{sec:twist.conformal.blocks} one can prove that the power divergent part of H-functions factorises 
\begin{equation}
\boldsymbol{H}_{\tau_0}^{(m,\log^n)}(z,\zb)\doteq\frac{z^{-1}}{\zb-z}k_{\frac{\tau_0}{2}+1}(z)\overline{\boldsymbol{H}}_{\tau_0}^{(m,\log^n)}(\zb),
\end{equation}
where we have again defined H-function depending only on $\zb$ as
\begin{equation}
\overline{\boldsymbol{H}}_{\tau_0}^{(m,\log^n)}(\zb)=\zb^{-1}\sum_{\ell=0}\langle A^{(0)}_{\tau_0,\ell}\rangle\frac{\log^n \mathcal J_{\tau_0}}{\mathcal J_{\tau_0}^{2m}}\,
k_{\frac{\tau_0}{2}+\ell+2}(\zb).
\end{equation}
Also, the action of Casimir operator \eqref{eq:supercasimir} simplifies when acting on the power divergent part
\begin{equation}
\mathcal C_S\boldsymbol{H}_{\tau_0}^{(m,\log^n)}(z,\zb)\doteq\frac{z^{-1}}{\zb-z}k_{\frac{\tau_0}{2}+1}(z)\overline{\mathcal D}_S\overline{\boldsymbol{H}}_{\tau_0}^{(m,\log^n)}(\zb),
\end{equation}
where we defined
\begin{equation}
\overline{\mathcal D}_S=\zb^{-2}\,\overline{\mathcal D}\, \zb^2=-\zb+(2-3\zb)\zb\partial_\zb+(1-\zb)\zb^2\partial_\zb^2\,.
\end{equation}
Finally, the H-functions $\overline{\boldsymbol{H}}^{(m,\log^n)}(\zb)$ satisfy the following recursion relation
\begin{equation}\label{eq:recursion.super.Hbar}
\overline{\boldsymbol{H}}_{\tau_0}^{(m,\log^n)}(\zb)= \overline{\mathcal D}_S\overline{\boldsymbol{H}}_{\tau_0}^{(m+1,\log^n)}(\zb).
\end{equation}

In the following, we will compute the one-loop perturbative correction to the function $\mathcal H(u,v)$, in exactly the same way as we did in the ordinary, non-superconformal case. In particular, in analogy with \eqref{eq:oneloop.in.Hfunctions} its power divergent part can be expanded using H-functions as
\begin{align}
\mathcal H(z,\zb)\doteq \sum_{\tau_0}\frac{ z^{-1}}{\zb-z}\sum_\rho\left( A_{\tau_0,\rho}\,k_{\frac{\tau_0}{2}+1}(z)+B_{\tau_0,\rho}\,\left(\frac\partial{\partial\tau} k_{\frac{\tau}{2}+1}(z)\right)\big|_{\tau\to\tau_0}\right) \overline {\boldsymbol{H}}_{\tau_0}^{\,\rho}(\zb).\label{eq:oneloop.in.Hfunctions.super}
\end{align}
Here, $A_{\tau_0,\rho}$ and $B_{\tau_0,\rho}$ are large-$\mathcal{J}$ expansion coefficients of the modified structure constants and the anomalous dimensions, respectively. Again, in order to extract the CFT-data, we will need only an explicit form of the power divergent part of the H-functions for $\rho=(m,\log^n)$ with $m\leqslant0$ and $n=0,1$. All these functions can be easily obtained from $\overline{\boldsymbol{H}}_{2}^{(0)}(\zb)$ and $\overline{\boldsymbol{H}}_{2}^{(0,\log)}(\zb)$ using the recursion relation \eqref{eq:recursion.super.Hbar} and
\begin{equation}
\overline{\boldsymbol{H}}_{\tau_0}^{(m,\log^n)}(\zb)\doteq \frac{\Gamma(\frac{\tau_0}{2}+1)^2}{\Gamma(\tau_0+1)}\frac1{\tilde c}\left(\left(
\tilde c(-1)^{\frac{\tau_0}2}-\tfrac{\tau_0}{2}(\tfrac{\tau_0}{2}+1)\right)\overline{\boldsymbol{H}}_{2}^{(m,\log^n)}(\zb)+
\overline{\boldsymbol{H}}_{2}^{(m-1,\log^n)}(\zb)
\right).
\end{equation}

In the superconformal case, we have not been able to compute the exact form of the complete $\boldsymbol{H}_{2}^{(0)}(u,v)$, in contrast to the conformal case. Therefore, in principle, both $\overline{\boldsymbol{H}}_{2}^{(0)}(\zb)$ and $\overline{\boldsymbol{H}}_{2}^{(0,\log)}(\zb)$ could contain enhanced divergent terms proportional to $\log^2(1-\zb)$. It turns out that this is not the case\footnote{The fact that we can take $\overline{\boldsymbol{H}}_{2}^{(0)}(\zb)$ free from powers of logarithms can be seen by explicitly computing the power divergent terms of $\overline{\boldsymbol{H}}_{2}^{(m)}(\zb)$ for some $m<0$ using the kernel method, and see that they can be obtained by acting $m$ times with $\overline{\mathcal D}_S$ on $\overline{\boldsymbol{H}}_{2}^{(0)}(\zb)$.} and we end up with expressions analogous to the conformal case
\begin{align}
\overline{\boldsymbol{H}}_{2}^{(0)}(\zb)&=\frac{\tilde{c}}{1-\zb},\\\label{eq:H0log.super}
\overline{\boldsymbol{H}}_{2}^{(0,\log)}(\zb)&=-\frac{\log(1-\zb)}{2(1-\zb)}\tilde c-\frac{\gamma_E}{1-\zb}\tilde c+
\left(
-\frac{1}{12}-\frac{1-\zb}{15}+\ldots
\right)\tilde c\log^2(1-\zb).
\end{align}
More terms in the expansion of $\overline{\boldsymbol{H}}_{2}^{(0,\log)}(\zb)$ can be found in appendix~\ref{app:H0log}.

Equipped with the supersymmetric H-functions we are now ready to find the form of one-loop correction to the function $\mathcal{H}(u,v)$. Following a similar discussion as in section~\ref{sec:ansatz}, we start by observing that again all power divergent contributions to $\mathcal{H}(u,v)$ are completely captured by the twist-two operators. These terms come either from an infinite towers of twist-two operators or from solution truncated in spin. The general ansatz for leading-$u$ contribution of $\mathcal{H}(u,v)$ is therefore
\begin{align}\label{eq:Hansatz}
\mathcal{H}^{(1)}_{\mathrm{L.T.}}(u,v)&=\frac{u}{v} \left(\beta_{11}\log u \log v+\beta_{10}\log u+\beta_{01}\log v+\beta_{00}\right)\tilde c+\ldots\\&+\sum_{\ell=0}^L\langle A^{(0)}_{2,\ell}\rangle u^{-2}\left( \kappa_\ell \,G_{6,\ell}(u,v)+\lambda_\ell \left(\partial_{\tau} G_{\tau+4,\ell}(u,v)\right)\big|_{\tau\to2}  \right).
\end{align} 
for some $L$. The bootstrap equation \eqref{eq:superbootstrap} immediately implies that
\begin{equation}
\beta_{10}=\beta_{01}.
\end{equation}
Moreover, by direct application of the method described in the previous section, one can check that the truncated solutions cannot be completed to a crossing symmetric function. It implies that
\begin{equation}
\kappa_\ell=0,\qquad \lambda_\ell=0, \qquad\mathrm{for}\quad\ell=0,2,4,\ldots,L .
\end{equation} 
This stays in contrast to the ordinary conformal case where the spin-zero truncated solution was allowed.

We now use the H-function method explained in section \ref{sec:method} to complete the power divergent part of \eqref{eq:Hansatz} to a full leading-$u$ answer. In particular, the H-function method allows us to find the CFT-data for twist-two operators
\begin{align}\label{eq:BPS.gamma}
\overline{\langle \gamma_{2,\ell}\rangle}&=-4\beta_{11}S_{1}(\ell+2)+2\beta_{10},\\\label{eq:BPS.aa}
\overline{\langle \hat{\alpha}_{2,\ell}\rangle}&=-2\beta_{10}S_{1}(\ell+2)+\beta_{00}.
\end{align}
We could in principle continue as in the previous section and find a general solution as a function of three constants $(\beta_{11},\beta_{10},\beta_{00})$. Instead we will focus purely on the case of four half-BPS operators for which we can use additional information about the CFT-data found in the literature. In particular, it is known that the twist-two operators are not degenerate and the anomalous dimensions $\gamma_{2,\ell}$ have been found by direct calculations in e.g.\ \cite{Plefka:2012rd}
\begin{equation}\label{eq:BPS.an.res}
\gamma_{2,\ell}=2S_1(\ell+2),\qquad \ell=0,2,4,\ldots.
\end{equation} 
Additionally, the structure constants for two half-BPS operators and twist-two operators can also be found in \cite{Plefka:2012rd} and for $\ell=0$ it is
\begin{equation}\label{eq:BPS.str.cons.res}
a^{(1)}_{2,0}=-\tilde c.
\end{equation} 
Using the first two values in \eqref{eq:BPS.an.res} together with \eqref{eq:BPS.str.cons.res} we can fix our constants to\footnote{Notice that these values could also be found by considering \eqref{eq:BPS.gamma} and \eqref{eq:BPS.aa} for $\ell=-2$. This should correspond to a BPS current in the symmetric traceless representation of R-symmetry which implies $\gamma_{2,-2}=0$ and $a^{(1)}_{2,-2}=0$. It leads to $\beta_{10}=0$ and $\beta_{00}=2\, \zeta_2\,\beta_{11}$. The remaining constant can be reabsorbed into the definition of the coupling constant, leading to the result \eqref{eq:BPS.four.point}.}
 \begin{equation}\label{eq:BPS.four.point}
 \beta_{11}=-\frac{1}{2},\qquad \beta_{10}=0,\qquad\beta_{00}=-\zeta_2 .
 \end{equation}
Then the leading-$u$ result takes the form
\begin{equation}
\mathcal{H}_{\mathrm{L.T.}}(u,v)=-\tilde c\,\frac{ z \left(2\,\text{Li}_2(\zb)+\left(\log \left(z\right)+\log \left(\zb\right)\right)\log \left(1-\zb\right) \right)}{2 \left(1-\zb\right)}.
\end{equation}

Now we can use the bootstrap equation \eqref{eq:superbootstrap} to find the complete power divergent part of the function $\mathcal{H}(u,v)$. Subsequently, we use the H-function method to find the CFT-data for all twists which we collect in appendix \ref{app:superconformal.results}. Plugging it back to the superconformal block decomposition we can find the complete one-loop correlator which takes the form 
\begin{equation}\label{eq:BPS.final}
\mathcal{H}(u,v)=-\frac{\tilde c\,u}{2\,v}\Phi(u,v).
\end{equation}
This agrees with the known one-loop result for the four-point correlation function of four half-BPS operators in $\mathcal{N}=4$ SYM found in \cite{Arutyunov:2001mh}.

\section{Conclusions and outlook}
In this chapter we found a family of solutions to the conformal bootstrap equation relevant for the one-loop perturbation of four-dimensional conformal gauge theories. We employed twist conformal blocks which allow a systematic expansion around the double lightcone limit, namely $u=0$, $v=0$. Starting from the most general leading expansion \eqref{eq:G.leading.twist} we were able to complete it to a full crossing symmetric function of the cross-ratios. For four-point correlator of scalar operators with dimension $\Delta=2+g \,\gamma_{\mathrm{ext}}+O(g^2)$ we found a four-parameter family of solutions. By supplementing this by a few additional pieces of CFT-data for the leading-twist spectrum of the theory, we extracted the known form of one-loop correlator of four Konishi operators. Repeating this analysis for half-BPS operators $\mathcal{O}_{\boldsymbol{20'}}$ in $\mathcal{N}=4$ SYM and employing the superconformal block expansion we have also found an explicit form of the one-loop correlator of four such operators. 

There are many directions one could pursue using the method we described in this chapter. First of all, the four-point correlator of Konishi operators is only one representative of the family of solutions we found. A natural question is whether we can identify how other scalar correlators fit into our solution. Secondly, it should be possible to generalise our construction and apply it to correlation functions of operators with higher classical dimension. This would allow to find a large class of one-loop correlation functions in conformal gauge theories. Furthermore, there should be no conceptual obstruction to generalise it to mixed correlators.

The H-function technology can be in principle applied also to higher orders in the perturbation theory. Also in this case, the CFT-data can be expanded around the infinite spin and one can extract expansion coefficients for infinite towers of operators by focusing on the enhanced divergent part of the four-point function. In contrast with the one-loop case, where the complete enhanced divergent part was captured by power divergent terms, at higher orders it is possible to get other types of enhanced divergences. For example, at two loops there can be terms proportional to $\log^2 v$ which were prohibited by the conformal block expansion and bootstrap equation at one loop, see section \ref{sec:ansatz}. By examining an explicit form of conformal blocks and using the bootstrap equation it is easy to see that all such contributions come from $\langle (\gamma_{\tau_0,\ell}^{(1)})^2\rangle$. They are therefore determined by the one-loop CFT-data. Unfortunately, we are unable to access this information from our previous discussion since there is a degeneracy in the spectrum. It implies that, in general, $\langle (\gamma_{\tau_0,\ell}^{(1)})^2\rangle\neq \langle (\gamma_{\tau_0,\ell}^{(1)})\rangle^2$ and therefore we cannot use the one-loop averages we have calculated to determine the enhanced divergent part of the two-loop answer. In order to find it we would need to solve the mixing problem at one loop completely. This has been successfully done for the large-$N$ expansion of the correlators of four half-BPS operators in \cite{AldayBissi2017,Aprile:2017bgs,Aprile:2017xsp}. There, it has been possible to solve the mixing problem by using the knowledge of an infinite family of one-loop four-point correlators $\langle \mathcal{O}_p(x_1)\mathcal{O}_p(x_2)\mathcal{O}_q(x_3)\mathcal{O}_q(x_4)\rangle$, for $p,q\geqslant 2$, where $\mathcal{O}_p(x)$ is an $\mathcal{N}=4$ SYM half-BPS operator with R-symmetry labels $[0,p,0]$. Similar analysis should be possible also at weak coupling. In particular, it would allow us to find the two-loop correlation function of four Konishi operators, which is not known at the moment. We postpone it to  future work.

\chapter{More applications}\label{ch:more}

\section[The $\OO N$ model at large $N$]{The $\boldsymbol{\OO N}$ model at large $\boldsymbol N$}
\label{sec:paper4short}

In this section, which is a summary of \cite{Paper4}, we show how large spin perturbation theory can be used to study the critical $\OO N$ model in the large $N$ expansion. Contrary to chapters~\ref{ch:paper2} and \ref{ch:paper1}, which followed the presentation of the respective publications, we will here make direct reference to the formalism laid out in chapters~\ref{ch:analyticstudy} and \ref{ch:practical}.

In section~\ref{sec:ONspectrum} we gave an overview of the operator content of the critical $\OO N$ model. We will consider the four-point function of $\varphi^i$ transforming in the vector representation $V$, which means that operators in the OPE transform in irreps in the tensor product $V\otimes V=S\oplus T\oplus A$. Using tensor structures
\beq{
\mathsf T^{ijkl}_S=\delta^{ij}\delta^{kl}, \qquad \mathsf T^{ijkl}_T=\frac{\delta^{ik}\delta^{jl}+\delta^{il}\delta^{jk}}2-\frac1N \delta^{ij}\delta^{kl}, \qquad \mathsf T^{ijkl}_A=\frac{\delta^{ik}\delta^{jl}-\delta^{il}\delta^{jk}}2,
}
the crossing matrix takes the form
\beq{\label{eq:crossingmatrixON}
{\renewcommand*{\arraystretch}{2.15}
M^{\OO N}=\begin{pmatrix}
\dfrac1N &\dfrac{ (N+2)(N-1)}{2N^2}  & \dfrac{1-N}{2N}\\
1&\dfrac{N-2}{2N}&\dfrac12\\
-1&\dfrac{N+2}{2N}&\dfrac12
\end{pmatrix}
}.}
Since we work in an expansion in $1/N$, the crossing matrix itself will affect the order at which crossed-channel operators appear with a non-zero double-discontinuity. To make this more clear, we write
\beqa{
&\G_S(u,v)=\parr{\frac uv}^{\Delta_\varphi}\!\!\parr{
\frac{\G_S(v,u) }N\! +\! \frac{\G_T(v,u)-\G_A(v,u)}2\!+\! \frac{\G_T(v,u)+\G_A(v,u)}{2N} -\frac{\G_T(v,u)}{N^2}
}\!,
\\
&\G_{T/A}(u,v)=\parr{\frac uv}^{\Delta_\varphi}\!\!\parr{
\pm\G_S(v,u)+\frac{\G_T(v,u)+\G_A(v,u)}2\mp\frac1N\G_T(v,u)\label{eq:crossingTA}
}\!,
}
where the upper sign refers to $T$ and the lower to $A$.

The first operator to contribute is the identity, which gives GFF OPE coefficients \eqref{eq:aGFF}
\beq{
a_{T,n,\ell}=-a_{A,n,\ell}=Na_{S,n,\ell}=
a^{\mathrm{GFF}}_{n,\ell}|_{\Delta_\varphi}.
}
As discussed immediately after proposition~\ref{prop:doubletwistsexist} in section~\ref{sec:directchannelstructure}, from the explicit expansion of these OPE coefficients with $\Delta_\varphi=\mu-1+\gamma^{(1)}_\varphi/N+\ldots$, it follows that the OPE coefficients for $n\geqslant1$ are suppressed by an additional order $1/N$. To the order we consider, they will not contribute to the double-discontinuity and therefore decouple from the problem. In the following, we concentrate on the operators at $n=0$, which are the weakly broken currents $\mathcal J_{R,\ell}$.

The next crossed-channel operator to generate a double-discontinuity is the auxiliary field $\sigma$, where we assume that
\beq{
\Delta_\sigma=2+O(N^{-1}), \qquad a_\sigma=c_{\varphi\varphi\sigma}^2=\frac{a_\sigma^{(0)}}N+O(N^{-2}).
}
Because of the particular value of the scaling dimension, the contribution from $\sigma$ can be computed using inversion~\ref{inv:scalarN}. In the $T$ and $A$ representations, $\1$ and $\sigma$ are the only operators contributing to order $1/N$, and the CFT-data to this order can therefore be extracted. From \eqref{eq:crossingTA} we see that between $T$ and $A$, the function $\mathbf U_R(\log z,\hb)$ only differs by a sign, which has the important consequence that the leading anomalous dimensions agree,
\beq{
\gamma_{T,\ell}=\gamma_{A,\ell}=-\frac{2(\mu-2)^2}{J^2}\frac{a^{(0)}_{\sigma}}{N}+O(N^{-2}),
}
where $J^2=(\mu-1+\ell)(\mu-2+\ell)$ to this order. 

So far, we have introduced two free parameters in the problem: the OPE coefficient $a^{(0)}_\sigma$ and the anomalous dimension $\gamma_\varphi^{(1)}$. We can now use conservation of the global symmetry current, $\gamma_{A,1}=-2\gamma_\varphi$, to deduce an equation relating these parameters,
\beq{
\label{eq:asigma0res}
a^{(0)}_{\sigma}=\frac{\mu(\mu-1)}{(\mu-2)^2}\gamma^{(1)}_\varphi.
}

Next we consider corrections to the currents in the $S$ representation. Here the contribution from $\sigma$ appears at the same order as the first contribution from the weakly broken currents in the $T$ and $A$ representations. Recall that by proposition~\ref{prop:firstappearance}, the contribution from these currents is proportional to the square of their anomalous dimensions. This means that first we need to first compute an infinite sum of the form
\beqa{\label{eq:ITAsum}
I_{T-A}=\left(\frac{\zb}{1-\zb}\right)^{\Delta_\varphi}\frac12\Bigg(
&\sum_{\ell=0,2,\ldots}\!\!\!a_{T,0,\ell} \gamma_{T,\ell}^2 G^{(d)}_{2\mu-2+\ell,\ell}(1-\zb,1-z)\\&-\!\!\!
\sum_{\ell=1,3,\ldots}\!\!\!a_{A,0,\ell} \gamma_{A,\ell}^2 G^{(d)}_{2\mu-2+\ell,\ell}(1-\zb,1-z)
\Bigg)\frac{\log^2(1-\zb)}{8} + O(N^{-3}).\nonumber
}
This sum is computed by invoking the twist conformal blocks that we introduced in section~\ref{sec:familiesofoperators}. More precisely, $I_{T-A}$ corresponds to the double-lightcone limit of the difference of the level $2$ twist conformal blocks in the $T$ and $A$ representations. It can therefore be found by solving the differential equation
\beq{\label{eq:H2caseq}
{\cas}^2\mathcal H^{(2)}_{T-A}(u,v)=\mathcal H^{(0)}_{T-A}(u,v)=\parr{\frac uv}^{\mu-1},
}
where we used the explicit form of the $T-A$ tree-level correlator. Since the sum \eqref{eq:ITAsum} is evaluated at the unitarity bound, equation \eqref{eq:H2caseq} can be supplemented by the equation $\dsat \mathcal H^{(m)}=0$ for $\dsat$ in \eqref{eq:boundspinning}. The combined system of differential equations was solved in the appropriate limit in \cite{Paper4}. A similar computation can be done for $\mathcal H^{(2)}_{T+A}(u,v)$, and we reproduce the result in the following inversion.
\begin{invtool}\label{inv:fromcurrents} 
Consider the contribution from an infinite sum over $\ell$ of
broken currents $\mathcal J_\ell$ with anomalous dimensions
$\gamma_\ell:=\tau_\ell-2\Delta_\phi=\frac \kappa{J^2}$ and OPE coefficients $\alpha a^{\mathrm{GFF}}_{0,\ell}|_{\mu-1}$. In the $\phi$ four-point function, the leading contribution takes the form
\beq{
\mathbf U(\log z,\hb)=
-\frac{1}{2(\mu-2)^2J^2}\left[\pm1+(\mu-2)\pi\csc(\pi\mu)\right]\log
z+E_\pm,
}
where we use the $+$ ($-$) sign if the broken currents have even (odd) spin, and where $E_\pm$ are some rather lengthy expressions\footnote{More precisely,
\beqa{
E_\pm=\frac{\alpha \kappa^2}{4}\Bigg[&\frac{2\pi \csc(\pi\mu)\left((\mu-2)(\SS[\mu-1](\hb)+\pi\cot(\pi\mu))-1\right)\pm 2\left(\frac{3-\mu}{\mu-2}-S_1(\mu-2)\right)}{J^2(\mu-2)^2}\nonumber\\
&\pm\frac{\mu(\mu-1)B(\hb,\mu)+2(S_2(\mu-2)-\zeta_2)+\frac{2(\mu-3)}{(\mu-2)^2}}{(J^2-(\mu-1)(\mu-2))(\mu-2)}-\frac{2\pi\csc(\pi\mu)}{J^4(\mu-2)}
\Bigg],
\nonumber
}
where
\beq{\nonumber
B(\hb,\mu)=\frac{{\displaystyle {_4F_3}\!\left(\!\left. {1,1,2,\mu+1}~\atop~{\!\!3,3-\hb,\hb+2} \right|1\right)}}{J^2(J^2-2)}
-\frac{2 \pi  \Gamma (\hb) \Gamma (\mu +\hb-1)}{J^2 \Gamma (\mu +1) \sin (\pi  \hb)  \Gamma (2 \hb)} {_3F_2}\!\left(\!\left. {\hb-1,\hb,\hb+\mu -1}~\atop~{\!\!2 \hb,\hb+1} \right|1\right).}}.
\end{invtool}\vspace{2.5ex}

\noindent This inversion is used to find the anomalous dimensions in the $S$ representation,
\beq{
\label{eq:gammaSres}
\gamma_{S,\ell}^{(1)} = -\frac{2\gamma_{\varphi}^{(1)}}{J^2}\left(   (\mu -1) \mu +  \gamma_{\varphi}^{(1)} \frac{\pi  \csc (\pi  \mu ) \Gamma (\mu +1)^2 \Gamma (\ell+1)}{(\mu -2) \Gamma (\ell+2 \mu -3)}   \right),
}
where we used the relation \eqref{eq:asigma0res} to eliminate $a_\sigma^{(0)}$. Conservation of the stress tensor, $\gamma_{S,2}=-2\gamma_\varphi$ now gives a quadratic equation for $\gamma^{(1)}_\varphi$, which has the solutions
\beq{\label{eq:findingON}
\gamma^{(1)}_{\varphi,\mathrm{free} }=0,\qquad \gamma^{(1)}_{\varphi,\OO N }= \frac{(\mu-2)\Gamma(2\mu-1)}{\Gamma(\mu+1)\Gamma(\mu)^2\Gamma(1-\mu)}.
}
The second of these solutions exactly agrees with the result \eqref{eq:gammaphival} in the critical $\OO N$ model. Notice that this equation fixes all parameters that enter the problem at order $1/N$, and has followed simply from analytic bootstrap considerations\footnote{This was shown already in \cite{Alday2016b}, which focussed completely on computing anomalous dimensions.}.

Finally, the anomalous dimension of $\sigma$ can be extracted from our results by assuming the shadow relation
\beq{\label{eq:shadowrelation}
\Delta_\sigma=d-\Delta_{S,\ell=0}.
}
Pluggin in the anomalous dimensions \eqref{eq:gammaSres}, evaluated at spin zero, reproduces exactly the literature value \eqref{eq:Deltasigma} for $\Delta_\sigma$. Although \eqref{eq:shadowrelation} clearly holds at tree-level, it is not obvious why this should survive in perturbation theory. Perhaps this question can be addressed using the framework of \cite{Giombi:2018vtc}, which treated deformations like the Hubbard--Stratonovich transform in a unified way, reproducing for instance the relation \eqref{eq:asigma0res} for $a_\sigma$ from general considerations.

In \cite{Paper4} one further step was worked out, namely the contribution to the CFT-data in the $T$ and $A$ representations at order $1/N^2$. Now we have the same terms as described above, including subleading corrections to the contribution from $\sigma$, but also a new contribution from the operators $[\sigma,\sigma]_{n,\ell}$. This involves a rather complicated sum, and unfortunately the contributions to $U^{(0)}_{T/A,\hb}$ was only determined numerically, giving a numerical prediction for the current central charge at order $1/N^2$. All scaling dimensions computed in \cite{Paper4} agree perfectly with the literature \cite{Derkachov:1997ch,Manashov:2017xtt}, and many of the OPE coefficients were new results\footnote{More precisely, the only results for OPE coefficients in the literature prior to \cite{Paper4} were the order $1/N$ corrections to the central charges \cite{Petkou:1994ad} ($C_J$ in 3d already in \cite{Cha1991}), the order $1/N$ OPE coefficients in the $T$ and $A$ irreps \cite{Dey:2016mcs}, and the correction $a^{(1)}_\sigma$ to the $\sigma$ OPE coefficient \cite{Lang:1993ct}.}.

\section[\texorpdfstring{$\phi^4$}{phi**4} theories with any global symmetry]{$\boldsymbol{\phi^4}$ theories with any global symmetry}\label{sec:paper5short}

In \cite{Paper5} it was realised that the considerations for the critical $\OO N$ model, both in section~\ref{sec:ONwithinPaper2} for the $\epsilon$ expansion and in section~\ref{sec:paper4short} for the large $N$ expansion, will generalise to $\phi^4$ theories in any global symmetry group. There is a variety of interesting such models, some of which correspond to critical phenomena in three dimensions \cite{Pelissetto:2000ek,Osborn:2017ucf}, and relevant symmetry groups are for instance (hyper)cubic models \cite{Stergiou:2018gjj} and product groups such as $\OO m\times\OO n/\Z_2$ \cite{Paper5} and $\OO m^n\rtimes \mathrm S_n$ \footnote{$\mathrm S_n$ denotes the permutation group which acts by permuting the $\OO m$ factors.}\cite{Stergiou:2019dcv}. 

The existence of an expansion corresponding to the large $N$ limit of $\OO N$ depends on the group in question, and can be determined by studying the scaling dimensions of the bilinear scalars $\phi^2_R$ in the various representations. If a large $N$ expansion exists for a given fixed-point, the scaling dimensions of the bilinear scalars will approach either $\Delta=2+O(N^{-1})$ or $\Delta=2\Delta_\phi+O(N^{-1})$. The first case signals that the operator should be promoted to an auxiliary field $\mathcal R$. We have seen that this happens in the singlet representation in the critical $\OO N$ model, where $\mathcal S=\sigma$. The second case happens in the $T$ representation, where we have that $\Delta_{\varphi^2_T}=2-\epsilon+O(N^{-1})$ using \eqref{eq:ONgslitt}. If the bilinear scalars approach any other value than these two, the fixed-point does not admit a large $N$ expansion described by a Hubbard--Stratonovich field. It is therefore recommended to consider a specific model first in the $\epsilon$ expansion, even if one is primarily interested in the behaviour at large $N$ in e.g.\ three dimensions.

\subsection[General solution in the \texorpdfstring{$d=4-\veps$}{d=4-epsilon} expansion]{General solution in the $\boldsymbol{d=4-\veps}$ expansion}\label{sec:anyglobalsymmetryeps}
Consider first the contribution from the identity operator, appearing in
the singlet representation. This will give rise to the leading
contribution to $U^{(0)}_{R,\hb}$ in all representations. Since this is the
only operator contributing until order $\veps^2$, we get, using
inversion~\ref{inv:identity},
\begin{equation}
  \mathbf U_{R}(\log z,\hb)=M_{RS}  \AA[\Delta_\phi](\hb)+O(\veps^2),
\end{equation}
where $\Delta_\phi=1-\veps/2+\gamma_\phi$ and $\gamma_\phi=O(\veps^2)$.
From this expression, the leading order OPE coefficients can be extracted:
\begin{equation}
c^2_{\phi\phi\mathcal
J_{R,\ell}}=\frac{2  \Gamma(\ell+1)^2}{\Gamma(2\ell+1)}M_{RS}+O(\veps).
\end{equation}
Here $\ell$ takes even (odd) values for $R$ being an even (odd) representation. The scalar bilinears $\phi^2_R$ in the even representations have OPE coefficients
\begin{equation}
  c^2_{\phi\phi\phi^2_{R}}=2M_{RS}+O(\veps).
\end{equation}
These scalars are the next operators to contribute to the
double-discontinuity. Assume that they have dimension
$\Delta_{\phi^2_{R}}=2\Delta_\phi+g_{R}  \veps+O(\veps^2)$. Then,
using inversion~\ref{inv:scalarEps} we get the order $\veps^2$ corrections
\beq{
\mathbf U_{R}(\log z,\hb)=M_{RS}  \AA[\Delta_\phi](\hb)-M_{RS}  \Gamma^{\{2\}}_R\parr{\frac{1}{J^4}+\frac{\log z}{J^2}}\veps^2+O(\veps^3),
}
where
\begin{equation}
  \Gamma_R^{\{2\}}=\frac{1}{M_{RS}}\sum_{\rp \text{ even}}M_{R\rp }  
  g_{\rp }^2   M_{\rp S}.
\end{equation}
Using \eqref{eq:aellfromU} we can thus write down the leading correction to the anomalous dimension,
\begin{equation}
\Delta_{R,\ell}=2\Delta_\phi+\ell+\gamma_{R}(\hb), \quad
\gamma_{R}(\hb)=-\frac{\Gamma^{\{2\}}_R\veps^2}{J^2},
\end{equation}
where $J^2=\hb(\hb-1)$ and $\hb=\Delta_\phi+\ell$.

Next, as observed in chapter~\ref{ch:paper2} for the Wilson--Fisher fixed-point, we assume that it is possible to
analytically continue the result $\gamma_R(\hb)$ to spin zero, by making
the change of variables $\hb\to \hb_{\mathrm f}=\frac{\Delta+\ell}2$, i.e.\
by replacing the \emph{bare} with the \emph{full} conformal spin. For spin zero we should evaluate this at
$\hb_\mathrm{f}=\Delta_{\phi^2_{R}}/2=1-\veps/2+g_{R}  \veps/2+O(\veps^2)$. This leads to a system of equations
\begin{equation}\label{eq:matchingEps}
g_{R} \overset!=\left.\gamma_{R}(\hb)\right|_{\hb=\Delta_\phi+\frac{g_{R}}2
}, \qquad \text{$R$ even,}
\end{equation}
at order $\veps$, where now one power of $\veps$ in the $\gamma_{R}(\hb)$ cancels against the factor $\hb_{\mathrm f}-1=(g_{R}-1)\veps/2$ in the denominator. This simplifies to
\begin{equation}\label{eq:eqsystScalars}
M_{RS}\, g_{R}(g_{R}-1)+2\sum_{\rp \text{ even}}M_{R\rp }M_{\rp S}\,
g_{\rp }^2=0, \qquad \text{$R$ even,}
\end{equation}
which is a system of $k$ quadratic equations for the $k$ constants $g_R$, where $k$ is the number of even representations, or equivalently the number of scalar bilinears. Solving \eqref{eq:eqsystScalars} gives all possible fixed-points in the $\veps$ expansion with the given symmetry.

As an example, we go back to the $\OO N$ model with crossing matrix \eqref{eq:crossingmatrixON}. The even representations are $S$ and $T$, and the bilinear scalars are $\varphi^2_S=\varphi^i\varphi^i$ and $\varphi^2_T=\varphi^{\{i}\phi^{j\}}$. There are two solutions to \eqref{eq:eqsystScalars}, $g_S=g_T=0$, which is the theory of $N$ free fields, and
\begin{equation}\label{eq:ONscalardims}
  g_S=\frac{N+2}{N+8}, \qquad g_T=\frac{2}{N+8}.
\end{equation}
which is exactly the values for the critical $\OO N$ model \cite{Paper3} quoted in \eqref{eq:ONgslitt}.

The singlet spin-two current in any global symmetry group is the
stress tensor with dimension $\Delta_{S,2}=4-\veps$. This gives the constraint
\begin{equation}
  \gamma^{(2)}_\phi=\tfrac{1}{12}\Gamma^{\{2\}}_S,
\end{equation}
where $\Delta_\phi=1-\veps/2+\gamma^{(2)}_\phi\veps^2+O(\veps^3)$. Using this we write down the full dimension of the broken currents to order $\veps^2$
\begin{equation}
  \Delta_{R,\ell}=2-\veps+\ell+2\gamma_\phi^{(2)}\veps^2-\frac{\Gamma^{\{2\}}_R\veps^2}{\ell(\ell+1)},
\end{equation}
determined completely by the solutions to \eqref{eq:eqsystScalars}.
The OPE coefficients are extracted using \eqref{eq:aellfromU} and \eqref{eq:aArel},
\begin{equation}
a_{R,\ell}=M_{RS}a^{\mathrm{GFF}}_{0,\ell}|_{\Delta_\phi}+M_{RS}\frac{2\Gamma(\ell+1)^2}{\Gamma(2\ell+1)}\frac{\Gamma_R^{\{2\}}\veps^2}{\ell(\ell+1)}\left(
S_1(2\ell)-S_1(\ell)+\frac{1}{\ell+1}
\right)+O(\veps^3),
\end{equation}
where we evaluate the GFF OPE coefficients \eqref{eq:aGFF} at $\Delta_\phi=1-\veps/2+\gamma^{(2)}_\phi\veps^2$ and expand to order $\veps^2$.

From the $\ell=2$ OPE coefficient in the singlet representation we extract the central charge correction using \eqref{eq:centralchargesinglet},
\begin{equation}
\frac{C_T}{C_{T,\mathrm{free}}}=1-\frac{5\gamma_\phi^{(2)}}{3}\veps^2+O(\veps^3)=1-\frac{5  \Gamma_S^{\{2\}}}{36}\veps^2+O(\veps^3),
\end{equation}
which is consistent with~(E.1) of \cite{Osborn:2017ucf}.  We emphasise
that the considerations here are valid with any global symmetry group. The
input needed to specialise to a given symmetry group is the crossing matrix
$M_{R\rp }$ and the division of the representations into even and odd.
By solving the system of equations \eqref{eq:eqsystScalars}, one finds all
fixed-points in the $\veps$ expansion compatible with that symmetry group
and derives the leading (order $\veps$) anomalous dimensions of the
bilinear scalars.  Conservation of the stress tensor allows one to
compute the leading (order $\veps^2$) anomalous dimension of $\phi$.

In one or several of the odd representations $R$, the current at spin $\ell=1$ may be conserved, being the Noether current of a continuous global symmetry. This gives further constraints $\Delta_{R,\ell}=d-1$, which must be explicitly checked. The corresponding OPE coefficient is related to the $C_J$ of that symmetry current:
\begin{equation}
  \frac{C_{J_R}}{C_{J_R,\mathrm{free}}}=1-3\gamma_\phi^{(2)}\veps^2+O(\veps^3)=1-\frac{3  \Gamma_R^{\{2\}}}4\veps^2+O(\veps^3).
\end{equation}

Let us discuss the extension to higher orders in the $\veps$ expansion. To
order $\veps^3$, the operators contributing with a nonzero
double-discontinuity are the same as at the previous order, namely the
bilinear scalars $\phi^2_R$. At higher orders, infinite families of
operators contribute. In the $\OO N$ model, the only such families at order $\veps^4$ are operators of approximate twist $2$ and $4$, and we expect that this generalises to any global symmetry. However, to compute the contribution from approximate twist $4$ requires detailed knowledge of the operator content of the theory in question. This was worked out in the $\OO N$ model in \cite{Paper3}.

All constants that enter the problem at order $\veps^2$ can be fixed using continuation to spin zero and conservation of the stress tensor. This is no longer true at higher orders. At order
$\veps^3$ a total of $2k+1$ new constants appear: $\gamma_\phi^{(3)}$, the
second order correction to $\gamma_{\phi^2_R}=g_R  \veps(1+g^{(2)}_R\veps)+\ldots$, and the corrections $\alpha_R$ to the OPE coefficients defined by
\begin{equation}
  c^2_{\phi\phi\phi^2_R}=2M_{RS}(1+\alpha_R  \veps)+O(\veps^2).
\end{equation}
Based on experience from the $\OO N$ model, the
order $\veps^2$ continuation to spin zero requires order $\veps^4$ results
for the currents, so the only new equations at order $\veps^3$ are the
conservation of the symmetry currents (including the stress tensor). In general, this will not provide enough equations to fix all constants, but in many cases we can still make progress. Firstly, from the crossing analysis of \eqref{eq:alpha1correctionp2}, it follows that we must have $\alpha_R=-g_R$. Secondly, the second order corrections $g_R^{(2)}$ to the bilinear scalar dimensions, as well as $\gamma_\phi^{(3)}$, are in many cases known in the literature, and can be taken as input.

Using inversion~\ref{inv:scalarEps}, it is then straightforward to derive expressions for $U^{(p)}_{R,\hb}$ at order $\veps^3$. The anomalous dimensions extracted from these expressions take the form
\begin{equation}
  \gamma_R(\hb)=-\frac{\Gamma_R^{\{2\}}}{J^2}\veps^2+\frac{\Gamma_R^{\{2\}}-2  \Gamma_R^{\{2,1\}}+\big(\Gamma_R^{\{3\}}-\Gamma_R^{\{2\}}\big)S_1(\hb-1)}{J^2}\veps^3+O(\veps^4),
\end{equation}
where $\hb=1-\frac\epsilon2+\ell+O(\epsilon^2)$ and 
\begin{equation}
\Gamma_R^{\{3\}}=\frac{1}{M_{RS}}\sum_{\rp \text{ even}}M_{R\rp }   g_{\rp }^3M_{\rp S}, \qquad
\Gamma_R^{\{2,1\}}=\frac{1}{M_{RS}}\sum_{\rp \text{ even}}M_{R\rp }   g_{\rp }^2g^{(2)}_{\rp }M_{\rp S}.
\end{equation}

From the corresponding expression for the OPE coefficients using inversion~\ref{inv:scalarEps}, we can extract the central charge correction,
\begin{equation}
  \frac{C_T}{C_{T,\mathrm{free}}}=1-\frac53\left(\gamma_\phi^{(2)}\veps^2+\gamma_\phi^{(3)}\veps^3\right)-\frac{29}{18}\gamma_\phi^{(2)}\veps^3+\frac5{48}\Gamma^{\{3\}}_S\veps^3+O(\veps^4).
\end{equation}
Here we used that the stress tensor conservation eliminates the
dependence on $g_R^{(2)}$ in favour of $\gamma_\phi^{(3)}$. 
Similarly, for the current central charges we derive the expression
\begin{equation}
  \frac{C_{J_R}}{C_{J_R,\mathrm{free}}}=1-3\left(\gamma_\phi^{(2)}\veps^2+\gamma_\phi^{(3)}\veps^3\right)-\frac94\gamma_\phi^{(2)}\veps^3+\frac{1}{4}\Gamma^{\{3\}}_R\veps^3+O(\veps^4).
\end{equation}

\subsection[General solution in the large \texorpdfstring{$N$}{N} expansion]{General solution in the large $\boldsymbol N$ expansion}
We will now derive the form of the CFT-data in the large $N$ expansion for a
generic symmetry group, parametrised by some number $N$. Compared to the
$\veps$ expansion, the situation is a bit more complicated, since the
parameter $N$ enters in the crossing matrix $M_{R\rp }$ itself. In a given
even representation $R$, there are two options for the smallest dimension
scalar. As discussed above, it is either a scalar bilinear $\phi^2_R$ with dimension
$2\Delta_\phi+O(N^{-1})$, or a Hubbard--Stratonovich field $\mathcal
R$ with dimension $2+O(N^{-1})$. We assume that the Hubbard--Stratonovich fields $\mathcal R$ have OPE coefficients $c^2_{\phi\phi\mathcal R}={a_{\mathcal R}}/N+O(N^{-2})$, which will generate corrections to the free theory.

In order to provide some structure of the subsequent computations we define the following subsets of the representations in $V\otimes V=\mathrm I\cup\mathrm{II}$:

\begin{itemize}
\item Group $\mathrm I$: Representations whose only corrections at order $1/N$ come from crossed-channel Hubbard--Stratonovich fields.
\item Group $\mathrm{II}$: Representations where the corrections at order $1/N$ come from Hubbard--Stratonovich fields as well as from broken currents in group $\mathrm I$ representations in the crossed channel.
\item Group $\mathrm{III}$: Representations that admit a Hubbard--Stratonovich field. Typically $\mathrm{III}\subset\mathrm{II}$.
\end{itemize}
For instance, in the $\OO N$ model we have $S\in\mathrm{II}\cap\mathrm{III}$ and $T,A\in\mathrm I$. Our strategy will then be the following. First, as in the $\veps$ expansion, the identity operator creates the leading contribution to $U^{(0)}_{R,\hb}$ for all representations. Next, we turn to the representations in Group $\mathrm I$. The contributions from Hubbard--Stratonovich fields will give the order $1/N$ anomalous dimensions in these representations. Using inversion~\ref{inv:scalarN} we see that these corrections will be proportional to $1/J^2$. Finally, we consider the representations in the Group $\mathrm{II}$. Here we get contributions from both the Hubbard--Stratonovich fields, using inversion~\ref{inv:scalarN}, and from the currents in Group $\mathrm{I}$. Due to the particular form of the anomalous dimensions of these currents, we can use inversion~\ref{inv:fromcurrents} to find the complete order $1/N$ CFT-data.

The expressions will depend on $|\mathrm{III}|+1$ free parameters: the OPE coefficients $a_{\mathcal R}=c^2_{\phi\phi\mathcal R}$ for $R\in \mathrm{III}$, as well as the leading order anomalous dimension of $\phi$.
The consistency conditions available to fix these constants are the conservation of the symmetry currents (including the stress tensor), and depending on the number of conserved currents this may or may not be enough. As in the order $\veps^3$ results above, literature values can be used to fix the remaining constants if the conservation equations are not
sufficient. Finally, the leading anomalous dimensions of the
Hubbard--Stratonovich fields may be extracted by imposing the shadow relation $\Delta_{\mathcal R}+\Delta_{R,0}\overset !=d$
in similarity with \eqref{eq:shadowrelation} in the $\OO N$ model.

Let us now execute the strategy in full generality.
The contribution from the identity operator gives
\begin{equation}
\mathbf U_R(\log z,\hb)=M_{RS}  \AA[\Delta_\phi](\hb),
\end{equation}
where $\Delta_\phi=\mu-1+\gamma^{(1)}_\phi/N+O(N^{-2})$ and $\mu=d/2$. For the representations in group $\mathrm I$ we have contributions from Hubbard--Stratonovich fields in group $\mathrm{III}$. Using inversion~\ref{inv:scalarN} we get
\begin{equation}
U^{(1)}_{R,\hb}=-\sum_{ \rp \in\mathrm{III}}M_{R \rp
}  \, 2(\mu-2)^2\,\frac{a_{\oprp}}N\frac{\AA[\mu-1](\hb)}{J^2}, \quad R\in \mathrm
I,
\end{equation}
and a corresponding expression for $U^{(0)}_{R,\hb}$. From this we can extract the order $N^{-1}$ anomalous dimensions of currents in group $\mathrm I$ representations:
\begin{equation}
  \gamma_{R,\hb}=-\frac{2(\mu-2)^2K_R}{J^2N}+O(N^{-2}), \quad
  K_R=\frac{1}{M_{RS}}\sum_{\rp \in\mathrm{III}}M_{R\rp }\, a_{\oprp},
  \quad R\in \mathrm I.
\end{equation}
In step 3 we consider the second group of operators, $\mathrm{II}$. They get contributions both from $\mathcal R$ for $R\in \mathrm{III}$ and from $\mathcal J_{R,\ell}$ for $R\in\mathrm I$. We get
\begin{align}\nonumber
U^{(1)}_{R,\hb}&=-\sum_{ \rp \in\mathrm{III}}2M_{R \rp }(\mu-2)^2\frac{a_{\oprp}}N\frac{\AA[\mu-1](\hb)}{J^2}
\\
&\quad-\sum_{\rp _\pm\in\mathrm{I}}4M_{R\rp }K_{\rp }^2M_{\rp
S}\frac{(\mu-2)^2(2\hb-1)}{J^2N^2}\left(\pm1+(\mu-2)\pi\csc(\pi\mu)\right),\quad
R\in \mathrm{II},
\end{align}
where the $+$ ($-$) sign is used if the operators in the $\rp $ representations have even (odd) spin. This means that the dimensions of the group $\mathrm{II}$ double-twist operators are
\begin{equation}
\Delta_{R,\ell}=2(\mu-1)+\ell+\frac{2\gamma_\phi^{(1)}}N-\frac{2(\mu-2)^2K_R}{J^2N}-\frac{\widehat
K_{R}}{J^2N^2}\frac{(\mu-2)^2\Gamma(\mu-1)^2\Gamma(\ell+1)}{\Gamma(2\mu+\ell-3)},\quad
R\in\mathrm{II},
\end{equation}
where in the above expressions we have $J^2=(\mu-1+\ell)(\mu-2+\ell)$ and
\begin{equation}
\widehat K_R=\frac{1}{M_{RS}}\sum_{\rp _\pm\in\mathrm I}2M_{R\rp }K_{\rp
}^2M_{\rp S}\left(\pm1+(\mu-2)\pi\csc(\pi\mu)\right), \quad R\in
\mathrm{II}.
\end{equation}
Conservation equations for the stress-tensor and for global symmetry currents may now be used to fix the free parameters $\gamma_\phi^{(1)}$ and $a_{\mathcal R}$ for $R\in \mathrm{III}$.

\section[Multicritical theories in \texorpdfstring{$d_{\mathrm c}(\theta)-\epsilon$}{d\_c(theta)-epsilon} dimensions]{Multicritical theories in $\boldsymbol{d_{\mathrm c}(\theta)-\epsilon}$ dimensions}
\label{sec:multicritical}

In this section we apply large spin perturbation theory to the multicritical theories in an $\epsilon$ expansion near their critical dimensions. As far as we are aware, this has not previously been done in the literature. Let us recall from figure~\ref{fig:operahouse} in the introductory chapter, that the multicritical theories, labelled by integers $\theta$, correspond to scalar theories with $\lambda\phi^{2\theta}$ interactions and are defined below the critical dimensions $2\mu_\theta:=d_{\mathrm c}(\theta)=\frac{2\theta}{\theta-1}$ where these interactions become marginal. 
The interacting fixed-points in $2\mu_\theta-\epsilon$ dimensions have $\theta-1$ relevant deformations and are individually referred to as tricritical ($\theta=3$), tetracritical ($\theta=4$) etc. In an $\epsilon$ expansion near the critical dimension, $\lambda$ takes a value of order $\epsilon$ at the fixed-points, which are reached by short RG flows.

Despite not being defined in any integer dimension $\geqslant3$, these theories have received some attention from the conformal bootstrap. In particular, the methods of \cite{Rychkov:2015naa} using multiplet recombination have been generalised to the tricritical \cite{Basu:2015gpa} and general multicritical case \cite{Codello:2017qek}\footnote{I thank G. P. Vacca for useful discussions of the literature on multicritical theories.}. Using traditional diagrammatic techniques, some scaling dimensions have been computed. Here we only need the anomalous dimension for $\phi$, which is given by \cite{Lewis:1978zz}
 \beq{
\Delta_\varphi=\mu-1+\epsilon^2\gamma_\phi^{(2)}+O(\epsilon^3), \qquad \gamma_\phi^{(2)}=\frac{2(\theta-1)^2\Gamma(\theta+1)^6}{\Gamma(2\theta+1)^3}.
}

We now follow the procedure of large spin perturbation theory for these theories. We are interested in computing CFT-data of weakly broken currents $\mathcal J_\ell=\phi\de^\ell\phi$ in the direct channel, given by inverting contributions from crossed-channel operators. Clearly the identity operator $\1$ appears, and from proposition~\ref{prop:firstappearance} we know that the contribution from the currents themselves is suppressed. To determine the next operator to contribute, we use the heuristic diagrammatic method of section~\ref{sec:heuristicmethod}. From studying some possible diagrams we realise that the next operator after $\1$ to consider is $\phi^{2\theta-2}$, where the corresponding diagram is displayed in figure~\ref{fig:multidiagram}. 
This operator has dimension $\Delta_{\phi^{2\theta-2}}=(2\theta-2)\Delta_\phi+O(\epsilon)=2+O(\epsilon)$, and from the diagram it follows that its OPE coefficient is of order $\epsilon^2$. The contributions from $\1$ and $\phi^{2\theta-2}$ can this be readily computed using inversions~\ref{inv:identity} and \ref{inv:scalarN}, and take the form
\beqa{
U^{(1)}_\hb&=-2(\mu_\theta-2)^2\frac{\AA[\mu_\theta-1](\hb)}{J^2}a^{(0)}_{\phi^{2\theta-2}}\epsilon^2+O(\epsilon^3),
\\
U^{(0)}_\hb&=\AA[\Delta_\phi](\hb)+(\mu_\theta-2)^2\frac{\AA[\mu_\theta-1](\hb)}{J^2}a^{(0)}_{\phi^{2\theta-2}}\epsilon^2\parr{\SS[\mu_\theta-1](\hb)-\frac1{J^2}}+O(\epsilon^3).
}
From these functions, the anomalous dimensions, and therefore the scaling dimensions, of the weakly broken currents are extracted:
\beq{\label{eq:dimensionsmultibefore}
\Delta_{\ell}=2(\mu-1)+\ell+2\epsilon^2\gamma_\phi^{(2)}-\frac{2(\mu_\theta-2)^2a^{(0)}_{\phi^{2\theta-2}}\epsilon^2}{(\mu_\theta-1+\ell)(\mu_\theta-2+\ell)}+O(\epsilon^3),
}
where $\mu=\mu_\theta-\frac\epsilon2$. 
\begin{figure}
  \centering
\includegraphics{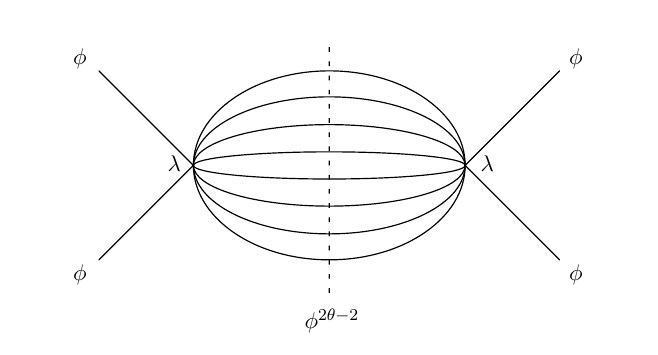}
\caption[Feynman diagram for operators in the crossed channel in multicritical theories.]{Feynman diagram for operators in the crossed channel, showing that the operator $\phi^{2\theta-2}$ contributes at order $\lambda^2$, or equivalently at order $\epsilon^2$.}
\label{fig:multidiagram}
\end{figure}
Solving $\Delta_2=2\mu$ for the stress tensor at spin $2$ gives the relation
\beq{
a^{(0)}_{\phi^{2\theta-2}}=\frac{\theta(2\theta-1)}{(\theta-2)^2}\gamma_\phi^{(2)}.
}
Plugging this into \eqref{eq:dimensionsmultibefore}, as well as the corresponding expression for the OPE coefficients, gives a complete determination of the CFT-data at order $\epsilon^2$ in terms of $\gamma_\phi^{(2)}$. The scaling dimensions \eqref{eq:dimensionsmultibefore} precisely reproduce the result derived in \cite{Gliozzi:2017gzh}, whereas the corresponding OPE coefficients are new. 

Assuming that the CFT-data can be extended to spin zero generates the dimension of $\phi^2$, which reproduces the relation derived in \cite{Codello:2017epp}
\beq{\label{eq:multirelationphiphi2}
\gamma^{(2)}_{\phi^2}=\frac{2\theta(2\theta-1)}{\theta-2}\gamma_\phi^{(2)} ,
}
where $\Delta_{\phi^2}=2\Delta_\phi+\gamma_{\phi^2}$. Finally, from the OPE coefficient of the stress tensor we can extract the correction to the central charge,
\beq{\label{eq:CTmulti}
\frac{C_T}{C_{T,\mathrm{free}}}=1-\frac{2(\theta-1)(3\theta-1)}{\theta(2\theta-1)}\epsilon^2\gamma_\phi^{(2)}+O(\epsilon^3).
}
We believe that this is a new result.

\subsection{Generalisations}
The generalisation to multicritical theories with global $\OO N$ symmetry follows in a straightforward way. We use the crossing matrix \eqref{eq:crossingmatrixON}, and scalar operators in both $S$ and $T$ representations contribute. Conservation of the global symmetry current and the stress tensor gives two constraints solved by
\beq{
a_S^{(0)}=\frac{\theta(2N\theta-2N+1)\gamma_\varphi^{(2)}}{(\theta-2)^2N}, \quad a^{(0)}_T=\frac{2\theta(2N\theta-N+1)\gamma_\varphi^{(2)}}{(\theta-2)^2(N+2)},
}
which again fix all CFT-data in terms of the anomalous dimension of $\varphi$. For this we use the literature value \cite{Hofmann:1991ge}
\beq{
\gamma_\varphi^{(2)}=\frac{(\theta-1)^2\Gamma(\theta)^2\Gamma(\theta+1)\Gamma(\theta+\frac N2)}{4\Gamma(2\theta)^2\Gamma(1+\frac N2)}{_3F_2}
\left( {1-\frac N2-\theta,\frac12-\frac\theta2,-\frac\theta2}~\atop~{ 1,\frac12-\theta}\middle|1\right)^{-2}\!,
}
where for any integer $\theta$ the hypergeometric function truncates and gives a polynomial in $N$. The relations~\eqref{eq:multirelationphiphi2} and \eqref{eq:CTmulti} are unchanged, and the current central charge is given by
\beq{
\frac{C_J}{C_{J,\mathrm{free}}}=1-\frac{2(\theta^2-1)}{\theta}\epsilon^2\gamma_\varphi^{(2)}+O(\epsilon^3).
}
Contrary to the critical ($\theta=2$) $\OO N$ model, it is believed that a large $N$ expansion of the multicritical theories does not exist for odd $\theta$. In the tricritical case, there is evidence for a curve with approximate equation $N_c=\frac{3.6}{3-d}$, at which the fixed-point vanishes by annihilating a non-perturbative fixed point. In \cite{Yabunaka:2017uox} it was argued that this curve ends near the point $(N,d)=(19,2.8)$, and that passing below this point and then towards large $N$ for $d<2.8$ corresponds to reaching another non-perturbative fixed-point. It would be interesting to investigate this using numerical bootstrap.

We could also extend the analysis to interactions with an odd number of fields, by looking at theories with interaction $\lambda\phi^{2t+1}$. These fixed-points are non-unitary, and are believed to be related to non-unitary minimal models in two dimensions, although the mapping is not completely resolved \cite{Codello:2017epp}. Now $\lambda\sim i\sqrt{\epsilon}$ and we have a similar diagram to figure~\ref{fig:multidiagram}. The results for these theories can be recovered from the considerations above by the substitutions $\theta\rightsquigarrow t+\frac12$ and $\gamma_{\phi}^{(2)}\epsilon^2\rightsquigarrow\gamma_{\phi}^{(1)}\epsilon$, where $\gamma_{\phi}^{(1)}$ is negative. No closed form expression exists for $\gamma_{\phi}^{(1)}$, but it is given as a sum representation in \cite{Codello:2017epp} which can be evaluated case by case in $t$. Relations like \eqref{eq:multirelationphiphi2} still hold in these cases. 

Finally, we would like to mention a couple of interesting aspects of multicritical theories with global symmetries. As discussed in section~\ref{sec:ONspectrum}, the critical $\OO N$ model is related to a cubic theory in $6-\epsilon$ dimensions, and a generalisation of this relation has been suggested involving multicritical theories for even $\theta$ and theories with odd power interactions \cite{Gracey:2017okb}. Theories with $\phi^5$ interactions ($t=2$) are also interesting since their critical dimension $\frac{10}3$ is above three dimensions, which suggests that physical theories in 3d may be reached starting from theories with $\phi^5$ interactions in $\frac{10}3-\epsilon$ dimensions \cite{Zinati:2019gct}.

\section[Four-point function of \texorpdfstring{$\varphi^2$ in the $\epsilon$}{phi**2 in the epsilon} expansion]{Four-point function of $\boldsymbol{\varphi^2}$ in the $\boldsymbol\epsilon$ expansion}\label{sec:WFfrompaper1}

Our final example is an application of the result from chapter~\ref{ch:paper1} to the four-point function of bilinear scalars $\varphi^2=\varphi^2_S$ in the critical $\OO N$ model $4-\epsilon$ dimensions\footnote{I thank M. van Loon for suggesting this idea.}. At first, this does not seem to be a allowed, since the $\epsilon$ expansion neither is a gauge theory nor defined in four dimensions. However, by comparing figures~\ref{fig:spectrumWF}\footnote{Figure~\ref{fig:spectrumWF} shows the $N=1$ case, but as discussed in section~\ref{sec:ONspectrum}, the corresponding spectrum for general $N$ takes the same form but with larger number of degenerate operators at each point. When studying the $\expv{\varphi_S^2\varphi_S^2\varphi_S^2\varphi_S^2}$ correlator, the grey bands of figure~\ref{fig:spectrumWF} will be adjusted to match those of figure~\ref{fig:spectrumNN4}.} and \ref{fig:spectrumNN4}, we see that the spectrum of operators in the singlet representation has the same structure as in \NN4 SYM, and furthermore, by Wick contractions, the tree-level correlator $\expv{\varphi_S^2\varphi_S^2\varphi_S^2\varphi_S^2}$ takes the form
\beq{
\G^{(0)}(u,v)=1+\frac{u^2}{v^2}+u^2+\frac 4N\parr{u+\frac uv+\frac{u^2}v},
}
which exactly matches \eqref{tree.level} for $c=\frac4N$. Working at leading order in $\epsilon$, the conformal blocks can be evaluated in four dimensions, and finally, the $\epsilon$ expansion is consistent with the perturbative structure~\eqref{eq:ansatzforgaugetheory}, or equivalently \eqref{Ginf.leading}, which was derived in \cite{AldayBissi2013} for conformal gauge theories.

The only piece of CFT-data that can not be extracted by this method are the corrections to the OPE coefficients $\langle a^{(1)}_{\tau_0,\ell}\rangle$ for $\tau_0\geqslant 4$. This is because the expression \eqref{eq:structure.constants.sum}, from which they would follow, corresponds to expanding the $\expv{\varphi_S^2\varphi_S^2\varphi_S^2\varphi_S^2}$ correlator in terms of four-dimensional conformal blocks, whereas a proper decomposition should be in terms of $(4-\epsilon)$-dimensional blocks. At leading twist, $\tau_0=2$, we do not have this issue since the collinear blocks are independent of spacetime dimension. 

We will now determine the values of the constants in the results of chapter~\ref{ch:paper1} by matching with CFT-data of leading twist operators, i.e.\ the weakly broken currents. We compare their dimensions with respect to the free four-dimensional theory, and we therefore have
\beq{
\gamma_{\mathrm{ext}}=-1+\gamma_{\varphi^2_S}^{(1)}=-\frac6{N+8}
,\qquad \gamma_{\ell}^{\mathrm{LT}}=-\epsilon+(1+\gamma_{\mathrm{ext}})\epsilon\delta_{\ell0},
}
from which it follows that
\beq{
 \alpha_{11}=0,\qquad \alpha_{10}=-\frac12,\qquad
\nu_0=1+\gamma_{\mathrm{ext}}=\frac{N+2}{N+8}.
}
The central charge gives the final constraint. Since by \eqref{eq:CentralCharge}, the central charge is not corrected until order $\epsilon^2$, this corresponds to a purely dimensional correction which we can extract from $C_T=N\frac{4-\epsilon}{3-\epsilon}$ by \eqref{eq:CTNscalars}. Matching with \eqref{eq:twisttwoalphaa} gives
\beq{
\alpha_{00}=-2(1+\gamma_{\mathrm{ext}})=-\frac{2N+4}{N+8}.
}

There are some consistency conditions we can check for our result, namely for operators where there is no mixing. For instance, we get $\gamma_{4,0}=0$, which is consistent with $\Delta_{\varphi^4_S}=4+O(\epsilon^2)$. In the $N=1$ case we can also check that $\gamma_{4,2}=-\frac59$, which is consistent with $\Delta_{\de^2\varphi^4_S}=4-2\epsilon+\frac{13}9\epsilon+O(\epsilon^2)$, and that $a_{6,0}=0$, which is consistent with the fact mentioned at the end of section~\ref{sec:WFspectrum} that there is no operator at this point constructed from four fields $\phi$.

Finally, let us give two pieces of our results on explicit form, namely the leading twist OPE coefficients,
\beq{
a^{\mathrm{LT}}_{\ell}=c^2_{\varphi_S^2\varphi_S
^2\mathcal J_{S,\ell}}=\frac{8\Gamma(\ell+1)^2}{N\Gamma(2\ell+1)}\big(1+\epsilon\parr{S_1(2\ell)-2-\delta_{\ell0}}+\gamma_{\mathrm{ext}} \epsilon\parr{2S_1(\ell)-2-\delta_{\ell0}}
\big),
}
and the whole correlator,
\beqa{
\G^{(1)}
(u,v)&=
-(1+\gamma_{\mathrm{ext}})\frac{4\, u}{N v}\parrk{\left(u+v+uv\right)\Phi(u,v)+2+2u+2v}\nonumber\\
&\quad+\gamma_{\mathrm{ext}}\left[\frac{u}{v}\left(\frac{u}{v}+uv+\frac{8u}N\right)\log u+\frac{u}{v}\left(-\frac{u}{v}-\frac4N-\frac{4u}N+\frac{4v}N \right)\log v\right]
\nonumber\\\label{eq:varphi.final}
&\quad+\frac{4u}{Nv}\parrk{\parr{u-\frac v2-\frac12}\log u+\parr{v-\frac u2-\frac12}\log v}\!,
}
where $\Phi(u,v)$ is given in \eqref{eq:boxfunction}.

\chapter{Discussion}
\label{ch:disc}

In this thesis we have aimed to give a comprehensive account of large spin perturbation theory and its application to conformal field theories with a suitable expansion parameter. In chapter~\ref{ch:intro} we gave an introduction where we said that the goal of the analytic bootstrap is to make general statements about the spectrum, the OPE and the correlators in a CFT, without any reference to additional tools such as a Lagrangians or supersymmetry. Chapter~\ref{ch:analyticstudy} contained a technical background where we fixed some essential conventions and definitions, reviewed the developments that led to the proposal of large spin perturbation theory and gave a precise derivation of the perturbative inversion formula. It also contained a short exposition of three commonly studied CFTs, which were revisited in the later chapters. Chapter~\ref{ch:practical} took the form of a practical guide providing the tools needed for a step-by-step application of the computational framework. 
The three following chapters gave explicit examples of how large spin perturbation theory can be applied to generate a variety of new results and insights. Following closely the original publications, we derived in detail new results for the Wilson--Fisher fixed-point at order $\epsilon^4$ \cite{Paper2}, and the most general perturbative four-point function of bilinear scalars in a conformal gauge theory \cite{Paper1}. Chapter~\ref{ch:more} was more linked to the earlier chapters and made direct use of several results derived there. Sections~\ref{sec:paper4short} and \ref{sec:paper5short} studied critical $\phi^4$ theories with $\OO N$ and generic global symmetry group, based on \cite{Paper4} and \cite{Paper5} respectively. Section~\ref{sec:multicritical} contained a previously unpublished computation of the leading CFT-data at the multicritical fixed-points in an $\epsilon$ expansion near their critical dimensions, including new results for OPE coefficients and the central charge. Finally, section~\ref{sec:WFfrompaper1} was an adaptation of the results of chapter~\ref{ch:paper1} to the correlator of bilinear singlet scalars in the $\epsilon$ expansion, also that a previously unpublished result.

Let us recapitulate the main line of arguments in slightly new words. Thinking about CFT correlators as constructed from a sum of conformal blocks, it appears as if there would be an enormous amount of freedom in a generic CFT, where all pieces of CFT-data in principle could be uncorrelated. The crossing equation, which is equivalent to associativity of the OPE, massively reduces this freedom. The Lorentzian inversion formula, developed in analogy with similar dispersion relations in scattering amplitudes in generic quantum field theories, shows that physical considerations in the form of a bounded Regge limit constrains the CFT-data even more. Specifically, the inversion formula connects data of all spinning operators into analytic functions. Of course, such functions could in principle be very complicated, but at lower orders in perturbation theory they tend to take simple forms, often guided by some transcendentality principle. 

In some contexts, the transcendentality principle is comparatively well understood, for instance in scaling dimensions of individual operators in the Wilson--Fisher model. The transcendental numbers appearing there can be viewed as the result of Feynman integrals computed to high order---collectively known as periods---and the order at which each number appears has been analysed in terms of a Galois coaction principle \cite{Panzer:2016snt}. For the CFT-data of entire twist families, transcendentality was an important organisational principle in computing leading twist anomalous dimensions in QCD and \NN4 SYM, and there is a dictionary between the generalisations of harmonic numbers and harmonic polylogarithms through the Mellin transform \eqref{eq:mellinsplitting} \cite{Remiddi:1999ew}. It is an interesting direction for future research to connect these two cases in order to give a deeper theoretical explanation for the structure of the inversion dictionary worked out in appendix~\ref{integrals}.

The fact that the CFT-data of spinning operators is captured in the double-discontinuity of the correlator has interesting physical interpretations. It means that the intuition from holography, where multi-trace operators on the boundary correspond to multi-particle states in the bulk, in a suitable approximation extends beyond holographic CFTs. This was referred to as superhorizon locality in \cite{Fitzpatrick:2012yx}. From a more practical point of view, it means that we are able to capture the whole CFT-data in terms of a few parameters. As we have seen, in perturbative CFTs this can be developed in a systematic way, where it is possible to make precise statements about exactly which operators contribute at a given order. However, many ideas from this approach persist non-perturbatively, as demonstrated in \cite{Simmons-Duffin:2016wlq}.
\\

\noindent Two main types of applications emerge from the framework described in this thesis, although there is some overlap. On the one hand, the analytic bootstrap can be used to make universal statements, valid in all or a wide range of conformal field theories. Here, large spin perturbation theory is very suitable, since it relies only on the CFT axioms and the existence of an expansion parameter $g$ such that the spectrum at $g=0$ has twist degeneracy. By explicitly stating what further assumptions are made, one can prove statements of the form \emph{any CFTs satisfying $A$, also has properties $B$}. Chapter~\ref{ch:paper1} is an obvious example of this, where we find most general perturbative four-point function of bilinear scalars in a conformal gauge theory. From the simple assumption that such a theory admits a global symmetry singlet $\O$ with dimension $\Delta_\O=2+O(g)$, it follows that the space of possible four-point functions is constrained to a five-parameter family. However, also the work in chapter~\ref{ch:paper2} and section~\ref{sec:paper4short} on the Wilson--Fisher and $\OO N$ models can be said to adhere to this principle. In these cases, we initially only make some crude assumptions on the theory, which in principle could apply to a larger range of models. Yet we saw that these assumptions led to quadratic equations---\eqref{eq:betalikeequation} in the $\epsilon$ expansion and \eqref{eq:findingON} in the large $N$ expansion---which precisely single out the only two known possibilities: the free theory and the critical Ising or $\OO N$ model. In section~\ref{sec:paper5short} we proposed that this principle can be used to classify fixed-points with any global symmetry group. The only input needed is the crossing matrix $M_{R\rp}$ and the parity of the respective irreps $R$.

On the other hand, we have also demonstrated how large spin perturbation theory can be used to derive new quantities in a number of theories. In particular, it treats anomalous dimensions and OPE coefficients on an almost equal level and is therefore particularly useful for the latter, where diagrammatic approaches such as the skeleton expansion \cite{Petkou:1994ad} become very complicated. For instance, until five years ago the central charge $C_T$ was only known to order $\epsilon^2$ in the Wilson--Fisher \cite{Cappelli:1990yc} and $\OO N$ models \cite{Petkou:1994ad}. It was then computed to order $\epsilon^3$ in \cite{Gopakumar:2016cpb} and \cite{Dey:2016mcs} respectively, and to order $\epsilon^4$ in \cite{Paper2} and \cite{Paper3}, reproduced in chapter~\ref{ch:paper2}. An extension to order $\epsilon^5$ may be possible in the near future; the main obstacle is to determine the exact contribution from the operators $\de^\ell\phi^4$, which participate in a non-trivial mixing. When using large spin perturbation theory to generate new results in a given theory, one can follow a less purist approach, and allow oneself to rely on existing results in the literature and on methods specific to the theory under consideration. As we reviewed in section~\ref{sec:applicationsofLSPT}, this combination of methods has generated many results in the case of strongly coupled \NN4 SYM.
\\

\noindent While we have aimed to present a consistent and coherent account of analytic conformal bootstrap and large spin perturbation theory, one aspect is not completely addressed, namely the question about analyticity in spin at spin zero. At first, the statement made in \cite{Caron-Huot2017} and repeated in section~\ref{sec:Inversionformula} seems decisive---the Lorentzian inversion formula will give the correct result for $\ell>1$, and one would expect there to be a non-analytic contribution at spin zero\footnote{Of course, similar considerations would apply at spin one, but we limit the discussion here to spin zero.}. However, as reviewed in section~\ref{sec:spinzerogendisc}, we have seen that in many cases, analyticity can in fact be extended to spin zero, albeit sometimes in a subtle way\footnote{The most obvious exception is the Konishi four-point function in chapter~\ref{ch:paper1}, where we explicitly had to supplement the averaged anomalous dimensions $\overline{\langle \gamma_{2,\ell} \rangle}_{\mathrm{inf}}=\gamma_{\mathrm{univ.}}(\ell)$ with a contribution at spin zero. However, there the spin zero anomalous dimension, corresponding to the Konishi operator, matches instead the dimension of the spin zero operator in the $\Xi$ family: $\gamma_{\mathcal K}=\gamma_{\Xi_0}$ in \eqref{eq:anomalousdimensionsLTSYM}.}. It remains as an open problem to clear up exactly when and why this is possible. 

A particularly interesting case is the Wilson--Fisher model, where we analysed the analytic continuation in section~\ref{sec:mathingconditions}, as illustrated in figure~\ref{fig:plot}. If one expands the scaling dimensions \eqref{deltaell} of the weakly broken currents in the coupling constant, one gets a pole $\sim\frac{g^2}\ell$, as in \eqref{eq:gammaWFlitt}. If one instead evaluates the full conformal spin at $\ell=0$, assuming an anomalous dimension of order $g$, the pole gets cancelled at the cost of one factor $g$, and one arrives at the quadratic equation~\eqref{eq:betalikeequation}. Conjecturing that this extends to $\phi^4$ theories in any symmetry group underpins the method of section~\ref{sec:paper5short}, and so far no counterexamples have been found\footnote{The case for $\OO m\times \OO n/\Z_2$ symmetry is described in \cite{Paper5}. We have also checked this for theories with $\OO m^n\rtimes \mathrm S_n$, hypercubic and hypertetrahedral symmetry, where the latter two do not admit a Hubbard--Stratonovich description at large $N$.}. It is interesting to compare with the situation in the multicritical case for generic $\theta$, where such continuation to spin zero also passes beyond the pole $\hb=1$ but generates $\gamma_{\phi^2}$ correctly without any subtleties. This may be explained through the observation made in \cite{Alday2016b} by explicitly studying crossing in almost free theories, that a solution truncated in spin is allowed only in four dimensions. Perhaps more can be learnt by investigating various limits in the $(\theta,d)$ plane near the point $(\theta=2,d=4)$. This idea is inspired by the interpolation made in \cite{Badel:2019oxl} between universal non-perturbative behaviour at large $\OO 2$ charge and the $4-\epsilon$ expansion  (see also \cite{Arias-Tamargo:2019xld,Watanabe:2019pdh}). 

Holographic CFTs in a strong coupling expansion explicitly violate analyticity at spin two, since they are assumed to have a stress tensor but no other single-trace operators of spin $\ell\geqslant2$. This holds for both \NN4 SYM at strong coupling and in the minimal gravity theories considered in the heavy-light bootstrap mentioned in section~\ref{sec:applicationsofLSPT}. Indeed, in formulating proposition~\ref{prop:currentsexist} we had to limit ourselves to the case where the expansion parameter is the coupling constant and the weakly broken currents can be explicitly constructed, since at strong coupling the range of analyticity is raised from $\ell>1$ to some other small value \cite{Alday:2017vkk}. We note that the truncated solutions of \cite{Heemskerk2009,Alday:2014tsa} which violate analyticity in spin also violate Regge boundedness used in the derivation of the Lorentzian inversion formula.

A promising tool in addressing the question of spin zero is the notion of light-ray operators \cite{Kravchuk:2018htv}. They are non-local, intrinsically Lorentzian, operators, which can take arbitrary values of spin. For integer spins, they reduce to the integration of local operators along a null direction, called the \emph{light transform}. In this perspective, the function $C(\Delta,\ell)$ generated from the Lorentzian inversion formula \eqref{eq:invformulagen} should be viewed as CFT-data for a family of such non-local operators. An integral transform similar to the light transform is the \emph{shadow transform}, which generates a non-local operator with scaling dimension $d-\Delta$. Perhaps a combination of these transforms might explain our final example of analytic continuation to spin zero, namely the shadow relation \eqref{eq:shadowrelation}, which relates the dimension of a Hubbard--Stratonovich auxiliary field to the shadow dimension of the would-be $\ell=0$ operator in the leading twist family. While such relation at tree-level is clear from the Lagrangian, it is not obvious why it should hold in perturbation theory or if this continues to all orders\footnote{We have checked that it holds at order $1/N^2$ in three dimensions, using the literature values for $\gamma^{(2)}_\sigma$ \cite{Vasiliev:1981dg} and $\gamma^{(2)}_{S,\ell}$ \cite{Manashov:2017xtt}. Notice the typo in the expression $\gamma_2(s)$ in the published version of \cite{Manashov:2017xtt}, one of the harmonic numbers should read $S_1(s-\frac12)$ instead of $S_1(s+\frac12)$.}.
\\

\noindent When twist additivity was derived, which opened up a new, analytic, direction of the bootstrap, the authors of \cite{Fitzpatrick:2012yx} said that the conformal bootstrap since its revival had already ``led to a great deal of progress, but perhaps the best is yet to come''. Now, three quarters of a decade later, the research has gone into a more mature phase. While there is certainly room for more numerical bootstrap, it is perhaps the analytic techniques that will generate the greatest long-term impact on fundamental physics as a whole. Conformal field theories will continue to play an important role, both as toy models and as \emph{bona fide} models for physics, and will be studied by a wide range of existing and future techniques. It is my hope that the methods presented in this thesis will become a part of the toolbox for anyone who is interested in studying various conformal field theories. More precisely, a large spin analysis should be one of the aspects considered when giving a presentation of the fundamental characteristics of a given CFT.

Apart from obtaining a complete description of the situation around spin zero, there is also room for other developments of large spin perturbation theory. This includes applications to new specific models as well as theoretical and technical improvements. One such direction is to extend the inversion formula to non-scalar correlators. While this may not be so conceptually hard---the essential ingredients could perhaps be extracted from \cite{Kravchuk:2018htv}---one has to reduce the setup to a manageable problem to facilitate a practical implementation. Initial examples to consider may be correlators involving fermions \cite{Elkhidir:2017iov,Albayrak:2019gnz} or conserved currents $J_\mu$, for instance analysing the $\epsilon$ expansion of the the $\expv{\varphi\varphi JJ}$ correlator in the critical $\OO2$ model, which was studied numerically in three dimensions in \cite{Reehorst:2019pzi}.

More generally, the results and methods from the analytic bootstrap may be used in combination with other methods for studying conformal field theories, both within the bootstrap and more generally. One promising result is the systematic study of twist families in two dimensional CFTs in \cite{Collier:2018exn}. It is also desirable to make a more direct contact with the numerical bootstrap in the search for a more powerful implementation to be used for finding non-perturbative CFTs. This direction is connected with great challenges in combining the respective rigid assumptions on both sides---strict inequalities in the numerical bootstrap and perturbative expansion in large spin perturbation theory as presented here.

In the future, mathematical consistency will continue to be a leading principle within theoretical physics. This is increasingly true as the field develops in a direction away from standard methods within perturbative Lagrangian quantum field theory. In this thesis we have presented a complete framework for a perturbation theory that is independent on any Lagrangian description, and therefore applies also to expansion parameters different from the coupling constant. The future will show what other results and technologies will follow from clever application of mathematical consistency conditions.

\appendix

\renewcommand{\appendixname}{\xspace}

\chapter{Appendices to chapters \ref{ch:analyticstudy} and \ref{ch:practical}}

\section{Subleading corrections to collinear conformal blocks}\label{app:subcollinearblocks}
In this appendix we give some more details on the subleading corrections to the conformal blocks in the collinear limit, referring back to section~\ref{sec:blockology}. We follow closely appendix~A of \cite{Paper3}. In the collinear limit the conformal blocks expand as \cite{Simmons-Duffin:2016wlq}
\beq{
G^{(d)}_{2h+\ell,\ell}(z,\zb)=z^h\sum_{k=0}^\infty z^k\sum_{m=-k}^k c_{k,m}k_{h+\ell+m}(\zb)
.}
The coefficients $c_{k,m}$ can be computed order by order by solving the Casimir equation
\beq{
\cas G^{(d)}_{2h+\ell,\ell}(z,\zb) =(h+\ell)(h+\ell-1) G^{(d)}_{2h+\ell,\ell}(z,\zb),
}
where
\beq{
\cas=\mathcal C_2-h(h+1-2\mu),
}
and $\cas$ is given in \eqref{eq:quadraticCasimir}. The results for the first two subleading orders are

\beqa{\nonumber
c_{0,0}&=1,\qquad\qquad c_{1,-1}=\frac{\ell  (\mu -1)}{\ell +\mu -2},\qquad\qquad c_{1,0}=\frac h2,\phantom{alotofequationspaceneedd}\\\nonumber
c_{1,1}&=\frac{(\mu -1) (h+\ell )^2 (2 h+\ell -1)}{4 (2 h+2 \ell -1) (2 h+2 \ell +1) (2 h+\ell -\mu +1)},}\beqa{
\nonumber
c_{2,-2}&=\frac{(\ell -1) \ell  (\mu -1) \mu }{2 (\ell +\mu -3) (\ell +\mu -2)}, \qquad\qquad c_{2,-1}=\frac{(h+1) \ell  (\mu -1)}{2 (\ell +\mu -2)},
\\\nonumber
c_{2,0
}&=\tfrac{h (h+1)^2 (2 h-3) (\mu -2) (2 h-\mu +1)}{4 (2 h+2 \ell -3) (2 h+2 \ell +1) (\ell +\mu -2) (2 h+\ell -\mu +1)}
\\\nonumber
&\hspace{-8mm}
+\!
\tfrac{(2 h-1) \ell  (4 h^5+4 h^4 (\mu -2)+h^3 (46 \mu -10 \mu ^2-59)+h^2 (2 \mu ^3-31 \mu ^2+92 \mu -81)+h (8 \mu ^3-38 \mu ^2+64 \mu -37)+(\mu -1)^2 (2 \mu -3))}{4 (2 h+2
   \ell -3) (2 h+2 \ell +1) (2 h-2 \mu +3) (\ell +\mu -2) (2 h+\ell -\mu +1)}
   \\\nonumber
   &\quad
   +\tfrac{\ell ^2 (20 h^5-12 h^4 (\mu -2)+h^3 (-2 \mu ^2+22 \mu -55)-h^2 (6 \mu ^3+11 \mu ^2-92 \mu +101)+h (16 \mu ^3-64 \mu ^2+90 \mu -41)-\mu +1)}{4 (2 h+2 \ell -3) (2 h+2 \ell +1) (2
   h-2 \mu +3) (\ell +\mu -2) (2 h+\ell -\mu +1)}
\\\nonumber
&\quad
+
\tfrac{ \ell ^3(\ell+4 h-2) (4 h^3-4 h^2 (\mu -3)+2 h (\mu ^2-5 \mu +6)-2 \mu ^3+7 \mu ^2-9 \mu +4)}{4 (2 h+2 \ell -3) (2 h+2 \ell +1) (2 h-2 \mu +3) (\ell +\mu -2) (2 h+\ell -\mu +1)}
\\\nonumber
 c_{2,1}&=\frac{(h+1) (\mu -1) (h+\ell )^2 (2 h+\ell -1)}{8 (2 h+2 \ell -1) (2 h+2 \ell +1) (2 h+\ell -\mu +1)}, 
\\
c_{2,2}&=\frac{(\mu -1) \mu  (h+\ell )^2 (h+\ell +1)^2 (2 h+\ell -1) (2 h+\ell )}{32 (2 h+2 \ell -1) (2 h+2 \ell +1)^2 (2 h+2 \ell +3) (2 h+\ell -\mu +1) (2 h+\ell -\mu +2)}.
}
Unfortunately, we have not been able to find any closed form for the coefficients $c_{k,m}$, except for the four sequences $c_{\pm k,\pm k}$ and $c_{\pm k,\pm k\mp 1}$.

\section{Some useful identities}\label{app:identities}
In this appendix we collect some useful identities used throughout the thesis. Some of these are not symbolically implemented in Mathematica \cite{Mathematica}, but can be checked numerically.

The hypergeometric function is defined by
\beq{\label{eq:hyperdef}
_{p+1}F_{p}\left( {{a}_1\,\dots,\,{a}_{p+1}}~\atop~{ \!\!\!{b}_1,\dots,{b}_p}\middle|x\right)=\sum_{k=0}^{\infty}\frac{\prod_{i=1}^{p+1}(a_i)_k}{\prod_{j=1}^{p}(b_j)_k}\frac{x^k}{k!}.
}
The case $p=1$ is referred to as \emph{Gau\ss's hypergeometric function}, for which we write $_2F_1(a,b;c;x)$. It satisfies the differential equation
\beq{
x(1-x)F''(x)+\parr{c-(a+b+1)x}F'(x)-a b F(x)=0.
}
The special case $a=b=c/2=\hb$ that appears in the collinear blocks has a series expansion near unit argument which contains logarithms,
\beq{\label{eq:gausshyperexp}
{_2F_1}(\hb,\hb;2\hb;1-w)=\sum_{k=0}^\infty\frac{\Gamma(2\hb)}{\Gamma(\hb)^2}\parr{\frac{(\hb)_k}{k!}}^2\parr{2S_1(k)-2S_1(k+\hb-1)-\log w}w^k.
}
For generic parameter values satisfying $a+b\neq c$, the expansion consists of two superimposed power series:
\beqa{\nonumber
{_2F_1}(a,b;c;1-w)&=\parr{\frac{\Gamma(c)\Gamma(c-b-a)}{\Gamma(c-a)\Gamma(c-b)}+O(w)}\\&\quad-w^{c-a-b}\parr{\frac{\Gamma(c)\Gamma(c-a-b)\Gamma(1+a+b-c)}{\Gamma(a)\Gamma(b)\Gamma(1+c-a-b)}+O(w)}.
\label{eq:generic2F1expansion}
}
When $a+b<c$, the limit $w\to0$ is finite and gives 
\beq{\label{eq:Buhringsrelation}
_2F_1(a,b;c;1)=\frac{\Gamma(c)\Gamma(c-b-a)}{\Gamma(c-b)\Gamma(c-a)}.
}

For $x\in (0,1)$ the following holds for Gau\ss's hypergeometric function
\begin{equation} \label{eq:2F1zbtoxit}
_2 F_1\left(a,c-b;c;\frac{x}{x-1}\right) = (1-x)^a\, _2 F_1(a,b;c;x),
\end{equation}
and for the polylogarithms
\beqa{\label{eq:polyLogid2}
\mathrm{Li}_2(x)&=\zeta_2-\mathrm{Li}_2(1-x)-\log(1-x)\log x,
\\\label{eq:polyLogid3}
\mathrm{Li}_3(x)&=\zeta_3-\mathrm{Li}_3(1-x)-\mathrm{Li}_3\parr{\frac{x-1}x}+\zeta_2\log x+\frac{\log^3 x}6-\frac{\log(1-x)\log^2x}2.
}

Finally, we have the integrals
\beqa{\label{eq:BetaIntegralDef}
\int\limits_0^1 \df x\, x^a(1-x)^b&=\mathrm B(a+1,b+1)=\frac{\Gamma(a+1)\Gamma(b+1)}{\Gamma(a+b+2)},
\\\label{eq:BetaIntegralmod}
\int\limits_0^1 \df x\,x^a(1-x)^b(1-\gamma x)^c&=\frac{\Gamma(a+1)\Gamma(b+1)}{\Gamma(a+b+2)}{_2F_1}(a+1,-c;a+b+2;\gamma),
}
for $a,b>-1$ and $\gamma\in(0,1)$,
and the practical relation for the $\Gamma$ function:
\beq{\label{eq:subSinGamma}
\sin(\pi p)=\frac{\pi}{\Gamma(p)\Gamma(1-p)}.
}

\chapter{Appendices to chapter \ref{ch:paper2}}

\section{Some inversion integrals}
\label{integrals}

In chapter~\ref{ch:paper2} of the main text we arrived at the inversion integral
\begin{equation}
 A(J)= \frac{1}{\pi^2} \int_0^1 \df t \df\bar z \frac{\bar z^{\bar h-2}(t(1-t))^{\bar h-1}}{(1-t \bar z)^{\bar h}}  \dDisc \left[G(\bar z)\right],
\end{equation}
where $J^2=\bar h(\bar h-1)$. In table~\ref{tab:allinversions} we present a number of inversions used in the main text. In this table we use the nested harmonic sums $S_{\mathbf{a}}=S_{\mathbf{a}}(\bar h-1) $ which for integer arguments take the values
\begin{table}
\begin{center}
\caption{Inversions used in the $\epsilon$ expansion in chapter~\ref{ch:paper2}.}\label{tab:allinversions}
{\small
\vspace{4pt}
\begin{tabular}{|l|l|}\hline
$ G(\bar{z}) $ & $ A(J)$
\rule{0pt}{2.7ex} \rule[-1.2ex]{0pt}{0pt}
\\\hline
$\log ^2\left(1-\bar{z}\right)$&$\displaystyle \dfrac{4 }{J^2}$
\rule{0pt}{4.2ex} \rule[-4.2ex]{0pt}{0pt} \\
$\log^3(1-\bar z)$ & $\displaystyle -\dfrac{24 S_1}{J^2}$
 \rule{0pt}{4.2ex} \rule[-4.2ex]{0pt}{0pt} \\
$\log^4(1-\bar z)$ & $\displaystyle\dfrac{96}{J^2}\left(S_1^2- \zeta _2- S_{-2}\right)$
 \rule{0pt}{4.2ex} \rule[-4.2ex]{0pt}{0pt} \\
$ \log ^2\left(1-\bar{z}\right)\text{Li}_2\left(1-\bar{z}\right)$ & $\displaystyle\dfrac{4 }{J^2}\left( \zeta _2+2 S_{-2}\right)$
 \rule{0pt}{4.2ex} \rule[-4.2ex]{0pt}{0pt} \\
$ \log ^3\left(1-\bar{z}\right)\text{Li}_2\left(1-\bar{z}\right)$ & $\displaystyle\dfrac{24 }{J^2}\left( \left( S_{-3}-2 S_{-2,1}\right)-3 \left( S_{-3}-2 S_{1,-2}\right)+3 \zeta _2 S_1-2
   S_3\right)$
 \rule{0pt}{4.2ex} \rule[-4.2ex]{0pt}{0pt} \\\hline
$ \log ^2\left(1-\bar{z}\right)\text{Li}_3\left(1-\bar{z}\right)$ & $\displaystyle\dfrac{4}{J^2}\left(-2 \left( S_{-3}-2 S_{1,-2}\right)+\zeta _3+2 \zeta _2 S_1-2 S_3\right)$
 \rule{0pt}{4.2ex} \rule[-4.2ex]{0pt}{0pt} \\
$ \log ^2\left(1-\bar{z}\right)\text{Li}_3\left(\dfrac{\bar{z}-1}{\bar{z}}\right)$ & $\displaystyle\dfrac{4}{J^2}\left(-2 \zeta _3-\frac{1}{J^6}-\frac{2}{J^4}+2 S_3\right)$
 \rule{0pt}{4.2ex} \rule[-4.2ex]{0pt}{0pt} \\
$\log ^2\left(1-\bar{z}\right) \log\bar z$ & $\displaystyle-\dfrac{4 }{J^4}$
 \rule{0pt}{4.2ex} \rule[-4.2ex]{0pt}{0pt} \\
$\log ^2\left(1-\bar{z}\right)\log^2\bar z$ & $\displaystyle\dfrac{8 }{J^2}\left(- \zeta _2+\frac{1}{J^4}+\frac{1}{J^2}-2 S_{-2}\right)$
 \rule{0pt}{4.2ex} \rule[-4.2ex]{0pt}{0pt} \\
$\log ^2\left(1-\bar{z}\right) \log^3\bar z$ & $\dfrac{24 }{J^2}\left(2 \left( S_{-3}-2 S_{1,-2}\right)+ \zeta _3-\dfrac{1}{J^6}-\dfrac{2}{J^4}+\dfrac{ \zeta
   _2}{J^2}+\dfrac{2 S_{-2}}{J^2}-2 \zeta _2 S_1\right)$\hspace{-5mm}
 \rule{0pt}{4.2ex} \rule[-4.2ex]{0pt}{0pt} \\\hline
$\log ^3\left(1-\bar{z}\right) \log \bar z$ & $\displaystyle\dfrac{24 }{J^2}\left(- \zeta _2+\frac{ S_1}{J^2}-2 S_{-2}\right)$
 \rule{0pt}{4.2ex} \rule[-4.2ex]{0pt}{0pt} \\
$\log ^4\left(1-\bar{z}\right) \log\bar z$ & $\displaystyle\begin{matrix}\dfrac{48 }{J^2}\left(-\dfrac{2S_1^2}{J^2}-4 \left( S_{-3}-2 S_{-2,1}\right)+6 \left(S_{-3}-2
   S_{1,-2}\right)\right.\\\qquad\qquad\qquad\left.+3 \zeta _3+\dfrac{2 \zeta _2}{J^2}+\dfrac{2 S_{-2}}{J^2}-6 \zeta _2 S_1+2 S_3\right)\end{matrix}$
 \rule{0pt}{6.2ex} \rule[-6.2ex]{0pt}{0pt} \\
$\log ^2(1-\bar{z})\text{Li}_2(1-\bar{z})\!  \log \bar{z}$ & $
\displaystyle\dfrac{4 }{J^2}\left(-6 \left( S_{-3}-2 S_{1,-2}\right)-\dfrac{ \zeta
   _2}{J^2}-3 \zeta _3-\dfrac{2 S_{-2}}{J^2}+6 \zeta _2 S_1\right)$
  \rule{0pt}{6.2ex} \rule[-6.2ex]{0pt}{0pt} \\\hline
\end{tabular}
}
\end{center}
\end{table}
\begin{equation}
S_{a_1,a_2,\ldots}(n)=\sum_{b_1=1}^n\frac{(\mathrm{sgn}\, a_1)^{b_1}}{b_1^{|a_1|}}\sum_{b_2=1}^{b_1}\frac{(\mathrm{sgn}\, a_2)^{b_2}}{b_2^{|a_2|}}\sum_{b_3=1}^{b_2}\frac{(\mathrm{sgn}\, a_3)^{b_3}}{b_3^{|a_3|}}\cdots.
\end{equation}
For non-integer values of $\bar h$ we make the standard analytic continuation from even arguments $n$, see e.g.\ \cite{Albino2009}, so that for instance
\begin{equation}
S_{-2}(x)=\frac{1}{4}\left(\psi^{(1)}\left(\tfrac{x+1}{2}\right)-\psi^{(1)}\left(\tfrac{x+2}{2}\right)\right)-\frac{\zeta_2}{2},
\end{equation}
where $\psi^{(1)}(x)$ is the trigamma function.

Evaluating these inversion integrals is non-trivial, but one can proceed as follows. Expanding the function to invert in powers of $\frac{1-\bar z}{\bar z}$ we are led to the integral entering in \eqref{eq:invXiDphi}. We then by \eqref{eq:sumexpansionbeforeinversion} obtain a series expansion for large $J^2$, which can be identified as a linear combination of suitable functions. The final result is checked numerically, for finite values of $\bar h$, to very high precision.

\section{Double-discontinuity at fourth order}
\label{ddisc}
To order $g^4$ the terms contributing to the double discontinuity from the bilinear operators are
\begin{align}
I_{\varphi^2}&= \log^4(1-\bar z) \frac{\log \bar z-\log z}{192} + \log^3(1-\bar z) \frac{1}{24} (\text{Li}_2(1-\bar z)+3\log z-3\log \bar z+2\zeta_2)\nonumber\\
&+ \log^2(1-\bar z)\Bigg(\frac{5}{8} \text{Li}_3(1-\bar z)-\log z \frac{46+3 \text{Li}_2(1-\bar z)+\log\bar z+12 \zeta_2}{48}  +\frac{1}{2} \text{Li}_3 \left(\!\frac{\bar z-1}{\bar z}\! \right)\nonumber \\
&+ \frac{2 (23+6 \zeta_2) \log \bar z-\text{Li}_2(1-\bar z) (21 \log \bar z+34)-106 \zeta_2-4 \log^3\bar z+\log ^2\bar z+24 \zeta_3}{48}\Bigg)\!,\\
I_{2} &=\frac{\log^2(1-\bar z)}8 \left( \log z (\zeta_2-2)  +2 \log \bar z+ \frac{1}{6}\log^3 \bar z +\text{Li}_3(1-\bar z)-\text{Li}_3\left(\!\frac{\bar z-1}{\bar z}\!\right)-\zeta_3\right)\!.
\end{align}
In order to compute the first expression we used the value of the OPE coefficient for the bilinear scalar operator, which takes the form $a_0=2(1-g-g^2+\ldots)$, as well as the precise relation between $g$ and $\epsilon$.

\chapter{Appendices to chapter \ref{ch:paper1}}

\section{Superconformal blocks}
\label{app:superblocks}
In this appendix we present an explicit form of the superconformal blocks appearing in the expansion of correlation functions of four half-BPS operators in $\mathcal{N}=4$ SYM. We closely follow \cite{Doobary:2015gia} and restrict to the case $p_1=p_2=p_3=p_4=2$, which is the one relevant for this thesis. All supermultiplets appearing in the intermediate channel of such correlation functions can be labelled by a Young tableau $\underline\lambda=[\lambda_1,\lambda_2]$, with $\lambda_1\geq\lambda_2$, consisting of maximally two rows, and a charge $\gamma=0,2,4$. We distinguish three types of multiplets: half-BPS, quarter-BPS and long, whose representation labels are summarised in the table~\ref{tab:supermultiplets}. Notice that the only long multiplets are in the singlet representation $[0,0,0]$ of the $\SU4$ R-symmetry. 

\begin{table}[ht]

\centering
\caption{Supermultiplets appearing in the superconformal partial waves of $\langle \mathcal{O}_{\boldsymbol{20'}}\mathcal{O}_{\boldsymbol{20'}}\mathcal{O}_{\boldsymbol{20'}}\mathcal{O}_{\boldsymbol{20'}}\rangle$.}\label{tab:supermultiplets}
{\small
\vspace{4pt}
\renewcommand{\arraystretch}{1.25}
\begin{tabular}{|c|c|c|c|c|}
\hline
Young tableau $\underline\lambda$& twist $\tau$ & spin $\ell$&R-symmetry representation& multiplet type\\
\hline
$[0,0]$&$\gamma$&0&$[0,\gamma,0]$&half-BPS\\
\hline
$[\lambda_1,0],\lambda_1\geq 2$&$\gamma$&$\lambda_1-2$&$[0,\gamma-2,0]$&quarter-BPS\\
$[\lambda_1,1],\lambda_1\geq 2$&$\gamma$&$\lambda_1-2$&$[1,\gamma-4,1]$&quarter-BPS\\
$[1,0]$&$\gamma$&$0$&$[1,\gamma-2,1]$&quarter-BPS\\
$[1,1]$&$\gamma$&$0$&$[2,\gamma-4,2]$&quarter-BPS\\
\hline
$[\lambda_1,\lambda_2],\lambda_2\geq 2$&$2\lambda_2$&$\lambda_1-\lambda_2$&$[0,0,0]$&long\\
\hline
\end{tabular}
}
\end{table}

The superconformal blocks are given by
\begin{equation}
\mathcal{S}_{\mathcal{R}}(z,\zb,\alpha,\alphab)=\left(\frac{z \zb}{\alpha \alphab}\right)^{\gamma/2}\mathcal{F}^{\gamma, \underline \lambda}(z,\zb,\alpha,\alphab),
\end{equation}
where
\begin{equation}
\mathcal{F}^{\gamma, \underline \lambda}(z,\zb,\alpha,\alphab)=(-1)^{\frac{\gamma}{2}-1}D^{-1}\,\det \begin{pmatrix}
F^X_{\underline \lambda}(z,\zb)&R\\K_{\underline\lambda}&F^Y(\alpha,\alphab)
\end{pmatrix}.
\end{equation}
The explicit form of all ingredients (with $1\leq i,j\leq 2$ and $1\leq m,n\leq\gamma/2$) is
\begin{align}
(F_{\underline\lambda}^X(z,\zb))_{in}&=[x_i^{\lambda_n-n}{_2F_1}(\lambda_n+1-n+\tfrac{\gamma}{2},\lambda_n+1-n+\tfrac{\gamma}{2},2\lambda_n+2-2n+\gamma;x_i)],\\
(F^Y(\alpha,\alphab))_{mj}&=(y_j)^{m-1}{}_2F_1(m-\tfrac{\gamma}{2},m-\tfrac{\gamma}{2},2m-\gamma;y_j),
\end{align}
where $x_1=z$, $x_2=\zb$ and $y_1=\alpha$, $y_2=\alphab$, 
and
\begin{align}
(K_{\underline\lambda})_{mn}&=-\delta_{m,n-\lambda_n},\\
R&=\begin{pmatrix}
\frac{1}{z-\alpha}&\frac{1}{z-\alphab}\\\frac{1}{\zb-\alpha}&\frac{1}{\zb-\alphab}
\end{pmatrix},\\
D&=\frac{(z-\zb)(\alpha-\alphab)}{(z-\alpha)(z-\alphab)(\zb-\alpha)(\zb-\alphab)}.
\end{align}  
 Here, the square bracket in the definition of $F^X$ indicates that we keep only the regular part, namely
 \begin{equation}
 [x^{-\alpha}{}_2F_1(a,b,c;x)]=x^{-\alpha}{}_2F_1(a,b,c;x)-\sum_{k=0}^{\alpha-1}\frac{(a)_k(b)_k}{(c)_k k!}x^{k-\alpha}=\sum_{k=0}^{\infty}\frac{(a)_{k+\alpha}(b)_{k+\alpha}}{(c)_{k+\alpha}(k+\alpha)!}x^{k}\,.
 \end{equation}
Importantly, for long multiplets have $\gamma=4$, $\lambda_2=\tfrac{\tau}{2}$, $\lambda_1=\ell+\tfrac{\tau}{2}$, $\tau\geq4$ and $\alpha\geq 0$. Then, the superconformal blocks can be written in a more explicit form as
\begin{align}\label{eq:super.long}
\mathcal{F}_{\mathrm{long}}(z,\zb,\alpha,\alphab)=\frac{(z-\alpha)(z-\alphab)(\zb-\alpha)(\zb-\alphab)}{(z \,\zb)^4}G_{\tau+4,\ell}(z,\zb) \,,
\end{align}   
where $G_{\tau,\ell}(z ,\zb)$ is the ordinary conformal block in four dimensions \eqref{eq:ConformalBlock} as found in~\cite{Dolan:2001tt}.
  
At the unitarity bound, quarter-BPS multiplets can combine to form a long multiplet in the interacting theory. This is exactly the case for the twist-two multiplets in the singlet representation
\begin{equation}
(\gamma=2,\underline\lambda=[\ell+2,0])\oplus (\gamma=4,\underline\lambda=[\ell+1,1])\longrightarrow (\gamma=4,\underline\lambda=[\ell+1,1])_\mathrm{long}\,.
\end{equation}
Using the explicit form of superconformal blocks one can write
\begin{equation}\label{eq:twist2.recomb}
\frac{\alpha \alphab}{z \zb}\mathcal{F}^{2,[\ell+2,0]}(x,y)+\mathcal{F}^{4,[\ell+1,1]}(x,y)=\frac{(z-\alpha)(z-\alphab)(\zb-\alpha)(\zb-\alphab)}{(z \,\zb)^4}G_{6,\ell}(z,\zb) ,
\end{equation} 
which agrees with \eqref{eq:super.long} for $\tau=2$.

\section[More details on \texorpdfstring{$\overline{H}^{(0,\log)}(\zb)$}{H**(0,log)(zb)}]{More details on $\overline{H}^{\boldsymbol{(0,\log)}}\boldsymbol{(\zb)}$}
\label{app:H0log}
In the expression~\eqref{eq:H0log} for $\overline H^{(0,\log)}(\zb)$, the coefficients $e_i$ multiplying $(1-\zb)^i\log^2(1-\zb)$ for $i=\{0,1,2,\ldots\}$ are given by the sequence
\begin{equation}
\left\{-\frac{1}{12},\frac{1}{10},-\frac{5}{504},-\frac{8}{2835} ,-\frac{251}{199584}, 
-\frac{55967}{81081000},-\frac{2499683}{5837832000},-\frac{50019793}{173675502000},
\ldots 
\right\}.
\end{equation}
The corresponding values in the superconformal case \eqref{eq:H0log.super} are
\begin{equation}
\left\{-\frac{1}{12},-\frac{1}{15},-\frac{151}{2520},-\frac{127}{2268},-\frac{53219}{997920},-\frac{8327609}{162162000},-\frac{290756381}{5837832000},
\ldots
\right\}.
\end{equation}

\section{Modified structure constants for the conformal case}
\label{app:results.nonsuper}
In this appendix we present an exact form of the modified structure constants that appear in \eqref{eq:structure.constants.sum}. The equations below are valid for $\tau_0>2$.
\begin{align}
\overline{\langle\hat\alpha_{\tau_0,\ell}\rangle}_{11}&=\frac{4\,c\,\eta}{P_{\tau_0,\ell}}\Big(
-\zeta_2+S_1\left( \tfrac{\tau_0}{2}-2 \right)^2-S_1\left( \tfrac{\tau_0}{2}-2 \right)S_1\left( \tau_0-4 \right)-\frac12 S_2\left( \tfrac{\tau_0}{2}-2 \right)
\nonumber\\
& \quad-\frac{\delta_{\tau_0,4}}{2} + \left[2S_1\left( \tfrac{\tau_0}{2}-2 \right)-S_1\left(\tau_0-4\right)+\tfrac{\delta_{\tau_0,4}}{4}\right]S_1\left( \tfrac{\tau_0}{2}+\ell-1 \right)
\Big),
\\
\overline{\langle\hat\alpha_{\tau_0,\ell}\rangle}_{10}&=\frac{2\,c\,\eta}{P_{\tau_0,\ell}}\Big(-3S_1\left( \tfrac{\tau_0}{2}-2 \right)+2S_1\left(\tau_0-4\right)-\tfrac{3\delta_{\tau_0,4}}{4}-S_1\left( \tfrac{\tau_0}{2}+\ell-1 \right)\Big),
\\
\overline{\langle\hat\alpha_{\tau_0,\ell}\rangle}_{00}&=\frac{c\,\eta}{P_{\tau_0,\ell}},
\\
\overline{\langle\hat\alpha_{\tau_0,\ell}\rangle}_{\mathrm{ext}}&=\frac{2\,c\,\eta}{P_{\tau_0,\ell}}\Big(1+ S_1\left( \tfrac{\tau_0}{2}-2 \right)-S_1\left( \tau_0-4 \right)+\frac{\delta_{\tau_0,4}}{2}\Big)-\frac{\tau_0-3}{P_{\tau_0,\ell}}
\nonumber\\
&\quad+2\left[-1+2 S_1\left( \tfrac{\tau_0}{2}-2 \right)-S_1\left(\tau_0-4 \right)+S_1\left( \tfrac{\tau_0}{2}+\ell-1 \right) \right],
\\
\overline{\langle\hat\alpha_{\tau_0,\ell}\rangle}_{\mu_0}&=\frac{4\,c}{P_{\tau_0,\ell}}\Big(-S_1\left( \tfrac{\tau_0}{2}-2 \right)+ S_1\left(\tau_0-4\right)\Big),
\\
\overline{\langle\hat\alpha_{\tau_0,\ell}\rangle}_{\nu_0}&=\frac{2\,c}{P_{\tau_0,\ell}}\Big(-\zeta_2-2S_1\left( \tfrac{\tau_0}{2}-2 \right)-2S_1\left( \tfrac{\tau_0}{2}-2 \right)^2+2S_1\left(\tau_0-4\right) +S_2\left( \tfrac{\tau_0}{2}-2 \right)
\nonumber\\
& +2S_1\left( \tfrac{\tau_0}{2}-2 \right)S_1\left(\tau_0-4\right)+2\left[S_1\left( \tfrac{\tau_0}{2}-2 \right)- S_1\left(\tau_0-4\right) \right]S_1\left( \tfrac{\tau_0}{2}+\ell-1 \right)\Big).
\end{align}
As in \eqref{eq:gamma.HT}, we have $\eta=(-1)^{\frac{\tau_0}{2}}$ and $P_{\tau_0,\ell}=c\eta+\left(\tau_0+\ell-2\right)\left(\ell+1\right)$.

\section{Konishi CFT-data}
We present here an explicit form of the CFT-data for operators present in the conformal partial wave decomposition of \eqref{eq:Konishi.final}.

The anomalous dimensions are given by
\begin{align}
\overline{\langle \gamma_{2,\ell}\rangle}&=2S_1(\ell)+3\delta_{\ell,0},
\\
\overline{\langle \gamma_{\tau_0,\ell}\rangle}&=6+\frac{12 c}{P_{\tau_0,\ell}}\left[-S_1\left(\tfrac{\tau_0}{2}-2\right)+S_1\left(\tfrac{\tau_0}{2}+\ell-1\right)\right]
\nonumber\\&\quad
+\frac{c\,\eta}{P_{\tau_0,\ell}}\left[
6-\delta_{\tau_0,4}-2S_1\left(\tfrac{\tau_0}{2}-2\right)-2\left(\tfrac{\tau_0}{2}+\ell-1\right)
\right],\qquad \tau_0>2,
\end{align}
and the modified structure constants take the form
\begin{align}
\overline{\langle \hat\alpha_{2,\ell}\rangle}&=-6-3\delta_{\ell,0}-\zeta_2+6 S_1(\ell),
\\
\overline{\langle \hat\alpha_{\tau_0,\ell}\rangle}&=6\left[
-1+2S_1\left(\tfrac{\tau_0}{2}-2\right)-S_1\left(\tau_0-4\right)+S_1\left(\tfrac{\tau_0}{2}+\ell-1\right)
\right]-\frac{3}{P_{\tau_0,\ell}}(\tau_0-3)
\nonumber\\
&\quad+\frac{c}{P_{\tau_0,\ell}}\Big(12\left[ 
S_1\left(\tfrac{\tau_0}{2}-2\right)-S_1\left(\tau_0-4\right)
\right]
\nonumber\\&\qquad\qquad\quad
+\eta\left[-4 S_1\left(\tfrac{\tau_0}{2}-2\right)+2S_1\left(\tau_0-4\right)-\tfrac{\delta_{\tau_0,4}}{2}\right]\Big)S_1\left(\tfrac{\tau_0}{2}+\ell-1\right)
\nonumber\\
&\quad+\frac{6\,c}{P_{\tau_0,\ell}}\Big(
-\zeta_2-2S_1\left(\tfrac{\tau_0}{2}-2\right)^2+2 S_1\left(\tfrac{\tau_0}{2}-2\right)S_1\left(\tau_0-4\right)+S_2\left(\tfrac{\tau_0}2-2\right)
\Big)
\nonumber\\
&\quad+\frac{c\,\eta}{P_{\tau_0,\ell}}\Big(\zeta_2+6 S_1\left(\tfrac{\tau_0}{2}-2\right)-2S_1\left(\tfrac{\tau_0}{2}-2\right)^2-6 S_1\left(\tau_0-4\right)
\nonumber\\&\qquad\qquad\quad
+2S_1\left(\tfrac{\tau_0}{2}-2\right)S_1\left(\tau_0-4\right)+S_2\left(\tfrac{\tau_0}{2}-2\right)+4\delta_{\tau_0,4}
\Big) , \qquad \tau_0>2.
\end{align}
Recall that the one-loop structure constants $\langle a^{(1)}_{\tau_0,\ell}\rangle$ can be found using \eqref{eq:modified.structure.constant}.

\section{Half-BPS CFT-data}
\label{app:superconformal.results}
We present here an explicit form of the CFT-data for long supermultiplets present in the superconformal block decomposition of \eqref{eq:BPS.final}.

The anomalous dimensions are given by
\begin{align}
\overline{\langle \gamma_{2,\ell}\rangle}&=2S_1(\ell+2),
\\
\overline{\langle \gamma_{\tau_0,\ell}\rangle}&=-\frac{2\,\tilde c}{\mathcal P_{\tau_0,\ell}}\Big((\eta+1)S_1\left(\tfrac{\tau_0}2\right)+(\eta-1)S_1\left(\tfrac{\tau_0}2+\ell+1\right) \Big), \qquad \tau_0>2\,,
\end{align}
and the modified structure constants by
\begin{align}
\overline{\langle \hat\alpha_{2,\ell}\rangle}&=-\zeta_2,
\\
\overline{\langle \hat\alpha_{\tau_0,\ell}\rangle}&=-\frac{2\,\tilde c}{\mathcal P_{\tau_0,\ell}}\Big(\left[
(2\eta-1)S_1\left(\tfrac{\tau_0}2\right)+(1-\eta)S_1(\tau_0)
\right]S_1\left(\tfrac{\tau_0}2+\ell+1\right)+(1+\eta)S_1\left(\tfrac{\tau_0}{2}\right)^2
\nonumber\\
&\qquad\qquad
-(1+\eta)S_1\left(\tfrac{\tau_0}{2}\right)S_1(\tau_0)-\frac{1+\eta}{2}S_2\left(\tfrac{\tau_0}{2}\right)+\frac{1-\eta}{2}\zeta_2
\Big),\qquad \tau_0>2,
\end{align}
where $\mathcal P_{\tau_0,\ell}=\tilde c\,\eta+(\tau_0+\ell+2)(\ell+1)$ is the factor appearing in the higher twist structure constants \eqref{eq:structure.constants.Free}. The one-loop structure constants $\langle a^{(1)}_{\tau_0,\ell}\rangle$ can be found using the supersymmetric version of \eqref{eq:modified.structure.constant}.

\renewcommand{\bibname}{References}

\let\oldbibliography\thebibliography
\renewcommand{\thebibliography}[1]{%
  \oldbibliography{#1}%
\addcontentsline{toc}{chapter}{References}
  \setlength{\itemsep}{-0pt}%
}

\bibliography{biblthesis}
\bibliographystyle{JHEPthesis}

\end{document}